\pdfoutput=1

\documentclass[acmsmall,screen,natbib]{acmart}

\newcommand{\pdfver}{arxiv} 

\ifdefstring{\pdfver}{camred}{
  \setcopyright{rightsretained}
  \acmPrice{}
  \acmDOI{10.1145/3571205}
  \acmYear{2023}
  \copyrightyear{2023}
  \acmSubmissionID{popl23main-p85-p}
  \acmJournal{PACMPL}
  \acmVolume{7}
  \acmNumber{POPL}
  \acmArticle{12}
  \acmMonth{1}
  \received{2022-07-07}
  \received[accepted]{2022-11-07}
}{
  \acmJournal{PACMPL}
  \acmVolume{0}
  \acmNumber{0}
  \acmArticle{0}
  \acmYear{2023}
  \acmMonth{1}
  \acmDOI{10.1145/nnnnnnn.nnnnnnn}
  \startPage{1}
  \setcopyright{rightsretained}
}

\usepackage{booktabs}   
\usepackage{subcaption} 

\usepackage[all=normal,paragraphs,wordspacing]{savetrees}
\input{packages}

\usepackage{euscript}
\newcommand{\mathcalcommand}[1]{\mathcal{#1}}

\newcommand{\mcC}{\mathcalcommand{C}}

\DeclareMathAlphabet{\mathpzc}{T1}{pzc}{m}{it}

\allowdisplaybreaks

\newcommand*{\commentout}[1]{}
\newcommand{\rev}[1]{{#1}}

\definecolor{lred}{rgb}{1.0, 0.5, 0.5}
\definecolor{lorange}{rgb}{1.00, 0.90, 0.20}
\definecolor{lgreen}{rgb}{0.35, 0.95, 0.35}
\definecolor{lime}{rgb}{0.9, 1.0, 0.6}
\newcommand*{\hltred}[1]{{\setlength{\fboxsep}{2pt}\colorbox{lgreen}{#1}}}
\newcommand*{\hltorg}[1]{{\setlength{\fboxsep}{2pt}\colorbox{lorange}{#1}}}

\newcommand*{\hlmAux}[1]{{\setlength{\fboxsep}{2pt}\colorbox{lgreen}{$\displaystyle#1$}}}
\newcommand*{\hlm}[1]{
  \mathchoice
      {\hlmAux{\displaystyle{#1}}}
      {\hlmAux{\textstyle{#1}}}
      {\hlmAux{\scriptstyle{#1}}}
      {\hlmAux{\scriptscriptstyle{#1}}}%
}
\newcommand*{\hlnAux}[1]{{\setlength{\fboxsep}{2pt}\colorbox{lorange}{$\displaystyle#1$}}}
\newcommand*{\hln}[1]{
  \mathchoice
      {\hlnAux{\displaystyle{#1}}}
      {\hlnAux{\textstyle{#1}}}
      {\hlnAux{\scriptstyle{#1}}}
      {\hlnAux{\scriptscriptstyle{#1}}}%
}

\newcommand{\ind}[1]{\smash{\mathbf{1}_{[{#1}]}}}
\newcommand*{\dom}{\mathrm{dom}}

\newcommand*{\defeq}{\triangleq} 
\DeclareMathOperator*{\argmax}{argmax}
\newcommand{\lip}{\mathrm{Lip}}

\newcommand*{\B}{\mathbb{B}}
\newcommand*{\D}{\mathbb{D}}
\newcommand*{\E}{\mathbb{E}}
\newcommand*{\N}{\mathbb{N}}
\newcommand*{\R}{\mathbb{R}}

\newcommand*{\cD}{\mathcal{D}}

\newcommand*{\cL}{\mathcal{L}}
\newcommand*{\cN}{\mathcal{N}}

\newcommand*{\cP}{\mathcal{P}}
\newcommand*{\cT}{\mathcal{T}}

\newcommand*{\aD}{\smash{\cD^\sharp}}

\newcommand*{\db}[1]{\ensuremath{{\llbracket #1 \rrbracket}}}
\newcommand*{\edb}[1]{\ensuremath{\smash{{\llparenthesis #1 \rrparenthesis}^\sharp}}}
\newcommand*{\cdb}[1]{\ensuremath{\smash{{\llbracket #1 \rrbracket}^\sharp}}}

\newcommand*{\ctr}[2]{\ensuremath{{\overline{#1}^{#2}}}}

\newcommand*{\code}[1]{\mathtt{{#1}}}
\newcommand*{\cname}{\code{name}}

\newcommand*{\ctrue}{\code{true}}

\newcommand*{\cdist}[1]{\code{dist}_\code{#1}}
\newcommand*{\cpdf}[1]{\code{pdf}_\code{#1}}
\newcommand*{\cnor}{\cdist{N}} 
\newcommand*{\cpdfnor}{\cpdf{N}}
\newcommand*{\csample}{\code{sam}}
\newcommand*{\cobs}{\code{obs}}
\newcommand*{\cskip}{\code{skip}}
\newcommand*{\cif}{\code{if}}

\newcommand*{\celse}{\code{else}}
\newcommand*{\cwhile}{\code{while}}
\newcommand*{\cstr}[1]{\texttt{"}\mathtt{#1}\texttt{"}}

\newcommand*{\smath}[1]{\mathit{{#1}}}
\newcommand*{\scond}{\smath{cond}}
\newcommand*{\sseq}{\smath{seq}}
\newcommand*{\strue}{\smath{true}}
\newcommand*{\sfalse}{\smath{false}}

\newcommand*{\fv}{\smath{fv}}
\newcommand*{\fix}{\smath{fix}}

\newcommand*{\norm}[1]{{\lVert}{#1}{\rVert}_2}

\newcommand*{\sdom}[1]{\mathsf{{#1}}}
\newcommand*{\Fn}{\sdom{Fn}}
\newcommand*{\Name}{\sdom{Name}}
\newcommand*{\NameEx}{\sdom{NameEx}}

\newcommand*{\State}{\sdom{St}}
\newcommand*{\Str}{\sdom{Str}}
\newcommand*{\Var}{\sdom{Var}}
\newcommand*{\AVar}{\sdom{AVar}}

\newcommand*{\PVar}{\sdom{PVar}}

\newcommand*{\DistEx}{\sdom{DistEx}}
\newcommand*{\LamEx}{\sdom{LamEx}}

\newcommand*{\op}{\mathit{op}}
\newcommand*{\pr}{\mathit{pr}}

\newcommand*{\like}{\mathit{like}}

\newcommand*{\val}{\mathit{val}}
\newcommand*{\sampled}{\mathit{cnt}}

\newcommand*{\pfun}[2]{\ensuremath{\smash{ \ifstrempty{#2} {p_{#1}} {p_{#1}^{\langle#2\rangle}} }}}
\newcommand*{\vfun}[1]{\ensuremath{\smash{ v_{#1} }}}

\newcommand*{\repname}{\smath{rv}}

\newcommand*{\Statesub}{\State_{\square}}

\newcommand*{\getbase}[1]{\smath{#1}}
\newcommand*{\getpvar}{\getbase{pvars}}
\newcommand*{\getval }{\getbase{vals}}
\newcommand*{\getpr  }{\getbase{prs}}
\newcommand*{\getpvarsub}{\getpvar_\square}
\newcommand*{\getvalsub }{\getval _\square}
\newcommand*{\getprsub  }{\getpr  _\square}
\newcommand*{\getprsubi }[1]{\smash{\getpr_\square^{\langle #1 \rangle}}}

\newcommand*{\noerr}{\smath{noerr}}
\newcommand*{\used }{\smath{used}}
\newcommand*{\usedm}{\smath{used}_{-}}

\newcommand{\oo}{\raisebox{-0.1pt}{\scalebox{1.3}{$\circ$}}}
\newcommand{\xx}{\hltred{$\times$}}
\newcommand{\xxadm}{\hltorg{$\times$}}

\newcommand{\ours}{\mathit{ours}}

\newcommand{\mytt}[1]{{\tt\small #1}}

\begin{document}

\title[Smoothness Analysis and Selective Reparameterisation]
      {Smoothness Analysis for Probabilistic Programs with Application to Optimised Variational Inference}


\author{Wonyeol Lee}
\affiliation{
  \department{Computer Science}          
  \institution{Stanford University}      
  \country{USA}                          
}
\email{wonyeol@cs.stanford.edu}          

\author{Xavier Rival}
\affiliation{
  \institution{INRIA Paris, Département d'Informatique de l'ENS, and CNRS, PSL University, Paris}          
  \country{France}                   
}
\email{rival@di.ens.fr}         

\author{Hongseok Yang}
\affiliation{
  \department{School of Computing and Kim Jaechul Graduate School of AI}        
  \institution{KAIST} 
  \country{South Korea}                  
}
\email{hongseok.yang@kaist.ac.kr}          
\affiliation{
  \department{Discrete Mathematics Group}
  \institution{Institute for Basic Science (IBS)}
  \country{South Korea}                 
}

\begin{abstract}
We present a static analysis for discovering differentiable or more generally smooth parts of a given probabilistic program, and show how the analysis can be used to improve the pathwise gradient estimator, one of the most popular methods for posterior inference and model learning. Our improvement increases the scope of the estimator from differentiable models to non-differentiable ones without requiring manual intervention of the user; the improved estimator automatically identifies differentiable parts of a given probabilistic program using our static analysis, and applies the pathwise gradient estimator to the identified parts while using a more general but less efficient estimator, called score estimator, for the rest of the program. Our analysis has a surprisingly subtle soundness argument,
partly due to the misbehaviours of some target smoothness properties when viewed from the perspective of program analysis designers. For instance, some smoothness properties, such as partial differentiability and partial continuity, are not preserved by function composition, and this makes it difficult to analyse sequential composition soundly without heavily sacrificing precision. We formulate five assumptions on a target smoothness property, prove the soundness of our analysis under those assumptions, and show that our leading examples satisfy these assumptions.
We also show that by using information from our analysis instantiated for differentiability,
\rev{our improved gradient estimator satisfies an important differentiability requirement and thus computes the correct estimate on average (i.e., returns an unbiased estimate) under a regularity condition.}
Our experiments with representative probabilistic programs in the Pyro language show that our static analysis is capable of identifying smooth parts of those programs accurately, and making our improved pathwise gradient estimator exploit all the opportunities for high performance in those programs.
\end{abstract}

\ccsdesc[500]{Software and its engineering~Correctness}
\ccsdesc[500]{Software and its engineering~Automated static analysis}
\ccsdesc[500]{Mathematics of computing~Bayesian computation}
\ccsdesc[500]{Mathematics of computing~Variational methods}

\keywords{{smoothness}, static analysis, probabilistic programming, variational inference}  

\maketitle


\section{Introduction}
\label{sec:intro}

Probabilistic programs define models from machine learning and statistics, and are used to analyse datasets from a wide range of applications~\cite{Stan,InferNet,Tabular,goodman_uai_2008,Mansinghka-venture14,wood-aistats-2014,TolpinMYW16,Hakaru,psi:cav:16,TranKDRLB16,SalvatierWF16,siddharth2017learning,pmlr-v84-ge18b,TranHMSVR18,BinghamCJOPKSSH19}. These programs are written in languages with special runtimes, called inference engines, which can be used to answer probabilistic queries, such as posterior inference and marginal likelihood estimation, or to learn model parameters in those programs, such as weights of neural networks. Whether a probabilistic program is useful for, for instance, discovering a hidden pattern in a given dataset or making an accurate prediction largely lies in these inference engines. These engines should compute accurate approximations or good model parameters within a fixed time budget. It is, thus, not surprising that substantial research efforts have been made to develop efficient inference algorithms and their implementations (as inference engines)~\cite{WingateGSS11,WingateSG11,WingateBBVI13,KucukelbirRGB15,NoriHRS14,ChagantyNR13,Mansinghka-venture14,SchulmanHWA15,RitchieHG16,ZhouYTR20,HoltzenBM20,van2018introduction}.

We are concerned with smoothness\footnote{%
In mathematics, ``smoothness'' typically refers to the {\em specific} property of functions: being infinitely differentiable.
In this paper, we override the term to denote a {\em set} of properties of functions describing well-behavedness (e.g., differentiability).%
} properties of probabilistic programs,
which have been exploited by performant posterior-inference and model-learning algorithms and engines. For instance, when probabilistic programs are differentiable (in the sense that they define differentiable unnormalised densities), their posteriors can be inferred by Hamiltonian Monte Carlo~\cite{Neal-HMC}, one of the best performing MCMC algorithms. Also, in that case, their posteriors and model parameters can be inferred or learnt using the pathwise gradient estimator~\cite{KingmaICLR14,RezendeICML14},  a popular technique for estimating the gradient of a function using samples. We also point out that the need for smoothness arises in a broader context of machine learning and computer science; Lipschitz continuity is one of the desired or at least recommended properties for neural networks \cite{pmlr-v70-arjovsky17a,KimPM21}, and also differentiability commonly features as a requirement for pieces of code inside simulation software and cyber physical systems, where differential equations are used to specify the environment~\cite{Platzer18}.

We present a static analysis that enables optimised posterior inference and model learning for probabilistic programs. We develop a static analysis that discovers differentiable or more generally smooth parts of given probabilistic programs, and show how the analysis can be used to improve the pathwise gradient estimator. Our improvement increases the scope of the estimator from differentiable to non-differentiable models, without requiring any intervention from the user; the improved estimator automatically identifies differentiable parts of probabilistic programs using our static analysis, and applies the pathwise gradient estimator to the identified parts while using a more general but less efficient estimator, called score estimator~\cite{WilliamsMLJ1992,RanganathGB14}, for the rest of the programs.

Our static analysis for smoothness has a surprisingly subtle soundness argument, partly due to the misbehaviours of some target smoothness properties when viewed from the perspective of program analysis designers. For instance, some smoothness properties, such as partial differentiability and partial continuity, are not preserved by function composition, and this makes it difficult to analyse sequential composition soundly without heavily sacrificing precision. In fact, overlooking such misbehaviours has been a source of errors in published static analyses for continuity~\cite{cgl:popl:10,cgl:jacm:12}.\footnote{{%
  The analysis in \cite{cgl:popl:10} infers the continuity property for multivariate programs, but it incorrectly joins two input-variable sets if a program is continuous with respect to each set {\em jointly}. Such a rule would hold if 
  separate per-input-variable continuity were considered, but it does not hold for multivariate joint continuity. Conversely, the analysis in \cite{cgl:jacm:12} considers a {\em per-input-variable} definition of continuity, but incorrectly assumes that this per-input-variable continuity is preserved by function composition.
  \rev{These (and other) issues make the two analyses unsound in several aspects (see \cref{sec:intro-more} for details).}
  We do not claim that these unsoundness issues are hard to fix. Instead, our point is that a similar issue may be introduced easily and remain undetected due to the subtlety in the soundness of a static analysis.%
}} We formulate five assumptions that clearly identify what a smoothness property should satisfy in order to avoid unsound analysis. Interestingly, these assumptions also determine what the property is allowed to violate.
For instance, they reveal that the smoothness property does not have to be closed under the limits of chains of smooth (partial) functions, although the closure under such limits, called admissibility, has often been used to justify proof rules about or static analysis of loops. Dispensing with the admissibility requirement broadens the scope of our program analysis non-trivially; some useful smoothness properties from mathematics fail to meet the requirement.

Our variant of the pathwise gradient estimator works by program transformation and non-standard execution.
It first transforms given probabilistic programs based on the results of our static analysis.
\rev{Then, our estimator executes the transformed programs according to a standard (sampling) semantics, and collects sampled values during the execution. Finally, using the collected values, the estimator
executes the original programs according to a non-standard (density) semantics this time. During the execution,}
our estimator computes a quantity involving differentiation, which becomes the estimate of the target gradient.
\rev{We prove that our estimator satisfies an important differentiability requirement and thus, under a regularity condition, it is correct}:
the computed estimate is unbiased, i.e., it is the target gradient when averaged over random choices made during execution. 

Our static analysis and variant of the pathwise gradient estimator have been implemented for a subset of the Pyro probabilistic programming language~\cite{BinghamCJOPKSSH19}. They have been successfully applied to the 13 representative Pyro examples, which include advanced models with deep neural networks, such as attend-infer-repeat~\cite{EslamiNIPS16} and single-cell annotation using variational inference~\cite{xu2021probabilistic}.
\rev{For each of these examples, Pyro provides a (default) selective use of the pathwise gradient estimator but without any correctness guarantee.
  Our analysis and improved estimator automatically reproduced those uses,
  and proved that in those use cases, the estimator satisfies an important differentiability requirement and it is, thus, correct (i.e., unbiased) under a regularity condition (which needs to be discharged separately).}

\rev{We summarise the main contributions of the paper:
\begin{itemize}
\item
  We present a program analysis for smoothness properties such as differentiability, and explain a subtle soundness argument for the analysis. Our argument identifies five assumptions for target smoothness properties, which are violated by some well-known smoothness properties and can help detect and prevent soundness errors in static analyses for smoothness properties~(\cref{s:sa:new}).  
\item
  We present a gradient estimator for probabilistic programs that improves the well-known pathwise gradient estimator using our program analysis.
  We also prove that our estimator satisfies an important differentiability requirement  and it is, thus, correct (namely unbiased) under a regularity condition (\cref{sec:generalised-gradient-estimator-new} and \cref{sec:var-sel-alg}).
\item We show that our program analysis and gradient estimator can be successfully applied to representative probabilistic programs in Pyro,
  and can prove that existing unproved optimisations for these programs satisfy the differentiability requirement (\cref{sec:impl-eval}).
\end{itemize}
The appendix (i.e., \cref{sec:intro-more}--\cref{sec:impl-eval-more}) includes omitted proofs and details, and can be found in \cite{LeeRY22}.}

\section{Informal Description of Basic Concepts and Our Approach}
\label{sec:overview}

We start by describing informally basic concepts and the goal of our approach, which we hope helps the reader to see the big picture of our technical contributions. To  simplify presentation, we use toy examples in the section. But we emphasise that our approach has been applied to representative Pyro programs that describe advanced machine-learning models with deep neural networks.


\vspace{2mm}
\noindent{\bf Probabilistic programming and variational inference.}\
In a probabilistic programming language (PPL), a program expresses a probabilistic model.
As an example, consider the program $c_m$ in \cref{fig:eg-model-guide}, which
describes a probabilistic model of the random variables $z_1$ and $z_2$ in $\R$
by specifying their \emph{unnormalised} density
\begin{align*}
 & p_{c_m}(z_1,z_2) = \cN(z_1; 0, 5) \cdot \cN(z_2; z_1, 3) \cdot (\ind{z_2 > 0} \cdot \cN(0; 1,1) + \ind{z_2 \leq 0} \cdot \cN(0; -2,1)),
\end{align*}
where $\cN(x; a,b)$ is the probability density of a normal distribution with mean $a$ and variance $b$,
and $\ind{\varphi}$ is the indicator function that returns $1$ if $\varphi$ holds and $0$ otherwise.
The first two $\cN$ factors in the equation come from the sample commands ($\csample$) in $c_m$. They are called prior distributions, and describe prior knowledge on two random variables named $z_1$ and $z_2$. The last factor comes from the if and observe commands ($\cif$ and $\cobs$), which express that an unnamed random variable is sampled and \rev{observed to be $0$} and its distribution is $\cnor(1,1)$ or $\cnor(-2,1)$ depending on whether $z_2$ is positive or not. This factor is called likelihood, and it states information about $z_1$ and $z_2$ that comes from an observed data point $0$. Ignore the third arguments of the sample commands of $c_m$ for now, which have no effect on $p_{c_m}$; they will be explained later.

\begin{figure*}[t]
  \begin{subfigure}[c]{\textwidth}
    \centering
    \small
    \begin{align*}
      c_m & =  \left(\,
      \begin{alignedat}{3}
        & \rlap{$x_1 := \csample(\cstr{z_1}, \cnor(0  , 5), \lambda y.y);$}
        \\[-0.5ex]
        & \rlap{$x_2 := \csample(\cstr{z_2}, \cnor(x_1, 3), \lambda y.y);$}
        \\[-0.5ex]
        & \cif \, (x_2 > 0) \, && \{ \cobs(\cnor(1,1), 0) \}
        \\[-0.5ex]
        & \celse\, && \{ \cobs(\cnor(-2,1), 0) \} 
      \end{alignedat}
      \,\,\right),
      &
      c_g & = \left(\,
      \begin{alignedat}{2}
        & x_1 := \csample(\cstr{z_1}, \cnor(\theta_1, 1), \lambda y.y);
        \\[-0.5ex]
        & x_2 := \csample(\cstr{z_2}, \cnor(\theta_2, 1), \lambda y.y)
        \end{alignedat}
      \,\right)\!.
    \end{align*}
  \end{subfigure}
  \vspace{-1.0em}
  \caption{A model $c_m$ and a guide $c_g$ in a PPL. Here $\cnor(a,b)$ is the distribution expression, and denotes the normal distribution with mean $a$ and variance $b$. }
  \label{fig:eg-model-guide}
\end{figure*}

\begin{figure*}
  \vspace{-0.3em}
  \begin{subfigure}[c]{\textwidth}
    \centering
    \small
    \begin{align*}
      c_g' &= \left(\,
      \begin{alignedat}{2}
        & x_1 := \csample(\cstr{z_1}, \hlm{\cnor(0, 1)}, \hlm{\lambda y.y + \theta_1});
        \\[-1.1ex]
        & x_2 := \csample(\cstr{z_2}, \hlm{\cnor(0, 1)}, \hlm{\lambda y.y + \theta_2})
        \end{alignedat}
      \,\right)\!,
      &
      c_g'' & = \left(\,
      \begin{alignedat}{2}
        & x_1 := \csample(\cstr{z_1}, \hlm{\cnor(0, 1)}, \hlm{\lambda y.y + \theta_1});
        \\[-0.5ex]
        & x_2 := \csample(\cstr{z_2}, \cnor(\theta_2, 1), \lambda y.y)
        \end{alignedat}
      \,\right)\!.
    \end{align*}
  \end{subfigure}
  \vspace{-1.0em}
  \caption{A fully (or selectively) reparameterised guide $c_g'$ (or $c_g''$).}
  \vspace{-0.7em}
  \label{fig:eg-repar}
\end{figure*}

The purpose of writing $c_m$ in a PPL, called \emph{model},
is to infer its \emph{normalised} probability density 
\begin{align*}
  \overline{p}_{c_m}(z_1,z_2) \defeq {p_{c_m}(z_1,z_2)} / {\smash{\medint\int} p_{c_m}(z_1,z_2)\, dz_1dz_2},
\end{align*}
also called normalised posterior density.
Intuitively, this normalised density
brings together two types of information about $z_1$ and $z_2$, the first from their prior distributions (expressed in the first and second lines of $c_m$), and the second from the observed data point $0$ that depends on $z_1$ and $z_2$ (the third and fourth lines of $c_m$).
This inference task is called \emph{posterior inference} problem.
Among a wide range of approaches to the problem, we focus on the approach called \emph{variational inference}, which forms the core of the recent combination of PPLs and deep learning.

In variational inference, we posit another program $c_g$, called \emph{guide},  that is simpler than $c_m$ and parameterised by $\theta$. Then, we approximate the normalised density of $c_m$ by $c_g$ with an optimal choice of $\theta$.
For instance, consider the program $c_g$ in \cref{fig:eg-model-guide}.
The program specifies the following already-normalised probability density
\begin{align*}
 p_{c_g, \theta}(z_1,z_2) &= \cN(z_1; \theta_1, 1) \cdot \cN(z_2; \theta_2, 1).
\end{align*}
It can serve as a guide program for $\smash{ \overline{p}_{c_m} }$.
To best approximate $\smash{ \overline{p}_{c_m} }$ by $p_{c_g, \theta}$,
variational inference aims at finding $\theta$ that minimises some notion of the discrepancy (called KL divergence) between $p_{c_g,\theta}$ and $\smash{ \overline{p}_{c_m} }$,
or equivalently that maximises the objective function $\cL$ (called evidence lower bound):
\begin{align*}
  & \arg\max_\theta \cL(\theta)
  & \text{for}
  \
  \cL(\theta) \defeq \E_{p_{c_g, \theta}(z_1,z_2)} \left[  f_\theta(z_1,z_2) \right]
  \ \text{with}\ 
  f_\theta(z_1,z_2) \defeq \log (p_{c_m}(z_1,z_2) / p_{c_g,\theta}(z_1,z_2)).
\end{align*}
A standard way to solve this optimisation problem is to apply the gradient-ascent algorithm:
starting from an initial value $\smash{ \theta^{(0)} }$ of $\theta$, compute $\smash{ \theta^{(t)} }$ iteratively by
$\theta^{(t+1)} \defeq \theta^{(t)} + \eta \cdot \nabla_\theta \cL(\theta^{(t)})$,
and return $\smash{ \theta^{(T)} }$ for a sufficiently large $T \in \N$. Here $\eta \in \R_{>0}$ denotes a learning rate.

A challenging part in the 
algorithm is
to compute $\nabla_\theta \cL(\theta)$. 
An exact computation of the gradient is mostly intractable due to the expectation inside $\cL$,
which hinders the gradient from having a closed-form formula.
Hence, in practice, we rather \emph{estimate} (not exactly compute) the gradient via a Monte Carlo method:
draw a random sample $(\hat{z}_1,\hat{z}_2)$ from some distribution $q_\theta$, apply some function $g_\theta$ to the sample,
and use the result as an estimate to the gradient, i.e.,
\begin{equation}
  \label{eq:grad-expt-estm}
  g_\theta(\hat{z}_1,\hat{z}_2) \approx \nabla_\theta \cL(\theta)
  \quad\text{for a sample $(\hat{z}_1,\hat{z}_2)$ drawn from $q_\theta(\hat{z}_1,\hat{z}_2)$}.
\end{equation}
An important desired property of such a gradient estimator is \emph{unbiasedness}, which states that the estimate is accurate in expectation:
$\E_{q_\theta(z_1,z_2)}[g_\theta(z_1,z_2)] = \nabla_\theta \cL(\theta)$.
\rev{The property is necessary for the algorithm to converge to a local optimum, and is, thus, desired.}

\vspace{2mm}
\noindent{\bf Gradient estimators for variational inference: SCE, PGE, and SPGE.}\ 
A standard estimator for 
$\nabla_\theta \cL(\theta)$ is the \emph{score estimator} (SCE)~\cite{WilliamsMLJ1992,RanganathGB14}, which is unbiased under mild conditions.
It estimates $\nabla_\theta \cL(\theta)$ by using the recipe in \cref{eq:grad-expt-estm}
with 
$q_\theta(z_1,z_2) = p_{c_g,\theta}(z_1,z_2)$
and
\begin{align*}
 g_\theta(z_1,z_2) &= f_\theta(z_1,z_2) \cdot \nabla_\theta \log q_\theta(z_1,z_2). 
\end{align*}
That is, the estimator draws a 
sample from the guide distribution $p_{c_g, \theta}$
and applies the above $g_\theta$ to obtain a gradient estimate.%
\footnote{\rev{The log term in $g_\theta$ comes from the well-known log-derivative trick:
  $\nabla_\theta q_\theta(z_1,z_2) = q_\theta(z_1,z_2) \cdot \nabla_\theta \log q_\theta(z_1,z_2)$.}}
It is applicable to a wide range of model-guide pairs while remaining unbiased,
but it is known to have a large approximation error (i.e., have a large variance).

The \emph{pathwise gradient estimator} (PGE)~\cite{KingmaICLR14,RezendeICML14} is another standard gradient estimator,
which is known to have a smaller approximation error than the SCE
and thus has been a preferred option against the SCE.
The PGE requires an additional program $c_g'$ that is a $\theta$-independent reparameterisation of the guide $c_g$.
A program $c'$ is said to be \emph{$\theta$-independent} if the probability densities of the sampled random variables in $c'$ are $\theta$-independent.
It is called a \emph{reparameterisation} of $c$ if $c$ and $c'$ sample the same set of random variables
and they have the same semantics on those variables in the following sense: when there are $n$ random variables, 
for any measurable $h : \R^n \to \R$,
we have $\E_{p_c(z)} [ h(v_c(z)) ] = \E_{p_{c'}(z)} [ h(v_{c'}(z)) ]$,
where $p_c : \R^n \to \R$ is the probability density of all $n$ random variables in $c$,
and $v_c : \R^n \to \R^n$ is the so called value function of $c$,
which applies the lambda functions in the third arguments of $c$'s sample commands to the corresponding random variables.
For example, $c_g'$ in \cref{fig:eg-repar} is a $\theta$-independent reparameterisation of $c_g$ for $\theta = (\theta_1,\theta_2)$. 
It has the probability 
density $p_{c_g'}(z_1,z_2) = \cN(z_1; 0, 1) \cdot \cN(z_2; 0, 1)$, and the value function
 $v_{c_g', \theta}(z_1,z_2) = (z_1 + \theta_1, z_2 + \theta_2)$.
Note that $p_{c_g'}$ does not depend on $\theta$, as required by the $\theta$-independence of $c_g$.
We can show this $c_g'$ is a reparameterisation of $c_g$ in \cref{fig:eg-model-guide}
by using the fact that $v_{c_g}$ is the identity function and $y=x + a$ for $x$ drawn from $\cN(x; 0,1)$ 
follows the distribution $\cN(y; a,1)$.

Given a reparameterised guide $c_g'$, the PGE estimates $\nabla_\theta \cL(\theta)$
by again following the recipe in \cref{eq:grad-expt-estm} this time with 
$q'(z_1,z_2) = p_{c_g'}(z_1,z_2)$ and
\begin{align*}
  g'_\theta(z_1,z_2) &= \nabla_\theta f_\theta(z'_1,z'_2) \quad\text{for}\ (z'_1,z'_2) =v_{c_g',\theta}(z_1,z_2).
\end{align*}
This estimator differs from the SCE in two aspects. First,
a random sample is drawn from a reparam\-eterised-guide distribution $p_{c_g'}$, not from $p_{c_g, \theta}$. Next,
the estimation function $g'_\theta$ computes the derivative of $f_\theta(z'_1,z'_2)$ with respect to $\theta$ and $(z'_1,z'_2)$
(not with respect to $\theta$ only), since the argument of $f_\theta(-)$ in $g'_\theta$ depends on $\theta$ via $v_{c_g',\theta}$.\footnote{%
By the chain rule, \rev{we have the following for each $i \in \{1,2\}$:}
\belowdisplayskip=0pt
\belowdisplayshortskip=0pt
\begin{multline*}
  \hspace{-2em}
  \smash{
\frac{\partial f_\theta(z'_1,z'_2)}{\partial \theta_i}
\hspace{-0.0ex}{{}={}}\hspace{-0.6ex}
\left.
\frac{\partial f_\theta(x_1,x_2)}{\partial \theta_i}\right|{}_{\begin{subarray}{l}(x_1,x_2,\theta)\\ = (z'_1,z'_2,\theta)\end{subarray}} 
\hspace{-0.4ex}{{}+{}}\hspace{-0.1ex}
\left\langle
\left(
\left.\frac{\partial f_\theta(x_1,x_2)}{\partial x_1}\right|{}_{\begin{subarray}{l}(x_1,x_2,\theta)\\ = (z'_1,z'_2,\theta)\end{subarray}}
\hspace{-0.3ex},
\left.\frac{\partial f_\theta(x_1,x_2)}{\partial x_2}\right|{}_{\begin{subarray}{l}(x_1,x_2,\theta) \\= (z'_1,z'_2,\theta)\end{subarray}}\right)
\hspace{-0.3ex},
\left(\left.
\frac{\partial v_{c_g',\theta}(y_1,y_2)}{\partial \theta_i}\right|{}_{\begin{subarray}{l}(y_1,y_2,\theta) \\= (z_1,z_2,\theta)\end{subarray}}\right)\right\rangle
\hspace{-0.3ex}.
  }
  \hspace{-2em}
\end{multline*}%
}

While having a small approximation error, the PGE requires more than the SCE to ensure the unbiasedness.
An important additional requirement for the PGE
is that (i) $p_{c_m}(z_1,z_2)$ and $p_{c_g,\theta}(z_1,z_2)$ should be differentiable in $\theta$ and $z_1,z_2$
and (ii) $v_{c_g', \theta}(z_1,z_2)$ be differentiable in $\theta$ for all $z_1,z_2$.
The requirement is imposed
partly to ensure that no differentiation error arises in computing $g'_\theta$.
This differentiability requirement, however, can be easily violated if a model or a guide starts to use branches or loops.
For instance, it is violated by our example in \cref{fig:eg-model-guide,fig:eg-repar}
as $p_{c_m}(z_1,z_2)$ is not differentiable in $z_2$. 
This violation makes the PGE biased for the example, i.e.,
\begin{align*}
  \E_{q'_\theta(z_1,z_2)}[g'_\theta(z_1,z_2)]
  &= (\cdots, {\textstyle \frac{1}{3}}(\theta_1 -\theta_2))
  \neq (\cdots, {\textstyle \frac{1}{3}}(\theta_1 -\theta_2) + {\textstyle \frac{3}{2}} \cN(-\theta_2; 0,1))
  = \nabla_\theta \cL(\theta),
\end{align*}
and thus causes the gradient-ascent algorithm to converge to a suboptimal $\theta$:
applying the PGE to the example produces a suboptimal solution $\theta = (0,0)$,
whereas the optimal solution is $\theta \approx (0.95, 1.52)$.

The \emph{selective pathwise gradient estimator} (SPGE)~\cite{SchulmanHWA15} combines the two previous gradient estimators to alleviate their limitations:
one has a large approximation error, and the other imposes a strong requirement for unbiasedness. The SPGE requires an additional program $c_g''$ that is a reparameterisation of the guide $c_g$
but needs not be $\theta$-independent (unlike the PGE).
An instance of $c_g''$ for our example is given in \cref{fig:eg-repar}, which changes the sample command for $z_1$ in $c_g$ but keeps the one for $z_2$.
Note that the changed sample command for $z_1$ in $c_g''$ uses a $\theta$-independent probability distribution.
Typically, $c_g''$ is obtained by \emph{selecting} a subset of the random variables in $c_g$
and changing the sample commands for the selected variables such that their probability distributions become $\theta$-independent;
the sample commands for the unselected remain as they are.
Given $c_g''$, the SPGE estimates $\nabla_\theta \cL(\theta)$
by following the recipe in \cref{eq:grad-expt-estm} with 
$q''_\theta(z_1,z_2) = p_{c_g'', \theta}(z_1,z_2)$ and
\begin{equation}
  \label{eq:spge-toy}
  g''_\theta(z_1,z_2) = \nabla_\theta f_\theta(z''_1,z''_2) + f_\theta(z''_1,z''_2) \cdot \nabla_\theta \log q''_\theta(z_1,z_2)
  \quad \text{for}\ (z''_1,z''_2) = v_{c_g'',\theta}(z_1,z_2). 
\end{equation}
Note that the estimation function $g''_\theta$ consists of two terms, which come from that of the PGE and the SCE. 
The second term adjusts the PGE to correctly account for unchanged random variables
\rev{(e.g., $z_2$ in the example of \cref{fig:eg-repar})}.

By allowing a guide that makes only some selected (not all) random variables $\theta$-independent, the SPGE offers two advantages at the same time:
it achieves a smaller approximation error than the SCE, and imposes a weaker requirement for unbiasedness than the PGE. 
In particular, the differentiability requirement for the SPGE is weaker than that for the PGE, which is as follows for our example in \cref{fig:eg-model-guide,fig:eg-repar}:
(i) $p_{c_m}(z_1,z_2)$ and $p_{c_g, \theta}(z_1,z_2)$ be differentiable in $\theta$
and $z_1$ (but they may be non-differentiable in $z_2$);
and (ii) $v_{c_g'', \theta}(z_1,z_2)$ and $p_{c_g'',\theta}(z_1,z_2)$ 
be differentiable in $\theta$ for all $z_1,z_2$.
This requirement holds, and as a result, the SPGE with this $c_g''$ is unbiased
(whereas the PGE with the given $c_g'$ is biased as seen before).
\rev{You can find in \cref{sec:overview-more} a table summarising and comparing SCE, PGE, and SPGE.}


\vspace{2mm}
\noindent{\bf Variable-selection problem for SPGE.}\ 
To maximize the advantages offered by the SPGE, we consider the following algorithmic problem:
\begin{definition}[SPGE Variable-Selection Problem; Informal]
  \label{def:var-sel-prob}
  Assume that we are given a model $c_m$, a guide $c_g$, and a \emph{reparameterisation plan} $\pi$, i.e., a map from sample commands to sample commands.
  Then, find automatically a large subset $S$ of random variables such that if we let $\smash{\ctr{c_g}{\pi,S}}$ be the result of $\pi$-transforming every sample command in $c_g$ that defines a random
  variable in $S$, then $\smash{ \ctr{c_{g}}{\pi,S} }$ is a reparameterisation of $c_g$
  and $\smash{ (c_m, c_g, \ctr{c_g}{\pi,S}) }$ satisfies the differentiability requirement for the SPGE.
  \qed
\end{definition}

An instantiation of the problem for our example is that $c_m$ and $c_g$ are programs in \cref{fig:eg-model-guide} and $\pi$ transforms commands of the form
$y \,{:=}\, \csample(n, \cnor(e',1), \lambda y.y)$ to 
$y \,{:=}\, \csample(n, \cnor(0,1),\lambda y.y+e')$,
while leaving all the other sample commands as they are.  In this instantiation, the condition in the problem is met by $S = \emptyset$ and $S = \{z_1\}$, and the latter option is preferred due to its size.
Note that the solution $S=\{z_1\}$ yields the guide $c_g''$ in \cref{fig:eg-repar}, that is, $\smash{ \ctr{c_g}{\pi,S} = c_g'' }$.

Existing PPLs, when applying the SPGE, choose an $S$ without checking the differentiability requirement,
and this can make the requirement easily violated.
For instance, given a model-guide pair, in one of its standard settings, 
Pyro automatically applies the SPGE with $S$ being the set of all continuous random variables in the guide.
This choice of $S$, however, does not guarantee the requirement is met.
For our example in \cref{fig:eg-model-guide}, Pyro chooses $S=\{z_1,z_2\}$, but this $S$ violates the requirement;
due to this, the SPGE becomes biased and Pyro returns a suboptimal $\theta=(0,0)$.

In the rest of the paper, we will present our solution to the SPGE variable-selection problem. A core component of our solution is a general static analysis framework for smoothness properties  such as differentiability \rev{(\cref{s:sa:new})}, which our solution uses to discharge the differentiability requirement for the SPGE correctly and automatically. As we briefly mentioned in the introduction, automatically analysing the smoothness properties of a program in a sound manner is surprisingly subtle. Our analysis framework identifies five assumptions for smoothness properties, and prove that the analysis is sound if a target smoothness property satisfies these assumptions. 

Our solution for the SPGE variable-selection problem~\rev{(\cref{sec:var-sel-alg})} runs the static analysis on given $c_m$ and $c_g$, and computes a maximal set $S'$ of random variables in which $p_{c_m}$ and $p_{c_g,\theta}$ are differentiable. Then, it heuristically searches for a subset of $S'$ starting from $S'$ itself such that $\smash{ \ctr{c_g}{\pi,S'} }$ satisfies the differentiability requirement.  For instance, for our example in \cref{fig:eg-model-guide}, our differentiability analysis infers that
$p_{c_m}$ and $p_{c_g,\theta}$ are differentiable in $\{z_1\}$ and $\{z_1,z_2\}$, respectively. 
From this, we set $S' = \{z_1\}$, run our analysis again on $\smash{ \ctr{c_g}{\pi,S'} }$, and get confirmation that $\smash{ \ctr{c_g}{\pi,S'}} $ meets the requirement, i.e., $\smash{ p_{\ctr{c_g}{\pi,S'},\theta} }$ and $\smash{ v_{\ctr{c_g}{\pi,S'},\theta} }$ are differentiable in $\theta$. Thus, this $S'$ becomes the final result. In fact, this first-round success appeared in our experiments~\rev{(\cref{sec:impl-eval})}: our implementation shows on all tested examples that the initial choice of $S'$ is indeed valid in the above sense so that no subset search is necessary.

We point out that to mathematically develop and analyse our solution for the SPGE variable-selection problem, 
we formalise the SPGE in the PPL setting and formally derive a sufficient condition for its unbiasedness, which includes the differentiability requirement~\rev{(\cref{sec:generalised-gradient-estimator-new})}.


\section{Setup}
\label{sec:setup}
We use a simple imperative probabilistic programming language, which models the core of popular imperative PPLs, such as Pyro. Programs in the language describe densities, which are sometimes unnormalised (i.e., they do not integrate to $1$). In this section, we describe the syntax and semantics of the language, and also variational inference for the language.

\vspace{2mm}
\noindent{\bf Syntax of a simple imperative PPL.}\
Let $\PVar$ be a finite set of program variables, $\Str$ be a finite set of strings, and $\Fn$ be a set of function symbols that represent measurable maps of type $\R^k \to \R$. The language has the following syntax:
\begin{align*}
\hspace{-0.5em}
\text{Real Expr.}\
e & {} ::= x \mid r \mid \op(e_1,\ldots,e_k)
&
\text{Boolean Expr.}\
b & {} ::= \ctrue \mid e_1 < e_2 \mid b_1 \wedge b_2 \mid \neg b \hspace{4em}
\\[-0.5ex]
\hspace{-0.5em}
\text{Name Expr.}\
n & {} ::= \cname(\alpha,e)
&
\text{Distribution Expr.}\
d & {} ::= \cnor(e, e')
\\[-0.5ex]
\hspace{-0.5em}
\text{Command}\
c & \rlap{${} ::= \cskip \mid x \,{:=}\, e \mid c;c'
\mid \cif\,b\,\{c\}\,\celse\,\{c'\} \mid \cwhile\,b\,\{c\} \mid
x\,{:=}\,\csample(n, d, \lambda y.e)
\mid \cobs(d, r)$}
\end{align*}
Here $x$, $r$, $\op$, and $\alpha$ stand for a program variable in $\PVar$, a real number, a function symbol in $\Fn$, and a string in $\Str$, respectively. 

The language has four kinds of expressions, which denote maps from states to values of appropriate types. All the real and boolean expressions are standard. The name expressions $n$ denote the identifiers of drawn samples (i.e., random variables).
They are built by appending an integer (obtained by the floor of a real) to a string in $\Str$;
e.g., $\cname(\cstr{z},3.2)$ denotes the name $(\cstr{z},3)$.%
\footnote{\rev{The $\cname$ construct has the second argument to easily support the sampling of (conditionally) {i.i.d.} random variables.}}
The distribution expression $\cnor(e,e')$ denotes the normal distribution with mean $e$ and variance $e'$.
The language supports standard commands for imperative computation, and additionally has sample and observe for probabilistic programming.
The sample command $x:=\csample(n,d,\lambda y.e)$ creates a random variable
named $n$ by drawing a sample $r$ from $d$; then, it transforms $r$ to $e[r/y]$
and stores the result in the program variable $x$.
In the programs written by the user of 
the language, only the identity function $\lambda y.y$ appears as the third argument of the sample commands. But as we explain later, when a program is constructed from another by a gradient estimator, such as the SPGE, it may contain sample commands with non-identity function arguments. The observe command $\cobs(d,r)$ describes that an unnamed random variable is drawn from $d$ and is immediately observed to have the value~$r$. Computationally, $\cobs(d,r)$ calculates the probability density of $d$ at $r$ and updates a variable that tracks the product of these probabilities from all the observations, by multiplying the variable with the calculated density.


\vspace{2mm}
\noindent{\bf Density semantics of the PPL.\footnotemark}
We use a semantics of our language where commands are interpreted as calculators for densities, which may be unnormalised. Commands transform states, but in so doing, they compute densities of sampled random variables. More precisely, in the semantics, a command starts with an initial state that fixes not just the values of program variables but also those of all the random variables that are to be sampled during execution. When the command runs, it calculates the densities of those random variables at their given initial values, and also computes the probability density of all the observations, called \emph{likelihood}. The product of all these densities and the likelihood becomes the so called \emph{unnormalised posterior density}.
\footnotetext{\rev{%
  Our semantics is an instance (or variant) of existing density semantics (e.g., \cite{LeeYRY20}),
  and is different from sampling semantics (e.g., \cite{StatonYWHK16}).
  Although the density semantics and the sampling semantics have different presentations,
  they are closely related and equivalent in a formal sense (see, e.g., \cite{LeeYRY20}).
  We use the density semantics instead of the sampling semantics,
  because the gradient estimator (\cref{sec:generalised-gradient-estimator-new}) of our interest performs computation on (unnormalised) densities
  and it is easier for a program analysis (\cref{s:sa:new}) to work with the density semantics than the sampling semantics.}}

Let $\N$ be the set of natural numbers. Fix $N \in \N$ with $N \geq 1$.
Formally, the semantics uses the states of the following form:
\begin{align*}
\mu \in \Name & \defeq \{(\alpha, i) \mid \alpha \in \Str, i \in \N \cap [0, N)\},
\\[-0.5ex]
a \in \AVar &\defeq \{\like\} \cup \{\pr_\mu, \val_\mu, \sampled_\mu \mid \mu \in \Name \},
\\[-0.5ex]
u,v \in \Var & \defeq \Name \uplus \PVar \uplus \AVar,
\\[-0.5ex]
\sigma \in \State & \defeq [\Var \to \R],
\quad
\State[K] \defeq [K \to \R] \ \text{for}\ K \subseteq \Var.
\end{align*}
Here $\sigma(\mu)$ for $\mu \in \Name$ is the initial value of the random variable $\mu$, which is used by the sample command and does not change during execution.
For technical simplicity, the set $\Name$ has the restriction that the integer part of a name must be in $[0,N)$.\footnote{%
  \rev{This restriction is often respected by probabilistic programs in practice,
    since they commonly sample random variables whose number is uniformly bounded over all traces;
    note, however, that it is not always respected (e.g., by programs from Bayesian nonparametrics).}
  The uniform bound $N$ can often be found by a simple static analysis.
  This restriction along with the finiteness of $\PVar$ and $\Str$ implies the finiteness of $\Var$,
  and this makes our technical development simpler since $\sigma \in \State$ becomes a function on a finite-dimensional space.
  \rev{Lifting the restriction would make the technical development more complicated,
    since this would require $\State$ to be isomorphic to $[\R^\infty \to \R]$ or $\biguplus_k [\R^k \to \R]$
    and the former (or latter) choice of $\State$ makes defining differentiability (or formalising our program analysis) nontrivial;
    we leave it as future work.}}
The set $\AVar$ consists of four types of auxiliary variables.
The auxiliary variable $\like$ stores the likelihood (i.e., the probability density of all the observations), and its value is initialised to $1$ and changes whenever the observe command $\cobs(d,r)$ runs; the new value becomes the old times the density of the probability distribution $d$ at $r$.
The other auxiliary variables $\pr_\mu$, $\val_\mu$, and $\sampled_\mu$ are associated with a random variable $\mu$,
\rev{standing for the ``prior'', ``value'', and ``counter'' of $\mu$}.
They are initialised with $\cN(\sigma(\mu); 0,1)$ (i.e., the density of the standard normal distribution at $\sigma(\mu)$),
$\sigma(\mu)$, and $0$, respectively, and get updated by the sample command $x:=\csample(n,d,\lambda y.e)$ where $n$ denotes $\mu$. The command increases $\sampled_\mu$ by $1$, so as to record the occurrence of a sampling event for $\mu$. Then, it
looks up the given value $\sigma(\mu)$ of the random variable $\mu$,
transforms the value to $e[\sigma(\mu)/y]$, and stores the result in $x$ and $\val_\mu$.  Finally, the command computes the density of the distribution $d$ at the looked-up value $\sigma(\mu)$, and updates $\pr_\mu$ with this density.
The unnormalised posterior density (i.e., the joint density of all the random variables and observations) is then obtained by multiplying at the end of program execution the values of $\like$ and $\pr_\mu$ for all $\mu \in \Name$.

The formal semantics of expressions is standard, and has the following types:
\begin{align*}
  \db{e} &: \State \to \R, &
  \db{b} &: \State \to \B, &
  \db{n} &: \State \to \Name, &
  \db{d} &: \State \to \D.
\end{align*}
Here
 $\B$ is the set of booleans, i.e., $\strue$ and $\sfalse$,
and $\D$ is the set of positive probability-density functions on $\R$, i.e., a subset of $[\R \to (0,\infty)]$ whose elements
are measurable functions that integrate to $1$. The semantics is defined for a minor extension of the set of expressions 
where non-program variables are allowed to appear, such as $(\mu+x)$. The interpretation of expressions is mostly standard.
We show only the case for the name expressions $n \equiv \cname(\alpha,e)$:
$\db{\cname(\alpha,e)}\sigma \defeq \mathit{create\_name}(\alpha,\db{e}\sigma)$,
\rev{where $\mathit{create\_name} : \Str \times \R \to \Name$ is an operator converting a string-real pair to a name,
  defined by $\mathit{create\_name}(\alpha,r) = (\alpha, \min\{\max\{\lfloor r \rfloor, 0\}, N-1\})$.}%
\footnote{\rev{There are multiple valid ways to convert a string-real pair to a name (i.e., to define $\mathit{create\_name}$); we choose just one of them.}}

\rev{Note that the types of the semantics of expressions imply that}
the evaluation of an expression always produces a value of the right type. In particular, 
it never generates an error. For instance, when an argument of an operator $\op$ or a distribution constructor $\cnor$ is outside its intended domain as 
in $\log(-3)$ and $\cnor(0, -2)$, or when the integer part of a name expression is outside $[0, N)$ as in $\cname(\cstr{z},-1)$, our semantics
does not generate an error. Instead, it returns some pre-chosen default value of the right type. 
This slightly unusual way of handling errors is also adopted in our semantics of commands to be presented shortly, and it 
lets us avoid the complexity caused by error handling when we formalise variational inference and develop our program analysis 
for smoothness properties. 

The formal semantics of commands is also mostly standard with the handling of errors via default values, 
although its interpretation of sample and observe commands deserves special attention.
Let $\bot$ be an element not in $\State$, and define $\State_\bot$ to be the usual lifting of $\State$ with $\bot$.
That is, $\State_\bot$ is a partially-ordered set $\State \uplus \{\bot\}$ with the following order: for $\xi,\xi' \in \State_\bot$,
we have $\xi \sqsubseteq \xi'$ if and only if $\xi = \bot$ or $\xi = \xi'$.
We write the standard lifting of a function $f : \State \to \State_\bot$ by $f^\dagger : \State_\bot \to \State_\bot$
(i.e., $f^\dagger(\xi) \defeq \text{if}\ (\xi = \bot)\ \text{then}\ \xi\ \text{else}\ f(\xi)$). The semantics of a command $c$ is a map $\db{c} : \State \rightarrow \State_\bot$, and is defined inductively as shown below:
\begin{align*}
  \db{\cskip}\sigma &\defeq \sigma,
  \\[-0.5ex]
  \db{x := e}\sigma &\defeq \sigma[x\mapsto \db{e}\sigma], 
  \\[-0.5ex]
  \db{c;c'}\sigma &\defeq \smash{\db{c'}^\dagger}(\db{c}\sigma),
  \\[-0.5ex]
  \db{\cif\,b\,\{c\}\, \celse\,\{c'\}}\sigma &\defeq
  \text{if } (\db{b}\sigma = \strue)
  \text{ then }
  \db{c}\sigma 
  \text{ else }
  \db{c'}\sigma,
  \\[-0.5ex]
  \db{\cwhile\,b\,\{c\}}\sigma 
  & \defeq 
  (\fix\ F)(\sigma)
  \quad \text{where}\
  F(f)(\sigma)    
  \defeq
  \text{if } (\db{b}\sigma=\strue)         
  \text{ then }
  \smash{f^\dagger}(\db{c}\sigma)  
  \text{ else }
  \sigma,
  \\[-0.5ex]
  \db{x := \csample(n,d,\lambda y.e')}\sigma 
  &\defeq 
  \sigma[x \mapsto r,\, \val_\mu \mapsto r,
    \pr_\mu \mapsto \db{d}\sigma(\sigma(\mu)),\, \sampled_\mu \mapsto \sigma(\sampled_\mu)+1]
  \\[-0.5ex]
  & \phantom{ {} \defeq 
  (\fix\ F)(\sigma)
  \quad } \text{where}\
  \mu \defeq \db{n}\sigma
  \ \text{and}\ 
  r \defeq {} \db{e'[\mu/y]}\sigma,
  \\[-0.5ex]
  \db{\cobs(d,r)}\sigma &\defeq  \sigma[\like\mapsto \sigma(\like) \cdot \db{d}\sigma(r)].
\end{align*}
The interpretation uses the least fixed-point operator $\fix$ for continuous maps $F$ on the function space $[\State \to \State_\bot]$, where the function space is ordered pointwise and continuity means the one with respect to this order.
\rev{It also uses the notation $e'[e''/y]$ to denote the substitution of $y$ in $e'$ by $e''$.}
According to this interpretation, $x := \csample(n,d,\lambda y.e')$ increments the $\sampled_\mu$ variable for the name $n=\mu$ so that the variable, which has $0$ initially, records the number of times that the random variable with the same name $n$ is sampled during execution. 

Having some $\sampled_\mu$ variable increased by $2$ or larger at some point of execution is not an intended behaviour of a command $c$. That is, if $c$ is a well-designed command, every random variable with a fixed name should be sampled at most once during the execution of $c$. This intended behaviour of commands plays an important role in our results, and we refer to it using the following terminology.
\begin{definition}
An always-terminating command $c$ \emph{does not have a double-sampling error} if for any $\sigma \in \State$, we have $\db{c}\sigma(\sampled_{\mu})  - \sigma(\sampled_\mu) \leq 1$ for all $\mu \in \Name$.
\qed
\end{definition}

\vspace{-2mm}
\rev{\begin{example}[Density semantics]
    Consider the following state $\sigma \in \State$: 
    $\sigma \defeq [x \mapsto 0, \allowbreak y \mapsto 0, \allowbreak
      (\cstr{a},0) \allowbreak \mapsto 2, \allowbreak (\cstr{b},0) \allowbreak \mapsto 4, \allowbreak
      \sampled_{(\cstr{a},0)} \allowbreak \mapsto 0, \allowbreak \sampled_{(\cstr{b},0)} \allowbreak \mapsto 0, \cdots],$
    where $x,y \in \PVar$ denote program variables and $(\cstr{a},0), (\cstr{b},0) \in \Name$ denote random variables.
    Note that $\sigma$ consists of three parts:
    the $\PVar$ part says that the values of $x$ and $y$ are both 0;
    the $\Name$ part says that the values of $(\cstr{a},0)$ and $(\cstr{b},0)$ are 2 and 4;
    and the $\AVar$ part says that $(\cstr{a},0)$ and $(\cstr{b},0)$ have not been sampled yet.
    
    Next, consider a command
    $c \equiv \big(x := \csample(\cname(\cstr{a},0), \allowbreak \cnor(-3,1), \allowbreak \lambda z.z); \allowbreak \;
    y := \csample(\cname(\cstr{b},0), \allowbreak \cnor(5,1), \allowbreak \lambda z.z)\big).$
    Given the command $c$ and the input state $\sigma$, our density semantics computes the following output state $\sigma' \in \State$:
    $\sigma' \defeq \db{c}\sigma = [x \mapsto 2, \allowbreak y \mapsto 4, \allowbreak
      (\cstr{a},0) \allowbreak \mapsto 2, \allowbreak (\cstr{b},0) \allowbreak \mapsto 4, \allowbreak
      \sampled_{(\cstr{a},0)} \allowbreak \mapsto 1, \allowbreak \sampled_{(\cstr{b},0)} \allowbreak \mapsto 1, \allowbreak
      \pr_{(\cstr{a},0)} \allowbreak \mapsto \cN(2;-3,1), \allowbreak \pr_{(\cstr{b},0)} \allowbreak \mapsto \cN(4;5,1), \cdots].$
    The input/output states $\sigma$ and $\sigma'$ illustrate two aspects of our semantics.
    First, the semantics uses the $\Name$ part of an input state to determine the sampled values of sample commands:
    $x$ and $y$ in $\sigma'$ take the values of $\sigma((\cstr{a},0))=2$ and $\sigma((\cstr{b},0))=4$.
    Second, the semantics records, in the $\AVar$ part of an output state, the densities of sample commands:
    $\pr_{(\cstr{a},0)}$ and $\pr_{(\cstr{b},0)}$ in $\sigma'$ store the densities of the two sample commands in $c$ evaluated at 2 and 4. 
    \qed
\end{example}}


\noindent{\bf Variational inference.}\
We consider the most common form of \emph{variational inference} for Pyro-like PPLs where we are asked to learn a good approximation of the posterior of a given model, i.e., the conditional distribution of the model given a dataset. Typically, a parameterised approximate posterior is given in variational inference, and learning corresponds to finding good values of those parameters. A popular approach is to measure the quality of parameter values by the so called evidence lower bound (ELBO), and to optimise ELBO.

To translate what we have described so far to our context, we need to explain a general recipe for generating a density $p_c$
for a command~$c$, which is in general unnormalised (i.e., does not integrate to $1$). The recipe specifies $p_c$ \rev{as follows:}
for each $\sigma_\theta \in \State[\theta]$ \rev{with $\theta \subseteq \PVar$},
  $p_{c,\sigma_\theta} : \State[\Name] \to [0,\infty)$ is defined by
\begin{align}
 \label{eq:density-of-c}
 p_{c,\sigma_\theta}(\sigma_n) & \defeq
 \begin{cases}
   \db{c}\sigma(\like) \cdot \prod_{\mu \in \Name} \db{c}\sigma(\pr_\mu)
   & \text{if $\db{c}\sigma \in \State$ and $\db{c}\sigma(\sampled_\mu) \leq 1$ for all $\mu$}
   \\
   0 & \text{otherwise}
 \end{cases}
\end{align}
where $\sigma = \sigma_0 \oplus \sigma_\theta \oplus \sigma_n \in \State$, and
the $\oplus$ operator combines two real-valued maps with disjoint domains in the standard way. Also,
$\sigma_0 \in \State[(\PVar \setminus \theta) \uplus \AVar]$ maps $\like$ to $1$,
and $\pr_\mu$ to $\cN(\sigma_n(\mu);0,1)$ and $\val_\mu$ to $\sigma_n(\mu)$ for every $\mu \in \Name$,
and all other variables to $0$.
Here $\State[\Name]$ is understood as a measurable space constructed by taking the product of the $|\Name|$ copies of the measurable space $\R$
and the integral is taken over the uniform measure on $\State[\Name]$ (i.e., the product of the $|\Name|$ copies of the Lebesgue measure on $\R$).

In variational inference in our PPL context, we are given two commands $c_m$ and $c_g$, called \emph{model} and \emph{guide}. We assume that (i) these commands always terminate and do not have a double-sampling error, (ii) some variables $\theta = \{\theta_1,\ldots,\theta_k\} \subseteq \PVar$ that only appear in $c_g$, not in $c_m$, are identified as parameters to be optimised,
and (iii) the density $\smash{ p_{c_g,\sigma_\theta} }$ of the guide $c_g$ integrates to $1$ and defines a probability distribution.%
\footnote{%
In practice, one more assumption is required: the set of random variables sampled from the model should be the same as that from the guide.
This assumption can be checked automatically, e.g., by \cite{LeeYRY20,LewCSCM20}.
In this work, however, this assumption is always satisfied:
all random variables in $\Name$ are sampled by a sample command or at the beginning (via initialisation).
}
Given the model-guide pair $(c_m,c_g)$,
a popular approach for variational inference is to solve the following optimisation problem approximately,
\begin{align}
  \label{eqn:ELBO-pre}
  \argmax_{\sigma_\theta}\, \E_{p_{c_g,\sigma_\theta}(\sigma_n)}
  \left[\log (p_{c_m}(\sigma_n)/p_{c_g,\sigma_\theta}(\sigma_n))\right]\!,
\end{align}
when the expectation is well-defined for all $\sigma_\theta$. The objective of this optimisation is the ELBO that we mentioned earlier. Here $p_{c_m}$ means $p_{c_m,\sigma'_\theta}$ for some/any $\sigma'_\theta$; the choice of $\sigma'_\theta$ does not matter since $c_m$ does not access the parameters $\theta$ and so
$p_{c_m,\sigma'_\theta} = p_{c_m,\sigma''_\theta}$ for all $\sigma'_\theta,\sigma''_\theta \in \State[\theta]$. 

Often variational inference is applied when the model $c_m$ is parameterised as well. In those cases, it asks for finding good parameters of the model $c_m$ as well as those of the guide $c_g$. So, an algorithm for variational inference this time simultaneously learns a good model for given observations and a good approximate posterior for the learnt model. This more general form of variational inference  can be easily accommodated in our setup. We just need to drop the condition that the parameters may not appear in $c_m$, and to use $p_{c_m,\sigma_\theta}$ instead of $p_{c_m}$ in the optimisation objective in \cref{eqn:ELBO-pre}:
\begin{align}
  \label{eqn:ELBO}
  \argmax_{\sigma_\theta} \cL(\sigma_\theta)
  \ \text{ for }\ 
  \cL(\sigma_\theta) \defeq
  \E_{p_{c_g,\sigma_\theta}(\sigma_n)}
  \left[\log (p_{c_m,\sigma_\theta}(\sigma_n)/p_{c_g,\sigma_\theta}(\sigma_n))\right]\!.
\end{align}
The rest of the paper focuses on this general form of variational inference \rev{(often called model learning)}. 


\section{Selective Pathwise Gradient Estimator}
\label{sec:generalised-gradient-estimator-new} 

 
We consider a gradient-based algorithm for the optimisation problem in \cref{eqn:ELBO}.
The algorithm finds a good $\sigma_\theta$ by repeatedly estimating
the gradient of the optimisation objective at the current $\sigma_\theta$,
\[
\mathrm{grad\_est}(\sigma_\theta) \approx
\nabla_\theta
\E_{p_{c_g,\sigma_\theta}(\sigma_n)}
  \left[\log (p_{c_m,\sigma_\theta}(\sigma_n)/p_{c_g,\sigma_\theta}(\sigma_n))\right]\!,
\]
and updating $\sigma_\theta$ with the estimate under a learning rate $\eta > 0$, that is,
$\sigma_\theta \leftarrow \sigma_\theta + \eta \cdot \mathrm{grad\_est}(\sigma_\theta)$.
Note that the core of the algorithm lies in the computation of $\mathrm{grad\_est}(\sigma_\theta)$.

In this section, we describe a particular algorithm for the gradient computation, called \emph{selective pathwise gradient estimator (SPGE)}, which is often regarded as the algorithm of choice and corresponds to the inference algorithm developed for stochastic computation graphs \cite{SchulmanHWA15} and implemented for Pyro. Our description of the SPGE takes the often-ignored aspect of customising the SPGE algorithm for PPLs seriously,
and it is accompanied with a novel formal analysis of the customisation \rev{(\cref{sec:prog-trans} and \cref{sec:grad-estm-prog-trans})}.
Our analysis clearly identifies information about probabilistic programs that is useful for this customised SPGE algorithm, and prepares the stage for our program analysis for smoothness properties in \cref{s:sa:new} \rev{(\cref{sec:grad-estm-prog-trans} and \cref{sec:local-lips-req})}. 

\subsection{Program Transformation}
\label{sec:prog-trans}

We start by describing a program transformation that changes some sample commands in a given probabilistic program. This transformation is used crucially by the SPGE. 

The key component of the transformation is a partial function $\pi$ called \emph{reparameterisation plan}, which has the type
$\NameEx \times \DistEx \times \LamEx \rightharpoonup \DistEx \times \LamEx$.
Here $\NameEx$, $\DistEx$, and $\LamEx$ denote the sets of name expressions, distribution expressions, and lambda expressions of the form $\lambda y.e$, respectively. The plan $\pi$ specifies how we transform sample commands. Concretely, assume that we are given $x := \csample(n,d,\lambda y.e)$. We check whether $\pi(n,d,\lambda y.e)$ is defined or not. If not, we  keep the original sample command. Otherwise, \rev{if $\pi(n,d,\lambda y.e)$} is $(d',\lambda y'.e')$, we replace the command with $x := \csample(n,d',\lambda y'.e')$. 

A natural extension of this intended transformation of $\pi$ leads to the following program transformation for a general command $c$, denoted by $\ctr{c}{\pi}$:
\begin{align*}
  \smash{ \ctr{\cskip}{\pi} }
  &\defeq \cskip,
  \\[-0.5ex]
  \smash{ \ctr{x:=e}{\pi} }
  &\defeq x:=e,
  \\[-0.5ex]
  \smash{ \ctr{c;c'}{\pi} }
  &\defeq \smash{ \ctr{c}{\pi};\ctr{c'}{\pi}, }
  \\[-0.5ex]
  \smash{ \ctr{\cif\, b\, \{c\}\, \celse\, \{c'\}}{\pi} }
  &\defeq \smash{ \cif\, b\, \{\ctr{c}{\pi}\}\, \celse\, \{\ctr{c'}{\pi}\}, }
  \\[-0.5ex]
  \smash{ \ctr{\cwhile\ b\ \{c\}}{\pi} }
  &\defeq \smash{ \cwhile\ b\ \{\ctr{c}{\pi}\}, }
  \\[-1.2ex]
  \smash{ \ctr{x:=\csample(n,d,l)}{\pi} }
  &\defeq
  \begin{cases}
    x:=\csample(n, d', l')
    & \text{if}\ \exists (d',l').\, \pi(n,d,l) = (d',l')
    \\[-0.5ex]
    x:=\csample(n,d,l)
    & \text{otherwise},
  \end{cases}
  \\[-0.5ex]
  \smash{ \ctr{\cobs(d,r)}{\pi} }
  &\defeq \cobs(d,r).
\end{align*}
The transformation recursively traverses $c$, and applies $\pi$ to all the sample commands in $c$.
Note that for any $\pi$, there exists a total function $\pi'$ such that $\smash{ \ctr{c}{\pi} = \ctr{c}{\pi'} }$ for all $c$; the $\pi'$ coincides
with $\pi$ in the domain of $\pi$, and outside of this domain, it is the identity function. But such $\pi'$ loses
information about the domain of $\pi$, which plays a crucial role in our formalisation of the SPGE.

We are primarily interested in semantics-preserving instances of $\ctr{\,\cdot\,}{\pi}$. The next definition helps us to identify such instances.
\begin{definition}
  \label{def:valid-plan}
A reparameterisation plan
$\pi$ is \emph{valid} if
for all $n \in \NameEx$, $d,d' \in \DistEx$, and $(\lambda y.e),(\lambda y'.e') \in \LamEx$ such that
$\pi(n,d,\lambda y.e) = (d',\lambda y'.e')$,
the following condition holds: for all states $\sigma \in \State$ and measurable subsets $A \subseteq \R$,
  \begin{equation}
  \label{eqn:valid-plan}
    \int \ind{\rev{\db{e[r/y]}\sigma} \in A} \cdot \db{d}\sigma(r) \,dr
    = 
    \int \ind{\rev{\db{e'[r/y']}\sigma} \in A} \cdot \db{d'}\sigma(r) \,dr.
  \makeatletter\displaymath@qed\,
  \end{equation}
\end{definition}
\noindent
The condition says that the distribution obtained by sampling from $d$ and applying $\lambda y.e$
is the same as that obtained by sampling from $d'$ and applying $\lambda y'.e'$.
An example of a widely-used valid reparameterisation plan
maps its input as follows, whenever defined:
$\pi_0(n,\,\cnor(e_1,e_2),\,\lambda y.e_3) = (\cnor(0,1),\, \lambda y.e_3[(y \times \sqrt{e_2} + e_1)/y])$,
where we assume $y$ does not appear in $e_1$ and $e_2$, the substitution in $\pi_0$ expresses the composition of two functions
$\lambda y.e_3$ and $\lambda y.(y \times \sqrt{e_2} + e_1)$, and $\sqrt{-}$ denotes a square-root operator that handles non-positive arguments in the same way as $\cnor(e,-)$ does:
if $\db{e_2}\sigma \leq 0$ and $\db{\cnor(e_1,e_2)}\sigma = \lambda r.\, \cN(r;\db{e_1}\sigma, r_2)$ for some $r_2 > 0$,
then $\db{\sqrt{e_2}}\sigma = \sqrt{r_2}$. The above plan satisfies the condition in \cref{eqn:valid-plan}, because $y \times \sqrt{r_2} + r_1$ with a sample $y$ from $\cN(0,1)$ is distributed by $\cN(r_1,r_2)$.

We now  show that $\ctr{\,\cdot\,}{\pi}$ with a valid $\pi$ preserves semantics.
For a command $c$
and $\sigma_\theta \in \State[\theta]$, define the \emph{value function} $\vfun{c,\sigma_\theta} : \State[\Name] \to \State[\Name]$ as follows:
\[
  \vfun{c, \sigma_\theta}(\sigma_n)(\mu) 
  {} \defeq {} 
  \text{let $\sigma \defeq  \sigma_0 \oplus \sigma_\theta \oplus \sigma_n$ in }
  \begin{cases}
    \db{c}\sigma(\val_\mu)
    & \text{if $\db{c}\sigma \in \State$ and $\db{c}\sigma(\sampled_{\mu'}) \leq 1$ for all $\mu'$}
    \\
   0  &  \text{otherwise}
  \end{cases}
\]
where $\sigma_0 \in \State[(\PVar \setminus \theta) \uplus \AVar]$ maps $\like$ to $1$, and
$\pr_\mu$ to $\cN(\sigma_n(\mu);0,1)$ and $\val_\mu$ to $\sigma_n(\mu)$ for every $\mu \in \Name$,
and it also maps all the other variables to $0$.
The value function basically applies the lambda functions in $c$'s sample commands to the corresponding random variables.
The next theorem proves that if $\pi$ is valid,
 the program transformation $\ctr{\,\cdot\,}{\pi}$ preserves the semantics in the sense that 
the integral of a function $h$ remains the same under $c$ and $\ctr{c}{\pi}$ for any $c$.
Note that the two integrals in the theorem are connected via the value functions of $c$ and~$\ctr{c}{\pi}$.
\begin{theorem}
  \label{thm:unbiased-val}
  Let $\pi$ be a valid reparameterisation plan, and $c$ be a command.
  Then, for all $\sigma_\theta \in \State[\theta]$ and all measurable $h: \State[\Name] \to \R$, we have
  \begin{align*}
    { \int d\sigma_n \Big(\pfun{c,\sigma_\theta}{}(\sigma_n) \cdot h(\vfun{c, \sigma_\theta}(\sigma_n))\Big) }
    &=
    \int d\sigma_n \Big(\pfun{\ctr{c}{\pi},\sigma_\theta}{}(\sigma_n) \cdot h(\vfun{\ctr{c}{\pi}, \sigma_\theta}(\sigma_n))\Big)
  \end{align*}
  where the left integral is defined if and only if the right integral is defined.
\end{theorem}

\begin{remark}
\label{remark:p-probability}
One immediate yet important consequence of the theorem is that if $p_{c,\sigma_\theta}$ is a probability density, so is $p_{\ctr{c}{\pi}, \sigma_\theta}$. This consequence will be used in \cref{sec:grad-estm-prog-trans} and the proof of \cref{thm:unbiased-grad} later. 
\qed
\end{remark}

\subsection{Gradient Estimator via Program Transformation}
\label{sec:grad-estm-prog-trans}

Let $c$ be a command that always terminates and does not have a double-sampling error,
and let $\sigma_\theta \in \State[\theta]$. We define the \emph{partial density function $\pfun{c,\sigma_\theta}{S}$ of $c$ over a subset $S \subseteq \Name$} as {%
\begin{align*}
  \pfun{c, \sigma_\theta}{S} &
  {} : \State[\Name] \to (0, \infty),
  &
  \pfun{c, \sigma_\theta}{S}(\sigma_n) & 
  {} \defeq {} \prod_{\mu \in S} \db{c}(\sigma_0 \oplus \sigma_\theta \oplus \sigma_n)(\pr_\mu),
\end{align*}%
}%
where $\sigma_0$ is set as in the definition of $\pfun{c,\sigma_\theta}{}$ in \cref{eq:density-of-c}. The partial density $\pfun{c,\sigma_\theta}{S}$ is essentially the full density
$\pfun{c,\sigma_\theta}{}$ in \cref{eq:density-of-c} with the omission of the factors not mentioned in $S$. Intuitively, it computes the density of the random
variables in $S$ conditioned on the random variables outside of $S$.

The SPGE computes an approximate gradient of the objective $\cL$ in \cref{eqn:ELBO} using the program transformation 
in the previous subsection. Its inputs are a model $c_m$, a guide $c_g$, parameters $\theta$ to optimise, and 
a reparameterisation plan $\pi$,
where
{\abovedisplayskip=\topsep
  \abovedisplayshortskip=\topsep
  \belowdisplayskip=\topsep
  \belowdisplayshortskip=\topsep
\begin{align}
  \label{eq:spge-input-assume}
  \begin{array}{r@{\hskip\labelsep}l@{\hskip\leftmargin}}
    \text{\labelitemi} & \text{$c_m$, $c_g$, and $\ctr{c_g}{\pi}$ always terminate and do not have a double-sampling error, and} \\
    \text{\labelitemi} & \text{$c_g$ defines the normalised probability density $\pfun{c_g,\sigma_\theta}{}$ for all $\sigma_\theta \in \State[\theta]$.}
  \end{array}
\end{align}}%
Given these inputs, the SPGE computes an approximate gradient in three steps. First, it defines the set $\repname(\pi) \subseteq \Name$ of random variables to be reparameterised:
\begin{align*}
  \repname(\pi) \defeq \{(\alpha,i) \in \Name \,\mid\, {} 
  (\cname(\alpha,\_), \_, \_) \in \dom(\pi)\}
\end{align*}
where $\_$ means some existentially quantified (meta) variable. Second, the SPGE transforms the guide $c_g$ to $\ctr{c_g}{\pi}$,
and draws a sample $\hat{\sigma}_n$ from $\pfun{\ctr{c_g}{\pi},\sigma_\theta}{}$.%
\footnote{%
In practice, the SPGE often draws a fixed number of independent samples $\smash{ \hat{\sigma}^{(1)}_n,\ldots,\hat{\sigma}^{(M)}_n }$ from $\smash{ \pfun{\ctr{c_g}{\pi},\sigma_\theta}{} }$
and computes $\smash{ \frac{1}{M} \sum_{i = 1}^M \mathrm{grad\_est}(\sigma_\theta; \hat{\sigma}^{(i)}_n) }$
as an estimate of $\nabla_\theta \cL(\sigma_\theta)$.
The presented results hold for this more general case as well.%
}
Drawing a sample $\hat{\sigma}_n$ makes sense here since $\pfun{\ctr{c_g}{\pi},\sigma_\theta}{}$ is a probability density (i.e., it normalises to $1$) {by \cref{remark:p-probability}}. Another important point is that drawing $\hat{\sigma}_n$ can be done simply by executing $\ctr{c_g}{\pi}$ in the standard sampling semantics (not in our density semantics), where each sample command is interpreted as a random draw, not as a density calculator. Third, the SPGE computes the following approximation of $\nabla_\theta \cL(\sigma_\theta)$ and returns it as a result:
\begin{equation}
\begin{aligned}
  \label{eq:sel-grad-est}
  \mathrm{grad\_est}\rev{({\sigma}_\theta; \hat{\sigma}_n)} & \defeq
      \left(\nabla_\theta \log \pfun{c_g, \sigma_\theta}{\Name \setminus \repname(\pi)}(\sigma_n')\right)
      \cdot \log (\pfun{c_m,\sigma_\theta}{}(\sigma_n')/\pfun{c_g,\sigma_\theta}{}(\sigma_n'))
   \\
   & \qquad{}
    - \nabla_\theta \log{\pfun{c_g,\sigma_\theta}{\repname(\pi)}(\sigma_n')}
    + \nabla_\theta \log{\pfun{c_m,\sigma_\theta}{}(\sigma_n')},
    \qquad
    \text{for}\ \sigma'_n \defeq \vfun{\ctr{c_g}{\pi}, \sigma_\theta}(\hat{\sigma}_n).
\end{aligned}
\end{equation} 
Recall that if a command $c$ always terminates, both the partial density $\pfun{c,\sigma_\theta}{S}(\sigma_n)$ and the full density $p_{c,\sigma_\theta}(\sigma_n)$ can be computed simply by executing $c$ in our semantics and calculating the defining formulas of both densities from the final state of the execution. Thus, all the terms in $\mathrm{grad\_est}$ can be computed by executing $c_g$ and $c_m$ according to our density semantics or differentiating the results of these executions via, for instance, automatic differentiation as done in Pyro.
\rev{Note that $\mathrm{grad\_est}$ applies two non-trivial optimisations, when compared with the (naive) SPGE explained in \cref{eq:spge-toy}:
  its  first term involves a partial density of $c_g$ instead of the full density of $\ctr{c_g}{\pi}$,
  and its second term involves (again) a partial density of $c_g$ instead of the full density of $c_g$.}

Is the SPGE correct in any sense? The answer depends on its inputs. If the inputs satisfy the requirements that we will explain soon, the result of the SPGE is precisely $\nabla_\theta \cL(\sigma_\theta)$ on average, that is,
$\nabla_\theta \cL(\sigma_\theta) = \E[ \mathrm{grad\_est}(\sigma_\theta; \hat{\sigma}_n)]$,
where the expectation is taken over the sample $\hat{\sigma}_n$
used by the SPGE. This property is called \emph{unbiasedness}, and it plays the crucial role for ensuring that
parameters updated iteratively with estimated gradients converge to a local optimum.

Let us now spell out the requirements on the inputs of the SPGE. To do so, we need to introduce one further concept for the reparameterisation plans $\pi$.
\begin{definition}
  A reparameterisation plan $\pi$ is \emph{simple}
  if for all $(n,d,\lambda y.e)$ and $(n',d',\lambda y'.e')$ in $\NameEx \times \DistEx \times \LamEx$ such that
  $n$ and $n'$ have the same string part, we have
  $(n,d,\lambda y.e) \in \dom(\pi) \iff (n',d',\lambda y'.e') \in \dom(\pi)$.
  \qed
\end{definition}
\noindent
The simplicity is one of the requirements that the SPGE imposes on $\pi$.
\rev{It ensures the following property of the set $\repname(\pi)$, which the SPGE relies on when computing $\mathrm{grad\_est}$ by \cref{eq:sel-grad-est}:
  $\repname(\pi)$ (and $\Name \setminus \repname(\pi)$) over-approximates the set of random variables
  that, if sampled, are (and are not) reparameterised by $\ctr{\,\cdot\,}{\pi}$.}
Specifically, it forbids $\pi$ from using any syntax-specific information of the arguments of a sample command when it decides whether to transform the command or not.
All the requirements of the SPGE, including the simplicity just explained, are summarised in the next theorem.
\begin{theorem}
  \label{thm:unbiased-grad}
  \rev{Let $c_m$, $c_g$, and $\pi$ be the inputs to the SPGE (i.e., they satisfy the assumptions in \cref{eq:spge-input-assume}).}
  Suppose that $\cL(\sigma_\theta)$ and $\nabla_\theta \cL(\sigma_\theta)$ are well-defined for every $\sigma_\theta \in \State[\theta]$.
  Further, assume that every sample command in $c_g$ has $\lambda y.y$ as its third argument, 
  and $c_g$ does not have observe commands.
  Then,  for all $\sigma_\theta \in \State[\theta]$,
  \begingroup
  \addtolength{\abovedisplayskip}{-0.4em}
  \begin{align}
    \label{eq:spge-unbiased}
    & \nabla_\theta \cL(\sigma_\theta) = 
    \E_{\pfun{\ctr{c_g}{\pi},\sigma_\theta}{}(\hat{\sigma}_n)}\left[ \mathrm{grad\_est}(\sigma_\theta; \hat{\sigma}_n)\right]
  \end{align}
  \endgroup
  if $\pi$ satisfies the following requirements: 
  \begin{enumerate}[label=(R\arabic*),ref=(R\arabic*),itemindent=0.5em]
    \item \label{cond:pi}
  $\pi$ is valid and simple.
  \item \label{cond:dens-diff}
    The below functions from $\State[\theta] \times \State[\Name]$ to $(0,\infty)$ are differentiable in $\theta \cup \repname(\pi)$ jointly:
    \begin{align*}
      (\sigma_\theta, \sigma_n) 
      &\longmapsto \pfun{c_m, \sigma_\theta}{}(\sigma_n),
      &
      (\sigma_\theta, \sigma_n) 
      &\longmapsto \pfun{c_g, \sigma_\theta}{\repname(\pi)}(\sigma_n),
      &
      (\sigma_\theta, \sigma_n)
      &\longmapsto \pfun{c_g, \sigma_\theta}{\Name \setminus \repname(\pi)}(\sigma_n).
    \end{align*}
  \item \label{cond:val-diff}
  For all $\sigma_n \in \State[\Name]$,
    the below functions on $\State[\theta]$ are differentiable in $\theta$ jointly:
    \begin{align*}
      \sigma_\theta &\longmapsto \vfun{\ctr{c_g}{\pi}, \sigma_\theta}(\sigma_n), 
      &
      \sigma_\theta &\longmapsto \pfun{\ctr{c_g}{\pi}, \sigma_\theta}{\repname(\pi)}(\sigma_n), 
      &
      \sigma_\theta &\longmapsto \pfun{\ctr{c_g}{\pi}, \sigma_\theta}{\Name \setminus \repname(\pi)}(\sigma_n). 
    \end{align*}
  \item \label{cond:grad-dens-zero}
    For all $\sigma_\theta \in \State[\theta]$ and $\sigma_n \in \State[\Name]$, we have
    $\nabla_\theta \pfun{\ctr{c_g}{\pi}, \sigma_\theta}{\repname(\pi)}(\sigma_n) = 0$.
  \item \label{cond:int-diff}
    The below equations hold for all $\sigma_\theta \in \State[\theta]$:
    \begingroup
    \addtolength{\abovedisplayskip}{-0.3em}
    \addtolength{\belowdisplayskip}{0.05em}
    \begin{align*}
      \nabla_\theta \int  d\sigma_n \Big(\pfun{\ctr{c_g}{\pi}, \sigma_\theta}{}(\sigma_n)\Big)
      &= \int  d\sigma_n \nabla_\theta \Big(\pfun{\ctr{c_g}{\pi}, \sigma_\theta}{}(\sigma_n)\Big),
      \\[-0.45em]
      \nabla_\theta \int d \sigma_n \Big(
      \pfun{{\ctr{c_g}{\pi}},\sigma_\theta}{}(\sigma_n)
      \cdot \log\frac{\pfun{c_m, \sigma_\theta}{}({\sigma_n'})}{\pfun{c_g, \sigma_\theta}{}({\sigma_n'})}\Big)
      &=
      \int d \sigma_n\,\nabla_\theta \Big(
      \pfun{{\ctr{c_g}{\pi}},\sigma_\theta}{}(\sigma_n)
      \cdot \log\frac{\pfun{c_m, \sigma_\theta}{}({\sigma_n'})}{\pfun{c_g, \sigma_\theta}{}({\sigma_n'})} \Big).
    \end{align*}
    \endgroup
    In the second equation, we write $\sigma'_n$ for $\,\vfun{\ctr{c_g}{\pi}, \sigma_\theta}({\sigma_n})$.
  \end{enumerate}
\end{theorem}
\noindent
To be clear, $f : \State[K] \to \R^n$ for $K \subseteq \Var$ is said to be {\em differentiable in $K' \subseteq K$ jointly}
if for any $\tau \in \State[K \setminus K']$, $f|_{[\tau]} : \State[K'] \to \R^n$ is (jointly) differentiable,
where $f|_{[\tau]}(\sigma) \defeq f(\sigma \oplus \tau)$.

In this work, we focus on the requirements \cref{cond:dens-diff,cond:val-diff} about smoothness.
They require that five 
density functions of $c_m$, $c_g$, and $\ctr{c_g}{\pi}$,
and the value function of $\ctr{c_g}{\pi}$ be differentiable in certain variables.
We will develop a program-analysis framework to check these differentiability requirements soundly and automatically (\cref{s:sa:new}),
and will describe an algorithm to the SPGE variable-selection problem, 
using the developed analysis framework (\cref{sec:var-sel-alg}).
The remaining requirements \cref{cond:pi}, \cref{cond:grad-dens-zero}, and \cref{cond:int-diff} are of less interest in this work. 
\cref{cond:pi,cond:grad-dens-zero} can be guaranteed by simple syntactic checks and our way of constructing $\pi$ 
(\cref{lem:suff-cond-grad-dens-zero}). \cref{cond:int-diff} follows from \cref{cond:dens-diff} and \cref{cond:val-diff}
under a regularity condition on densities (\cref{lem:suff-cond-int-diff}) which is usually satisfied in practice;
\rev{in some cases, however, the condition might not hold,
  and we leave it as future work to automatically discharge the condition (which is about the integrability of local Lipschitz constants of densities)
  or more generally \cref{cond:int-diff}.%
  \footnote{\rev{Some sufficient conditions for a regularity condition similar to the one considered here
  have been studied for certain classes of model-guide pairs, e.g., in \cite{LeeYRY20}.}}}

We point out that Pyro uses the SPGE in their inference engine, but without checking the above requirements.
In particular, its default option simply uses the $\pi$ that transforms all the continuous random variables in a guide,
and this can easily violate the requirements and make the SPGE biased. 


\subsection{Local Lipschitzness for Relaxed Requirements}
\label{sec:local-lips-req}

In \cref{thm:unbiased-grad}, we considered the requirements \cref{cond:dens-diff,cond:val-diff}
about the differentiability of density and value functions,
as a sufficient condition for the unbiasedness of the SPGE.
They are, however, sometimes too strong to hold in practice due to the use of popular non-differentiable functions (e.g., ReLU).
As we will see in \cref{sec:impl-eval}, the requirements are indeed violated by some representative Pyro programs
even though the conclusion of \cref{thm:unbiased-grad} holds for those programs (i.e., the estimated gradients
by the SPGE for those programs are unbiased).

To validate the unbiasedness of the SPGE for more examples in practice,
we consider the following relaxation of the requirements \cref{cond:dens-diff,cond:val-diff}, which changes differentiability to local Lipschitzness:
\begin{enumerate}[label=(R\arabic*'),ref=(R\arabic*'),itemindent=0.8em]
\setcounter{enumi}{1}
\item \label{cond:dens-lips}
  The functions in \cref{cond:dens-diff} are locally Lipschitz in $\theta \cup \repname(\pi)$ jointly.
\item \label{cond:val-lips}
  For every $\sigma_n \in \State[\Name]$,
  the functions in \cref{cond:val-diff} are locally Lipschitz in $\theta$ jointly.
\end{enumerate}
Here \rev{$f : V \to \R^m$ for $V \subseteq \R^n$} is \emph{Lipschitz}
if there is $C>0$ such that $\|f(x)-f(x')\|_2 \leq C\|x-x'\|_2$ for~all \rev{$x,x' \in V$};
and $f$ is \emph{locally Lipschitz}
if for all $x \in \R^n$, there is an open neighborhood $U \subseteq \R^n$ of $x$ such that
$f|_U : U \to \R^m$ is Lipschitz.
Further, $f : \State[K] \to \R^m$ for $K \subseteq \Var$ is {\em locally Lipschitz in $K' \subseteq K$ jointly}
if for any $\tau \in \State[K \setminus K']$, $f|_{[\tau]} : \State[K'] \to \R^m$ is locally Lipschitz,
where $f|_{[\tau]}(\sigma) \defeq f(\sigma \oplus \tau)$.
Although differentiability does not imply local Lipschitzness,
\emph{continuous} differentiability does. Since most differentiable functions used in practice are
continuously differentiable, asking for \cref{cond:dens-lips,cond:val-lips} amounts to relaxing the requirements
of \cref{cond:dens-diff,cond:val-diff} in practice.

We choose local Lipschitzness as an alternative to differentiability in \cref{cond:dens-diff,cond:val-diff} for two main reasons.
First, local Lipschitzness is satisfied by most functions used in practice, which can even be non-differentiable (e.g., ReLU).
This would allow us to validate the unbiasedness of the SPGE for more programs, even when they use non-differentiable functions.
Second, using local Lipschitzness in \cref{cond:dens-diff,cond:val-diff} does not break the results given in \cref{sec:grad-estm-prog-trans},
even though local Lipschitzness is more practically permissive than differentiability (as explained above).
\rev{In particular, we have a counterpart of \cref{thm:unbiased-grad}
  that uses \cref{cond:dens-lips} and \cref{cond:val-lips} instead of \cref{cond:dens-diff} and \cref{cond:val-diff}:
  \begin{theorem}
    \label{thm:unbiased-grad-lip}
    Consider the setup of \cref{thm:unbiased-grad}.
    Then, for all $\sigma_\theta \in \State[\theta]$,
    \cref{eq:spge-unbiased} holds
    if $\pi$ satisfies the requirements \cref{cond:pi}, \cref{cond:dens-lips}, \cref{cond:val-lips}, \cref{cond:grad-dens-zero}, and \cref{cond:int-diff}.
  \end{theorem}
  \noindent
  We can obtain \cref{thm:unbiased-grad-lip} thanks to three properties of local Lipschitzness:
  locally Lipschitz functions are closed under function composition, have well-defined gradients almost everywhere,
  and satisfy the chain rule almost everywhere in restricted settings (\cref{lem:diff-rule-loclip}).
  Although the latter two properties are weaker than the corresponding properties of differentiability
  (i.e., differentiable functions always have well-definedness gradients and satisfy the chain rule),
  they are still strong enough to prove the theorem.
  
  As in \cref{sec:grad-estm-prog-trans}, we focus on the requirements \cref{cond:dens-lips} and \cref{cond:val-lips}
  and less on \cref{cond:pi}, \cref{cond:grad-dens-zero}, and \cref{cond:int-diff}.
  Note that \cref{cond:pi} and \cref{cond:grad-dens-zero} can be checked syntactically in the same way as discussed in \cref{sec:grad-estm-prog-trans}.
  On the other hand, \cref{cond:int-diff} follows from \cref{cond:dens-lips} and \cref{cond:val-lips}
  now under two (instead of one) regularity conditions on densities (\cref{lem:suff-cond-int-diff-lip}),
  where the first condition is the one mentioned in \cref{sec:grad-estm-prog-trans}
  and the second (new) condition is about some form of almost-everywhere differentiability of densities.
  Although we believe that the two regularity conditions would usually hold in practice,
  they can be violated in some cases that involve locally Lipschitz, non-differentiable functions;
  hence, it would be worthwhile to devise an automatic way of checking the two conditions or more generally \cref{cond:int-diff},
  which we leave as future work.}

Because local Lipschitzness property has wider coverage than differentiability in practice while ensuring that the results in \cref{sec:grad-estm-prog-trans} remain valid,
our implementation and experiments consider the option of using \cref{cond:dens-lips,cond:val-lips}
as well as that of using \cref{cond:dens-diff,cond:val-diff};
\rev{see \cref{:4:inst}--\cref{sec:impl-eval} for details.}

\section{Program Analysis for Smoothness}
\label{s:sa:new} \label{s:sa:n}

Recall that our goal is to develop an algorithm for the SPGE variable-selection problem in \cref{def:var-sel-prob}, which asks for finding a large set $S$ of random variables with a certain property when given a model $c_m$, a guide $c_g$, and a reparameterisation plan $\pi_0$. When rephrased using the terminologies that we covered so far, finding such an $S$ amounts to finding a restriction $\pi$ of the given $\pi_0$ such that $(c_m, c_g,\pi)$ satisfies the requirements in \cref{thm:unbiased-grad} \rev{(or \ref{thm:unbiased-grad-lip})}. Thus, the key for developing a desired algorithm for the problem lies in constructing an automatic method for proving that the requirements in \cref{thm:unbiased-grad} \rev{(or \ref{thm:unbiased-grad-lip})}, in particular, 
 the smoothness requirements \cref{cond:dens-diff,cond:val-diff} \rev{(or \cref{cond:dens-lips,cond:val-lips})} are met.
In this section, we propose a program analysis for smoothness properties, which can help find $\pi$ that meets \rev{the smoothness requirements}, 
and which, together with the optimiser in the next section, leads to an algorithm for solving the SPGE variable-selection problem.

We first define a parametric abstraction for smoothness properties (\cref{:1:abs}).
We then describe a program analysis based on this abstraction
and prove the soundness of the analysis (\cref{:2:sa}). 
We finally instantiate the analysis to differentiability and local Lipschitzness (\cref{:4:inst}).
The results in this section are not limited to PPLs,
but are applicable to general imperative programming languages.

\subsection{Parametric Abstraction for Smoothness Properties}
\label{:1:abs}
At a high level, our parametric abstraction for smoothness properties is built out of two components. 
The first is a predicate over commands that expresses a target smoothness property but in a conditional form. The predicate is parameterised by two sets of variables, $K$ for the input variables and $L$ for the output variables. Intuitively, the predicate holds for a command if conditioning the input variables outside of $K$ to any fixed values and varying only the ones in $K$ makes the command a smooth function on the output variables in $L$. Our program analysis tracks a conditional smoothness property formalised by this predicate, and in so doing, it identifies a smooth part of a given command, even when the command fails to be so with respect to some variables.

The second component is also a predicate over commands, but it deals with dependency, instead of smoothness. It is again parameterised by $K$ and $L$, and expresses that to compute the output variables in $L$, a command accesses only the input variables in~$K$. Our program analysis tracks dependency formalised by this predicate, so as to achieve high precision, especially when handling sequential composition.
To see this, imagine that we want to check the differentiability of a sequence $c;c'$.
A natural approach is to use the chain rule. If the dependency-tracking part of our analysis is missing, in order to establish that the output on a variable $v$ by the sequence is differentiable in an input variable $u$, the analysis should show the output on $v$ by the second command $c'$ is differentiable in \emph{all} variables, and the output on any variable by the first command $c$ is differentiable in $u$. This requirement on $c$ and $c'$ is too strong. Often, $c'$ uses only a small number of variables to compute $v$, and it is sufficient to require that just on those used variables, $c$ should be differentiable in $u$. Similarly, $c$ commonly updates only a small number of variables using $u$, and it is enough to require that just in those $u$-dependent variables, the second command $c'$ is differentiable when computing the output $v$. The dependency-tracking part lets our analysis carry out such reasoning and achieve better precision. Formally, this means our analysis uses a version of reduced product~\cite{cc:popl:79} between dependency analysis and the analysis that tracks the target smoothness property.

We now formally describe each of these components as well as their combination.

\subsubsection{Family of Smoothness Predicates.} Our program analysis assumes that a target smoothness property is specified in terms of a family of predicates,
\[
\phi = (\phi_{K,L} : K,L \subseteq \Var),
\]
where $\phi_{K,L}$ is a set of partial functions from $\State[K]$ 
to $\State[L]$ (i.e., $\phi_{K,L} \subseteq [\State[K] \rightharpoonup \State[L]]$).
\begin{example}[Differentiability]
  \label{e:1:diff}
  In the instantiation of our program analysis for differentiability, we use the
   family $\smash{ \phi^{(d)} }$ where for all $K,L \subseteq \Var$, a partial function $f : \State[K] \to \State[L]$ belongs to $\smash{ \phi^{(d)}_{K,L} }$ 
   if and only if (i) $\dom(f)$ is open and (ii) $f$ is (jointly) differentiable in its domain.
   \qed
\end{example}
At first, one may wonder why we use a family of $\phi_{K,L}$ predicates
instead of a single predicate $\phi_0$ over $[\State \to \State_\bot]$. The reason is that, as mentioned above, the analysis aims at a conditional
variant of the traditional notion of smoothness.
For instance, instead of checking that a function $f : \State \to \State_\bot$
is differentiable on $\State \setminus f^{-1}(\{\bot\})$, the analysis proves differentiability conditioned
on certain variables being fixed:
if we fix the input variables in $\Var \setminus K$
and vary just those in $K$ in the initial state, and look at the output variables in $L$ only,
then the function $f$ becomes differentiable, although it might not be so when all input/output variables are considered.
To express this, we need the whole family of $\phi_{K,L}$ predicates.

This notion of conditional differentiability is similar to, but
not the same as, so called partial differentiability.
Partial differentiability in $K$ says that, for every
$v \in K$, if we fix all the input variables except $v$,
including those in $K \setminus \{v\}$, and consider the output variables in
$L$ only, $f$ becomes differentiable.%
\footnote{\rev{Conditional differentiability extends partial differentiability
  in the sense that the latter can be expressed as a conjunction of the former (but not vice versa):
  $f$ is partially differentiable in $K$ if and only if $f$ is conditionally differentiable in $\{v\}$ for all $v \in K$.}}
As we will show in \cref{remark:phi-examples-nonexamples}, the set of partially-differentiable functions 
is not closed under a certain operator, but we need the closure to ensure
that our program analysis is sound. Our conditional
differentiability does not suffer from this issue. 

\subsubsection{Smoothness Abstraction.}
Based on the family $\phi$, we build a predicate $\Phi$ that captures
the smoothness of commands. The $\Phi$ constrains functions from $\State$ to $\State_\bot$, unlike $\phi_{K,L}$.
For $K,L \subseteq \Var$ with $K \supseteq L$, define $\pi_{K,L}$ to
be the projection from $\State[K]$ to $\State[L]$.
\begin{definition}
The \emph{smoothness abstraction} $\Phi$ is the predicate over a function $f \in [\State \to \State_\bot]$
and variable sets $K,L \subseteq \Var$. It is satisfied by $(f,K,L)$
if for all  $\tau \in \State[\Var \setminus K]$, the predicate $\phi_{K,L}$ holds
for the following partial function $g : \State[K] \rightharpoonup \State[L]$:
$\dom(g) \defeq \{\sigma \in \State[K] \,\mid\, f(\sigma \oplus \tau) \in \State \}$
and
$g(\sigma) \defeq (\pi_{\Var,L} \circ f)(\sigma \oplus \tau) \ \text{for}\ \sigma \in \dom(g)$.
We denote the satisfaction of $\Phi$ by \[{} \models \Phi(f,K,L). \makeatletter\displaymath@qed\]
\end{definition}
\noindent
Note that the function $g$ is constructed from $f$ by fixing the $\Var
\setminus K$ part of the input state to~$\tau$, and looking at only the
$L$ part of the output.
This construction is precisely the one used in the informal definition of conditional
differentiability described above, and its use reflects the fact that our
program analysis attempts to prove a conditional smoothness property.

\subsubsection{Dependency Abstraction.}
Our abstract domain has a component 
for tracking dependency between input-output variables. Dependency here means that 
a given input variable is used for computing a given output variable. We define 
a predicate $\Delta$ that has a similar format as $\Phi$. Intuitively, $\Delta(f,K,L)$ holds
if and only if the $L$ part of the output of $f$ depends at most on the $K$ part of the input to $f$.
To define $\Delta$ formally, for $K \subseteq \Var$, let $\sim_K$ be the following
equivalence relation over states:
$\sigma \sim_K \sigma'$ if and only if $\sigma(v) = \sigma'(v)$ for all $v \in K$.
\begin{definition}
The \emph{dependency abstraction} $\Delta$ is the predicate on $f \in [\State \to \State_\bot]$ and $K,L \subseteq \Var$ that holds if for all $\sigma,\sigma' \in \State$ with $\sigma \sim_K
\sigma'$, we have
$(f(\sigma) \in \State \,{\iff}\, f(\sigma') \in \State)$
and $(f(\sigma) \in \State
\,{\implies}\, f(\sigma) \sim_L f(\sigma'))$.
We denote the satisfaction of $\Delta$ by \[{} \models \Delta(f,K,L). \makeatletter\displaymath@qed\]
\end{definition}

\subsubsection{Combined Abstraction.}
We bring together the two abstractions that we just defined, and construct the final abstract domain $\aD$ used by our program analysis.

Intuitively, each element of $\aD$ is a predicate on a function $f \in [\State \to \State_\bot]$ expressed by the conjunction 
of the following form: $\bigwedge_{i = 1}^m \Phi(f,K_i,L_i) \wedge \bigwedge_{j = 1}^n \Delta(f,K'_j,L'_j)$.
A direct but naive way of implementing this intuition is to let $\aD$ be the collection of all the constraints of this form, but it permits too many constraints and leads to a costly program analysis. We take a more economical alternative that further restricts the allowed form of the constraints. The alternative requires that the conjunction from above should be constructed out of two mappings $p$ and $d$ from output variables to input variable sets, and a set $V$ of input variables.
The $p$ component describes smoothness, and the $d$ and $V$ components dependency.
They together encode the constraint
\[
\bigwedge_{v \in \Var} \Phi(f,p(v),\{v\})
\wedge
\bigwedge_{u \in \Var} \Delta(f,d(u),\{u\})
\wedge
\Delta(f,V,\emptyset).
\] 
Thus, a function $f \in [\State \to \State_\bot]$ satisfies the constraint encoded by $p$, $d$, and $V$ if (i) for every output variable $v$, when we fix the values of all the input
variables outside of $p(v)$, the (partial) function $\sigma \longmapsto f(\sigma)(v)$
is smooth (e.g., differentiable); (ii) for every output variable $u$, the (partial) function $\sigma \longmapsto f(\sigma)(u)$ does not access any variable outside of $d(u)$ to compute the value of $u$; and (iii) the values of input variables in $V$ determine whether $f$ returns $\bot$  or not.

\begin{definition}
The \emph{abstract domain} $\aD$ consists of triples 
$(p,d,V) \in [\Var \to \cP(\Var)]^2 \times \cP(\Var)$,
called \emph{abstract state}, such that 
$p(v) \supseteq d(v)^c$ and $d(v) \supseteq V$ for all $v \in \Var$, where $-^c$ is the standard operation for set complement. That is,
\begin{align*}
\aD \defeq \{(p,d,V) 
\in [\Var \to \cP(\Var)]^2 \times \cP(\Var)
\, \mid \,p(v) \supseteq d(v)^c\ \text{and}\ d(v) \supseteq V\ \text{for all } v \in \Var\}.
\end{align*}
We order abstract states as follows:
$(p,d,V) \sqsubseteq (p',d',V')$ if and only if
$V \subseteq V'$ and for all $v \in \Var$,
$p(v) \supseteq p'(v)$ and
$d(v) \subseteq d'(v)$.
These abstract states are concretised by 
$\gamma : \aD \to \cP([\State \rightarrow \State_\bot])$:
\begin{equation}
\label{eqn:definition-gamma}
  f \in 
  \gamma(p,d,V)
  \!\iff\!
  {\!\!}\models \Delta(f,V,\emptyset),\
  {\!\!}\models \Phi(f,p(v),\{v\}),\
  \text{and}\ {\!}\models \Delta(f,d(v),\{v\})\
  \text{for all $v \in \Var$}.
  \makeatletter\displaymath@qed\,
\end{equation}
\end{definition}
\noindent
Note that the definition of $\aD$ contains two conditions. The first condition $p(v) \supseteq d(v)^c$ comes from our assumption that if a function does not depend on a variable $u$, it is smooth in $u$. This and other assumptions of the analysis will be explained shortly in \cref{subsec:soundness-assumptions}. The other condition $d(v) \supseteq V$ originates from the relationship that if $\Delta(f,K,\{v\})$ holds, so does $\Delta(f,K,\emptyset)$.
 
\begin{example}[Differentiability]
  \label{e:2:diff}
  Consider the setup of \cref{e:1:diff} and the program
  $ c \equiv (
    y := x \ast x;\
    \cif\ ( x \geq 0 ) \ \{ s := 1 \}\ 
    \celse \ \{ s := -1 \} ).
    $
  Let $(p,d,V)$ be the smallest abstract state that describes the program.
  In this program, $s$ is not differentiable in $x$, but  $y$ is. So,  $p(s) = \Var \setminus \{x\} \supseteq \{y,s\}$ and $p(y) = \Var \supseteq \{x,y,s\}$. Note that $p(s)$ contains the input variables $s$ and $y$ because by not depending on those input variables, the output $s$ is differentiable in those variables. 
  For the dependency part, we have $d(s)=d(y)=\{x\}$ and $V = \emptyset$.
  \qed
\end{example}

\subsection{Parametric Static Program Analysis}
\label{:2:sa}
Our program analysis is based on abstract interpretation
\cite{cc:popl:77}, and computes an approximation of the concrete semantics
$\db{c}$ of a given command $c$ using the abstract domain $\aD$.
We formalise this computation by the \emph{abstract semantics} of $c$, which defines
$\cdb{c} \in \aD$ by induction on the structure of commands, and
over-approximates $\db{c}$ in the sense of the concretization $\gamma$ in \cref{eqn:definition-gamma}.

\subsubsection{Analysis Definition.} \Cref{f:asem} shows the abstract semantics of $\cdb{c}$. The overall structure of the semantics follows the standard compositional semantics of an imperative language. For instance, the abstract semantics of sequential composition is defined in terms of those of constituent commands, and the semantics of a loop is the least fixed point of a monotone operator over $\aD$. However, the specifics of the semantics include non-standard details, and we spell them out by going through the defining clauses of $\cdb{c}$.

\begin{figure}[t]
  \begin{equation*}
  \hspace{-1.1em}
  \begin{array}{r@{}l}
    \cdb{\cskip}
    & {} \defeq (\lambda v.\Var,\, \lambda v.\{v\},\, \emptyset),
    \\[+1.0ex]
    \cdb{x \,{:=}\, e}
    & {} \defeq  (
    (\lambda v.\,v \,{\equiv}\, x\,?\, \edb{e} : \Var),\,
    (\lambda v.\,v \,{\equiv}\, x\,?\, \fv(e) : \{v\}),\,
    \emptyset),
    \\[+1.0ex]
    \cdb{c;c'}
    & {} \defeq {} 
    \text{let}\
    (p,d,V) \defeq \cdb{c} \ \text{and}\
    (p',d',V') \defeq \cdb{c'}\ \text{in}\
    (p'',d'',V'')
    \\
    \text{where}\
    p''(v) \hspace{-1em}\! & \hspace{1em} {} \defeq {} (V  \cup p_\cap(d'(v))^c \cup d_\cup(p'(v)^c))^c,
    \\
    d''(v) \hspace{-1em}\! & \hspace{1em} {} \defeq {} V \cup d_\cup(d'(v)), \text{ and }
    V'' \defeq V \cup d_\cup(V'),
    \\[+1.0ex]
    \cdb{\cif\,b\,\{c\}\, \celse\,\{c'\}}
    & {} \defeq {}
    \text{let}\ (p,d,V) \defeq \cdb{c}\ \text{and}\ (p',d',V') \defeq \cdb{c'}
    \ \text{in}\ (p'', d'',  V'')
    \\
    \text{where}\ 
    p''(v) \hspace{-1em}\! & \hspace{1em} {} \defeq {} \fv(b)^c\cap p(v) \cap p'(v),
    \\
    d''(v) \hspace{-1em}\! & \hspace{1em} {} \defeq {} \fv(b) \cup d(v) \cup d'(v), \text{ and }
    V'' \defeq \fv(b) \cup V \cup V',
    \\[+1.0ex]
    \cdb{\cwhile\,b\,\{c\}}
    & {} \defeq {}
    \text{let}\ (p,d,V) \defeq \cdb{c}\ \text{in}\ \fix\ \smash{F^\sharp}
    \\
    \text{where}\
    \smash{F^\sharp}(p_0,d_0,V_0) \hspace{-1em}\! & \hspace{1em} {} \defeq {} (p',d',V'),
    \\
    p'(v) \hspace{-1em}\! & \hspace{1em} {} \defeq {} \fv(b)^c \cap
    (
    V \cup 
    p_\cap(d_0(v))^c \cup d_\cup(p_0(v)^c))^c,
    \\
    d'(v) \hspace{-1em}\! & \hspace{1em} {} \defeq {} \fv(b) \cup  (V \cup  d_\cup(d_0(v))) \cup \{v\}, \text{ and }
    V' \defeq \fv(b) \,{\cup}\, (V \,{\cup}\, d_\cup(V_0)),
    \\[+1.0ex]
    \cdb{\cobs(\cnor(e_1,e_2),r)}
    & {} \defeq {} (p,d,\emptyset)\quad
    \\
    \text{where}\ 
    p(v) \hspace{-1em}\! & \hspace{1em} {} \defeq {} (v \equiv \like)\, ?\, \edb{\like \times \cpdfnor(r; e_1, e_2)} : \Var,
    \\
    \text{and}\
    d(v) \hspace{-1em}\! & \hspace{1em} {} \defeq {} (v \equiv \like)\, ?\, \{\like\} \cup \fv(e_1) \cup \fv(e_2) : \{v\},
    \\[+1.0ex]
    \cdb{x\,{:=}\,\csample(n,\cnor(e_1,e_2), \lambda y.e')} & {} \defeq (p,d,\emptyset)
    \quad\text{for $n = \cname(\alpha,r)$ with $r \in \R$}
    \\
    \text{where}\
    \mu \hspace{-1em}\! & \hspace{1em} {} \defeq {} \mathit{create\_name}(\alpha, r),
    \\
    p(v) \hspace{-1em}\! & \hspace{1em} {} \defeq {}
    \begin{cases}
      \edb{e'[\mu/y]}  & \text{if } v \in \{x, \val_\mu\} 
      \\[-0.5ex]
      \edb{\cpdfnor(\mu;e_1,e_2)}  & \text{if } v \equiv \pr_{\mu}
      \\[-0.5ex]
      \edb{\code{\sampled_\mu + 1}}  & \text{if } v \equiv \sampled_{\mu}
      \\[-0.5ex]
      \Var & \text{otherwise},
    \end{cases}
    \\
    \text{and }
    d(v) \hspace{-1em}\! & \hspace{1em} {} \defeq {}
    \begin{cases}
      \fv(e'[\mu/y])  & \text{if } v \in \{x, \val_\mu\}
      \\[-0.5ex]
      \{\mu\} \cup \fv(e_1) \cup \fv(e_2) & \text{if } v \equiv \pr_{\mu}
      \\[-0.5ex]
      \{v\} & \text{otherwise},
    \end{cases}
    \\[-1.0em]
    \\[+1.0ex]
    \cdb{x\,{:=}\,\csample(n, \cnor(e_1,e_2),\lambda y.e')} & {} \defeq (p,d,\emptyset)
    \quad\text{for $n = \cname(\alpha,e)$ with $e \notin \R$}
    \\
    \text{where}\
    p(v) \hspace{-1em}\! & \hspace{1em} {} \defeq {}
    \begin{cases}
      \fv(e)^c \cap \bigcap_{\mu = (\alpha,\_) \in \Name} \edb{e'[\mu/y]}
      & \!\text{if}\ v \equiv x
      \\[-0.5ex]
      \fv(e)^c \cap \edb{e'[\mu/y]}
      & \!\text{if}\ v \equiv \val_\mu\ \text{for $\mu=(\alpha,\_)$}
      \\[-0.5ex]
      \fv(e)^c \cap \edb{\cpdfnor(\mu;e_1,e_2)}
      & \!\text{if } v \equiv \pr_{\mu} \text{ for $\mu=(\alpha,\_)$}
      \\[-0.5ex]
      \fv(e)^c \cap \edb{\sampled_\mu+1}
      & \!\text{if } v \equiv \sampled_{\mu} \text{ for $\mu=(\alpha,\_)$}
      \\[-0.5ex]
      \Var & \!\text{otherwise},
    \end{cases}
    \\
    \text{and }
    d(v) \hspace{-1em}\! & \hspace{1em} {} \defeq {}
    \begin{cases}
      \fv(e) \cup  \bigcup_{\mu=(\alpha,\_) \in \Name} \fv(e'[\mu/y])
      & \!\text{if } v \equiv x
      \\[-0.5ex]
      \fv(e) \cup  \{ \val_\mu \} \cup \fv(e'[\mu/y]) 
      & \!\text{if } v \equiv \val_{\mu} \text{ for $\mu=(\alpha,\_)$}
      \\[-0.5ex]
      \fv(e) \cup \{ \pr_\mu,\mu \} \cup \fv(e_1) \cup \fv(e_2)
      & \!\text{if } v \equiv \pr_{\mu} \text{ for $\mu=(\alpha,\_)$}
      \\[-0.5ex]
      \fv(e) \cup \{ \sampled_\mu \}
      & \!\text{if } v \equiv \sampled_{\mu} \text{ for $\mu=(\alpha,\_)$}
      \\[-0.5ex]
      \{ v \} & \!\text{otherwise}.
    \end{cases}
  \end{array}
  \end{equation*}
  \vspace{-1.5em}
  \caption{Abstract semantics of commands defining $\cdb{c} \in \aD$.}
  \vspace{-0.7em}
  \label{f:asem}
\end{figure}

The definition of $\cdb{\cskip}$ formalises the effect of $\cskip$ on smoothness and dependency. The definition says that $\cskip$ computes each output variable $v$ in a smooth manner in all input variables, and in so doing, it creates the dependency between the variable $v$ to itself at the input state. The $V$ part of $\cdb{\cskip}$ is the empty set since $\cskip$ always terminates. 

The next case is $x := e$. Its abstract semantics records the smoothness and dependency information of the updated variable $x$ by analysing the expression $e$. For the smoothness part, the semantics invokes the subroutine $\edb{e}$ that computes an under-approximation of the set of variables 
in which the expression $e$ is smooth \rev{(and thus over-approximates the smoothness property of $e$)}.%
\footnote{\rev{The subroutine $\edb{e} \subseteq \Var$ is defined inductively on $e$, and differs for different target smoothness properties.
  For instance, if the target property is differentiability,
  we have $\edb{r} \defeq \edb{x} \defeq \Var$, $\edb{e_1 + e_2} \defeq \edb{e_1} \cap \edb{e_2}$, $\edb{\mathrm{ReLU}(e)} \defeq \Var \cap \fv(e)^c$, etc.
  On the other hand, if the target property is local Lipschitzness,
  the subroutine changes for some cases: e.g., $\edb{\mathrm{ReLU}(e)} \defeq \edb{e}$.}}
For the dependency part, $\cdb{x:=e}$ computes the set of all the free variables of $e$ so as to get an over-approximation of all variables that may affect the value of $e$. 
For variables other than $x$, $\cdb{x:=e}$ behaves like $\cdb{\cskip}$.

The abstract semantics of a sequence $c;c'$ composes those of the sub-commands $c$ and $c'$.
It uses the liftings $f_\cup, f_\cap : \cP(\Var) \to \cP(\Var)$
of functions $f$ of type $\Var \to \cP(\Var)$, which are defined as follows:
$f_\cup(V) \defeq \bigcup_{v \in V} f(v)$ and
$f_\cap(V) \defeq \bigcap_{v \in V} f(v)$.
The abstract semantics $\cdb{c;c'}$ constructs the dependency part $d''$ by composing $d$ from $\cdb{c}$ and $d'$ from $\cdb{c'}$ after lifting the former. Note the inclusion of the set $V$ in the definition of $d''$. This is to account for the case that $d'(v)$ in the definition is the empty set; in that case, $d_\cup(d'(v))$ is empty as well and
does not have any information about termination. The smoothness part $p''$ of $\cdb{c;c'}$ is more involved, and implements the intuition described briefly in \cref{:1:abs}.
In order to conclude that input variables in $V_0$ together smoothly affect an output variable $v$ in the computation of $c;c'$, the $p''$ considers the intermediate state after the first command $c$, and forms two groups of variables at that intermediate state: $d'(v)$ and $p'(v)^c$. Note that the desired smoothness property for the input variables in $V_0$ and the output variable $v$ may fail if the first command $c$ uses some variable $u_0 \in V_0$ non-smoothly to update a variable $u'_0$ in $d'(v)$, or it uses some variable $u_1 \in V_0$ to compute the value of a variable $u'_1 \in p'(v)^c$. In the former case, the non-smoothness of $c$ causes an issue, and in the latter case, the non-smoothness of $c'$ causes an issue. The $p''$ collects the input variables that avoid these two failure modes and also do not influence the termination of the sequence. As we show in our soundness theorem, doing so is sufficient because it amounts to using a version of chain rule for the target smoothness property.

The abstract semantics of an if command conservatively assumes that any
 variable in its condition $b$ may affect the value of any output variable (by influencing whether 
the true or false branch of the command gets executed) and this influence is potentially non-smooth.
For every output variable $v$, the smooth set $p''(v)$ for the if command implements this assumption by excluding 
free variables in  $b$, and the computed dependency set does the same but this time by including
free variables in  $b$.

The abstract semantics of a loop computes the least fixed point of a monotone operator $F^\sharp : \aD \to \aD$ using the standard Kleene iteration. 
The operator $F^\sharp$ describes the effect of one iteration of the loop, and it is derived from the standard loop unrolling and our abstract semantics of sequencing and the if command.

The abstract semantics of an observe command $\cobs(\cnor(e_1,e_2),r)$ uses the fact
that the command has the same concrete semantics as the assignment $\like := \like \times \cpdfnor(r;e_1,e_2)$, where $\cpdfnor$ is the density function of the normal distribution. The semantics computes $(p,d,V)$ according to that of the assignment, which we explained earlier.

The final case is the abstract semantics of a sample command. The semantics performs a case analysis on the first argument of $\csample$. If it is a constant expression not involving any variables, then the abstract semantics constructs the name $\mu$ of the sampled random variable, and updates $p$, $d$, and $V$ according to the concrete semantics of the command. Otherwise, the abstract semantics acknowledges that the precise name $\mu$ of the random variable cannot be known statically, and performs so called \emph{weak update} by joining two pieces of information before and after the update of the command in the concrete semantics. Note that the abstract semantics does not require the third argument of $\csample$ should be the identity function. The ability of dealing with a general function in the third argument is needed since our analysis is intended to be applied to programs after the transformation of the SPGE, which may introduce such an argument. 

The abstract semantics is well-defined under the following relatively weak assumption:
\begin{assumption}[Expression analysis and free variables]
\label{assm:expression-analysis-fv}
$\edb{e} \supseteq \fv(e)^c$ for all expressions $e$.
\end{assumption}
\noindent
This assumption is satisfied by the instantiations of the semantics with differentiability and local Lipschitzness, which are used in our implementation. It will be assumed in the rest of the paper.
\begin{theorem}
\label{thm:abstract-semantics-well-formedness}
    If \cref{assm:expression-analysis-fv} holds, then
    for all commands $c$, we have $\cdb{c} \in \aD$, that is, when we let $(p,d,V) \defeq \cdb{c}$, we have $p(v) \supseteq d(v)^c$ and $d(v) \supseteq V$ for
    all variables $v \in \Var$.
\end{theorem}

\begin{example}[Differentiability]
  Consider the differentiability property and the example program of \cref{e:2:diff}. 
  Let $(p_1,d_1,V_1)$ and $(p_2,d_2,V_2)$ be the results of analysing the first assignment command $y := x\ast x$ 
  and the following if command of the program. Then,
  \[
  (p_1,d_1,V_1) \,{=}\, (\lambda v.\Var,\, \lambda v.\, (v {\equiv} y)\,?\, \{x\} \,{:}\, \{v\},\, \emptyset),
  \ \ 
  (p_2,d_2,V_2) \,{=}\, (\lambda v.\{x\}^c,\, \lambda v.\, (v{\equiv} s)\, ?\, \{x\} \,{:}\, \{x,v\},\, \{x\}).
  \]
  Let $(p,d,V)$ be the analysis result for the entire program. Then,
  \begin{align*}
  p(v) & = \big(V_1 \cup (p_1)_\cap(d_2(v))^c \cup (d_1)_\cup(p_2(v)^c)\big)^c
  = d_1(x)^c = \{x\}^c,
  & 
  V & = V_1 \cup (d_1)_\cup(V_2) = \{x\}.
  \end{align*}
  Also, $d(v) = V_1 \cup (d_1)_\cup(d_2(v)) = \rev{\{x\}}$.
  As shown in \cref{f:asem}, the variable $x$ that may affect the condition
  expression of the if command is removed from the smoothness sets, and
  $p(s)=p(y)=\{x\}^c$. Note that this result is conservative with respect to $y$.
  \qed
\end{example}

\subsubsection{Analysis Soundness and Assumptions.}
\label{subsec:soundness-assumptions}
The soundness of our analysis states that for every command $c$, 
its abstract semantics $\cdb{c} \in \aD$ over-approximates the concrete semantics $\db{c}$ via $\gamma$:
$\db{c} \in \gamma( \cdb{c} )$.
The soundness is conditioned on \cref{assm:expression-analysis-fv} and six new assumptions.
\rev{One of the new assumptions is about the soundness of $\edb{e}$. The remaining five assumptions} are concerned with
the predicate family for the target smoothness property $\phi = (\phi_{K,L} : K,L \subseteq \Var)$, and say that certain 
canonical operators are smooth according to $\phi$ so that using them in the abstract semantics should not cause an issue.
In this subsection, we present the six assumptions one by one, and sketch how those assumptions are related to the soundness.

We start with the assumption that the analysis of each expression $\edb{e}$
 under-approximates the set of variables in
 which the evaluation of $e$ is smooth \rev{(and thus over-approximates the smoothness property of $e$)}.
\begin{assumption}[Expression analysis soundness]
  \label{assm:expression-analysis-soundness}
  For all expressions $e$, variables $x$, subsets $K,L \subseteq \Var$, and states $\tau \in \State[\Var \setminus K]$ such that $K = \edb{e}$ and $L = \{x\}$, if we let $g : \State[K] \to \State[L]$ be the function defined by 
 $g(\sigma) \defeq [x\mapsto \db{e}(\sigma \oplus \tau)]$,
  the function $g$ satisfies $\phi_{K,L}$ (i.e., $g \in \phi_{K,L}$).
\end{assumption}
\noindent{This assumption is used in our soundness argument whenever the abstract semantics uses $\edb{e}$ for computing smoothness information about an expression $e$.} 

The next two assumptions assert the smoothness of the standard operators on the product spaces.
\begin{assumption}[Projection]
\label{assm:projection}
  For all $K,L \subseteq \Var$ with $K \supseteq L$, the projection $\pi_{K,L}$
  satisfies $\phi_{K,L}$.
\end{assumption}
\begin{assumption}[Pairing]
\label{assm:pairing}
  For all $K,L,M \subseteq \Var$ with $L \cap M = \emptyset$,
  if $f \in \phi_{K,L}$ and $g \in \phi_{K,M}$, we have $\langle f, g \rangle \in \phi_{K, L \cup M}$,
  where $\langle f, g \rangle$ is the pairing of two partial functions:
   $\langle f, g\rangle(\sigma) \defeq 
    \text{\rm if}\ (\sigma \in \dom(f) \cap \dom(g))\ 
    \text{\rm then}\ 
    f(\sigma) \oplus g(\sigma)\
    \text{\rm else}\
    \text{\rm undefined}$.
\end{assumption}
\noindent{Note that $\State[L \cup M]$ is isomorphic to $\State[L] \times \State[M]$, the product space that we referred to above.} The assumptions say that the projection  is smooth, and the pairing of smooth functions is smooth.
Our analysis uses \cref{assm:projection} to deal with variables not modified by a command. 
For instance, when analysing an assignment $x:=e$, the analysis uses \cref{assm:projection} and concludes that on every output variable $v$ other than $x$, the assignment is smooth in all the input variables. 
\cref{assm:pairing} is used to justify the handling of a sequence $c;c'$ by our analysis, in particular, the part that the analysis combines smoothness information over multiple output variables after the first command~$c$.

  The projection and pairing assumptions are about how shrinking and expanding output variables affect the target smoothness property. 
  The next restriction assumption is about shrinking the input variables. 
  It validates the weakening of the $K$ part of ${}\models\Phi(f,K,L)$, and is used in the abstract semantics of 
  $c;c'$ (and other composite commands).
\begin{assumption}[Restriction]
  \label{assm:restriction}
  For all $K,K',L \subseteq \Var$ with $K \supseteq K'$, and $\tau \in \State[K \setminus K']$,
  if $f \in \phi_{K,L}$, then we have $g \in \phi_{K',L}$, where
    $g(\sigma) \defeq
    \text{\rm if}\ (\sigma \oplus \tau \in \dom(f))\ 
    \text{\rm then}\  f(\sigma \oplus \tau)\
    \text{\rm else}\ \text{\rm undefined}$.
\end{assumption}

The following assumption says that the function composition preserves smoothness. It is related to the chain rule for differentiation, and used to justify the abstract semantics of a sequence $c;c'$.
\begin{assumption}[Composition]
  \label{assm:composition}
  For all $K,L,M \subseteq \Var$, if $f \in \phi_{K,L}$ and $g \in
  \phi_{L,M}$, we have $g \circ f \in \phi_{K,M}$, where $g \circ f$ is
  the standard composition of two partial functions: 
  $(g \circ f)(\sigma) \defeq
  \text{\rm if}\ 
  (\sigma \in \dom(f) \land f(\sigma) \in \dom(g))\ 
  \text{\rm then}\
  g(f(\sigma))\
  \text{\rm else}\ 
  \text{\rm undefined}$.
\end{assumption}

The final assumption lets the analysis infer smoothness information about the completely-undefined function. It is used to justify the handling of loops by our analysis.
\begin{assumption}[Strictness]
  \label{assm:strictness}
  For all $K,L \subseteq \Var$,
  we have  $(\lambda \sigma \in \State[K].\text{undefined}) \in \phi_{K,L}$.
\end{assumption}

\begin{theorem}[Soundness]
  \label{thm:soundness-analysis}
  If \cref{assm:expression-analysis-fv,assm:expression-analysis-soundness,assm:projection,assm:pairing,assm:restriction,assm:composition,assm:strictness} hold,
  the analysis computes the sound abstraction of the concrete semantics of commands in the following sense:
  for all commands $c$, \[\db{c} \in \gamma( \cdb{c} ).\]
\end{theorem}

\begin{remark}
\label{remark:admissibility}
A standard method for proving a property of a loop or more generally a recursively defined function is so called Scott induction. In this method, we view a property as a set $\cT$ of state transformers and a loop as the least fixed point of a continuous function $F$ on state transformers. Then, we prove the three conditions: (i) $\cT$ contains the least state transformer, (ii) it is closed under the least upper bound of any increasing sequence of state transformers, and (iii) $\cT$ is preserved by $F$. The first and second conditions are called strictness and admissibility, respectively, and these three conditions imply that the least fixed point of $F$ belongs to $\cT$.

Our soundness proof for the loop case deviates slightly from this standard method. If it followed the method instead, we would need, in addition to the strictness assumption, the following assumption, which corresponds to the second admissibility condition:
\begin{assumption}[Admissibility]
\label{assm:admissibility} 
Let $K,L \subseteq \Var$, and order partial functions in $[\State[K] \rightharpoonup \State[L]]$ by the inclusion of the graphs of partial functions. Then, for every increasing sequence $\{f_n : \State[K] \rightharpoonup \State[L]\}_{n \in \N}$ (i.e., the graph of $f_{n+1}$ includes that of $f_n$ for all $n \in \N$), 
if every $f_n$ satisfies $\phi_{K,L}$, so does the least upper bound $f_\infty$ of the sequence (defined by its graph being the union of the graphs of all $f_n$'s).
\end{assumption}
\noindent
The inclusion of this admissibility assumption would, then, limit the applicability of our program analysis, since some well-known smoothness properties, such as (global) Lipschitz continuity and local boundedness, do not satisfy the assumption, although they satisfy our five assumptions
(\cref{assm:projection,assm:pairing,assm:restriction,assm:composition,assm:strictness}) \rev{(see \cref{tab:target-prop-list})}. On the plus side, the inclusion of the admissibility assumption could enable our analysis to handle loops more accurately, possibly by tracking the impact of the boolean condition of each loop on smoothness more precisely.
\rev{Our soundness proof avoids the admissibility assumption by exploiting the fact that our analysis handles loop conditions conservatively:
  our analysis drops all the variables that loop conditions may depend on from the set of smooth variables,
  so that it avoids finding too precise inductive predicates that can break soundness.%
  \footnote{\rev{To be precise, our analysis does not require the admissibility assumption,
    not because our abstract domain is finite,
    but because given a loop, our analysis finds an inductive predicate that is ``sufficiently admissible'' in the sense that
    it is closed under the least upper bound of any chain {\em that matters for soundness}:
    the chain should be definable by some loop.
    More details are in \cref{sec:pf:thm:soundness-phi}.}}
}
\qed
\end{remark}

%
%

\subsection{Instantiations}
\label{:4:inst}

Our program analysis requires that the family of smoothness predicates should satisfy  \cref{assm:projection,assm:pairing,assm:restriction,assm:composition,assm:strictness}. Although these
assumptions are violated by some smoothness properties, such as partial differentiability and partial continuity, they are met by our leading example 
$\smash{\phi^{(d)}}$ for differentiability (\cref{e:1:diff}), and also by the predicate family $\smash{\phi^{(l)}}$ for 
local Lipschitzness, which is used in our implementation. Recall the definitions of the predicate
families $\smash{\phi^{(d)}}$ and $\smash{\phi^{(l)}}$: for all $K,L \subseteq \Var$, 
\begin{align*} 
    \phi^{(d)}_{K,L}
    & \defeq \{f : \State[K] \rightharpoonup \State[L]
    \mid \dom(f) \text{ is open and } f \text{ is (jointly) differentiable in its domain}\},
    \\[-0.3em]
    \phi^{(l)}_{K,L}
    & 
    \begin{aligned}[t] 
    {} \defeq \{f & : \State[K] \rightharpoonup \State[L]
    \mid {} \dom(f) \text{ is open, and for all $\sigma \in \dom(f)$, there are $C > 0$ and}
               \\[-0.5ex]
               & \text{an open $O \subseteq \dom(f)$  s.t. $\sigma \in O$ and $\norm{f(\sigma_0) - f(\sigma_1)} \leq C \norm{\sigma_0 - \sigma_1}$ for all $\sigma_0, \sigma_1 \in O$}\}. 
    \end{aligned}
\end{align*}
\begin{theorem}
\label{thm:families-satisfy-assumptions}
      Both $\smash{\phi^{(d)}}$ and $\smash{\phi^{(l)}}$ satisfy \cref{assm:projection,assm:pairing,assm:restriction,assm:composition,assm:strictness}. 
\end{theorem}

\rev{%
\vspace{-2mm}
\begin{remark}
The requirement of open domain in $\smash{\phi^{(d)}}$ is sometimes too constraining and hurts the accuracy of the analysis. It can, however, be relaxed, and we can generalise $\smash{\phi^{(d)}}$ to the following predicate family $\smash{\phi^{(d')}}$, which corresponds to the standard definition of differentiability on a manifold {\em with boundary} in differential geometry~\cite[Chapter~2]{lee03}: 
\begin{align*}
   \phi^{(d')}_{K,L}
   &
   \begin{aligned}[t]
   \defeq \{f & : \State[K] \rightharpoonup \State[L]
   \,\mid\, {} \text{for all $\sigma \in \dom(f)$, there exist an open $U \subseteq \State[K]$ and $g : U \to \State[L]$}
              \\[-0.5ex]
              & \text{such that $\sigma \in U$, $f = g$ on $U \cap \dom(f)$, and $g$ is (jointly) differentiable}\}.
   \end{aligned}
\end{align*}
Note  the weakening of open-domain requirement in $\smash{\phi^{(d')}}$: the open domain $U$ in the above definition
does not have to be included in $\dom(f)$. The family $\smash{\phi^{(d')}}$ satisfies \cref{assm:projection,assm:pairing,assm:restriction,assm:composition,assm:strictness}, and can lead to a more permissive instantiation of our program analysis than the family $\smash{\phi^{(d)}}$, especially in handling atomic commands, such as assignment, sample, and observe.
We point out that $\smash{\phi^{(d')}}$ does not satisfy \cref{assm:admissibility} (i.e., the admissibility assumption), while $\smash{\phi^{(d)}}$ does satisfy it. As a result, if the handling of loops in our analysis is changed such that loop conditions are analysed more accurately, the analysis may remain sound only for  $\smash{\phi^{(d)}}$, not for $\smash{\phi^{(d')}}$, as explained in \cref{remark:admissibility}. 
\qed
\end{remark}}

\begin{table*}[t]
  \centering
  \renewcommand{\arraystretch}{.85}
  \aboverulesep=0.2ex
  \belowrulesep=0.2ex
  \tabcolsep=1.5ex
  \small
  \caption{%
    Failure cases of the composition assumption.
    For each given $\phi$, we have $f,g \in \phi$ but $g \circ f \notin \phi$,
    where $f$ and $g$ are interpreted as (total or partial) functions from $\State[K]$ to $\State[L]$. 
    Let $c_1 \equiv (y=x^2; z=g(y))$ and $c_2 \equiv (y=x; z=g(x,y))$.
    Then, for each $i$-th $\phi$, $\cdb{c_i}$ incorrectly concludes that $z$ is smooth with respect to $x$.
  }
  \label{tab:phi-counter-eg}
  \vspace{-0.7em}
  \begin{tabular}{lll}
    \toprule
    \multicolumn{1}{l}{$\phi$} & \multicolumn{1}{l}{$f$} & \multicolumn{1}{l}{$g$}
    \\ \midrule
    $\phi^{(d'')}_{K,L}$
    & \begin{tabular}{@{}c@{}}$f(x) = x^2$ defined on $\R$\end{tabular}
    & \begin{tabular}{@{}c@{}}$g(x) = \ind{x>0}$ defined on $[0,1]$\end{tabular}
    \\
    $\phi^{(\mathit{pd})}_{K,L}$,
    $\phi^{(\mathit{pc})}_{K,L}$
    & \begin{tabular}{@{}c@{}}$f(x) = (x,x)$ defined on $\R$\end{tabular}
    & \begin{tabular}{@{}c@{}}$g(x,y) = 
        \Big\{\!\!
        \begin{array}{l@{\;\;}l}
          xy/(x^2+y^2) & \text{if $(x,y) \neq (0,0)$}
          \\
          0 & \text{otherwise}
        \end{array}
        $ defined  on $\R^2$\end{tabular}
    \\ \bottomrule
  \end{tabular}
\end{table*}

\vspace{-2mm}
\begin{remark}
\label{remark:phi-examples-nonexamples}
At this point, the reader might feel that \cref{assm:projection,assm:pairing,assm:restriction,assm:composition,assm:strictness} are satisfied by nearly all smoothness properties. This impression is not accurate. For instance, the composition assumption does not hold for the notions of differentiability of partial functions formalised by the following $\smash{\phi^{(d'')}}$ and $\smash{\phi^{(\mathit{pd})}}$, nor for the partial continuity formalised by $\smash{\phi^{(\mathit{pc})}}$:
  \begin{align*}
    \phi^{(d'')}_{K,L}
    & 
    \defeq \{f : \State[K] \rightharpoonup \State[L] \mid \text{$f$ is (jointly) differentiable in the interior of its domain}\},
    \\[-0.3em]
    \phi^{(\mathit{pd})}_{K,L}
    &
    \defeq \{f : \State[K] \rightharpoonup \State[L] \mid 
    \text{$\dom(f)$ is open, and for all $v \in K$, $f$ is partially differentiable in $v$}\},
    \\[-0.3em]
    \phi^{(\mathit{pc})}_{K,L}
    &
    \defeq \{f : \State[K] \rightharpoonup \State[L] \mid
    \text{$\dom(f)$ is open, and for all $v \in K$, $f$ is partially continuous in $v$}\}. 
  \end{align*}
  \Cref{tab:phi-counter-eg} contains counterexamples that show the failure of the composition assumption for these predicate families. In fact, when instantiated with these families, our program analysis is not sound. The same table shows example programs and incorrect conclusions derived by our analysis.

  \rev{\cref{tab:target-prop-list} shows more target smoothness properties from mathematics,
    and whether each property satisfies \crefrange{assm:projection}{assm:strictness} (and \cref{assm:admissibility}).
    Recall that our program analysis does not require \cref{assm:admissibility} for soundness; the table shows the assumption just for reference.
    The target properties from ``cont.'' to ``real analytic'' in the table (and three more) satisfy \crefrange{assm:projection}{assm:strictness},
    so that our analysis can be immediately applied to those target properties while remaining sound.}
  \qed
\end{remark}

\begin{table*}[t]
  \centering
  \renewcommand{\arraystretch}{.80}
  \aboverulesep=0.25ex
  \belowrulesep=0.25ex
  \tabcolsep=0.5ex
  \small
  \caption{\rev{%
      Various well-known target smoothness properties from mathematics, and whether each of them satisfies \crefrange{assm:projection}{assm:strictness}
      (and \cref{assm:admissibility}).
      Here ``cont.'' and ``diff.'' denote ``continuous'' and ``differentiable''.
      The properties above the double horizontal line are defined such that they include the open-domain requirement
      as in $\smash{\phi^{(d)}}$ and $\smash{\phi^{(l)}}$,
      and the properties below the line are defined without the open-domain requirement.}}
  \label{tab:target-prop-list}
  \vspace{-0.8em}
  \rev{%
  \newcolumntype{X}{>{\centering\arraybackslash\hspace{0pt}}p{39pt}}
  \setlength\doublerulesep{1.0pt}
  \begin{tabular}{lXXXXX|X}
    \toprule
    Target smoothness property  & A3 (proj.) & A4 (pair.) & A5 (rest.) & A6 (comp.) & A7 (stri.) & A8 (admi.)
    \\
    \midrule
    cont. ($\smash{\mcC^0}$)                  & \oo & \oo & \oo & \oo & \oo & \oo \\
    locally Lipschitz ($=\smash{\phi^{(l)}}$) & \oo & \oo & \oo & \oo & \oo & \oo \\
    uniformly cont.                           & \oo & \oo & \oo & \oo & \oo & \xxadm \\
    Lipschitz cont.                           & \oo & \oo & \oo & \oo & \oo & \xxadm \\
    \midrule
    diff. ($=\smash{\phi^{(d)}}$)             & \oo & \oo & \oo & \oo & \oo & \oo \\
    continuously diff. ($\smash{\mcC^1}$)     & \oo & \oo & \oo & \oo & \oo & \oo \\
    smooth ($\smash{\mcC^\infty}$)            & \oo & \oo & \oo & \oo & \oo & \oo \\
    real analytic ($\smash{\mcC^\omega}$)     & \oo & \oo & \oo & \oo & \oo & \oo \\
    \midrule
    partially cont. ($=\smash{\phi^{(pc)}}$)  & \oo & \oo & \oo & \xx & \oo & \oo \\
    partially diff. ($=\smash{\phi^{(pd)}}$)  & \oo & \oo & \oo & \xx & \oo & \oo \\
    \midrule \midrule
    almost-everywhere cont.                   & \oo & \oo & \xx & \xx & \oo & \oo \\
    almost-everywhere diff.                   & \oo & \oo & \xx & \xx & \oo & \oo \\
    coordinatewise non-decreasing             & \oo & \oo & \oo & \oo & \oo & \oo \\
    \midrule
    locally bounded                           & \oo & \oo & \oo & \oo & \oo & \xxadm \\
    bounded                                   & \xx & \oo & \oo & \oo & \oo & \xxadm \\
    \midrule
    Borel measurable                          & \oo & \oo & \oo & \oo & \oo & \oo \\
    locally integrable                        & \oo & \oo & \xx & \xx & \oo & \xxadm \\
    integrable                                & \xx & \oo & \xx & \xx & \oo & \xxadm \\
    \bottomrule
  \end{tabular}}
  \vspace{-0.4em}
\end{table*}


\section{Algorithm for the SPGE Variable-Selection Problem}
\label{sec:var-sel-alg}

We now put together the results from \cref{sec:generalised-gradient-estimator-new} and \cref{s:sa:new}
to formally define and soundly (yet approximately) solve the SPGE variable-selection problem. 
We start with the formal definition of the problem:
\begin{definition}[SPGE Variable-Selection Problem; Formal]
  \label{def:var-sel-prob-formal}
  Assume we are given a model $c_m$, a guide $c_g$, and a (initial) simple reparameterisation plan $\pi_0$ such that
  $c_m$, $c_g$, and $\ctr{c_g}{\pi}$ always terminate and have no double-sampling errors for all $\pi \sqsubseteq \pi_0$. 
  Here we write $\pi \sqsubseteq \pi'$ if the graph of $\pi$ is included in that of $\pi'$.
  Given these $c_m$, $c_g$, and $\pi_0$,
  find a reparameterisation plan $\pi \sqsubseteq \pi_0$ such that
  (i) $\pi$ is simple and satisfies \cref{cond:dens-diff,cond:val-diff} in \cref{sec:grad-estm-prog-trans}, 
  and (ii) $|\repname(\pi)|$ is maximised.
  We say that $\pi$ is a {\em sound solution} if it satisfies (i), and an {\em optimal solution} if it satisfies (i) and (ii).
  \qed
\end{definition}

The input $\pi_0$ in the problem is a newcomer. It fixes a semantics-preserving transformation for all the sample commands. Typically, $\pi_0$ is defined on the entire $\NameEx \times \DistEx \times \LamEx$, and remains fixed across all input model-guide pairs $(c_m,c_g)$. More importantly, it is valid so that the change of any sample command by $\pi_0$ preserves the semantics of the command when we take into account both the second distribution argument and the third lambda argument of the sample command. The validity of $\pi_0$ is inherited by any sound solution $\pi$ of the SPGE variable-selection problem since validity as a property on reparameterisation plans is down-closed with respect to the $\sqsubseteq$ order. In our setup, $\pi_0$ is fixed to be the following reparameterisation plan from \cref{sec:prog-trans}:
\begingroup
\addtolength{\abovedisplayskip}{-0.5pt}
\addtolength{\belowdisplayskip}{-0.5pt}
\begin{align}
\label{eqn:initial-reparameterisation-plan}
  \pi_0(n, \cnor(e_1,e_2), \lambda y.e_3) & \defeq (\cnor(0,1), \lambda y. e_3[(y \times \sqrt{e_2} + e_1) / y])
\end{align}
\endgroup
for all $n \in \NameEx$ and expressions $e_1$, $e_2$, and $e_3$. 

As an example of the SPGE variable-selection problem, consider the problem for the $\pi_0$ in \cref{eqn:initial-reparameterisation-plan} and the model-guide pair $(c_m, c_g)$ given in \cref{fig:eg-model-guide},
where $\cstr{z_i}$ in the figure is interpreted as $\cname(\cstr{z_i}, 0)$.
Then, as discussed in \cref{sec:overview}, the problem has the following optimal solution:
$\pi \defeq \pi_0|_{S \,\times\, \DistEx \,\times\, \LamEx}$ for $S \defeq \{\cname(\alpha,e) \in \NameEx \mid \alpha \not\equiv \cstr{z_2}\}$.

We present an algorithm for computing a sound (yet possibly suboptimal) solution to the problem.
\begin{enumerate}
\item
By running our program analysis instantiated with differentiability (described in \cref{:2:sa,:4:inst}),
compute $(\mathbb{p}_m, \mathbb{d}_m, \mathbb{V}_m) \defeq \cdb{c_m}$ and $(\mathbb{p}_g, \mathbb{d}_g, \mathbb{V}_g) \defeq \cdb{c_g}$,
where we use $\mathbb{p}$, $\mathbb{d}$, and $\mathbb{V}$ for the output of the analysis
to distinguish them from densities $p$ and distributions $d$.

\item Using $\mathbb{p}_m$ and $\mathbb{p}_g$, check
\begingroup
\addtolength{\abovedisplayskip}{-2.5pt}
\addtolength{\belowdisplayskip}{-1pt}
\begin{align}
  \label{eq:var-sel-alg-check1}
  \theta \subseteq K,
  \quad\text{where}\
  K \defeq \mathbb{p}_m(\like) \cap \bigcap_{\mu \in \Name} \mathbb{p}_m(\pr_\mu) \cap \bigcap_{\mu \in \Name} \mathbb{p}_g(\pr_\mu).
\end{align}
\endgroup
If the check fails, return an error message that our analysis cannot discharge \cref{cond:dens-diff} for {\em any} $\pi$,
since the analysis concludes that the density function of $c_m$ or $c_g$ can be non-differentiable in $\theta$ (even when $\repname(\pi)=\emptyset$).
If the check passes, initialise the set of reparameterised random variables by 
\begingroup
\addtolength{\abovedisplayskip}{-0.5pt}
\addtolength{\belowdisplayskip}{-0.5pt}
\[
S \defeq \{(\alpha, i) \in \Name \,\mid\,
\rev{\text{for all}\ i' \in \N,\ (\alpha, i') \in \Name \implies (\alpha, i') \in K}\}.
\]
\endgroup

\item\label{alg:variable-selection-3}  
Using $S$ and $\pi_0$, construct a reparameterisation plan $\pi \sqsubseteq \pi_0$ by $\pi \defeq \pi_0[S]$, where
$\pi_0[S](n,d,l)$ is $\pi_0(n,d,l)$ if $(n,d,l) \in \dom(\pi_0)$, $n = \cname(\alpha, \_)$, and $(\alpha, \_) \in S$; otherwise,
it is undefined.

\item By running the differentiability analysis  on $\ctr{c_g}{\pi}$,
compute $\smash{ (\overline{\mathbb{p}_g}, \overline{\mathbb{d}_g}, \overline{\mathbb{V}_g}) } \defeq \cdb{\ctr{c_g}{\pi}}$
and check
\begin{align}
  \label{eq:var-sel-alg-check2}
  \theta \subseteq \bigcap_{{ \mu \in \Name }} \overline{\mathbb{p}_g}(\pr_\mu) \cap \bigcap_{\mu \in \Name} \overline{\mathbb{p}_g}(\val_\mu).
\end{align}
If the check passes, return $\pi$ as the output of the algorithm.
If not, update $S$ by $S \setminus \{(\alpha, i) \in \Name\}$ after choosing some $(\alpha, \_) \in S$,
and then repeat the above procedure (from the step \eqref{alg:variable-selection-3}, the point where we construct $\pi$ using $S$) until $S$ becomes empty.
\end{enumerate}

Our algorithm computes a sound solution \rev{(in the sense stated in \cref{def:var-sel-prob-formal})},
partly because of the soundness of our program analysis:
\begin{theorem}
  \label{thm:soundness-reparam}
  Let $c_m$, $c_g$, and $\pi_0$ be the inputs to the SPGE variable-selection problem.
  If the above algorithm returns $\pi$ for $(c_m,c_g,\pi_0)$, then $\pi$ is a sound solution for the problem.
\end{theorem}
\noindent
\rev{We point out that the theorem holds for the local Lipschitzness case as well (in addition to the differentiability case).
  That is, if the above algorithm runs our program analysis instantiated with local Lipschitzness (instead of differentiability),
  and if it returns an output $\pi$,
  then $\pi$ is a sound solution to the SPGE variable-selection problem
  that uses \cref{cond:dens-lips,cond:val-lips} (instead of \cref{cond:dens-diff,cond:val-diff}).}

Our algorithm solves the problem only approximately:
there is no formal guarantee that it always computes an optimal solution.
The suboptimality may arise due to two approximations:
the overapproximation of our program analysis when it computes differentiability \rev{(or local Lipschitzness)} information,
and the heuristic choices made by our algorithm when the algorithm computes the random-variable set $S$.
We demonstrate, however, that our algorithm finds optimal solutions for all the benchmarks in~\cref{sec:impl-eval}.

Our algorithm calls our program analysis at most $|\{ \alpha \in \Str \mid (\alpha, \_) \in S_0 \}|+2$ times,
where $S_0$ is the initial value of $S$ (i.e., the set of random variables whose sample commands are to be transformed) in the algorithm.
However, for all the benchmarks in \cref{sec:impl-eval},
our algorithm terminated with the initial set $S_0$ and thus called our analysis only $3$ times
(on the model, the guide, and the reparameterised guide according to $S_0$).
\rev{Based on these results and our intuition on the algorithm,
we conjecture that our algorithm always terminates with the initial set $S_0$ under a mild condition on $\pi_0$ and our analysis
(e.g., $\edb{e}$ is computed inductively on $e$).
Since the conjecture is still open, our algorithm might not succeed in the first iteration,
and if so, it continues to search for a sound solution greedily.
Note that there are many other ways to continue the search and our algorithm uses just one of them (as it is linear-time).}

\section{Experimental Evaluation}
\label{sec:impl-eval}


In our experiments, we consider two research questions.
First, can the analysis proposed in \cref{s:sa:n} be instantiated
  and implemented so that it can produce meaningful smoothness results on
  real-world probabilistic programs?
Second, can the algorithm proposed in \cref{sec:var-sel-alg} 
  find near-optimal solutions to the SPGE variable-selection problem on real-world probabilistic programs?
To assess the two questions, we have implemented a static smoothness analyser for Pyro programs based on \cref{s:sa:n},
and a variable selector based on \cref{sec:var-sel-alg} which (approximately) solves the variable-selection problem.%
\footnote{\rev{Our implementation is available at \url{https://github.com/wonyeol/smoothness-analysis}.}}
Our analyser and variable selector are implemented in OCaml, and support a subset of the Pyro PPL and two smoothness properties:
differentiability and local Lipschitzness. 

\vspace{2mm}
\noindent{\bf Implementation.}\ 
Although the analysis described in \cref{s:sa:n} may look simple when considering a basic PPL,
real-world PPLs such as Pyro are of a much higher degree of complexity.
First, they provide a large panel of continuous/discrete probability distributions for sample and observe commands,
and library functions for tensors and neural networks.
Second, programs in real-world PPLs may fail to be smooth for reasons other
than if-else and while commands.
In particular, values sampled from discrete distributions,
and arguments to operators and distribution constructors that are well-defined only on a strict subset of values, may induce non-smoothness.
A straightforward treatment of these will result in an overly
conservative analysis, treating far too many variables as potentially
non-smooth.
Third, Pyro programs typically rely on tensors (of large, statically unknown size) to deal with large datasets,
and it is generally infeasible to reason about each (real-valued) element of tensors individually.
In the following, we discuss how our static analyser addresses these issues and provides sound, useful information
about smoothness of Pyro programs. 

\vspace{1mm}
\noindent{\it Distributions and library functions.}
Our analyser supports 17 distributions (continuous or discrete). 
Each distribution is characterized by a pair $(b,a)$ for a boolean $b$ and an array of booleans $a$,
where $b$ (or $a_i$) denotes whether its probability density is differentiable or locally Lipschitz
with respect to the sampled value (or the $i$-th argument) of the distribution.
For example, a normal distribution is described by \mytt{(true,[true,true])} (assuming that the second argument is positive) and a Poisson distribution by \mytt{(false,[true])}.
Similarly, the analyser supports a large number of PyTorch/Pyro library
functions for tensors and neural networks, and assumes the correct smoothness information about them.
For instance, the ReLU function is considered locally Lipschitz but not differentiable.

 \vspace{1mm}
\noindent{\it Refining smoothness information based on safety pre-analysis.}
Although the expression \mytt{x/y} is generally non-smooth with respect to
\mytt{y} (even if it is well-defined for \mytt{y=0}), if more information is available, for instance that \mytt{y}
always lies in range $[1,10]$, we can safely consider it smooth with respect to
both \mytt{x} and \mytt{y}.
Likewise, the density of a normal distribution is generally non-smooth with respect to the standard deviation argument $\sigma$
(even if it is well-defined for $\sigma \leq 0$),
so more precise smoothness information can be produced when $\sigma$ is known to be always positive.
Thus, establishing precise smoothness information requires to first
establish safety properties related to program operations.
To achieve this, our tool actually performs two analyses in sequence:
(i) a safety pre-analysis infers ranges over all numerical variables
  and marks each argument to an operator or a distribution constructor
  as either ``safe'' or ``potentially unsafe'';
(ii) the program analysis formalised in \cref{s:sa:n}
  utilises information computed in the first phase to produce precise
  smoothness information.
The first phase boils down to a forward abstract interpretation based on
basic abstract domains like intervals and signs~\cite{cc:popl:77}.
It logs safety information for each program statement just like static
analyses for runtime errors and undefined behaviors~\cite{astree:pldi03}.
As formalised in \cref{:2:sa}, the second analysis is compositional.
Due to their different nature, the two analyses need to be done in
sequence.

\vspace{1mm}
\noindent{\it Tensors.}
Pyro programs commonly use nested loops and indexed tensors.
\rev{As the number of scalar values in each tensor is often statically unknown (or known but huge),
treating each scalar as a separate variable is not feasible; so we rely on a conservative summarisation of tensors.
Intuitively, we treat all scalars in a tensor as a ``block'':
e.g., when density might not be smooth with respect to some scalar(s) of a tensor, 
the analysis conservatively concludes that it might not be smooth with respect to all scalars in the tensor.}
In our experiments, this abstraction does not result in any precision loss.

\begin{table*}[t]
  \centering
  \renewcommand{\arraystretch}{.85}
  \aboverulesep=0.2ex
  \belowrulesep=0.2ex
  \small
  \caption{%
    Subset of Pyro examples used in experiments and their key features
    (see \cref{sec:impl-eval-more} for the rest).
    The last five columns show the total number of code lines (excluding comments),
    loops, sample commands, observe commands, and learnable parameters
    (declared explicitly by \mytt{pyro.param} or implicitly by a neural network module).
    Each number is the sum of the counts in the model and guide.
  }
  \label{t:examples}
  \vspace{-1em}
  \begin{tabular}{llrrrrr}
    \toprule
    Name & Probabilistic model & LoC
    & $\cwhile$ & $\csample$ & $\cobs$ & \mytt{param}
    \\ \midrule
    \mytt{spnor} & Splitting normal example in \cref{fig:eg-model-guide}
    & 16 
    & 0 & 2 & 1 & 2 
    \\
    \mytt{sgdef} & Deep exponential family
    & 105 
    & 0 
    & 12 & 1 & 12 
    \\
    \mytt{dmm} & Deep Markov model
    & 112 
    & 3 
    & 2 & 1 & 13 
    \\
    \mytt{mhmm} & Hidden Markov models
    & 137 
    & 1 & 5 & 5 & 12
    \\
    \mytt{scanvi} & Single-cell annotation using variational inference
    & 147 
    & 0 & 7 & 2 & 21
    \\
    \mytt{air} & Attend-infer-repeat
    & 174 
    & 2 
    & 6 & 1 & 16 
    \\
    \mytt{cvae} & Conditional variational autoencoder
    & 205 
    & 0 & 2 & 1 & 15 
    \\ \bottomrule
  \end{tabular}
  \vspace{-0.7em}
\end{table*}

\vspace{2mm}
\noindent{\bf Evaluation.}\
We evaluated our analyser and variable selector on $13$ representative Pyro examples from the Pyro webpage~\cite{pyro:examples}
that use standard SVI engines and contain explicitly written model-guide pairs (without AutoGuide).
They include advanced models with deep neural networks such as attend-infer-repeat \cite{EslamiNIPS16}
and single-cell annotation using variational inference \cite{xu2021probabilistic}.
Additionally, we included the example in \cref{fig:eg-model-guide},
for which Pyro offers an unsound reparameterisation plan.
\cref{t:examples} lists half of these 14 Pyro examples with their code size and conceptual complexity
(see \cref{sec:impl-eval-more} for the rest).
Experiments were performed on a Macbook Pro with 2.3GHz Core i9 and 32GB RAM. 

\begin{table*}[t]
  \centering
  \renewcommand{\arraystretch}{.75}
  \aboverulesep=0.2ex
  \belowrulesep=0.2ex
  \small
  \caption{%
    Results of smoothness analyses.
    ``Manual'' and ``Ours'' denote the number of continuous random variables and learnable parameters
    in which the density of the program is smooth, computed by hand and by our analyser.
    ``Time'' denotes the runtime of our analyser in seconds.
    ``\#CRP'' denotes the total number of continuous random variables and learnable parameters in the program.
    \mytt{-m} and \mytt{-g} denote model and guide.
    We consider $\{(\alpha,i) \in \Name\}$ as one random variable for each $\alpha \in \Name$.
    \rev{See \cref{sec:impl-eval-more} for the rest of results.}
  }
  \label{tab:analysis-result}
  \vspace{-1em}
  \begin{tabular}{lrrrrrrr}
    \toprule
    & \multicolumn{3}{c}{Differentiable}
    & \multicolumn{3}{c}{Locally Lipschitz}
    \\ \cmidrule(lr){2-4} \cmidrule(lr){5-7}
    \multicolumn{1}{l}{Name}
    & \multicolumn{1}{r}{Manual}
    & \multicolumn{1}{r}{Ours}
    & \multicolumn{1}{l}{Time}
    & \multicolumn{1}{r}{Manual}
    & \multicolumn{1}{r}{Ours}
    & \multicolumn{1}{l}{Time}
    & \multicolumn{1}{r}{\#CRP}
    \\ \midrule
    \mytt{spnor-m}   &  1 &  1 & 0.006 &  1 &  1 & 0.009 &  \hltred{2}\!
    \\
    \mytt{spnor-g}   &  4 &  4 & 0.007 &  4 &  4 & 0.008 &  4
    \\
    \mytt{sgdef-m}   &  6 &  6 & 0.003 &  6 &  6 & 0.006 &  6 
    \\
    \mytt{sgdef-g}   & 18 & 18 & 0.016 & 18 & 18 & 0.015 & 18 
    \\
    \mytt{dmm-m}     &  \hltorg{4} &  \hltorg{4} & 0.014 & 10 & 10 & 0.016 & 10 
    \\
    \mytt{dmm-g}     &  \hltorg{4} &  \hltorg{4} & 0.026 &  5 &  5 & 0.020 &  5 
    \\
    \mytt{mhmm-m}    & 10 & 10 & 0.063 & 10 & 10 & 0.075 & 10 
    \\
    \mytt{mhmm-g}    &  6 &  6 & 0.007 &  6 &  6 & 0.008 &  6 
    \\
    \mytt{scanvi-m}  &  \hltorg{6} &  \hltorg{6} & 0.032 & 12 & 12 & 0.032 & 12 
    \\
    \mytt{scanvi-g}  &  \hltorg{8} &  \hltorg{8} & 0.052 & 15 & 15 & 0.058 & 15 
    \\
    \mytt{air-m}     &  \hltorg{1} &  \hltorg{1} & 0.108 &  4 &  4 & 0.105 &  4 
    \\
    \mytt{air-g}     &  \hltorg{3} &  \hltorg{3} & 0.075 & 15 & 15 & 0.072 & \hltred{16}\! 
    \\
    \mytt{cvae-m}    &  \hltorg{3} &  \hltorg{3} & 0.025 &  8 &  8 & 0.027 &  8 
    \\
    \mytt{cvae-g}    &  \hltorg{5} &  \hltorg{5} & 0.031 &  9 &  9 & 0.023 &  9 
    \\ \bottomrule
  \end{tabular}
  \vspace{-0.4em}
\end{table*}

\vspace{1mm}
\noindent{\it Smoothness analyser.}
We assess our smoothness analyser on the 14 Pyro examples for differentiability and local Lipschitzness (\cref{:4:inst}),
and show a subset of results in \cref{tab:analysis-result} (see \cref{sec:impl-eval-more} for the rest).
The results demonstrate that our analysis can cope successfully with real-world Pyro programs.
First, our analysis is accurate.
For all examples, the analysis identifies the exact ground-truth set of random variables and parameters
in which the density of the program is differentiable (or locally Lipschitz).
In many of them, information computed by the pre-analysis is required to achieve these exact results;
e.g., some examples (e.g., \mytt{dpmm} and \mytt{air})
require precise information about which distribution arguments can be proved to be always in the proper range of values.
Second, the runtime of our analysis is low.
Typical probabilistic programming applications are not of a very large
size, and conceptual complexity is generally the main issue, thus the
analysis performance presents no scalability concern.

We draw two more observations from the results.
First, for \mytt{spnor-m} and \mytt{air-g}, the density of each program is not locally Lipschitz in one continuous random variable.
These non-local-Lipschitznesses arise as follows:
for the former, the random variable ({\small $\cstr{\small z_2}$} in \cref{fig:eg-model-guide}) is used in the branch condition of an if-else command
that contains observe commands, thereby creating discontinuity;
and for the latter, the random variable ({\small $\cstr{z\_where}$}) is passed into the denominator of a division operator,
thereby causing a division-by-zero error for some value.

Second, for all the other examples, the density is locally Lipschitz in all continuous random variables and parameters,
but is often non-differentiable in many parameters (and continuous random variables too); see, for instance, \mytt{scanvi} and \mytt{cvae}.
Due to this, the requirement \cref{cond:dens-diff} is not satisfied for these examples
even with the empty reparameterisation plan (corresponding to the score estimator);
that is, if we use the differentiability requirements \cref{cond:dens-diff,cond:val-diff} to validate the unbiasedness of gradient estimators,
even the score estimator cannot be validated for these examples.
From manual inspection, we checked that the non-differentiabilities from these examples
all arise by the use of locally Lipschitz but non-differentiable operators (e.g., \mytt{relu} and \mytt{grid\_sample}).
Since many practical models (and guides) use locally Lipschitz but non-differentiable operators,
this observation strongly suggests that a right smoothness requirement for validating gradient estimators
is not differentiability (which has been used as a standard requirement), but rather local Lipschitzness (e.g., \cref{cond:dens-lips,cond:val-lips}).



\begin{table*}[t]
  \centering
  \renewcommand{\arraystretch}{.80}
  \aboverulesep=0.2ex
  \belowrulesep=0.2ex
  \small
  \caption{%
    Results of variable selections.
    ``Ours-Time'' denote the runtime of our variable selector in seconds.
    ``Ours-Sound'' and ``Pyro $\setminus$ Ours'' denote the number of random variables in the example
    that are in $\pi_\ours$, and that are in $\pi_0$ but not in $\pi_\ours$, respectively,
    where $\pi_\ours$ and $\pi_0$ denote the reparameterisation plans given by our variable selector and by Pyro.
    ``Pyro $\setminus$ Ours'' is partitioned into ``Sound'' and ``Unsound'':
    the latter denotes the number of random variables that make \cref{cond:dens-lips} or \cref{cond:val-lips} violated when added to $\pi_\ours$,
    and the former denotes the number of the rest.
    ``\#CR'' and ``\#DR'' denote the total number of continuous and discrete random variables in the example.
    We consider $\{(\alpha,i) \in \Name\}$ as one random variable for each $\alpha \in \Name$.
    \rev{See \cref{sec:impl-eval-more} for the rest of results.}
  }
  \label{tab:reparam-result}
  \vspace{-0.5em}
  \begin{tabular}{lrrrrr|r}
    \toprule
    & \multicolumn{2}{c}{Ours} & \multicolumn{2}{c}{Pyro $\setminus$ Ours}
    & \multicolumn{1}{c|}{}
    \\ \cmidrule(lr){2-3} \cmidrule(lr){4-5}
    Name
    & Time
    & Sound
    & Sound
    & Unsound
    & \#CR
    & \#DR
    \\ \midrule
    \mytt{spnor}   & 0.021 & 1 & 0 & \hltred{1}\! & 2 & 0
    \\
    \mytt{sgdef}   & 0.034 & 6 & 0 & 0 & 6 & 0 
    \\
    \mytt{dmm}     & 0.054 & 1 & 0 & 0 & 1 & 0 
    \\
    \mytt{mhmm}    & 0.083 & 2 & 0 & 0 & 2 & 1 
    \\
    \mytt{scanvi}  & 0.143 & 3 & 0 & 0 & 3 & 1 
    \\
    \mytt{air}     & 0.247 & 1 & 0 & \hltred{1}\! & 2 & 1 
    \\
    \mytt{cvae}    & 0.063 & 1 & 0 & 0 & 1 & 0 
    \\ \bottomrule
  \end{tabular}
  \vspace{-0.5em}
\end{table*}

\vspace{1mm}
\noindent{\it Variable selector.}
To evaluate our variable selector, we consider the SPGE variable-selection problem with local Lipschitzness requirements,
i.e., the problem that uses \cref{cond:dens-lips,cond:val-lips} in \cref{sec:local-lips-req}
instead of \cref{cond:dens-diff,cond:val-diff} in \cref{sec:grad-estm-prog-trans}.
We do not consider the original problem (with differentiability requirements),
since for many examples the differentiability requirements are not satisfied
even by the empty reparameterisation plan (i.e., score estimator) as observed above.
For an initial reparameterisation plan $\pi_0$ for the problem, we use the plan given by Pyro's default variable selector:
it is defined for all continuous random variables and applies standard reparameterisations
(e.g., \cref{eqn:initial-reparameterisation-plan} for a normal distribution).
In this settings, we apply our variable selector to the problem on the 14 Pyro examples.
\cref{tab:reparam-result} displays the results (only for 7 examples; see \cref{sec:impl-eval-more} for the rest) and compares them with $\pi_0$.

The results demonstrate that for all examples, our variable selector finds the optimal reparameterisation plan with a small runtime.
We also observe that for all cases, it terminates in the first iteration
and calls our smoothness analyser only three times, as mentioned in \cref{sec:var-sel-alg}.
Note that the reparameterisation plan given by Pyro is also optimal for all but two examples.
We emphasise, however, that our variable selector not only finds a reparameterisation plan
but also verifies the local Lipschitzness requirements \cref{cond:dens-lips,cond:val-lips},
whereas Pyro's default variable selector does not do so. 
Indeed, for two examples, Pyro's reparameterisation plan is unsound as it violates the local Lipschitzness requirements.
Hence, these results should be interpreted as: for all but two examples,
our variable selector (and smoothness analyser) \rev{successfully verifies that
  the default gradient estimator used by Pyro satisfies important smoothness-related preconditions for unbiasedness,
  namely the local Lipschitzness requirements.}

The two examples for which Pyro becomes unsound are \mytt{spnor} and \mytt{air}.
Recall that they have two continuous random variables (one for each) in which their densities are not locally Lipschitz.
The unsoundness of Pyro on these examples stems precisely from the fact that it reparameterises the two non-locally-Lipschitz random variables
without checking any local Lipschitzness requirements.



\section{Related Work}

The high-level idea of using program transformation for improved posterior inference and model learning in PPLs has been explored previously~\cite{NoriHRS14,ClaretRNGB13,RitchieSG16,GorinovaMH20,SchulmanHWA15}. In particular, \citet{SchulmanHWA15} proposed a method for implementing the SPGE for stochastic computation graphs via graph transformation, and this method was adopted in the implementation of the same estimator in Pyro and also in our work. However, the method lacks a formal analysis on the implemented estimator especially in the context of probabilistic programs; it does not have a version of \cref{thm:unbiased-grad}, which formally identifies requirements for the unbiasedness of the estimator. Also, the method does not check the required smoothness properties of given probabilistic programs. Our work fills in these gaps. \citet{GorinovaMH20} proposed an automatic technique to transform models in a PPL using the same or closely-related transformation of sample commands in the SPGE. The work is, however, concerned with transforming models and taming their posterior distributions, while ours focuses on transforming guides. Also, the work does not check smoothness properties of transformed models that are required for running efficient inference algorithms, such as Hamiltonian Monte Carlo, on those models, while our work checks those properties using our program analysis.

\rev{Program analyses or type systems for PPLs have been developed to detect common errors~\cite{LewCSCM20,LeeYRY20,WangHR21},
  infer important probabilistic properties such as conditional independence~\cite{GorinovaGSV22},
  or automatically plan inference algorithms \cite{webppl:analysis} as in our work.
  In particular, \citeauthor{webppl:analysis} runs a simple program analysis (checking if there are interleaving sample and observe commands)
  to decide if it is worth applying sequential Monte Carlo.}

The smoothness properties computed by our program analysis, such as differentiability and local Lipschitzness,
fall in the scope of \emph{hyperliveness} in the hierarchy of hyperproperties~\cite{hyper:08}. Intuitively,
 hyperliveness properties are those that cannot be refuted based on any finite counterexample (i.e., made 
 of finitely-many finite execution traces), and counterexamples for differentiability and local Lipschitzness
 should indeed require infinitely-many execution traces due to the use of limit or all neighbouring inputs
 in their definitions. Not so many analyses have considered such hyperliveness properties. Among those, the most relevant to
 our work are the continuity analyses of \citet{cgl:popl:10,cgl:jacm:12}. It uses a program 
 abstraction that is rather similar to ours, but their analyses suffer from soundness issues, \rev{partly}
due to the incorrect joining of continuity sets~\cite{cgl:popl:10} and also to an unsound rule for sequential composition~\cite{cgl:jacm:12}
{(see \cref{sec:intro-more} for details)}.
We do not claim that these issues are difficult to fix. Our point is just that developing program analyses for smoothness properties
requires special care. \citeauthor{cgl:jacm:12}'s work focuses on proving smoothness properties of control software, or revealing 
the unexpected continuity of discrete algorithms. On the other hand, our program analysis is designed to assist variational 
inference and model learning for probabilistic programs.
\rev{\citet{BartheCLG20} proposed a refinement type system, which considers a higher-order functional language
  and ensures that every typable first-order program is continuous in all variables. 
  On the other hand, our program analysis considers a first-order imperative language
  and can prove that a program is continuous in some (not necessarily all) variables. 
  Other existing program analyses for smoothness properties include
  \cite{LaurelYSM22} which over-approximates the Clarke generalised Jacobian,
  and \cite{MangalSNO20} which proves probabilistic Lipschitzness.}

\begingroup
\vspace{-0.5mm}
\begin{acks}                            
  \vspace{-0.5mm}
  \rev{%
  We thank Hangyeol Yu for helping us prove \cref{thm:unbiased-val,thm:unbiased-grad},
  and anonymous reviewers for giving constructive comments.
  Lee was supported by Samsung Scholarship.
  Yang was supported by the Engineering Research Center Program through the National Research Foundation of Korea (NRF)
  funded by the Korean Government MSIT (NRF-2018R1A5A1059921) and also by the Institute for Basic Science (IBS-R029-C1).
  Rival was supported by the French ANR VeriAMOS project.}
\end{acks}
\endgroup

\bibliography{paper}

\AtEndDocument{\newpage\normalsize\appendix}
\AtEndDocument{

\rev{%
\section{Deferred Results in \S\ref{sec:intro}}
\label{sec:intro-more}

\subsection{{Unsoundness of Continuity Analyses in \cite{cgl:popl:10,cgl:jacm:12}}}
\label{sec:intro:remarks}


The continuity analysis in \cite{cgl:popl:10} considers joint continuity,
whereas the continuity analysis in \cite{cgl:jacm:12} considers partial continuity.
That is, given a command $c$, an output variable $v$ of $c$, and some input variables $u_1, \ldots, u_m$ to $c$,
the former analyses whether $v$ is continuous in $\{u_1, \ldots, u_m\}$ jointly,
whereas the latter analyses whether $v$ is continuous in $u_i$ separately for every $1 \leq i \leq m$.

\vspace{2mm}
\noindent{\bf Join and Sequence rules.}\
The former analysis contains a rule called Join \cite[Figure 3]{cgl:popl:10},
and the latter analysis contains a rule called Sequence \cite[Figure 1]{cgl:jacm:12}.
The two rules can be rewritten (with some simplifications) as follows, in terms of functions between $\R^n$:
for any $f,g : \R^n \to \R^n$ and $S, S', T, U \subseteq \{1, \ldots, n\}$,
\begin{gather*}
  \infer[\text{(Join)}]{
    \text{For each $j \in T$, $f_j$ is continuous in $\{x_i \mid i \in S \cup S'\}$}
  }{
    \begin{array}{l}
      \text{For each $j \in T$, $f_j$ is continuous in $\{x_i \mid i \in S\}$}
      \\
      \text{For each $j \in T$, $f_j$ is continuous in $\{x_i \mid i \in S'\}$}
    \end{array}
  }
  \\
  \infer[\text{(Sequence)}]{
    \text{For each $k \in U$, $(g \circ f)_k$ is continuous in $x_i$ for each $i \in S$}
  }{
    \begin{array}{c}
      \text{For each $j \in T$, $f_j$ is continuous in $x_i$ for each $i \in S$}
      \\
      \text{For each $k \in U$, $g_k$ is continuous in $y_j$ for each $j \in T$}
    \end{array}
  }
\end{gather*}
where $f$ and $g$ are functions of variables $x_1, \ldots, x_n$ and $y_1, \ldots, y_n$, respectively,
and $h_i \defeq \mathit{proj}_i \circ h$ for $h : \R^n \to \R^n$ and $i \in \{1, \ldots, n\}$ denotes the $i$-th component of $h$.
As mentioned above, the Join rule analyses joint continuity, while the Sequence rule analyses partial continuity.
Further, the Join rule says that joint continuity is preserved under the union of input variables,
while the Sequence rule says that partial continuity is preserved under the composition of functions.

The two rules, however, are unsound with the following counterexamples.
Let $h : \R^2 \to \R^2$ be the function
\[
h(x_1,x_2) \defeq
\begin{cases}
  (x_1 x_2 / (x_1^2 + x_2^2), x_2) & \text{if $(x_1,x_2) \neq (0,0)$}
  \\
  (0, x_2) & \text{otherwise}.
\end{cases}
\]
Note that $h_1$ is continuous in $x_1$ and in $x_2$ separately, but not in $\{x_1, x_2\}$ jointly.
First, for the Join rule, consider the following $f : \R^2 \to \R^2$ and $S, S', T \subseteq \{1,2\}$:
\[
f(x_1,x_2) \defeq h(x_1,x_2), \quad S \defeq \{1\}, \quad S' \defeq \{2\}, \quad T \defeq \{1,2\}.
\]
Then, the premise of the Join rule holds, and so the conclusion of the rule must hold.
But this is {\em not} the case since $f_1 = h_1$ is not continuous in $\{x_1,x_2\}$.
Hence, the Join rule is unsound.
Next, for the Sequence rule, consider the following $f, g : \R^2 \to \R^2$ and $S, T, U \subseteq \{1,2\}$:
\[
f(x_1,x_2) \defeq (x_1, x_1), \quad g(x_1,x_2) \defeq h(x_1,x_2), \quad S \defeq T \defeq U \defeq \{1,2\}.
\]
Then, the premise of the Sequence rule holds, and so the conclusion of the rule must hold.
But this is {\em not} the case since $(g \circ f)_1$ is not continuous in $x_1$
(due to $(g \circ f)_1(x_1,x_2) = \ind{x_1 \neq 0} \cdot \frac{1}{2}$).
Hence, the Sequence rule is unsound.
These counterexamples show that
joint continuity is {\em not} preserved under the union of input variables,
and partial continuity is {\em not} preserved under the composition of functions.

The two aforementioned counterexamples can be easily translated into programs:
the first becomes $c_1 \equiv (z := h_1(x,y))$ and the second becomes $c_2 \equiv (y:=x;\ z:=h_1(x,y))$,
where $x$, $y$, and $z$ are program variables and $h_1$ is the binary operator defined above.
The analysis in \cite{cgl:popl:10} deduces that in $c_1$, $z$ is continuous in $x$ and $y$ (jointly),
and the analysis in \cite{cgl:jacm:12} deduces that in $c_2$, $z$ is continuous in $x$ (separately).
Both deductions, however, are incorrect as seen above, and the two analyses are thus unsound.

\vspace{2mm}
\noindent{\bf Sequence rule (again).}\
Both of the two analyses in \cite{cgl:popl:10,cgl:jacm:12} contain the Sequence rule (discussed above) which has the following rule as an instance:
for all variables $v_0, v_1, v_2$ that are not necessarily distinct, and for all commands $c_1,c_2$,
\begin{align*}
  \infer[]{
    \text{In $(c_1;c_2)$, $v_2$ is continuous in $v_0$}
  }{
    \begin{array}{l}
      \text{In $c_1$, $v_1$ is continuous in $v_{0}$}
      \\
      \text{In $c_2$, $v_2$ is continuous in $v_{1}$}
    \end{array}
  }
\end{align*}

The above instance of the Sequence rule is, however, unsound because it incorrectly handles the dependencies between variables.
For instance, consider commands $c_1 \equiv (\cif\ (x<0)\ \{y:=0\}\ \celse\ \{y:=1\})$ and $c_2 \equiv (z:=x+y)$,
and variables $v_0 \equiv x$, $v_1 \equiv x$, and $v_2 \equiv z$.
Then, in $(c_1;c_2)$, $z$ is {\em not} continuous in $x$ 
(because $z$ after $(c_1;c_2)$ is $x$ if $x<0$, and is $x+1$ if $x \geq 0$).
However, the premise of the above rule holds for this case,
and so the rule incorrectly concludes that $z$ is continuous in $x$ under $(c_1;c_2)$,
by ignoring the dependency of $z$ on $y$ (which is discontinuous in $x$).

\vspace{2mm}
\noindent{\bf Loop rule.}\
The analysis in \cite{cgl:popl:10} contains a rule called Simple-loop (Figure~5 in the paper) to analyse loops,
and the analysis in \cite{cgl:jacm:12} contains a rule called Loop (Figure~1 in the paper) which is essentially the same as the Simple-loop rule.

The two rules, however, incorrectly assume that (joint or partial) continuity without any restriction on the domain of functions
satisfies the admissibility assumption discussed in \cref{remark:admissibility};
and this in turn makes the rules unsound.
To illustrate the unsoundness, consider a command $c \equiv (\cwhile\ (0<x<1)\ \{ c'\})$
with $c' \equiv (y:=y+f(x);\ x:=g(x))$,
where $f$ and $g$ are the continuous operators defined by:
\begin{align*}
  & f(x) \defeq
  \begin{cases}
    x & \text{if $0<x \leq 1/2$} \\
    1-x & \text{if $1/2<x \leq 1$} \\
    0 & \text{otherwise},
  \end{cases}
  & g(x) \defeq
  \begin{cases}
    0 & \text{if $x \leq 0$} \\
    2x & \text{if $0<x \leq 1/2$} \\
    1 & \text{otherwise}.
  \end{cases}
\end{align*}
Then, the premises of the two rules are satisfied,
mainly because $x$ and $y$ after $c'$ are continuous jointly in $x$ and $y$ before $c'$,
and $x$ and $y$ do not change in $c'$ if $x=0$ or $x=1$ (i.e., at the ``boundary'' of the loop condition of $c$).
Hence, the two rules conclude that $x$ and $y$ after $c$ are continuous in $x$ and $y$ before $c$ (jointly or separately).
However, this is an unsound conclusion because in $c$, $y$ is {\em not} continuous in $x$ at $x=0$:
we have $y'=y$ if $x=0$, but $y' \to y+1$ as $x \to 0^+$ (more precisely, $y' = y$ if $x \leq 0$ or $x \geq 1$, and $y' = y+(1-x)$ if $0<x < 1$),
where $x'$ and $y'$ denote the values of $x$ and $y$ after $c$.%
}
}      
\AtEndDocument{

\section{Deferred Results in \S\ref{sec:overview}}
\label{sec:overview-more}

\subsection{Table Summarising \cref{sec:overview}}

\begin{table*}[!ht]
  \renewcommand{\arraystretch}{.97}
  \aboverulesep=0.4ex
  \belowrulesep=0.4ex
  \small
  \vspace{-0.5em}
  \caption{\rev{Gradient estimators for variational inference, and requirements for each estimator.
      ``Req.'' denotes ``Requirement'' and ``diff.'' denotes ``differentiable''.
      Recall that $f_\theta(z) \defeq \log (p_{c_m}(z) / p_{c_g,\theta}(z))$.}}
  \label{tab:grad-est}
  \vspace{-1em}
  \begin{center}
    \rev{%
    \begin{tabular}{ll|lll}
      \toprule
      & & SCE & PGE & SPGE
      \\ \midrule
      Setup
      & $q_\theta(z)$
      & $p_{c_g,\theta}(z)$
      & $p_{c_g'}(z)$
      & $p_{c_g'',\theta}(z)$
      \\
      & $v_{\theta}(z)$
      & $z$
      & $v_{c'_g,\theta}(z)$
      & $v_{c''_g,\theta}(z)$
      \\
      & $g_\theta(z)$
      & $f_\theta(v_\theta(z)) \cdot \nabla_\theta \log q_\theta(z)$
      & $\nabla_\theta \big( f_\theta(v_\theta(z)) \big)$
      & $\nabla_\theta \big( f_\theta(v_\theta(z)) \big) + f_\theta(v_\theta(z)) \cdot \nabla_\theta \log q_\theta(z) $
      \\ \midrule
      Req.
      & $q_{\theta}(z)$
      & diff.~ in $\theta$
      & -- 
      & diff.~ in $\theta$
      \\
      & $v_{\theta}(z)$
      & --
      & diff.~ in  $\theta$
      & diff.~ in  $\theta$
      \\
      & $p_{c_m}(z)$
      & --
      & diff.~ in $\theta$ and $z$
      & diff.~ in $\theta$ and {``selected'' $z_i$'s}
      \\
      & $p_{c_g,\theta}(z)$
      & --
      & diff.~ in $\theta$ and $z$
      & diff.~ in $\theta$ and {``selected'' $z_i$'s}
      \\ \bottomrule
    \end{tabular}}
  \end{center}
  \vspace{-0.5em}
\end{table*}

\noindent
\cref{tab:grad-est} compares key aspects of the three gradient estimators (SCE, PGE, and SPGE) explained in \cref{sec:overview}.
}   
\AtEndDocument{

\section{Deferred Results in \S\ref{sec:prog-trans}}


\subsection{Proof of \cref{thm:unbiased-val}}
\label{sec:proof:unbiased-val}

We introduce several definitions, state lemmas, and prove \cref{thm:unbiased-val} using the lemmas.
We prove the lemmas in \cref{sec:proof:lemmas:unbiased-val1,sec:proof:lemmas:unbiased-val2}.

Recall the partition $\Var = \PVar \uplus \Name \uplus \AVar$ of $\Var$.
We use the following letters to denote the values of each part:
$\sigma_p \in \State[\PVar]$, $\sigma_n \in \State[\Name]$, and $\sigma_a \in \State[\AVar]$.
Based on the partition, we define the next functions:
\begin{align*}
  \getpr  (c) &: \State[\PVar] \times \State[\Name] \times \State[\AVar] \to [0,\infty),
  \\
  \getpr(c)(\sigma_p, \sigma_n, \sigma_a)
  &\defeq
  \begin{cases}
    \begin{array}{@{}l@{}}
      \db{c}(\sigma_p \oplus \sigma_n \oplus \sigma_a)(\like)
      \\ \;\; \cdot
      \prod_{\mu \in \Name} \db{c}(\sigma_p \oplus \sigma_n \oplus \sigma_a)(\pr_\mu)
    \end{array}
    & \text{if $\noerr(c, \sigma_p \oplus \sigma_n \oplus \sigma_a)$}
    \\
    0 & \text{otherwise},
  \end{cases}
  \\[1.5ex]
  \getval (c) &: \State[\PVar] \times \State[\Name] \times \State[\AVar] \to \State[\Name],
  \\
  \getval(c)(\sigma_p, \sigma_n, \sigma_a)
  &\defeq \begin{cases}
    \lambda \mu \in \Name.\, \db{c}(\sigma_p \oplus \sigma_n \oplus \sigma_a)(\val_\mu)
    & \text{if $\noerr(c, \sigma_p \oplus \sigma_n \oplus \sigma_a)$}
    \\
    \lambda \mu \in \Name.\, 0
    & \text{otherwise},
  \end{cases}
\end{align*}
where $\noerr(c, \sigma)$ is a predicate for a command $c$ and $\sigma \in \State$, defined by
\begin{align*}
  \noerr(c,\sigma)
  & \iff \db{c}\sigma \in \State
  \land \big( \forall \mu \in \Name.\, \db{c}\sigma(\sampled_\mu) - \sigma(\sampled_\mu) \leq 1 \big).
\end{align*}
The predicate $\noerr(c,\sigma)$ says that $c$ terminates for $\sigma$ without a double-sampling error.
The functions $\getpr$ and $\getval$ generalise the density function $p$ and the value function $v$, respectively;
in particular, they do not assume a particular initial state $\sigma_0$ used in \cref{eq:density-of-c}.
We consider the generalisation of $p$ and $v$ so as to enable inductive proofs. 

Although generalising $p$ and $v$, the functions $\getpr$ and $\getval$ are not sufficient to enable inductive proofs
since their inputs and outputs contain some unnecessary parts, which stops induction from working well (especially in the sequential composition case):
namely, the part of $\State[\Name]$ that is not read during execution, and the part of $\State[\AVar]$ that is not updated during execution.
To exclude those unnecessary parts, we first define the set of substates of $\State[\Name]$ as follows: 
\begin{align*}
  \xi_n \in \Statesub[\Name] &\defeq \bigcup_{K \subseteq \Name} \State[K].
\end{align*}
Based on these substates, we define the next functions:
\begin{align*}
  \getprsub  (c) &: \State[\PVar] \times \Statesub[\Name] \to [0,\infty),
  \\
  \getprsub(c)(\sigma_p, \xi_n)
  &\defeq \begin{cases}
    \begin{array}{@{}l@{}}
      \db{c}(\sigma_p \oplus \xi_n \oplus \sigma_r)(\like)
      \\
      \;\; \cdot
      \prod_{\mu \in \dom(\xi_n)} \db{c}(\sigma_p \oplus \xi_n \oplus \sigma_r)(\pr_\mu)
    \end{array}
    & \text{if $\exists \sigma_r.\, \used(c, \sigma_p \oplus \xi_n \oplus \sigma_r, \xi_n)$}
    \\
    0 & \text{otherwise},
  \end{cases}
  \\[1.5ex]
  \getvalsub (c) &: \State[\PVar] \times \Statesub[\Name] \to \Statesub[\Name],
  \\
  \getvalsub(c)(\sigma_p, \xi_n)
  &\defeq \begin{cases}
    \lambda \mu \in \dom(\xi_n).\, \db{c}(\sigma_p \oplus \xi_n \oplus \sigma_r)(\val_\mu)
    & \text{if $\exists \sigma_r.\, \used(c, \sigma_p \oplus \xi_n \oplus \sigma_r, \xi_n)$}
    \\
    \lambda \mu \in \dom(\xi_n).\, 0
    & \text{otherwise},
  \end{cases}
  \\[1.5ex]
  \getpvarsub(c) &: \State[\PVar] \times \Statesub[\Name] \to \State[\PVar],
  \\
  \getpvarsub(c)(\sigma_p, \xi_n)
  &\defeq \begin{cases}
    \lambda x \in \PVar.\, \db{c}(\sigma_p \oplus \xi_n \oplus \sigma_r)(x)
    & \text{if $\exists \sigma_r.\, \used(c, \sigma_p \oplus \xi_n \oplus \sigma_r, \xi_n)$}
    \\
    \lambda x \in \PVar.\, 0
    & \text{otherwise},
  \end{cases}
\end{align*}
where $\used(c, \sigma, \xi_n)$ is a predicate for a command $c$, $\sigma \in \State$, and $\xi_n \in \Statesub[\Name]$, defined by
\begin{align*}
  \used(c, \sigma, \xi_n) \iff
  & \noerr(c, \sigma) \land \big( \sigma(\like)=1 \big)
  \land \big( \xi_n = \sigma|_{\dom(\xi_n)} \big)
  \\
  & {} \land \big( \dom(\xi_n) = \{ \mu \in \Name \mid  \db{c}\sigma(\sampled_\mu) - \sigma(\sampled_\mu) = 1\} \big).
\end{align*}
The predicate $\used(c,\sigma,\xi_n)$ says that
$c$ terminates for $\sigma$ without a double-sampling error,
$\like$ is initialised to $1$ in $\sigma$,
and $\xi_n$ is the $\Name$ part of $\sigma$ that is sampled during the execution of $c$ from $\sigma$.
By using $\used(-,-,-)$, the three functions do not take the unnecessary part of a state as an input,
and do not return the unnecessary part of a state in the output.
The three functions are well-defined. 
\begin{lemma}
  \label{lem:subfns-well-defined}
  $\getprsub$, $\getvalsub$, and $\getpvarsub$ are well-defined,
  i.e., they do not depend on the choice of $\sigma_r$.
\end{lemma}

We now state two main lemmas for \cref{thm:unbiased-val}.
The first lemma describes how $\getpr$ and $\getval$ are connected with $\getprsub$ and $\getvalsub$. 
The second lemma says that a particular integral involving $\getprsub$, $\getvalsub$, and $\getpvarsub$
is the same for $c$ and $\ctr{c}{\pi}$ if a reparameterisation plan $\pi$ is valid.
\begin{lemma}
  \label{lem:integral-sampled-terms-only}
  Let $c$ be a command, and  $f_i : \R \to \R$ for $i\in\{1,2,3\}$ be measurable functions such that
  $f_1(r) \geq 0$ for all $r \in \R$. 
  Define $f_* : \State[\Name] \to \State[\AVar]$ by
  \begin{align*}
    f_*(\sigma_n)(a) &\defeq
    \begin{cases}
      1 & \text{if $a \equiv \like$}
      \\
      f_1(\sigma_n(\mu)) & \text{if $a \equiv \pr_\mu$ for $\mu \in \Name$}
      \\
      f_2(\sigma_n(\mu)) & \text{if $a \equiv \val_\mu$ for $\mu \in \Name$}
      \\
      f_3(\sigma_n(\mu)) & \text{if $a \equiv \sampled_\mu$ for $\mu \in \Name$}.
    \end{cases}
  \end{align*}
  Then, for all $\sigma_p \in \State[\PVar]$ and all measurable $h : \State[\Name] \to \R$,
  \begin{align*}
    &
    \int d\sigma_n \Big(
    \getpr(c)(\sigma_p, \sigma_n, f_*(\sigma_n))
    \cdot h\Big(
    \getval(c)(\sigma_p, \sigma_n, f_*(\sigma_n))
    \Big)\Big)
    \\
    &=
    \int d\xi_n \Big( \getprsub(c)(\sigma_p, \xi_n) \cdot
    g\Big(
    \getvalsub(c)(\sigma_p, \xi_n)
    \Big) \Big)
  \end{align*}
  where {the integral on the LHS is defined if and only if the one on the RHS is defined}, 
  and the function $g : \Statesub[\Name] \to \R$ is defined by 
  \begin{align}
    \label{eq:lem:integral-sampled-terms-only-g}
    g(\xi''_n ) 
    & =
    \int d\xi'_n \Big( \ind{\dom(\xi''_n) \uplus \dom(\xi'_n) = \Name}
    \cdot \Big(\prod_{\mu \in \dom(\xi'_n)} f_1(\xi'_n(\mu))\Big)
    \cdot h\Big(\xi''_n \oplus \lambda \mu \in \dom(\xi'_n).\, f_2(\xi'_n(\mu)) \Big)\Big). 
  \end{align}
  
\end{lemma}
\begin{lemma}
  \label{lem:integral-same-under-reparam}
  Let $c$ be a command and $g : \State[\PVar] \times \Statesub[\Name] \to \R$ be a measurable function.
  Then, for all $\sigma_p \in \State[\PVar]$,
  \begin{align*}
    &
    \int d\xi_n \Big( \getprsub(c)(\sigma_p, \xi_n) \cdot
    g\Big(\getpvarsub(c)(\sigma_p, \xi_n), \getvalsub(c)(\sigma_p, \xi_n) \Big)\Big)
    \\
    &=
    \int d\xi_n \Big( \getprsub(\ctr{c}{\pi})(\sigma_p, \xi_n) \cdot
    g\Big(\getpvarsub(\ctr{c}{\pi})(\sigma_p, \xi_n), \getvalsub(\ctr{c}{\pi})(\sigma_p, \xi_n) \Big)\Big)
  \end{align*}
  {where the integral on the LHS is defined if and only if the one on the RHS is defined.}
\end{lemma}

We now prove \cref{thm:unbiased-val} using these two lemmas.
\begin{proof}[Proof of \cref{thm:unbiased-val}]
  Let $\pi$ be a valid reparameterisation plan, $c$ be a command,
  $\sigma_\theta \in \State[\theta]$, and $h : \State[\Name] \to \R$ be a measurable function.
  Suppose that the integral on the LHS of \cref{thm:unbiased-val} is defined.
  Recall that for a given $\sigma_n \in \State[\Name]$, the definitions of $p$ and $v$ (in \cref{sec:setup,sec:prog-trans})
  use the initial state $\sigma \defeq \sigma_\theta \oplus \sigma_n \oplus \sigma_0 \in \State$,
  where $\sigma_0 \in \State[(\PVar \setminus \theta) \cup \AVar]$ depends on $\sigma_n$ and has the following definition:
  \begin{align*}
    \sigma_0(v) & \defeq
    \begin{cases}
      0 & \text{if $v \in \PVar \setminus \theta$}
      \\
      1 & \text{if $v \equiv \like$}
      \\
      \cN(\sigma_n(\mu); 0, 1) & \text{if $v \equiv \pr_\mu$ for $\mu \in \Name$}
      \\
      \sigma_n(\mu) &  \text{if $v \equiv \val_\mu$ for $\mu \in \Name$}
      \\
      0 & \text{if $v \equiv \sampled_\mu$ for $\mu \in \Name$}.
    \end{cases}
  \end{align*}
  The initial state $\sigma$ can be re-expressed as
  \begin{align}
    \label{eq:thm:unbiased-val-sigma}
    \sigma = \sigma_p \oplus \sigma_n \oplus f_*(\sigma_n)
  \end{align}
  using the following $\sigma_p \in \State[\PVar]$ and $f_* : \State[\Name] \to \State[\AVar]$:
  \begin{align*}
    \sigma_p(x) & \defeq
    \begin{cases}
      \sigma_\theta(x) & \text{if $x \in \theta$}
      \\
      0 & \text{if $x \in \PVar \setminus \theta$},
    \end{cases}
    &
    f_*(\sigma_n)(a) &\defeq 
    \begin{cases}
      1 & \text{if $a \equiv \like$}
      \\
      f_1(\sigma_n(\mu)) & \text{if $a \equiv \pr_\mu$ for $\mu \in \Name$}
      \\
      f_2(\sigma_n(\mu)) & \text{if $a \equiv \val_\mu$ for $\mu \in \Name$}
      \\
      f_3(\sigma_n(\mu)) & \text{if $a \equiv \sampled_\mu$ for $\mu \in \Name$},
    \end{cases}
  \end{align*}
  where $f_1(r) \defeq \cN(r; 0,1)$, $f_2(r) \defeq r$, and $f_3(r) \defeq 0$.
  Using this, we get the desired equation:
  \begin{align*}
    &\int d\sigma_n \Big( \pfun{c,\sigma_\theta}{}(\sigma_n) \cdot h\Big(\vfun{c, \sigma_\theta}(\sigma_n)\Big)\Big)
    \\
    &= \int d\sigma_n \Big( \getpr(c)(\sigma_p, \sigma_n, f_*(\sigma_n)) \cdot h\Big(\getval(c)(\sigma_p, \sigma_n, f_*(\sigma_n))\Big) \Big)
    \\
    &=
    \int d\xi_n \Big( \getprsub(c)(\sigma_p, \xi_n) \cdot
    g\Big( \getvalsub(c)(\sigma_p, \xi_n) \Big) \Big)
    \\
    &=
    \int d\xi_n \Big( \getprsub(\ctr{c}{\pi})(\sigma_p, \xi_n) \cdot
    g\Big( \getvalsub(\ctr{c}{\pi})(\sigma_p, \xi_n) \Big) \Big)
    \\
    &= \int d\sigma_n \Big( \getpr(\ctr{c}{\pi})(\sigma_p, \sigma_n, f_*(\sigma_n))
    \cdot h\Big(\getval(\ctr{c}{\pi})(\sigma_p, \sigma_n, f_*(\sigma_n))\Big) \Big)
    \\
    &= \int d\sigma_n \Big(  \pfun{\ctr{c}{\pi},\sigma_\theta}{}(\sigma_n) \cdot h\Big(\vfun{\ctr{c}{\pi}, \sigma_\theta}(\sigma_n)\Big) \Big)
  \end{align*}
  where  $g : \Statesub[\Name] \to \R$ is defined as \cref{eq:lem:integral-sampled-terms-only-g}.
  The first equality holds by \cref{eq:thm:unbiased-val-sigma} and the definition of $p_{c,\sigma_\theta}$, $v_{c,\sigma_\theta}$, $\getpr$, and $\getval$.
  The second equality holds by \cref{lem:integral-sampled-terms-only} (applied to $c$).
  The third equality follows from \cref{lem:integral-same-under-reparam}.
  The fourth equality holds by \cref{lem:integral-sampled-terms-only} (applied to $\ctr{c}{\pi}$).
  The fifth equality holds by \cref{eq:thm:unbiased-val-sigma} and the definitions of $p_{\ctr{c}{\pi},\sigma_\theta}$, $v_{\ctr{c}{\pi},\sigma_\theta}$, $\getpr$, and $\getval$.  
  {Note that the same equational reasoning with the reverse direction can be used to prove the claimed equation of the theorem when the integral on the RHS of the equation is defined.}
\end{proof}

\subsection{Proofs of \cref{lem:subfns-well-defined,lem:integral-sampled-terms-only}}
\label{sec:proof:lemmas:unbiased-val1}

\begin{proof}[Proof of \cref{lem:subfns-well-defined}]
  Let $c$ be a command, $\sigma_p \in \State[\PVar]$, and $\xi_n \in \Statesub[\Name]$.
  Consider $\sigma_r \in \State[\Var \setminus (\dom(\sigma_p) \cup \dom(\xi_n))]$
  such that $\used(c, \sigma_p \oplus \xi_n \oplus \sigma_r, \xi_n)$.
  We want to show that $\getprsub(c)(\sigma_p, \xi_n)$,  $\getvalsub(c)(\sigma_p, \xi_n)$, and  $\getpvarsub(c)(\sigma_p, \xi_n)$
  do not depend on the choice of $\sigma_r$.
  To do so, consider $\sigma_r' \in \State[\Var \setminus (\dom(\sigma_p) \cup \dom(\xi_n))]$
  such that $\used(c, \sigma_p \oplus \xi_n \oplus \sigma_r', \xi_n)$.
  Let $\sigma \defeq \sigma_p \oplus \xi_n \oplus \sigma_r$ and $\sigma' \defeq \sigma_p \oplus \xi_n \oplus \sigma_r'$.
  Then, it suffices to show that
  \begin{equation}
    \label{eq:lem:subfns-well-defined-goal}
    \begin{split}
      \textstyle
      \db{c}\sigma(\like)
      \cdot \prod_{\mu \in \dom(\xi_n)} \db{c}\sigma(\pr_\mu)
      &= \textstyle
      \db{c}\sigma'(\like)
      \cdot \prod_{\mu \in \dom(\xi_n)} \db{c}\sigma'(\pr_\mu),
      \\
      \db{c}\sigma(\val_\mu)
      &=\db{c}\sigma'(\val_\mu) \quad \text{for all $\mu \in \dom(\xi_n)$},
      \\
     \db{c}\sigma(x)
      &=\db{c}\sigma'(x) \quad \text{for all $x \in \PVar$}.
    \end{split}
  \end{equation}
  Since $\used(c, \sigma, \xi_n)$ and $\used(c, \sigma', \xi_n)$, we have $\sigma(\like) = 1 = \sigma'(\like)$
  and so $\sigma|_V = \sigma'|_V$ for $V = \PVar \cup \dom(\xi_n) \cup \{\like\}$.
  Using this and  $\used(c, \sigma, \xi_n)$,
  we can apply \cref{lem:used-properties}-\cref{lem:used-properties-3} and -\cref{lem:used-properties-3-like}
  to get $\db{c}\sigma(v) = \db{c}\sigma'(v)$ for all $v \in \PVar \cup \{\like\} \cup \{\pr_\mu, \val_\mu \mid \mu \in \dom(\xi_n)\}$.
  Hence, we obtain the desired equations in \cref{eq:lem:subfns-well-defined-goal}.
\end{proof}

\begin{proof}[Proof of \cref{lem:integral-sampled-terms-only}]
  Let $c$ be a command, $h : \State[\Name] \to \R$ be a measurable function,
  $f_* : \State[\Name] \to \State[\AVar]$ be the function defined in the statement of this lemma,
  and $\sigma_p \in \State[\PVar]$. 

  We first prove that the following equations hold for any measurable $h' : \State[\Name] \to \R$:
  \begin{align*}
    &  \int d\sigma_n \Big(
    \ind{\noerr(c, \sigma_p \oplus \sigma_n \oplus f_*(\sigma_n))} \cdot h'(\sigma_n)
    \Big)
    \\
    &= \int d\sigma_n \Big(
    \ind{\noerr(c, \sigma_p \oplus \sigma_n \oplus f_*(\sigma_n))} \cdot h'(\sigma_n)
    \cdot \sum_{K \subseteq \Name} \ind{\used(c, \sigma_p \oplus \sigma_n \oplus f_*(\sigma_n), \sigma_n|_K)}
    \Big) 
    \\
    &= \sum_{K \subseteq \Name} \int d\sigma_n \Big( 
    \ind{\used(c, \sigma_p \oplus \sigma_n \oplus f_*(\sigma_n), \sigma_n|_K)}
    \cdot \ind{\noerr(c, \sigma_p \oplus \sigma_n \oplus f_*(\sigma_n))} \cdot h'(\sigma_n)
    \Big)
    \\
    &= \sum_{K \subseteq \Name} \int_{[K \to \R]} d\xi_n \int_{[\Name \setminus K \to \R]} d\xi_n' \Big(
    \ind{\used(c, \sigma_p \oplus (\xi_n \oplus \xi_n') \oplus f_*(\xi_n \oplus \xi_n'), \xi_n)}
    \\
    & \qquad\quad
    \cdot \ind{\noerr(c, \sigma_p \oplus (\xi_n \oplus \xi_n') \oplus f_*(\xi_n \oplus \xi_n'))} \cdot h'(\xi_n \oplus \xi_n')
    \Big)
    \\
    &= \sum_{K \subseteq \Name} \int_{[K \to \R]} d\xi_n  \Big( \sum_{L \subseteq \Name} \int_{[L \to \R]} d\xi_n' \Big(
    \ind{\dom(\xi_n) \uplus \dom(\xi_n') = \Name}
    \cdot \ind{\used(c, \sigma_p \oplus (\xi_n \oplus \xi_n') \oplus f_*(\xi_n \oplus \xi_n'), \xi_n)}
    \\
    & \qquad\quad
    \cdot \ind{\noerr(c, \sigma_p \oplus (\xi_n \oplus \xi_n') \oplus f_*(\xi_n \oplus \xi_n'))} \cdot h'(\xi_n \oplus \xi_n')
    \Big) \Big)
    \\
    &= \int d\xi_n \int d\xi_n' \Big(
    \ind{\dom(\xi_n) \uplus \dom(\xi_n') = \Name}
    \cdot \ind{\used(c, \sigma_p \oplus (\xi_n \oplus \xi_n') \oplus f_*(\xi_n \oplus \xi_n'), \xi_n)}
    \\
    & \qquad\quad
    \cdot \ind{\noerr(c, \sigma_p \oplus (\xi_n \oplus \xi_n') \oplus f_*(\xi_n \oplus \xi_n'))} \cdot h'(\xi_n \oplus \xi_n')
    \Big).
  \end{align*}
  {
  All of these equations mean that one side of the equation is defined if and only if the other side is defined, and when both sides are defined, they are the same.}
  The first equality holds because $\noerr(c,\sigma_p \oplus \sigma_n \oplus f_*(\sigma_n))$
  implies that there exists a unique $K \subseteq \Name$ with $\used(c,\sigma_p \oplus \sigma_n \oplus f_*(\sigma_n), \sigma_n|_K)$;
  here we use $f_*(\sigma_n)(\like) = 1$.
  The second equality holds since $\Name$ is finite.
  The third equality holds because $\State[\Name]$ is isomorphic to $[K \to \R] \times [\Name \setminus K \to \R]$.
  The fourth equality holds 
  since $\xi_n' \in [L \to \R]$ with $L \neq \Name \setminus K$ implies $\ind{\dom(\xi_n) \uplus \dom(\xi_n') = \Name} = 0$.
  The fifth equality holds by the definition of $\Statesub[\Name]$ and its underlying measure.
  
  Using this result, we obtain the desired equation: 
  \begin{align*}
    & \int d\sigma_n \Big(
    \getpr(c)(\sigma_p, \sigma_n, f_*(\sigma_n))
    \cdot h\Big( \getval(c)(\sigma_p, \sigma_n, f_*(\sigma_n)) \Big)
    \Big)
    \\
    &= \int d\sigma_n \Big(
    \ind{\noerr(c, \sigma_p \oplus \sigma_n \oplus f_*(\sigma_n))}
    \cdot \getpr(c)(\sigma_p, \sigma_n, f_*(\sigma_n))
    \cdot h\Big( \getval(c)(\sigma_p, \sigma_n, f_*(\sigma_n)) \Big)
    \Big)
    \\
    &= \int d\xi_n \int d\xi_n' \Big(
    \ind{\dom(\xi_n) \uplus \dom(\xi_n') = \Name}
    \cdot \ind{\used(c, \sigma_p \oplus (\xi_n \oplus \xi_n') \oplus f_*(\xi_n \oplus \xi_n'), \xi_n)}
    \\
    & \qquad\quad
    {} \cdot \ind{\noerr(c, \sigma_p \oplus (\xi_n \oplus \xi_n') \oplus f_*(\xi_n \oplus \xi_n'))}
    \\
    & \qquad\quad
    {} \cdot \getpr(c)(\sigma_p, \xi_n \oplus \xi_n', f_*(\xi_n \oplus \xi_n'))
    \cdot h\Big( \getval(c)(\sigma_p, \xi_n \oplus \xi_n', f_*(\xi_n \oplus \xi_n')) \Big)
    \Big)
    \\
    &= \int d\xi_n \int d\xi_n' \Big(
    \ind{\dom(\xi_n) \uplus \dom(\xi_n') = \Name}
    \cdot \ind{\used(c, \sigma_p \oplus (\xi_n \oplus \xi_n') \oplus f_*(\xi_n \oplus \xi_n'), \xi_n)}
    \\
    & \qquad\quad
    \cdot \getpr(c)(\sigma_p, \xi_n \oplus \xi_n', f_*(\xi_n \oplus \xi_n'))
    \cdot h\Big( \getval(c)(\sigma_p, \xi_n \oplus \xi_n', f_*(\xi_n \oplus \xi_n')) \Big)
    \Big)
    \\
    &= \int d\xi_n \int d\xi_n' \Big(
    \ind{\dom(\xi_n) \uplus \dom(\xi_n') = \Name}
    \cdot \ind{\used(c, \sigma_p \oplus (\xi_n \oplus \xi_n') \oplus f_*(\xi_n \oplus \xi_n'), \xi_n)}
    \\
    & \qquad\quad
    \cdot \getprsub(c)(\sigma_p, \xi_n) \cdot \Big( \prod_{\mu \in \dom(\xi'_n)} f_1(\xi'_n(\mu)) \Big)
    \cdot h\Big( \getvalsub(c)(\sigma_p, \xi_n) \oplus \Big(\lambda \mu \in \dom(\xi'_n).\, f_2(\xi'_n(\mu))\Big) \Big)
    \Big)
    \\
    &= \int d\xi_n \int d\xi_n' \Big(
    \ind{\dom(\xi_n) \uplus \dom(\xi_n') = \Name}
    \\
    & \qquad\quad
    \cdot \getprsub(c)(\sigma_p, \xi_n) \cdot \Big( \prod_{\mu \in \dom(\xi'_n)} f_1(\xi'_n(\mu)) \Big)
    \cdot h\Big( \getvalsub(c)(\sigma_p, \xi_n) \oplus \Big(\lambda \mu \in \dom(\xi'_n).\, f_2(\xi'_n(\mu))\Big) \Big)
    \Big)
    \\
    &= \int d\xi_n \Big(  \getprsub(c)(\sigma_p, \xi_n)
    \cdot \int d\xi_n' \Big(
    \ind{\dom(\getvalsub(c)(\sigma_p, \xi_n)) \uplus \dom(\xi_n') = \Name}
    \cdot \Big( \prod_{\mu \in \dom(\xi'_n)} f_1(\xi'_n(\mu)) \Big)
    \\
    &\phantom{= \int d\xi_n \Big(  \getprsub(c)(\sigma_p, \xi_n) \cdot \int d\xi_n' \Big(}
    \cdot h\Big( \getvalsub(c)(\sigma_p, \xi_n) \oplus \Big(\lambda \mu \in \dom(\xi'_n).\, f_2(\xi'_n(\mu))\Big) \Big)
    \Big) \Big)
    \\
    &= \int d\xi_n \Big(
    \getprsub(c)(\sigma_p, \xi_n) \cdot g\Big( \getvalsub(c)(\sigma_p, \xi_n) \Big)
    \Big)
  \end{align*}
  where $g : \Statesub[\Name] \to \R$ is defined as in the statement of this lemma,
  {and each equation again means that one side of it is defined if and only if the other side is defined, and when both sides are defined, they are the same.}
  The first and third equalities hold because
  $\getpr(c)(\sigma_p, \sigma_n, f_*(\sigma_n)) \neq 0$ implies
  $\ind{\noerr(c, \sigma_p \oplus \sigma_n \oplus f_*(\sigma_n))} = 1$.
  The second equality uses the equation that we have shown in the previous paragraph.
  The fourth equality holds because of the following reason:
  if 
  \[
  \ind{\dom(\xi_n) \uplus \dom(\xi_n') = \Name}
  \cdot \ind{\used(c, \sigma_p \oplus \xi_n \oplus \sigma_r, \xi_n)}=1
    \quad \text{for}\  \sigma_r \defeq \xi_n' \oplus f_*(\xi_n \oplus \xi_n'),
  \]
  then
  \begin{align*}
    \getpr(c)(\sigma_p, \xi_n \oplus \xi_n', f_*(\xi_n \oplus \xi_n'))
    &= \textstyle
    \db{c}(\sigma_p \oplus \xi_n \oplus \sigma_r)(\like)
    \cdot \prod_{\mu \in \Name} \db{c}(\sigma_p \oplus \xi_n \oplus \sigma_r)(\pr_\mu)
    \\
    &= \textstyle
    \db{c}(\sigma_p \oplus \xi_n \oplus \sigma_r)(\like)
    \cdot \prod_{\mu \in \dom(\xi_n)} \db{c}(\sigma_p \oplus \xi_n \oplus \sigma_r)(\pr_\mu)
    \\
    &\qquad \textstyle
    \cdot  \prod_{\mu \in \Name \setminus \dom(\xi_n)} (\sigma_p \oplus \xi_n \oplus \sigma_r)(\pr_\mu)
    \\
    &=  \textstyle
    \getprsub(c)(\sigma_p, \xi_n)
    \cdot \prod_{\mu \in \Name \setminus \dom(\xi_n)} (\sigma_p \oplus \xi_n \oplus \sigma_r)(\pr_\mu)
    \\
    &= \textstyle
    \getprsub(c)(\sigma_p, \xi_n)
    \cdot \prod_{\mu \in \dom(\xi'_n)} f_1(\xi'_n(\mu))
    \end{align*}
    and
    \begin{align*} 
    \getval(c)(\sigma_p, \xi_n \oplus \xi_n', f_*(\xi_n \oplus \xi_n'))
    &= \lambda \mu \in \Name.\, \db{c}(\sigma_p \oplus \xi_n \oplus \sigma_r)(\val_\mu)
    \\
    &= \big(\lambda \mu \in \dom(\xi_n).\, \db{c}(\sigma_p \oplus \xi_n \oplus \sigma_r)(\val_\mu)\big)
    \\
    &\quad \textstyle
    \oplus \big( \lambda \mu \in \Name \setminus \dom(\xi_n).\, (\sigma_p \oplus \xi_n \oplus \sigma_r)(\val_\mu) \big)
    \\
    &= \getvalsub(c)(\sigma_p, \xi_n) 
    \oplus \big( \lambda \mu \in \Name \setminus \dom(\xi_n').\, (\sigma_p \oplus \xi_n \oplus \sigma_r)(\val_\mu) \big)
    \\
    &= \getvalsub(c)(\sigma_p, \xi_n) \oplus \big( \lambda \mu \in \dom(\xi'_n).\, f_2(\xi'_n(\mu)) \big).
  \end{align*}
  These equalities for $\getpr(c)$ and $\getval(c)$ themselves hold for the below reasons:
  \begin{itemize}
  \item The first equalities hold by
    $\used(c, \sigma_p \oplus \xi_n \oplus \sigma_r, \xi_n)$ and the definitions of $\getpr(c)$ and $\getval(c)$.
  \item The second equalities hold by \cref{lem:used-properties}-\cref{lem:used-properties-2},
    which is applicable since $\used(c, \sigma_p \oplus \xi_n \oplus \sigma_r, \xi_n)$.
  \item The third equalities hold by $\used(c, \sigma_p \oplus \xi_n \oplus \sigma_r, \xi_n)$ and the definitions of $\getprsub(c)$ and $\getvalsub(c)$.
  \item The fourth equalities hold by $\dom(\xi_n) \uplus \dom(\xi_n') = \Name$ and the definition of $f_*$.
  \end{itemize}
  Returning back to the main equations, we point out that the fifth equality comes from the next fact:
  \begin{align*}
    \Big(\ind{\dom(\xi_n) \uplus \dom(\xi_n') = \Name} \cdot \getprsub(c)(\sigma_p, \xi_n)\Big) \neq 0 
    \implies \ind{\used(c, \sigma_p \oplus (\xi_n \oplus \xi_n') \oplus f_*(\xi_n \oplus \xi_n'), \xi_n)} = 1.
  \end{align*}
  The justification for this implication is given below:
  \begin{itemize}
  \item If the premise holds, then there exists $\sigma_r$ such that $\used(c, \sigma_p \oplus \xi_n \oplus \sigma_r, \xi_n)$.
    Since $\sigma_r(\like)=1=f_*(-)(\like)$,
    $\sigma_p \oplus \xi_n \oplus \sigma_r$ and $\sigma_p \oplus (\xi_n \oplus \xi_n') \oplus f_*(\xi_n \oplus \xi_n')$
    coincide on $\PVar \cup \dom(\xi_n) \cup \{\like\}$.
    Thus, \cref{lem:used-sigma-indep} gives the conclusion.
  \end{itemize}
  Again back to the main equations, we note that
  the sixth equality holds since 
  \[
  \dom(\xi_n) = \dom(\getvalsub(c)(\sigma_p,\xi_n)),
  \]
  and the seventh equality follows from the definition of~$g$.
  This completes the proof.
\end{proof}

\begin{lemma}
\label{lem:cnt-increase-and-rv-nochange}
For all commands $c$ and states $\sigma \in \State$ such that $\db{c}\sigma \in \State$, we have 
\[
\db{c}\sigma(\sampled_\mu) \geq \sigma(\sampled_\mu)
\quad\text{and}\quad
\db{c}\sigma(\mu) = \sigma(\mu) 
\]
for all $\mu \in \Name$.
\end{lemma}
\begin{proof}
We prove the lemma by induction on the structure of $c$. Let $\sigma \in \State$ such that $\db{c}\sigma \in \State$.

\paragraph{\bf Cases $c \equiv \cskip$, or $c \equiv (x:=e)$, or $c \equiv \cobs(d,r)$} In these cases, $\db{c}\sigma(\sampled_\mu) = \sigma(\sampled_\mu)$ and $\db{c}\sigma(\mu) = \mu$ for all $\mu$. The claim of the lemma, thus, follows.

\paragraph{\bf Case $c \equiv (x:=\csample(n,d,\lambda y.e'))$} Let $\mu \defeq \db{n}\sigma$, $p \defeq \db{d}\sigma$, and $r \defeq \db{e'[\mu/y]}\sigma$. Then,
\[
\db{c}\sigma = \sigma[x \mapsto r, \val_\mu \mapsto r, \pr_\mu \mapsto p(r), \sampled_\mu \mapsto \sigma(\sampled_\mu) + 1].
\]
Thus, the claim of the lemma follows.

\paragraph{\bf Case $c \equiv (c';c'')$} Pick $\mu \in \Name$. Then,
\[
\db{c';c''}\sigma(\mu) = \db{c''}(\db{c'}\sigma)(\mu) = \db{c'}\sigma(\mu) = \sigma(\mu).
\]
Here the second and third equalities use induction hypothesis on $c'$ and $c''$, respectively. Also, 
\begin{align*}
\db{c';c''}\sigma(\sampled_\mu) - \sigma(\sampled_\mu)
& = 
\Big(\db{c''}(\db{c'}\sigma)(\sampled_\mu) - \db{c'}\sigma(\sampled_\mu)\Big)
+
\Big(\db{c'}\sigma(\sampled_\mu) - \sigma(\sampled_\mu)\Big)
\\
& \geq 0.
\end{align*}
The inequality here uses induction hypothesis on $c'$ and $c''$. 

\paragraph{\bf Case $c \equiv (\cif\ b\ \{c'\}\ \celse\ \{c''\}$)} Assume that $\db{b}\sigma = \strue$. We will prove the claims of the lemma under this assumption. The other case of $\db{b} = \sfalse$ can be proved similarly. Pick $\mu \in \Name$. Then, by induction hypothesis on $c'$,
\[
\db{c}\sigma(\mu) = \db{c'}\sigma(\mu) = \sigma(\mu)
\quad
\text{and}
\quad
\db{c}\sigma(\sampled_\mu) = \db{'}\sigma(\sampled_\mu) \geq \sigma(\sampled_\mu).
\]

\paragraph{\bf Case $c \equiv (\cwhile\ b\ \{c'\})$} Let $\cT$ be the following subset of $[\State \to \State_\bot]$:
\[
f \in \cT \iff \forall \sigma \in \State.\,\Big(f(\sigma) \neq \bot \implies \forall \mu \in \Name.\, f(\sigma)(\mu) = \sigma(\mu) \land f(\sigma)(\sampled_\mu) \geq \sigma(\sampled_\mu)\Big).
\]
Let $F$ be the operator on $[\State \to \State_\bot]$ whose least fixed point becomes the semantics of the loop $c$. The desired conclusion follows if we show that $\cT$ contains $\lambda \sigma.\bot$ and is closed under taking the limit of a chain in $\cT$, and $F$ preserves $\cT$. The least element $\lambda \sigma.\bot$ belongs to $\cT$ since there are no states $\sigma$ with $(\lambda \sigma.\bot)(\sigma) \neq \bot$. Consider an increasing sequence $f_0,f_1,\ldots$ in $\cT$, and let $f_\infty \defeq \bigsqcup_{n \in \N} f_n$. Pick $\sigma$ such that $f_\infty(\sigma) \neq \bot$. Then, $f_\infty(\sigma) = f_m(\sigma)$ for some $m \in \N$. Since $f_m \in \cT$, we have
\[
f_m(\sigma)(\mu) = \sigma(\mu)\quad\text{and}\quad f_m(\sigma)(\sampled_\mu) \geq \sigma(\sampled_\mu)
\]
for all $\mu \in \Name$. Since $f_m(\sigma) = f_\infty(\sigma)$, we also have, for every $\mu \in \Name$, $f_\infty(\sigma)(\mu) = \sigma(\mu)$ and $f_\infty(\sigma)(\sampled_\mu) \geq \sigma(\sampled_\mu)$, as desired. It remains to show that $F(f) \in \cT$ for all $f \in \cT$. Pick $f \in \cT$ and $\sigma \in \State$ such that
$F(f)(\sigma) \in \State$. If $\db{b}\sigma = \sfalse$, we have $F(f)(\sigma) = \sigma$, and the claims of the lemma follow. Otherwise, $F(f)(\sigma) = f(\db{c'}\sigma)$. Pick $\mu \in \Name$. Then, by induction hypothesis on $c'$ and the membership $f \in \cT$, 
\[
F(f)(\sigma)(\mu) = f(\db{c'}\sigma)(\mu) = \db{c'}\sigma(\mu) = \sigma(\mu),
\]
and
\[
F(f)(\sigma)(\sampled_\mu) = f(\db{c'}\sigma)(\sampled_\mu) \geq \db{c'}\sigma(\sampled_\mu) \geq \sigma(\sampled_\mu).
\]
We have just shown that $F(f) \in \cT$, as desired.
\end{proof}

\begin{definition}
  Define $\usedm$ as the predicate $\used$ but without the condition that $\like$ should be $1$. That is, for
  all commands $c$, states $\sigma \in \State$, and $\xi_n \in \Statesub[\Name]$,
  \[
  \usedm(c, \sigma, \xi_n)
  \iff
  \begin{aligned}[t]
    \db{c}\sigma \in \State 
    & {} \land \big(\db{c}\sigma(\sampled_\mu) - \sigma(\sampled_\mu) \leq 1\ \text{for all}\ \mu \in \Name\big)
    \\
    & {} \land \xi_n = \sigma|_{\dom(\xi_n)} 
    \\
    & {} \land \dom(\xi_n) = \{ \mu \in \Name \mid  \db{c}\sigma(\sampled_\mu) - \sigma(\sampled_\mu) = 1\}.
  \end{aligned}
  \]
\end{definition}

\begin{lemma}
  \label{lem:used-properties}
  Let $c$ be a command, $\sigma_0, \sigma_1 \in \State$, and $\xi_n \in \Statesub[\Name]$.
  Suppose that $\usedm(c, \sigma_0, \xi_n)$ 
  and $\sigma_1|_V = \sigma_0|_V$ for $V \defeq \PVar \cup \dom(\xi_n)$.
  Then, the following properties hold:
  \begin{enumerate}[label=(\arabic*),ref=(\arabic*)] 
  \item \label{lem:used-properties-1}
    $\db{c}\sigma_1 \in \State$.
  \item \label{lem:used-properties-2}
    $\db{c}\sigma_1(a) = \sigma_1(a)$ for all $a \in\{ \pr_\mu, \val_\mu, \sampled_\mu \mid \mu \in \Name \setminus \dom(\xi_n) \}$.
  \item \label{lem:used-properties-3}
    $\db{c}\sigma_1(v) = \db{c}\sigma_0(v)$ for all $v \in \PVar \cup \{ \pr_\mu, \val_\mu \mid \mu \in \dom(\xi_n) \}$.
  \item \label{lem:used-properties-3-like}
    $\db{c}\sigma_1(\like) = \db{c}\sigma_0(\like)$, if $\sigma_0(\like) = \sigma_1(\like)$.
  \item \label{lem:used-properties-4}
    $\db{c}\sigma_1(a)-\sigma_1(a) = \db{c}\sigma_0(a)-\sigma_0(a)$ for all $a \in \{ \sampled_\mu \mid \mu \in \Name \}$.
  \end{enumerate}
\end{lemma}
\begin{proof}
  For $\sigma_0',\sigma_1' \in \State$ and $\xi'_n \in \Statesub[\Name]$, write 
  \[
  \sigma_0' \sim_{\xi'_n} \sigma_1'
  \] 
  to mean that $\sigma_0'|_V = \sigma_1'|_V$ for $V \defeq \PVar \cup \dom(\xi'_n)$. Note that using this notation, we can write the conditions
  of the lemma as follows:
  \[
  \usedm(c,\sigma_0,\xi_n) \land \sigma_0 \sim_{\xi_n} \sigma_1.
  \]
  We will prove, by induction on the structure of $c$, that these conditions imply the five properties claimed by the lemma. Our proof will sometimes use a simple observation that the five properties claimed by the lemma and the relationship $\sigma_0 \sim_{\xi_n} \sigma_1$ imply $\usedm(c,\sigma_1,\xi_n)$. One consequence of the observation is that if our lemma holds, its five properties also hold with $\sigma_0$ and $\sigma_1$ swapped. We will often use this consequence.
  
  \paragraph{\bf Case $c \equiv \cskip$} In this case, $\db{c}\sigma_0 = \sigma_0$ and $\db{c}\sigma_1 = \sigma_1$. From these equalities, the claimed properties \cref{lem:used-properties-1,lem:used-properties-2,lem:used-properties-3-like,lem:used-properties-4} follow. For the remaining property \cref{lem:used-properties-3}, we note that $\dom(\xi_n) = \emptyset$ and the property, thus, follows from $\sigma_0 \sim_{\xi_n} \sigma_1$.
  
  \paragraph{\bf Case $c \equiv (x:=e)$} In this case, $\db{c}\sigma_0 = \sigma_0[x\mapsto \db{e}\sigma_0]$ and $\db{c}\sigma_1 = \sigma_1[x\mapsto \db{e}\sigma_1]$. The results are not $\bot$, and they are identical to the pre-states $\sigma_0$ and $\sigma_1$ as far as auxiliary variables in $\AVar$ are concerned. Also, expressions in commands do not depend on variables other than program variables, so that $\sigma_0 \sim_{\xi_n} \sigma_1$ gives $\db{e}\sigma_0 = \db{e}\sigma_1$ and $\db{c}\sigma_0(x) = \db{c}\sigma_1(x)$ for all $x \in \PVar$. From all of these observations, the claimed properties \cref{lem:used-properties-1,lem:used-properties-2,lem:used-properties-3,lem:used-properties-3-like,lem:used-properties-4} follow.
  
  \paragraph{\bf Case $c \equiv (x := \csample(n,d,\lambda y.e'))$} Since $\sigma_0(x) = \sigma_1(x)$ for all $x \in \PVar$, we have
  $\db{n}\sigma_0 = \db{n}\sigma_1$ and $\db{d}\sigma_0 = \db{d}\sigma_1$. Let $\mu \defeq \db{n}\sigma_0$, $p \defeq \db{d}\sigma_0$, and $r \defeq \db{e'[\mu/y]}\sigma_0$.
  By the semantics of the sample commands, we have
  \[
  \db{c}\sigma_0 = \sigma_0[x \mapsto r,\val_\mu \mapsto r,\pr_\mu \mapsto p(\sigma_0(\mu)), \sampled_\mu \mapsto \sigma_0(\sampled_\mu) + 1]. 
  \]
  Since $\usedm(c,\sigma_0,\xi_n)$ holds, we have $\xi_n = \sigma_0|_{\{\mu\}}$, which in turn implies $\sigma_0(\mu) = \sigma_1(\mu)$ because $\sigma_0 \sim_{\xi_n} \sigma_1$. Thus, $\db{e'[\mu/y]}\sigma_1 = \db{e'[\mu/y]}\sigma_0 = r$, and
  \begin{align*}
  \db{c}\sigma_1 
  & = \sigma_1[x \mapsto \db{e'[\mu/y]}\sigma_1,\val_\mu \mapsto \db{e'[\mu/y]}\sigma_1,\pr_\mu \mapsto p(\sigma_1(\mu)), \sampled_\mu \mapsto \sigma_1(\sampled_\mu) + 1]
  \\
  & = \sigma_1[x \mapsto r, \val_\mu \mapsto r, \pr_\mu \mapsto p(\sigma_0(\mu)), \sampled_\mu \mapsto \sigma_1(\sampled_\mu) + 1].
  \end{align*}
  The RHS of the last equality implies that the five properties claimed by the lemma hold.
  
  \paragraph{\bf Case $c \equiv \cobs(d,r)$} We have $\db{d}\sigma_0 = \db{d}\sigma_1$ since $\sigma_0(x) = \sigma_1(x)$ for all $x \in \PVar$. Let $p \defeq \db{d}\sigma_0$. Then,
  \[
  \db{c}\sigma_0 = \sigma_0[\like \mapsto \sigma_0(\like) \cdot p(r)]
  \quad\text{and}\quad
  \db{c}\sigma_1 = \sigma_1[\like \mapsto \sigma_1(\like) \cdot p(r)].
  \]
  Also, $\dom(\xi_n) = \emptyset$ since $\usedm(c,\sigma_0,\xi_n)$ holds.
  From what we have proved and also the agreement of $\sigma_0$ and $\sigma_1$ on program variables, the five properties claimed by the lemma follow.
  
  \paragraph{\bf Case $c \equiv (c';c'')$} Since $\db{c}\sigma_0 = \db{c''}^\dagger (\db{c'}\sigma_0) \in \State$, we have $\db{c'}\sigma_0 \in \State$. Let
  \begin{align*}
  \sigma'_0 & \defeq \db{c'}\sigma_0,
  \\
  N_0 & \defeq \{\mu \in \Name \mid \db{c''}\sigma'_0(\sampled_\mu) - \sigma_0(\sampled_\mu) = 1\},
  \\
  N'_0 & \defeq \{\mu \in \Name \mid \sigma'_0(\sampled_\mu) - \sigma_0(\sampled_\mu) = 1\}. 
  \end{align*}
  Then, $N_0 = \dom(\xi_n)$ because $\usedm(c';c'',\sigma_0,\xi_n)$ holds.  
  We will prove the following facts:
  \begin{enumerate}
  \item $N'_0 \subseteq N_0$.
  \item Let $\xi'_n \defeq \xi_n|_{N'_0}$, and $\xi''_n \defeq \xi_n|_{(N_0 \setminus N'_0)}$. Then, 
  $\usedm(c',\sigma_0,\xi'_n)$ and $\usedm(c'',\sigma'_0,\xi''_n)$ hold.
  \item $\db{c'}\sigma_1 \in \State$.
  \item Let $\sigma'_1 \defeq \db{c'}\sigma_1$. Then, $\sigma'_0 \sim_{\xi''_n} \sigma'_1$.
  \end{enumerate}
  These four facts imply the five properties claimed by the lemma. Here is the reason. Note that $\sigma_0 \sim_{\xi'_n} \sigma_1$ since $\dom(\xi'_n) = N'_0 \subseteq N_0 = \dom(\xi_n)$. This relationship between $\sigma_0$ and $\sigma_1$ and the second fact let us use induction hypothesis on $(c',\xi'_n,\sigma_0,\sigma_1)$. Also, the second and fourth facts allow us to use induction hypothesis on $(c'',\xi''_n,\sigma_0',\sigma_1')$. 
  We can derive the five properties from what we get from these two applications of induction hypothesis:
  \begin{enumerate}
  \item By induction hypothesis on $(c'',\xi''_n,\sigma_0',\sigma_1')$, we have $\db{c';c''}\sigma_1 = \db{c''}\sigma'_1 \in \State$.
  \item For all $a \in \{ \pr_\mu, \val_\mu, \sampled_\mu \mid \mu \in \Name \setminus \dom(\xi_n)\}$, 
  \[
  \db{c';c''}\sigma_1(a)
  =
  \db{c''}\sigma'_1(a)
  =
  \sigma'_1(a)
  = 
  \db{c'}\sigma_1(a)
  =
  \sigma_1(a).
  \]
  The second equality comes from induction hypothesis on $(c'',\xi''_n,\sigma_0',\sigma_1')$ and $\dom(\xi''_n) \subseteq \dom(\xi_n)$,
  and the fourth equality from induction hypothesis on $(c',\xi'_n,\sigma_0,\sigma_1)$ and $\dom(\xi'_n) \subseteq \dom(\xi_n)$.
  \item For all $v \in \PVar \cup \{\pr_\mu,\val_\mu \mid \mu \in \dom(\xi''_n)\}$, by induction hypothesis on $(c'',\xi''_n,\sigma_0',\sigma_1')$,
  \[
  \db{c';c''}\sigma_1(v)
  =
  \db{c''}\sigma'_1(v)
  =
  \db{c''}\sigma'_0(v)
  =
  \db{c';c''}\sigma_0(v).
  \]
  Also, for all $a \in \{\pr_\mu,\val_\mu \mid \mu \in \dom(\xi'_n)\}$, we have $a \in \{\pr_\mu,\val_\mu \mid \mu \in \Name \setminus \dom(\xi''_n)\}$,
  and we can calculate:
  \begin{align*}
  \db{c';c''}\sigma_1(a)
  & {} =
  \db{c''}\sigma'_1(a)
   =
  \sigma'_1(a)
  \\
  & {}
  =
  \db{c'}\sigma_1(a)
  =
  \db{c'}\sigma_0(a)
  \\
  & {}
  =
  \sigma'_0(a)
  =
  \db{c''}\sigma'_0(a)
  =
  \db{c';c''}\sigma_0(a).
  \end{align*}
  The second equality uses induction hypothesis on $(c'',\xi''_n,\sigma_0',\sigma_1')$, and the fourth comes from the induction hypothesis on $(c',\xi'_n,\sigma_0,\sigma_1)$. The sixth equality follows from induction hypothesis applied to $(c'',\xi''_n,\sigma'_0,\sigma'_1)$ and again to the same tuple but with $\sigma'_0$ and $\sigma'_1$ swapped.  
  \item If $\sigma_0(\like) = \sigma_1(\like)$, by induction hypothesis on $(c',\xi'_n,\sigma_0,\sigma_1)$, 
  \[
  \sigma'_0(\like) = \db{c'}\sigma_0(\like) = \db{c'}\sigma_1(\like) = \sigma'_1(\like),
  \]
  which in turn implies, by induction hypothesis on $(c'',\xi''_n,\sigma'_0,\sigma'_1)$, 
  \[
  \db{c';c''}\sigma_0(\like) = \db{c''}\sigma'_0(\like) = \db{c''}\sigma'_1(\like) = \db{c';c''}\sigma_1(\like).
  \]
  \item For all $a \in \{\sampled_\mu \mid \mu \in \Name\}$,
  \begin{align*}
  \db{c';c''}\sigma_1(a) - \sigma_1(a) & 
  =
  \db{c';c''}\sigma_1(a) - \db{c'}\sigma_1(a) + \db{c'}\sigma_1(a) - \sigma_1(a)
  \\
  & 
  =
  \db{c''}\sigma'_1(a) - \sigma'_1(a) + \db{c'}\sigma_1(a) - \sigma_1(a)
  \\
  & 
  =
  \db{c''}\sigma'_0(a) - \sigma'_0(a) + \db{c'}\sigma_0(a) - \sigma_0(a)
  \\
  &
  = 
  \db{c';c''}\sigma_0(a) - \sigma_0(a). 
  \end{align*}
  The only non-trivial inequality is the third one, and it follows from induction hypothesis on $(c',\xi'_n,\sigma_0,\sigma_1)$ and $(c'',\xi''_n,\sigma_0',\sigma_1')$.
  \end{enumerate}
  
  We prove the four facts as follows: 
  \begin{enumerate}
  \item Let $\mu \in N_0'$. Since $\usedm(c';c'',\sigma_0,\xi_n)$, we have 
  \[
  \db{c''}\sigma'_0(\mu) - \sigma_0(\mu) = \db{c';c''}\sigma_0 - \sigma_0(\mu) \leq 1.
  \]
  Also, by \cref{lem:cnt-increase-and-rv-nochange} and the definition of $N_0'$,
  \[
  \db{c''}\sigma'_0(\mu) - \sigma_0(\mu) \geq \sigma'_0(\mu) - \sigma_0(\mu) = 1.
  \] 
  Thus, $\db{c''}\sigma'_0(\mu) - \sigma_0(\mu) = 1$, which implies that $\mu \in N_0$, as desired.
  
  \item We should show that $\usedm(c',\sigma_0,\xi_n')$ and $\usedm(c'',\sigma_0',\xi_n'')$ hold. The conjuncts in the definition of $\usedm(c',\sigma_0,\xi_n')$ except the
  second follow immediately from $\usedm(c';c'',\sigma_0,\xi_n)$ and the definition of $\xi_n$. For the remaining second conjunct, we use
  \cref{lem:cnt-increase-and-rv-nochange} and $\usedm(c';c'',\sigma_0,\xi_n)$, and prove the conjunct as shown below: for all $\mu \in \Name$,
  \[
  \db{c'}\sigma_0(\sampled_\mu) - \sigma_0(\sampled_\mu) \leq \db{c''}(\db{c'}\sigma_0)(\sampled_\mu) - \sigma_0(\sampled_\mu) \leq 1.
  \]
  For $\usedm(c'',\sigma_0',\xi_n'')$, we first note that the first and third conjuncts in its definition are direct consequences of $\usedm(c';c'',\sigma_0,\xi_n)$ and the definition of $\xi''_n$. We prove the second conjunct in the definition as follows: for all $\mu \in \Name$,
  \begin{align*}
  \db{c''}\sigma_0'(\sampled_\mu) - \sigma_0'(\sampled_\mu) 
  & {} =
  \db{c';c''}\sigma_0(\sampled_\mu) - \db{c'}\sampled_0(\sampled_\mu)
  \\
  & {} \leq
  \db{c';c''}\sigma_0(\sampled_\mu) - \sampled_0(\sampled_\mu) 
  \\
  & {}
  \leq 1.
  \end{align*}
  The first inequality uses \cref{lem:cnt-increase-and-rv-nochange}, and the second comes from $\usedm(c';c'',\sigma_0,\xi_n)$. It remains to show the fourth conjunct
  in the definition of $\usedm(c'',\sigma_0',\xi_n'')$, which we do below: for all $\mu \in \Name$,
  \begin{align*}
  & {} \db{c''}\sigma'_0(\sampled_\mu) - \sigma'_0(\sampled_\mu) = 1
  \\
  & \qquad {} \iff
    \db{c''}\sigma'_0(\sampled_\mu) - \sigma'_0(\sampled_\mu) = 1
    \land
    \sigma'_0(\sampled_\mu) - \sigma_0(\sampled_\mu) = 0
  \\
  & \qquad {} \iff 
  \mu \in N_0 \land \mu \not\in N'_0
  \\
  & \qquad {} \iff
  \mu \in \dom(\xi''_n).
  \end{align*}
  The first equivalence comes from \cref{lem:cnt-increase-and-rv-nochange} and $\db{c''}\sigma_0'(\sampled_\mu) - \sigma_0(\sampled_\mu) \leq 1$, which holds because of $\usedm(c';c'', \sigma_0, \xi_n)$. The second equivalence follows from the definitions of $N_0$ and $N'_0$.
  \item Since $\sigma_0 \sim_{\xi_n} \sigma_1$ implies $\sigma_0 \sim_{\xi'_n} \sigma_1$ and we have $\usedm(c',\sigma_0,\xi'_n)$, we can apply induction hypothesis to $(c',\xi'_n,\sigma_0,\sigma_1)$, and get $\db{c'}\sigma_1 \in \State$.
  \item We continue our reasoning in the previous item, and derive from induction hypothesis on $(c',\xi'_n,\sigma_0,\sigma_1)$ the fact that for all $x \in \PVar$,
  \[
  \sigma'_0(x) = \db{c'}\sigma_0(x) = \db{c'}\sigma_1(x) = \sigma'_1(x).
  \]
  Also, for all $\mu \in \dom(\xi_n'')$, 
  \[
  \sigma'_0(\mu) = \sigma_0(\mu) = \sigma_1(\mu) = \sigma_1'(\mu),
  \]
  where the first and third equalities come from \cref{lem:cnt-increase-and-rv-nochange}, and the second equality follows from the assumption that $\sigma_0 \sim_{\xi_n} \sigma_1$.
  \end{enumerate}
  
  \paragraph{\bf Case $c \equiv (\cif\ b\ \{c'\}\ \celse\ c'')$} Assume that $\db{b}\sigma_0 = \strue$. Then, $\db{c}\sigma_0 = \db{c'}\sigma_0$. We prove the five properties claimed by the lemma under this assumption. The proof for the other possibility, namely, $\db{b}\sigma_0 = \sfalse$ is similar. Since $\sigma_0 \sim_{\xi_n} \sigma_1$, the states $\sigma_0$ and $\sigma_1$ coincide for the values of program variables. Thus, $\db{b}\sigma_1 = \strue$, and $\db{c}\sigma_1 = \db{c'}\sigma_1$. Since $\db{c}\sigma_0 = \db{c'}\sigma_0$ as well, it suffices to show the five properties claimed by the lemma for $(c',\sigma_0,\sigma_1,\xi_n)$. This sufficient condition follows from induction hypothesis on $(c',\sigma_0,\sigma_1,\xi_n)$, since $\usedm(c,\sigma_0,\xi_n)$ and $\db{b}\sigma_0 = \strue$ imply $\usedm(c',\sigma_0,\xi_n)$.
  
  \paragraph{\bf Case $c \equiv (\cwhile\ b\ \{ c' \})$} Consider the version of $\usedm$ where the first parameter can be a state transformer $f : \State \to \State_\bot$, instead of a command. Similarly, consider the version of the five properties claimed by the lemma where we use a state transformer $f : \State \to \State_\bot$, again instead of a command. We denote both versions by $\usedm(f,\sigma''_0,\xi''_n)$ and $\varphi(f,\xi''_n,\sigma''_0,\sigma''_1)$. Let $\cT$ be the subset of $\State \to \State_\bot$ defined by
  \[
  f \in \cT \iff \forall \sigma''_0,\sigma''_1 \in \State.\,\forall \xi''_n \in \Statesub.\, \Big(\Big(\usedm(f,\sigma''_0,\xi''_n) \land \sigma''_0 \sim_{\xi''_n} \sigma''_1\Big) \implies \varphi(f,\xi''_n,\sigma''_0,\sigma''_1)\Big),
  \]
  and $F : [\State \to \State_\bot] \to [\State \to \State_\bot]$ be the following operator used in the semantics of the loop $\db{c}$:
  \[
  F(f)(\sigma) \defeq \text{if}\ (\db{b}\sigma = \strue)\ \text{then}\ f^\dagger(\db{c'}\sigma)\ \text{else}\ \sigma.
  \]
  We will show that $\cT$ contains $\lambda \sigma.\,\bot$ and is closed under taking the least upper bound of an increasing chain in $[\State \to \State_\bot]$,
  and the operator $F$ preserves $\cT$. These three conditions imply that $\db{c}$ is in $\cT$, which in turn gives the five properties claimed by the lemma.
  
 The first condition holds simply because $\usedm((\lambda \sigma.\bot),\sigma''_0,\xi''_n)$ is false for all $\sigma''_0$ and $\xi''_n$. To prove the closure under the least upper bound of a chain, consider an increasing sequence $f_0,f_1,\ldots$ in $\cT$. Let $f_\infty \defeq \bigsqcup_{n \in \N} f_n$. Consider $\sigma''_0,\sigma''_1 \in \State$ and $\xi''_n \in \Statesub$ such that $\sigma_0'' \sim_{\xi''_n} \sigma''_1$ and $\usedm(f_\infty, \sigma''_0,\xi''_n)$. We should show that $\varphi(f_\infty,\xi''_n,\sigma''_0,\sigma''_1)$ holds. By the definition of $f_\infty$, there exists $m \in \N$ such that $f_\infty(\sigma_0) = f_m(\sigma_0)$. Then, the assumption $\usedm(f_\infty,\sigma''_0,\xi''_n)$ implies $\usedm(f_m,\sigma''_0,\xi''_n)$. This in turn gives $\varphi(f_m,\xi''_n,\sigma''_0,\sigma''_1)$ because $f_m \in \cT$. By what we have proved and the definition of $f_\infty$, we have 
  \[
  f_m(\sigma''_0) = f_\infty(\sigma''_0) \in \State
  \quad\text{and}\quad
  f_m(\sigma''_1) = f_\infty(\sigma''_1) \in \State.
  \]
Thus, $\varphi(f_m,\xi''_n,\sigma''_0,\sigma''_1)$ entails $\varphi(f_\infty,\xi_n'',\sigma''_0,\sigma''_1)$, as desired. It remains to show that $F(f) \in \cT$ for all $f \in \cT$. 
Pick $f \in \cT$. We first replay our proof for the sequential-composition case after viewing $f^\dagger \circ \db{c'}$ as the sequential composition of $c'$ and $f$. This replay, then, gives the membership $f^\dagger \circ \db{c'} \in \cT$. Next, we replay our proof for the if case on $F(f)$ after viewing $f^\dagger \circ \db{c'}$ as the true branch and $\lambda \sigma.\,\sigma = \db{\cskip}$ as the false branch. This replay implies the required $F(f) \in \cT$.
\end{proof}

\begin{lemma}
  \label{lem:used-sigma-indep}
  Let $c$ be a command, $\sigma, \sigma' \in \State$, and $\xi_n \in \Statesub[\Name]$.
  \begin{itemize}
  \item If $\sigma|_V = \sigma'|_V$ for $V \defeq \PVar \cup \dom(\xi_n) \cup \{\like\}$,
    then $\used(c, \sigma, \xi_n)$ implies $\used(c, \sigma', \xi_n)$.
  \item  If $\sigma|_U = \sigma'|_U$ for $U \defeq \PVar \cup \dom(\xi_n)$,
    then $\usedm(c, \sigma, \xi_n)$ implies $\usedm(c, \sigma', \xi_n)$.
  \end{itemize}
\end{lemma}
\begin{proof}
  Assume the settings in the statement of this lemma.
  For the first claim, assume $\used(c,\sigma,\xi_n)$.
  Then, by the definition of $\used$ and $\noerr$,
  \begin{align*}
    &\db{c}\sigma \in \State
    \land \big( \forall \mu \in \Name.\, \db{c}\sigma(\sampled_\mu) - \sigma(\sampled_\mu) \leq 1 \big)
    \land \big( \sigma(\like)=1 \big)
    \\
    &{}\land \big( \xi_n = \sigma|_{\dom(\xi_n)} \big)
    \land \big( \dom(\xi_n) = \{ \mu \in \Name \mid  \db{c}\sigma(\sampled_\mu) - \sigma(\sampled_\mu) = 1\} \big).
  \end{align*}
  From this and \cref{lem:used-properties}
  (which is applicable since $\used(c,\sigma,\xi_n)$ and $\sigma|_V = \sigma'|_V$),
  we obtain
  \begin{align*}
    &\db{c}\sigma' \in \State
    \land \big( \forall \mu \in \Name.\, \db{c}\sigma'(\sampled_\mu) - \sigma'(\sampled_\mu) \leq 1 \big)
    \land \big( \sigma'(\like)=1 \big)
    \\
    &{}\land \big( \xi_n = \sigma'|_{\dom(\xi_n)} \big)
    \land \big( \dom(\xi_n) = \{ \mu \in \Name \mid  \db{c}\sigma'(\sampled_\mu) - \sigma'(\sampled_\mu) = 1\} \big).
  \end{align*}
  Note that we have the first clause by \cref{lem:used-properties}-\cref{lem:used-properties-1},
  the second and fifth clauses by \cref{lem:used-properties}--\cref{lem:used-properties-4},
  and the third and fourth clauses by $\sigma|_V = \sigma'|_V$.
  Hence, $\used(c,\sigma',\xi_n)$ holds.
  The proof of the second claim is exactly the same except that
  we apply \cref{lem:used-properties} to $\sigma|_U = \sigma'|_U$
  to prove only the four clauses of $\used$ that exclude $\sigma'(\like)=1$.
\end{proof}

\subsection{Proof of \cref{lem:integral-same-under-reparam}}
\label{sec:proof:lemmas:unbiased-val2}

\begin{proof}[Proof of \cref{lem:integral-same-under-reparam}]
  We prove this lemma by induction on the structure of $c$.
  Let $g : \State[\PVar] \times \Statesub[\Name] \to \R$ be a measurable function and $\sigma_p \in \State[\PVar]$.
  In this proof, each equation involving integrals means (otherwise noted) that
  one side of the equation is defined if and only if the other side is defined,
  and when both sides are defined, they are the same.

  \paragraph{\bf Case $c \equiv \cskip$, $c \equiv (x:=e)$, or $c \equiv \cobs(d,r)$}
  In this case, $\ctr{c}{\pi} \equiv c$ so the desired equation holds.

  \paragraph{\bf Case $c \equiv (x := \csample(n,d,\lambda y.e))$}
  If $(n,d,\lambda y.e) \notin \dom(\pi)$, then $\ctr{c}{\pi} \equiv c$ and thus the desired equation holds.
  So assume that $\pi(n,d,\lambda y.e) = (d',\lambda y'.e')$ for some $d'$ and $\lambda y'.e'$.
  Then, $\ctr{c}{\pi} \equiv (x := \csample(n,d',\lambda y'.e'))$.

  First, by the validity of $\pi$,
  for all states $\sigma \in \State$ and measurable subsets $A \subseteq \R$,
  \begin{align*}
    \int \ind{\db{e[r/y]}\sigma \in A} \cdot \db{d}\sigma(r) \,dr
    = 
    \int \ind{\db{e'[r/y']}\sigma \in A} \cdot \db{d'}\sigma(r) \,dr,
  \end{align*}
  where both sides are always defined.
  Using this and the monotone convergence theorem, we can show that for all measurable $f : \R \to \R$,
  \begin{align}
    \label{eq:lem:integral-same-under-reparam-sam}
    \int f(\db{e[r/y]}\sigma) \cdot \db{d}\sigma(r) \,dr
    = 
    \int f(\db{e'[r/y']}\sigma) \cdot \db{d'}\sigma(r) \,dr.
  \end{align}
  
  Next, choose any $\sigma_{r_0} \in \State[\Var \setminus \PVar]$.
  Since $\fv(n) \subseteq \PVar$, there exists $\mu \in \Name$ such that
  \begin{align*}
    \db{n}(\sigma_p \oplus \sigma_r) &= \mu \quad\text{for all $\sigma_r \in \State[\Var \setminus \PVar]$}.
  \end{align*}
  Using this and $\fv(e),\fv(d) \subseteq \PVar$, we obtain the following: for any $\sigma_r \in \State[\Var \setminus \PVar]$,
  \begin{align*}
    \db{c}(\sigma_p \oplus \sigma_r)(\like)
    &= 1,
    \\
    \db{c}(\sigma_p \oplus \sigma_r)(\pr_\mu)
    &= \db{d}(\sigma_p \oplus \sigma_r)(\sigma_r(\mu))
    = \db{d}(\sigma_p \oplus \sigma_{r_0})(\sigma_r(\mu)),
    \\
    \db{c}(\sigma_p \oplus \sigma_r)(\val_\mu)
    &= \db{e[\sigma_r(\mu)/y]}(\sigma_p \oplus \sigma_r)
    =  \db{e[\sigma_r(\mu)/y]}(\sigma_p \oplus \sigma_{r_0}),
    \\
    \db{c}(\sigma_p \oplus \sigma_r)(\sampled_\mu)
    &= 1,
    \qquad \db{c}(\sigma_p \oplus \sigma_r)(\sampled_{\mu'})
    = 0 \quad\text{for $\mu' \not\equiv \mu$},
    \\
    \db{c}(\sigma_p \oplus \sigma_r)|_\PVar
    &= \sigma_p[ x \mapsto \db{e[\sigma_r(\mu)/y]}(\sigma_p \oplus \sigma_{r_0}) ].
  \end{align*}
  This implies that for any $\xi_n \in \Statesub[\Name]$, if $\getprsub(c)(\sigma_p, \xi_n) \neq 0$,
  then 
  \begin{align*}
    \dom(\xi_n) & = \{\mu\},
    \\
    \getprsub(c)(\sigma_p, \xi_n)
    &= 1 \cdot \db{d}(\sigma_p \oplus \sigma_{r_0})(\xi_n(\mu)),
    \\
    \getpvarsub(c)(\sigma_p, \xi_n)
    &= \sigma_p[x \mapsto \db{e[\xi_n(\mu)/y]}(\sigma_p \oplus \sigma_{r_0})], 
    \\
    \getvalsub(c)(\sigma_p, \xi_n)
    &= [\mu \mapsto \db{e[\xi_n(\mu)/y]}(\sigma_p \oplus \sigma_{r_0})].
  \end{align*}
  Note that the same equations hold for $\ctr{c}{\pi}$, except that we replace $d$, $e$, and $y$ in the RHS of the above equations by $d'$, $e'$, and $y'$.
  Using these, we obtain:
  \begin{align*}
    &
    \int d\xi_n \Big( \getprsub(c)(\sigma_p, \xi_n) \cdot
    g\Big(\getpvarsub(c)(\sigma_p, \xi_n), \getvalsub(c)(\sigma_p, \xi_n) \Big)\Big)
    \\
    &= \int_{[\{\mu\} \to \R]} d\xi_n \Big( \db{d}(\sigma_p \oplus \sigma_{r_0})(\xi_n(\mu)) \cdot
    \widehat{g}\Big( \db{e[\xi_n(\mu)/y]}(\sigma_p \oplus \sigma_{r_0}) \Big)\Big)
    \\
    &= \int_\R dr \Big( \db{d}(\sigma_p \oplus \sigma_{r_0})(r) \cdot
    \widehat{g}\Big( \db{e[r/y]}(\sigma_p \oplus \sigma_{r_0}) \Big)\Big)
    \\
    &= \int_\R dr \Big( \db{d'}(\sigma_p \oplus \sigma_{r_0})(r) \cdot
    \widehat{g}\Big( \db{e'[r/y']}(\sigma_p \oplus \sigma_{r_0}) \Big)\Big)
    \\
    &= \int_{[\{\mu\} \to \R]} d\xi_n \Big( \db{d'}(\sigma_p \oplus \sigma_{r_0})(\xi_n(\mu)) \cdot
    \widehat{g}\Big( \db{e'[\xi_n(\mu)/y']}(\sigma_p \oplus \sigma_{r_0}) \Big)\Big)
    \\
    &=
    \int d\xi_n \Big( \getprsub(\ctr{c}{\pi})(\sigma_p, \xi_n) \cdot
    g\Big(\getpvarsub(\ctr{c}{\pi})(\sigma_p, \xi_n), \getvalsub(\ctr{c}{\pi})(\sigma_p, \xi_n) \Big)\Big)
  \end{align*}
  where $\widehat{g} : \R \to \R$ is defined as
  $\widehat{g}(r) = g(\sigma_p[x \mapsto r], [\mu \mapsto r])$.
  Here the first and fifth equalities use the equations proven above,
  the second and fourth equalities use that $[\{\mu\} \to \R]$ is isomorphic to $\R$,
  and the third equality uses \cref{eq:lem:integral-same-under-reparam-sam}.
  This proves the desired equation.
  
  \paragraph{\bf Case $c \equiv (\cif\ b\ \{c'\}\ \celse\ c'')$}
  In this case, since $\fv(b) \subseteq \PVar$ and $\ctr{c}{\pi} \equiv \cif\ b\ \{\ctr{c'}{\pi}\}\ \celse\ \ctr{c''}{\pi}$,
  we have only two subcases:
  \begin{itemize}
  \item For all $\sigma_r \in \State[\Var \setminus \PVar]$,
    $\db{c}(\sigma_p \oplus \sigma_r) = \db{c'}(\sigma_p \oplus \sigma_r)$ and
    $\db{\ctr{c}{\pi}}(\sigma_p \oplus \sigma_r) = \db{\ctr{c'}{\pi}}(\sigma_p \oplus \sigma_r)$.
  \item For all $\sigma_r \in \State[\Var \setminus \PVar]$,    
    $\db{c}(\sigma_p \oplus \sigma_r) = \db{c''}(\sigma_p \oplus \sigma_r)$ and
    $\db{\ctr{c}{\pi}}(\sigma_p \oplus \sigma_r) = \db{\ctr{c''}{\pi}}(\sigma_p \oplus \sigma_r)$.
  \end{itemize}
  If the first subcase holds, we have
  \begin{align*}
    &
    \int d\xi_n \Big( \getprsub(c)(\sigma_p, \xi_n) \cdot
    g\Big(\getpvarsub(c)(\sigma_p, \xi_n), \getvalsub(c)(\sigma_p, \xi_n) \Big)\Big)
    \\
    &=
    \int d\xi_n \Big( \getprsub(c')(\sigma_p, \xi_n) \cdot
    g\Big(\getpvarsub(c')(\sigma_p, \xi_n), \getvalsub(c')(\sigma_p, \xi_n) \Big)\Big)
    \\
    &=
    \int d\xi_n \Big( \getprsub(\ctr{c'}{\pi})(\sigma_p, \xi_n) \cdot
    g\Big(\getpvarsub(\ctr{c'}{\pi})(\sigma_p, \xi_n), \getvalsub(\ctr{c'}{\pi})(\sigma_p, \xi_n) \Big)\Big)
    \\
    &=
    \int d\xi_n \Big( \getprsub(\ctr{c}{\pi})(\sigma_p, \xi_n) \cdot
    g\Big(\getpvarsub(\ctr{c}{\pi})(\sigma_p, \xi_n), \getvalsub(\ctr{c}{\pi})(\sigma_p, \xi_n) \Big)\Big)
  \end{align*}
  where the second equality is by IH on $c'$.
  If the second subcase holds, we obtain a similar equation by IH on $c''$.
  Hence, the desired equation holds in all subcases.

  \paragraph{\bf Case $c \equiv (c';c'')$}
  In this case, we obtain the following equation:
  \begin{align*}
    &
    \int d\xi_n \Big( \getprsub(c';c'')(\sigma_p, \xi_n) \cdot
    g\Big(\getpvarsub(c';c'')(\sigma_p, \xi_n), \getvalsub(c';c'')(\sigma_p, \xi_n) \Big)\Big)
    \\
    &=
    \int d\xi'_n \Big(
    \getprsub(c')(\sigma_p, \xi'_n) \cdot
    \int d\xi''_n \Big(
    \getprsub(c'')\big(\getpvarsub(c')(\sigma_p,\xi'_n), \xi''_n\big)
    \cdot \ind{\dom(\xi'_n) \cap \dom(\xi''_n) = \emptyset}
    \\
    & \qquad
    {} \cdot g\Big(
    \getpvarsub(c'')\big(\getpvarsub(c')(\sigma_p,\xi'_n),\xi''_n\big),
    \getvalsub(c')(\sigma_p,\xi'_n) \oplus
    \getvalsub(c'')\big(\getpvarsub(c')(\sigma_p,\xi'_n),\xi''_n\big)
    \Big)  \Big)\Big)
    \\
    &=
    \int d\xi'_n \Big(
    \getprsub(c')(\sigma_p, \xi'_n) \cdot
    g'\Big(
    \getpvarsub(c')(\sigma_p, \xi'_n),
    \getvalsub(c')(\sigma_p, \xi'_n) 
    \Big) \Big)
    \\
    & \qquad\text{where }
    g'(\widehat{\sigma_p'}, \widehat{\xi_n'})
    \defeq
    \int d\xi''_n \Big(
    \getprsub(c'')(\widehat{\sigma_p'}, \xi''_n)
    \cdot \ind{\dom(\widehat{\xi_n'}) \cap \dom(\xi''_n) = \emptyset}
    \\
    & \phantom{\qquad\text{where }
    g'(\widehat{\sigma_p'}, \widehat{\xi_n'})
    \defeq
    \int d\xi''_n \Big(}
    {} \cdot g\Big(
    \getpvarsub(c'')(\widehat{\sigma_p'},\xi''_n),
    \widehat{\xi_n'} \oplus \getvalsub(c'')(\widehat{\sigma_p'},\xi''_n)
    \Big)  \Big)
    \\
    &=
    \int d\xi'_n \Big(
    \getprsub(\ctr{c'}{\pi})(\sigma_p, \xi'_n) \cdot
    g'\Big(
    \getpvarsub(\ctr{c'}{\pi})(\sigma_p, \xi'_n),
    \getvalsub(\ctr{c'}{\pi})(\sigma_p, \xi'_n) 
    \Big) \Big) \quad\cdots\quad(*)
  \end{align*}
  where the first equality is from \cref{lem:integral-seq-decompose},
  the second equality uses $\dom(\xi_n') = \dom(\getvalsub(c')(\sigma_p,\xi_n'))$,
  and the third equality is by IH on $c'$.
  We now analyse $\smash{ g'(\widehat{\sigma_p'}, \widehat{\xi_n'}) }$ as follows: 
  \begin{align*}
    & g'(\widehat{\sigma_p'}, \widehat{\xi_n'})
    \\
    &=
    \int d\xi''_n \Big(
    \getprsub(c'')(\widehat{\sigma_p'}, \xi''_n)
    \cdot \ind{\dom(\widehat{\xi_n'}) \cap \dom(\xi''_n) = \emptyset}
    \cdot g\Big(
    \getpvarsub(c'')(\widehat{\sigma_p'},\xi''_n),
    \widehat{\xi_n'} \oplus \getvalsub(c'')(\widehat{\sigma_p'},\xi''_n)
    \Big)  \Big)
    \\
    &=
    \int d\xi''_n \Big(
    \getprsub(c'')(\widehat{\sigma_p'}, \xi''_n)
    \cdot g''\Big(
    \getpvarsub(c'')(\widehat{\sigma_p'}, \xi''_n),
    \getvalsub(c'')(\widehat{\sigma_p'}, \xi''_n)
    \Big) \Big)
    \\
    &\qquad \text{where }
    g''(\widehat{\sigma_p''}, \widehat{\xi_n''})
    \defeq
    \ind{\dom(\widehat{\xi_n'}) \cap \dom(\widehat{\xi''_n}) = \emptyset}
    \cdot g\Big(
    \widehat{\sigma_p''},
    \widehat{\xi_n'} \oplus \widehat{\xi''_n}
    \Big)
    \\
    &=
    \int d\xi''_n \Big(
    \getprsub(\ctr{c''}{\pi})(\widehat{\sigma_p'}, \xi''_n)
    \cdot g''\Big(
    \getpvarsub(\ctr{c''}{\pi})(\widehat{\sigma_p'}, \xi''_n),
    \getvalsub(\ctr{c''}{\pi})(\widehat{\sigma_p'}, \xi''_n)
    \Big) \Big)
    \\
    &=
    \int d\xi''_n \Big(
    \getprsub(\ctr{c''}{\pi})(\widehat{\sigma_p'}, \xi''_n)
    \cdot \ind{\dom(\widehat{\xi_n'}) \cap \dom(\xi''_n) = \emptyset}
    \cdot g\Big(
    \getpvarsub(\ctr{c''}{\pi})(\widehat{\sigma_p'},\xi''_n),
    \widehat{\xi_n'} \oplus \getvalsub(\ctr{c''}{\pi})(\widehat{\sigma_p'},\xi''_n)
    \Big)  \Big)
  \end{align*}
  where the second and fourth equalities use $\dom(\xi_n'') = \dom(\getvalsub(c'')(\widehat{\sigma_p'},\xi_n''))$,
  and the third equality is by IH on $c''$.
  Using this, we obtain the following equation for the main quantity $(*)$:
  \begin{align*}
    (*)
    &=
    \int d\xi'_n \Big(
    \getprsub(\ctr{c'}{\pi})(\sigma_p, \xi'_n) \cdot
    \int d\xi''_n \Big(
    \getprsub(\ctr{c''}{\pi})\big(\getpvarsub(\ctr{c'}{\pi})(\sigma_p,\xi'_n), \xi''_n\big)
    \cdot \ind{\dom(\xi'_n) \cap \dom(\xi''_n) = \emptyset}
    \\
    & \qquad
    {} \cdot g\Big(
    \getpvarsub(\ctr{c''}{\pi})\big(\getpvarsub(\ctr{c'}{\pi})(\sigma_p,\xi'_n),\xi''_n\big),
    \\
    & \phantom{
      \qquad {} \cdot g\Big(
    }
    \getvalsub(\ctr{c'}{\pi})(\sigma_p,\xi'_n) \oplus \getvalsub(\ctr{c''}{\pi})\big(\getpvarsub(\ctr{c'}{\pi})(\sigma_p,\xi'_n),\xi''_n\big)
    \Big)  \Big)\Big)
    \\
    &=
    \int d\xi_n \Big( \getprsub(\ctr{c'}{\pi};\ctr{c''}{\pi})(\sigma_p, \xi_n) \cdot
    g\Big(\getpvarsub(\ctr{c'}{\pi};\ctr{c''}{\pi})(\sigma_p, \xi_n), \getvalsub(\ctr{c'}{\pi};\ctr{c''}{\pi})(\sigma_p, \xi_n) \Big)\Big)
  \end{align*}
  where the first equality uses $\dom(\xi_n') = \dom(\getvalsub(c')(\sigma_p,\xi_n'))$,
  and the second equality is by \cref{lem:integral-seq-decompose}, as we did above.
  By $\smash{ \ctr{c';c''}{\pi} \equiv \ctr{c'}{\pi}; \ctr{c''}{\pi} }$, we get the desired equation.
  
  \paragraph{\bf Case $c \equiv (\cwhile\ b\ \{ c' \})$}
  In this case, $\ctr{c}{\pi} \equiv (\cwhile\ b\ \{ \ctr{c'}{\pi} \})$.
  Without loss of generality, assume that $g$ is a nonnegative function;
  we can prove the general case of $g$ directly from the nonnegative case of $g$,
  by considering the nonnegative part and the negative part of $g$ separately.

  Consider the version of $\getprsub(-)$, $\getpvarsub(-)$, and $\getvalsub(-)$,
  where the parameter can be a state transformer $f : \State \to \State_\bot$, instead of a command.
  We denote the versions by $\getprsub(f)$, $\getpvarsub(f)$, and $\getvalsub(f)$.
  Define $\cT \subseteq [\State \to \State_\bot]^2$ and $T: [\State \to \State_\bot]^2 \to [\State \to \State_\bot]^2$ by
  \begin{align*}
    (f, \overline{f}) \in \cT
    & \iff \int d\xi_n\, G_{g',\sigma_p'}(f)(\xi_n) = \int d\xi_n\, G_{g',\sigma_p'}(\overline{f})(\xi_n)
    \\ &\qquad\qquad
    \text{for all measurable $g' : \State[\PVar] \times \Statesub[\Name] \to \R_{\geq 0}$ and $\sigma_p' \in \State[\PVar]$},
    \\
    T(f, \overline{f}) &\defeq (F(f), \overline{F}(\overline{f})),
  \end{align*}
  where $G_{g',\sigma_p'}(f) \in \Statesub[\Name] \to \R_{\geq 0}$
  and $F, \overline{F} : [\State \to \State_\bot] \to [\State \to \State_\bot]$ are defined by
  \begin{align*}
    G_{g',\sigma_p'}(f)(\xi_n)
    & \defeq \getprsub(f)(\sigma_p', \xi_n) \cdot
    g'\Big(\getpvarsub(f)(\sigma_p', \xi_n), \getvalsub(f)(\sigma_p', \xi_n) \Big),
    \\
    F(f)(\sigma) &\defeq \smash{ \text{if}\ (\db{b}\sigma = \strue)\ \text{then}\ (f{}^\dagger \circ \db{c'})(\sigma)\ \text{else}\ \sigma, }
    \\
    \overline{F}(\overline{f})(\sigma) & \defeq \smash{ \text{if}\ (\db{b}\sigma = \strue)\ \text{then}\
    (\overline{f}{}^\dagger \circ \db{\ctr{c'}{\pi}})(\sigma)\ \text{else}\ \sigma. }
  \end{align*}
  Note that $F$ and $\overline{F}$ are the operators used in the semantics of the loops $\db{c}$ and $\smash{ \db{\ctr{c}{\pi}} }$, respectively.
  We will show that $\cT$ contains $(\lambda \sigma.\,\bot, \lambda \sigma.\, \bot)$,
  the operator $T$ preserves $\cT$,
  and $\cT$ is closed under taking the least upper bound of an increasing chain in $[\State \to \State_\bot]^2$,
  where the order on $[\State \to \State_\bot]^2$ is defined as:
  ${ (f_0, \overline{f_0}) \sqsubseteq (f_1, \overline{f_1}) } \iff f_0 \sqsubseteq f_1 \land { \overline{f_0} \sqsubseteq \overline{f_1}. }$
  These three conditions imply $(\db{c}, \db{\ctr{c}{\pi}}) \in \cT$, which in turn proves the desired equation:
  \begin{align*}
    &
    \int d\xi_n \Big( \getprsub(c)(\sigma_p, \xi_n) \cdot
    g\Big(\getpvarsub(c)(\sigma_p, \xi_n), \getvalsub(c)(\sigma_p, \xi_n) \Big)\Big)
    \\
    &= \int d\xi_n\, G_{{g},{\sigma_p}}(\db{c})(\xi_n)
    \\
    &= \int d\xi_n\, G_{{g},{\sigma_p}}(\db{\ctr{c}{\pi}})(\xi_n)
    \\
    &= \int d\xi_n \Big( \getprsub(\ctr{c}{\pi})(\sigma_p, \xi_n) \cdot
    g\Big(\getpvarsub(\ctr{c}{\pi})(\sigma_p, \xi_n), \getvalsub(\ctr{c}{\pi})(\sigma_p, \xi_n) \Big)\Big),
  \end{align*}
  where the second equality follows from $(\db{c}, \db{\ctr{c}{\pi}}) \in \cT$.
  
  The first condition holds simply because $G_{g', \sigma_p'}(\lambda \sigma.\, \bot)(\xi_n) = 0$ for all $g'$, $\sigma_p'$, and $\xi_n$.
  To show the second condition, pick $\smash{ (f, \overline{f}) \in \cT }$.
  Our goal is to show $\smash{ T(f, \overline{f}) \in \cT }$.
  We first replay our proof for the sequential-composition case on $(f{}^\dagger \circ \db{c'}, { \overline{f}{}^\dagger \circ \db{\ctr{c'}{\pi}} })$,
  after viewing $f{}^\dagger \circ \db{c'}$ and ${ \overline{f}{}^\dagger \circ \db{\ctr{c'}{\pi}} }$
  as the sequential composition of $c'$ and $f$, and of $\ctr{c'}{\pi}$ and ${ \overline{f} }$, respectively.
  This replay, then, gives the membership ${ (f{}^\dagger \circ \db{c'}, \overline{f}{}^\dagger \circ \db{\ctr{c'}{\pi}}) \in \cT }$.
  Next, we replay our proof for the if case on $\smash{ (F(f), \overline{F}(\overline{f})) }$,
  after viewing $f{}^\dagger \circ \db{c'}$ and ${ \overline{f}{}^\dagger \circ \db{\ctr{c'}{\pi}} }$ as the true branches,
  and $\lambda \sigma.\,\sigma = \db{\cskip}$ as the false branch.
  This replay implies the required $\smash{ T(f, \overline{f}) = (F(f), \overline{F}(\overline{f})) \in \cT }$.

  To show the third condition, consider an increasing sequence $\smash{ \{(f_k, \overline{f_k})\}_{k \in \N} }$ in $\cT$.
  Let $f_\infty \defeq \bigsqcup_{k \in \N} f_k$ and $\smash{ \overline{f}_\infty \defeq \bigsqcup_{k \in \N} \overline{f_k} }$.
  Consider a measurable $g' : \State[\PVar] \times \Statesub[\Name] \to \R_{\geq 0}$ and $\sigma_p' \in \State[\PVar]$.
  We should show that $\int d\xi_n\, G_{g',\sigma_p'}(f_\infty)(\xi_n) = { \int d\xi_n\, G_{g',\sigma_p'}(\overline{f_\infty})(\xi_n) }$.
  Since $\{f_k\}_{k \in \N}$ is  increasing, for any $\sigma \in \State$,
  $f_k(\sigma) \in \State$ implies that $f_{k'}(\sigma) = f_k(\sigma) \in \State$ for all $k' \geq k$.
  Hence, $\{ G_{g', \sigma_p'}(f_k) \}_{k \in \N}$ is a pointwise increasing sequence: for all $\xi_n \in \Statesub[\Name]$,
  \begin{align*}
    0 \leq G_{g', \sigma_p'}(f_k)(\xi_n) \leq G_{g', \sigma_p'}(f_{k+1})(\xi_n) \quad\text{for all $k \in \N$}.
  \end{align*}
  Also, by the definition of $f_\infty$,
  for any $\sigma \in \State$, there exists $K \in \N$ such that $f_\infty(\sigma) = f_K(\sigma)$;
  thus, $G_{g', \sigma_p'}(f_\infty)$ is the pointwise limit of $\{ G_{g', \sigma_p'}(f_k) \}_{k \in \N}$:
  for all $\xi_n \in \Statesub[\Name]$,
  \begin{align*}
    G_{g', \sigma_p'}(f_\infty)(\xi_n) = \lim_{k \to \infty} G_{g', \sigma_p'}(f_k)(\xi_n).
  \end{align*}
  Note that the corresponding results hold for  $\smash{ \overline{f_\infty} }$ and $\smash{ \overline{f_k} }$.
  Using these results, we finally obtain the following as desired:
  \begin{align*}
    \int d\xi_n\, G_{g', \sigma_p'}(f_\infty)(\xi_n)
    &= \lim_{k \to \infty} \int d\xi_n\, G_{g', \sigma_p'}(f_k)(\xi_n)
    \\
    &= \lim_{k \to \infty} \int d\xi_n\, G_{g', \sigma_p'}(\overline{f_k})(\xi_n)
    = \int d\xi_n\, G_{g', \sigma_p'}(\overline{f_\infty})(\xi_n).
  \end{align*}
  The first and third equalities follow from the monotone convergence theorem, applied to the above results.
  The second equality holds since $\smash{ (f_k, \overline{f_k}) \in \cT }$.
  This completes the proof of the while case.
\end{proof}

\begin{lemma}
  \label{lem:like-mult}
  Let $c$ be a command, $\sigma_0, \sigma_1 \in \State$, and $r_0 \in \R$.
  Suppose that $\sigma_1(\like) = \sigma_0(\like)\cdot r_0$ and
  $\sigma_1|_V = \sigma_0|_V$ for $V \defeq \Var \setminus \{\like\}$.
  If $\db{c}\sigma_0 \in \State$, then
  \begin{align*}
    &\db{c}\sigma_1 \in \State,
    &
    &\db{c}\sigma_1(\like) = \db{c}\sigma_0(\like) \cdot r_0,
    &
    &(\db{c}\sigma_1)|_V = (\db{c}\sigma_0)|_V.
  \end{align*}
\end{lemma}
\begin{proof}
Let $V \defeq \Var \setminus \{\like\}$. Pick an arbitrary command $c$. We prove the lemma by induction on the structure of $c$. 
Let $\sigma_0,\sigma_1 \in \State$ and $r_0 \in \R$ such that $\db{c}\sigma_0 \in \State$, $\sigma_1(\like) = \sigma_0(\like) \cdot r_0$, and $\sigma_1|_V = \sigma_0|_V$. We should show that $\db{c}\sigma_1 \in \State$, $\db{c}\sigma_1(\like) = \db{c}\sigma_0(\like) \cdot r_0$, and $\db{c}\sigma_1|_V = \db{c}\sigma_0|_V$. 

\paragraph{\bf Case $c \equiv \cskip$} In this case, what we need to prove is identical to the assumption on $(\sigma_0,\sigma_1,r_0)$. 

\paragraph{\bf Case $c \equiv (x:=e)$} By the semantics of the assignments, we have $\db{c}\sigma_1 \in \State$, and 
\[
\db{c}\sigma_1(\like) = \sigma_1(\like) = \sigma_0(\like) \cdot r_0 =  \db{c}\sigma_0(\like) \cdot r_0.
\]
The last requirement also holds since $\db{e}\sigma_0 = \db{e}\sigma_1$ and $\sigma_0|_V = \sigma_1|V$.

\paragraph{\bf Case $c \equiv (x:=\csample(n,d,\lambda y.e')$} By the semantics of the sample commands, we have $\db{c}\sigma_1 \in \State$. Also, the assignments do not change the value of $\like$, so that $\db{c}\sigma_1(\like) = \sigma_1(\like) = \sigma_0(\like) \cdot r_0 = \db{c}\sigma_0(\like) \cdot r_0$. It remains to show that $\db{c}\sigma_0|_V = \db{c}\sigma_1|_V$. Since $\sigma_0|_V = \sigma_1|_V$, we have $\db{n}\sigma_0 = \db{n}\sigma_1$ and $\db{d}\sigma_0 = \db{d}\sigma_1$. Let $\mu \defeq \db{n}\sigma_0$. Then, by the same reason, $\db{e'[\mu/y]}\sigma_0 = \db{e'[\mu/y]}\sigma_1$. Let $f \defeq \db{d}\sigma_0$ and $r \defeq \db{e'[\mu/y]}\sigma_0$. 
We prove the required equality as follows:
\begin{align*}
\db{c}\sigma_0|_V 
& = \sigma_0[x \mapsto r, \val_\mu \mapsto r, \pr_\mu \mapsto f(\sigma_0(\mu)), \sampled_\mu \mapsto \sigma_0(\sampled_\mu) + 1]|_V
\\
& = \sigma_1[x \mapsto r, \val_\mu \mapsto r, \pr_\mu \mapsto f(\sigma_0(\mu)), \sampled_\mu \mapsto \sigma_0(\sampled_\mu) + 1]|_V
\\
& = \sigma_1[x \mapsto r, \val_\mu \mapsto r, \pr_\mu \mapsto f(\sigma_1(\mu)), \sampled_\mu \mapsto \sigma_1(\sampled_\mu) + 1]|_V
\\
& = \db{c}\sigma_1|_V.
\end{align*}

\paragraph{\bf Case $c \equiv (\cobs(d,r))$} By the semantics of the observe commands, we have $\db{c}\sigma_1 \in \State$. Also, the observe commands do not change any variable except $\like$. So, $\db{c}\sigma_0|_V = \sigma_0|V = \sigma_1|_V = \db{c}\sigma_1|_V$. The remaining requirement for $\like$ can be proved as follows:
\[
\db{c}\sigma_1(\like) 
= 
\sigma_1(\like) \cdot \db{d}\sigma_1(r) 
=
\sigma_0(\like) \cdot r_0 \cdot \db{d}\sigma_1(r)
=
\sigma_0(\like) \cdot r_0 \cdot \db{d}\sigma_0(r)
=
\db{c}\sigma_0(\like) \cdot r_0.
\]

\paragraph{\bf Case $c \equiv (c';c'')$} We have $\db{c'}\sigma_0 \in \State$ and $\db{c''}(\db{c'}\sigma_0) \in \State$. We apply induction hypothesis first to $(c',\sigma_0,\sigma_1,r_0)$, and then to $(c'',\db{c'}\sigma_0,\db{c'}\sigma_1,r_0)$. Then, we get the requirements of the lemma.

\paragraph{\bf Case $c \equiv (\cif\ b\ \{c'\}\ \celse\ \{c''\})$} We deal with the case that $\db{b}\sigma_0 = \strue$. The other case of $\db{b}\sigma_0 = \sfalse$ can be proved similarly. Since $\db{b}\sigma_0 = \strue$, we have $\db{c'}\sigma_0 = \db{c}\sigma_0 \in \State$. Thus, we can apply induction hypothesis to $c'$.  If we do so, we get 
\[
\db{c'}\sigma_1 \in \State,\quad
\db{c'}\sigma_1(\like) = \db{c'}\sigma_0(\like) \cdot r_0,\quad
\text{and}\quad
\db{c'}\sigma_1|_V = \db{c'}\sigma_0|V.
\]
This gives the desired conclusion because $\db{b}\sigma_1 = \db{b}\sigma_0 = \strue$ and so $\db{c}\sigma_1 = \db{c'}\sigma_1$, and $\db{c}\sigma_0 = \db{c'}\sigma_0$.

\paragraph{\bf Case $c \equiv (\cwhile\ b\ \{c'\})$} Let $F$ be the operator on $[\State \to \State_\bot]$ such that $\db{c}$ is the least fixed point of $F$. Define a subset $\cT$ of $[\State \to \State_\bot]$ as follows: a function $f \in [\State \to \State_\bot]$ is in $\cT$ if and only if for all $\sigma'_0,\sigma'_1 \in \State$ such that 
$\sigma'_1|_V = \sigma'_0|_V$ and $\sigma'_1(\like) = \sigma'_0(\like) \cdot r_0$, we have
\begin{align*}
& f(\sigma'_0) \neq \bot \implies \Big(f(\sigma'_1) \neq \bot \land  f(\sigma'_1)|_V = f(\sigma'_0)|_V \land f(\sigma'_1)(\like) = f(\sigma'_0)(\like) \cdot r_0\Big).
\end{align*}
The set $\cT$ contains the least function $\lambda \sigma.\bot$, and is closed under the least upper bound of any chain in $[\State \to \State_\bot]$. It is also closed under $F$. This $F$-closure follows essentially from our arguments for sequential composition, if command, and skip, and induction hypothesis on $c'$. What we have shown for $\cT$ implies that $\cT$ contains the least fixed point of $F$, which gives the desired property for $c$.
\end{proof}

\begin{lemma}
  \label{lem:integral-seq-decompose}
  Let $c', c''$ be commands and $g : \State[\PVar] \times \Statesub[\Name] \to \R$ be a measurable function.
  Then, for any $\sigma_p \in \State[\PVar]$,
  \begin{align*}
    &
    \int d\xi_n \Big(
    \getprsub(c';c'')(\sigma_p, \xi_n) \cdot g\Big(\getpvarsub(c';c'')(\sigma_p,\xi_n), \getvalsub(c';c'')(\sigma_p,\xi_n)\Big)
    \Big)
    \\
    &=
    \int d\xi'_n \Big(
    \getprsub(c')(\sigma_p, \xi'_n) \cdot
    \int d\xi''_n \Big(
    \getprsub(c'')(\getpvarsub(c')(\sigma_p,\xi'_n), \xi''_n)
    \cdot \ind{\dom(\xi'_n) \cap \dom(\xi''_n) = \emptyset}
    \\
    & \qquad
    {} \cdot g\Big(
    \getpvarsub(c'')(\getpvarsub(c')(\sigma_p,\xi'_n),\xi''_n),
    \getvalsub(c')(\sigma_p,\xi'_n) \oplus \getvalsub(c'')(\getpvarsub(c')(\sigma_p,\xi'_n),\xi''_n)
    \Big)  \Big)\Big).
  \end{align*}
\end{lemma}
\begin{proof}
  Let $c', c''$ be commands, $g : \State[\PVar] \times \Statesub[\Name] \to \R$ a measurable function, and $\sigma_p \in \State[\PVar]$.
  In this proof, each equation involving integrals means (otherwise noted) that
  one side of the equation is defined if and only if the other side is defined,
  and when both sides are defined, they are the same.

  First, to convert a single-integral on $\xi_n$ to a double-integral on $\xi_n'$ and $\xi_n''$ as in the desired equation, we show the following claim:
  for any measurable $f : \Statesub[\Name] \to \State$ and $h : \Statesub[\Name] \to \R$
  such that $f(\xi_n)|_{\dom(\xi_n)} = \xi_n$ for all $\xi_n \in \Statesub[\Name]$, we have
  \begin{align*}
    & \int d\xi_n \Big(
    \ind{\used(c';c'', f(\xi_n), \xi_n)} \cdot h(\xi_n)
    \Big)
    \\
    &= \sum_{K \subseteq \Name} \int_{[K \to \R]} d\xi_n \Big(
    \ind{\used(c';c'', f(\xi_n), \xi_n)} \cdot h(\xi_n)
    \Big)
    \\
    &= \sum_{K \subseteq \Name} \int_{[K \to \R]} d\xi_n \Big(
    \ind{\used(c';c'', f(\xi_n), \xi_n)} \cdot h(\xi_n)
    \cdot \sum_{L \subseteq K} \ind{\used(c', f(\xi_n), \xi_n|_L)}
    \Big)
    \\
    &= \sum_{K \subseteq \Name} \sum_{L \subseteq K} \int_{[K \to \R]} d\xi_n \Big(
    \ind{\used(c', f(\xi_n), \xi_n|_L)}
    \cdot \ind{\used(c';c'', f(\xi_n), \xi_n)} \cdot h(\xi_n)
    \Big)
    \\
    &= \sum_{K \subseteq \Name} \sum_{L \subseteq K} \int_{[L \to \R]} \hspace{-1em} d\xi_n' \int_{[K \setminus L \to \R]} \hspace{-1.4em} d\xi_n'' \Big(
    \ind{\used(c', f(\xi_n' \oplus \xi_n''), \xi_n')}
    \cdot \ind{\used(c';c'', f(\xi_n' \oplus \xi_n''), \xi_n' \oplus \xi_n'')} \cdot h(\xi_n' \oplus \xi_n'')
    \Big)
    \\
    &= \sum_{L' \subseteq \Name} \sum_{\substack{M' \subseteq \Name \\ L'\cap M' = \emptyset}}
    \int_{[L' \to \R]} \hspace{-1em} d\xi_n' \int_{[M' \to \R]} \hspace{-1.4em} d\xi_n'' \Big(
    \ind{\used(c', f(\xi_n' \oplus \xi_n''), \xi_n')}
    \cdot \ind{\used(c';c'', f(\xi_n' \oplus \xi_n''), \xi_n' \oplus \xi_n'')} \cdot h(\xi_n' \oplus \xi_n'')
    \Big)
    \\
    &= \sum_{L' \subseteq \Name} \int_{[L' \to \R]} d\xi_n' \Big( \sum_{M' \subseteq \Name} \int_{[M' \to \R]} d\xi_n'' \Big(
    \ind{\dom(\xi_n') \cap \dom(\xi_n'') = \emptyset}
    \cdot \ind{\used(c', f(\xi_n' \oplus \xi_n''), \xi_n')}
    \\
    &\phantom{= \sum_{L' \subseteq \Name} \int_{[L' \to \R]} d\xi_n' \Big( \sum_{\substack{M' \subseteq \Name}} \int_{[M' \to \R]} d\xi_n'' \Big(}
    \cdot \ind{\used(c';c'', f(\xi_n' \oplus \xi_n''), \xi_n' \oplus \xi_n'')} \cdot h(\xi_n' \oplus \xi_n'')
    \Big) \Big)
    \\
    &= \int d\xi_n' \int d\xi_n'' \Big(
    \ind{\dom(\xi_n') \cap \dom(\xi_n'') = \emptyset}
    \cdot \ind{\used(c', f(\xi_n' \oplus \xi_n''), \xi_n')}
    \cdot \ind{\used(c';c'', f(\xi_n' \oplus \xi_n''), \xi_n' \oplus \xi_n'')} \cdot h(\xi_n' \oplus \xi_n'')
    \Big).
  \end{align*}
  The first equality uses the definition of $\Statesub[\Name]$ and its measure.
  The second equality uses that $\used(c';c'', \sigma, \xi_n)$ implies
  a unique existence of $L \subseteq \dom(\xi_n)$ such that $\used(c', \sigma, \xi_n|L)$;
  we already showed the existence of such $L$ in the proof of \cref{lem:used-properties} (for the sequential composition case),
  and the uniqueness follows from the definition of $\used$.
  The third equality uses that $K$ is finite,
  and the fourth equality uses that $[K \to \R]$ is isomorphic to $[L \to \R] \times [K \setminus L \to \R]$ for any $L \subseteq K$.
  The fifth equality holds uses that $\{(L, K \setminus L) \mid K \subseteq \Name, L \subseteq K\}
  = \{(L', M') \mid L', M' \subseteq \Name, L' \cap M' = \emptyset\}$.
  The sixth equality uses that $\Name$ is finite, $\dom(\xi_n')=L'$, and $\dom(\xi_n'')=M'$.
  The seventh equality uses the definition of $\Statesub[\Name]$ and its measure.

  Second, to decompose $\getprsub(c';c'')$, $\getpvarsub(c';c'')$, and $\getvalsub(c';c'')$ as in the desired equation,
  we show the following claim. Suppose that $\sigma \in \State$ and $\xi_n', \xi_n'' \in \Statesub[\Name]$
  with $\dom(\xi_n') \cap \dom(\xi_n'') = \emptyset$ satisfy
  $\used(c';c'', \sigma, \xi_n' \oplus \xi_n'')$ and $\used(c', \sigma, \xi_n')$.
  Then, we first get
  \begin{align*}
    \getprsub(c';c'')(\sigma|_\PVar, \xi_n' \oplus \xi_n'')
    &= \textstyle \db{c'; c''} \sigma(\like) \cdot \prod_{\mu \in \dom(\xi_n' \oplus \xi_n'')} \db{c'; c''} \sigma(\pr_\mu)
    \\
    &= \textstyle \db{c''}(\db{c'}\sigma)(\like)
    \\
    & \textstyle \quad
    \cdot \prod_{\mu \in \dom(\xi_n')}  \db{c''}(\db{c'}\sigma)(\pr_\mu)
    \cdot \prod_{\mu \in \dom(\xi_n'')} \db{c''}(\db{c'}\sigma)(\pr_\mu)
    \\
    &= \textstyle \db{c'}\sigma(\like) \cdot \db{c''}\big((\db{c'}\sigma)[\like \mapsto 1]\big)(\like)
    \\
    & \textstyle \quad 
    \cdot \prod_{\mu \in \dom(\xi_n')}  \db{c'}\sigma(\pr_\mu)
    \cdot \prod_{\mu \in \dom(\xi_n'')} \db{c''}\big((\db{c'}\sigma)[\like \mapsto 1]\big)(\pr_\mu)
    \\ 
    &= \textstyle \getprsub(c')(\sigma|_\PVar, \xi_n') \cdot \getprsub(c'')\big((\db{c'}\sigma)[\like \mapsto 1]|_\PVar, \xi_n''\big)
    \\
    &= \textstyle \getprsub(c')(\sigma|_\PVar, \xi_n') \cdot \getprsub(c'')\big(\getpvarsub(c')(\sigma|_\PVar, \xi_n'), \xi_n''\big).
  \end{align*}
  Here is the proof of each equality.
  \begin{itemize}
  \item
    The first equality uses $\used(c';c'', \sigma, \xi_n' \oplus \xi_n'')$.
  \item
    The second equality uses $\noerr(c';c'', \sigma)$, which comes from $\used(c';c'', \sigma, -)$.
  \item
    The third equality comes from \cref{lem:like-mult},
    \cref{lem:used-properties}-\cref{lem:used-properties-2},
    and \cref{lem:used-properties}-\cref{lem:used-properties-3}.
    The two applications of \cref{lem:used-properties} are valid
    since $\usedm(c'', \db{c'}\sigma, \xi_n'')$ and  $\dom(\xi_n') \cap \dom(\xi_n'') = \emptyset$,
    where the first predicate follows from $\used(c';c'', \sigma, \xi_n' \oplus \xi_n'')$ and $\used(c', \sigma, \xi_n')$
    by the claim in the proof of \cref{lem:used-properties} (for the sequential composition case).
  \item
    The fourth equality uses that $\used(c', \sigma, \xi_n')$ and $\used(c'', (\db{c'}\sigma)[\like \mapsto 1], \xi_n'')$,
    where the second predicate follows from $\usedm(c'', \db{c'}\sigma, \xi_n'')$ and \cref{lem:used-sigma-indep}.
  \item
    The fifth equality uses
    $(\db{c'}\sigma)[\like \mapsto 1]|_\PVar = (\db{c'}\sigma)|_\PVar = \getpvarsub(c')(\sigma|_\PVar, \xi_n')$,
    where the second part of the equation comes from $\used(c', \sigma, \xi_n')$.
  \end{itemize}
  By the same argument so far (except that $\pr_\mu$ and $\times$ are replaced by $\val_\mu$ and $\oplus$),
  we next get
  \begin{align*}
    \getvalsub(c';c'')(\sigma|_\PVar, \xi_n' \oplus \xi_n'')
    &= \getvalsub(c')(\sigma|_\PVar, \xi_n') \oplus \getvalsub(c'')\big(\getpvarsub(c')(\sigma|_\PVar, \xi_n'), \xi_n''\big).
  \end{align*}
  By a similar argument, we lastly get
  \begin{align*}
    \getpvarsub(c';c'')(\sigma|_\PVar, \xi_n' \oplus \xi_n'')
    &=\db{c';c''}\sigma|_\PVar
    \\
    &=\db{c''}(\db{c'}\sigma)|_\PVar
    \\
    &=\db{c''}\big((\db{c'}\sigma)[\like \mapsto 1]\big)|_\PVar
    \\
    &= \getpvarsub(c'')\big((\db{c'}\sigma)[\like \mapsto 1]|_\PVar, \xi_n''\big)
    \\
    &= \getpvarsub(c'')\big(\getpvarsub(c')(\sigma|_\PVar, \xi_n'), \xi_n''\big).
  \end{align*}
  Here is the proof of each equality.
  \begin{itemize}
  \item
    The first equality uses $\used(c';c'', \sigma, \xi_n' \oplus \xi_n'')$.
  \item
    The second equality uses $\noerr(c';c'', \sigma)$ (shown above).
  \item
    The third equality uses $\usedm(c'', \db{c'}\sigma, \xi_n'')$ (shown above) and \cref{lem:used-properties}-\cref{lem:used-properties-3}.
  \item
    The fourth equality uses $\used(c'', (\db{c'}\sigma)[\like \mapsto 1], \xi_n'')$ (shown above).
  \item
    The fifth equality uses $(\db{c'}\sigma)[\like \mapsto 1]|_\PVar = \getpvarsub(c')(\sigma|_\PVar, \xi_n')$ (shown above).
  \end{itemize}

  Third, to remove some indicator terms that will appear in our derivation, we show the next claim:
  for any $\sigma \in \State$ and $\xi_n', \xi_n'' \in \Statesub[\Name]$ with $\dom(\xi_n') \cap \dom(\xi_n'') = \emptyset$,
  $\usedm(c', \sigma, \xi_n')$ and $\usedm(c'', \db{c'}\sigma, \xi_n'')$ imply $\usedm(c';c'', \sigma, \xi_n' \oplus \xi_n'')$.
  Assume the premise. Then, we have
  \begin{align*}
    \db{c'}\sigma \in \State 
    & {} \land \xi_n' = \sigma|_{\dom(\xi_n')}
    \\
    & {} \land \big(\forall \mu \in \Name.\, \db{c'}\sigma(\sampled_\mu) - \sigma(\sampled_\mu) \leq 1\big)
    \\
    & {} \land \dom(\xi_n') = \{ \mu \in \Name \mid  \db{c'}\sigma'(\sampled_\mu) - \sigma(\sampled_\mu) = 1\},
    \\
    \db{c''}(\db{c'}\sigma) \in \State 
    & {} \land \xi_n'' = (\db{c'}\sigma)|_{\dom(\xi_n'')}
    \\
    & {} \land \big(\forall \mu \in \Name.\, \db{c''}(\db{c'}\sigma)(\sampled_\mu) - \db{c'}\sigma(\sampled_\mu) \leq 1 \big)
    \\
    & {} \land \dom(\xi_n'') = \{ \mu \in \Name \mid  \db{c''}(\db{c'}\sigma)(\sampled_\mu) - \db{c'}\sigma(\sampled_\mu) = 1\}.
  \end{align*}
  We should show
  \begin{align*}
    \db{c';c''}\sigma \in \State 
    & {} \land \xi_n' \oplus \xi_n'' = \sigma|_{\dom(\xi_n' \oplus \xi_n'')}
    \\
    & {} \land \big(\forall \mu \in \Name.\, \db{c';c''}\sigma(\sampled_\mu) - \sigma(\sampled_\mu) \leq 1\big)
    \\
    & {} \land \dom(\xi_n' \oplus \xi_n'') = \{ \mu \in \Name \mid  \db{c';c''}\sigma'(\sampled_\mu) - \sigma(\sampled_\mu) = 1\}.
  \end{align*}
  We obtain the four clauses as follows.
  The first clause follows from $\db{c''}(\db{c'}\sigma) \in \State$.
  The second clause comes from $\xi_n' = \sigma|_{\dom(\xi_n')}$ and $\xi_n'' = (\db{c'}\sigma)|_{\dom(\xi_n'')} = \sigma|_{\dom(\xi_n'')}$,
  where the last equality comes from \cref{lem:cnt-increase-and-rv-nochange}.
  The third and fourth clauses hold by the following:
  \begin{align*}
    &\db{c';c''}\sigma(\sampled_\mu) - \sigma(\sampled_\mu)
    \\
    &=
    \db{c''}(\db{c'}\sigma)(\sampled_\mu)  - \sigma(\sampled_\mu)
    \\
    &=
    \begin{cases}
      \db{c'}\sigma(\sampled_\mu) - \sigma(\sampled_\mu)
      = 1
      & \text{if $\mu \in \dom(\xi_n')$}
      \\
      \db{c''}(\db{c'}\sigma)(\sampled_\mu) - \db{c'}\sigma(\sampled_\mu)
      = 1
      & \text{if $\mu \in \dom(\xi_n'')$}
      \\
      \sigma(\sampled_\mu) - \sigma(\sampled_\mu)
      = 0
      & \text{if $\mu \in \Name \setminus (\dom(\xi_n') \cup \dom(\xi_n''))$}.
    \end{cases}
  \end{align*}
  \begin{itemize}
  \item 
    The first case uses \cref{lem:used-properties}-\cref{lem:used-properties-2}
    (applied to $\usedm(c'', \db{c'}\sigma, \xi_n'')$ and $\dom(\xi_n') \cap \dom(\xi_n'') = \emptyset$)
    and the fourth clause of $\usedm(c', \sigma, \xi_n')$.
  \item 
    The second case uses \cref{lem:used-properties}-\cref{lem:used-properties-2}
    (applied to $\usedm(c', \sigma, \xi_n')$ and $\dom(\xi_n') \cap \dom(\xi_n'') = \emptyset$)
    and the fourth clause of $\usedm(c'', \db{c'}\sigma, \xi_n'')$.
  \item 
    The third case uses \cref{lem:used-properties}-\cref{lem:used-properties-2}
    (applied to $\usedm(c'', \db{c'}\sigma, \xi_n'')$ and $\usedm(c', \sigma, \xi_n')$).
  \end{itemize}
    
  Finally, we put the three results together. Define $f : \Statesub[\Name] \to \State$ as
  $f(\xi_n) \defeq \sigma_p \oplus \xi_n \oplus (\lambda v \in \Var \setminus (\PVar \cup \dom(\xi_n)).\, 1)$.
  Then, we obtain the desired equation as follows:
  \begin{align*}
    &
    \int d\xi_n \Big(
    \getprsub(c';c'')(\sigma_p, \xi_n) \cdot g\Big(\getpvarsub(c';c'')(\sigma_p,\xi_n), \getvalsub(c';c'')(\sigma_p,\xi_n)\Big)
    \Big)
    \\
    &=
    \int d\xi_n \Big(
    \ind{\used(c';c'', f(\xi_n), \xi_n)} \cdot
    \getprsub(c';c'')(\sigma_p, \xi_n) \cdot g\Big(\getpvarsub(c';c'')(\sigma_p,\xi_n), \getvalsub(c';c'')(\sigma_p,\xi_n)\Big)
    \Big)
    \\
    &=
    \int d\xi_n' \int d\xi_n'' \Big(
    \ind{\dom(\xi_n') \cap \dom(\xi_n'') = \emptyset}
    \cdot \ind{\used(c', f(\xi_n' \oplus \xi_n''), \xi_n')} \cdot \ind{\used(c';c'', f(\xi_n' \oplus \xi_n''), \xi_n' \oplus \xi_n'')} 
    \\
    &\qquad\quad
    \cdot \getprsub(c';c'')(\sigma_p, \xi_n' \oplus \xi_n'') \cdot
    g\Big(\getpvarsub(c';c'')(\sigma_p,\xi_n' \oplus \xi_n''), \getvalsub(c';c'')(\sigma_p,\xi_n' \oplus \xi_n'')\Big)
    \Big)
    \\
    &=
    \int d\xi_n' \int d\xi_n'' \Big(
    \ind{\dom(\xi_n') \cap \dom(\xi_n'') = \emptyset}
    \cdot \ind{\used(c', f(\xi_n' \oplus \xi_n''), \xi_n')} \cdot \ind{\used(c';c'', f(\xi_n' \oplus \xi_n''), \xi_n' \oplus \xi_n'')} 
    \\
    &\qquad\quad
    \cdot \getprsub(c')(\sigma_p, \xi_n') \cdot \getprsub(c'')\big(\getpvarsub(c')(\sigma_p, \xi_n'), \xi_n''\big)
    \\
    &\qquad\quad
    \cdot
    g\Big(\getpvarsub(c'')\big(\getpvarsub(c')(\sigma_p, \xi_n'), \xi_n''\big),
    \getvalsub(c')(\sigma_p, \xi_n') \oplus \getvalsub(c'')\big(\getpvarsub(c')(\sigma_p, \xi_n'), \xi_n''\big)
    \Big)
    \Big)
    \\
    &=
    \int d\xi_n' \int d\xi_n'' \Big(
    \ind{\dom(\xi_n') \cap \dom(\xi_n'') = \emptyset}
    \cdot \getprsub(c')(\sigma_p, \xi_n') \cdot \getprsub(c'')\big(\getpvarsub(c')(\sigma_p, \xi_n'), \xi_n''\big)
    \\
    &\qquad\quad
    \cdot
    g\Big(\getpvarsub(c'')\big(\getpvarsub(c')(\sigma_p, \xi_n'), \xi_n''\big),
    \getvalsub(c')(\sigma_p, \xi_n') \oplus \getvalsub(c'')\big(\getpvarsub(c')(\sigma_p, \xi_n'), \xi_n''\big)
    \Big)
    \Big).
  \end{align*}
  The first equality uses that $\getprsub(c';c'')(\sigma_p, \xi_n) \neq 0$ implies $\ind{\used(c';c'', f(\xi_n), \xi_n)} =1$:
  \begin{itemize}
  \item Since $\getprsub(c';c'')${}$(\sigma_p, \xi_n) \neq 0$, there is $\sigma_r \in \State[\Var \setminus (\PVar \cup \dom(\xi_n))]$
    such that $\used(c';c'', \sigma_p \oplus \xi_n \oplus \sigma_r, \xi_n)$.
    Note $(\sigma_p \oplus \xi_n \oplus \sigma_r)|_{V} = f(\xi_n)|_{V}$ for $V = \PVar \cup \dom(\xi_n) \cup \{\like\}$.
    From these and \cref{lem:used-sigma-indep}, we get $\used(c';c'', f(\xi_n), \xi_n)$.
  \end{itemize}
  The second and third equalities use the first and second results we proved above, respectively.
  The fourth equality uses the next claim:
  $\getprsub(c')(\sigma_p, \xi_n') \cdot \getprsub(c'')\big(\getpvarsub(c')(\sigma_p, \xi_n'), \xi_n''\big) \neq 0$
  and $\dom(\xi_n') \cap \dom(\xi_n'') = \emptyset$
  imply $\ind{\used(c', f(\xi_n' \oplus \xi_n''), \xi_n')} \cdot \ind{\used(c';c'', f(\xi_n' \oplus \xi_n''), \xi_n' \oplus \xi_n'')} = 1$.
  We prove the claim using the third result we proved above:
  \begin{itemize}
  \item Assume the premise.
    From $\getprsub(c')(\sigma_p, \xi_n') \neq 0$,
    there is $\sigma_r' \in \State[\Var \setminus (\PVar \cup \dom(\xi_n'))]$
    such that $\used(c', \sigma_p \oplus \xi_n' \oplus \sigma_r', \xi_n')$.
    Note $(\sigma_p \oplus \xi_n' \oplus \sigma_r')|_{V'} = f(\xi_n' \oplus \xi_n'')|_{V'}$ for $V' = \PVar \cup \dom(\xi_n') \cup \{\like\}$.
    From these and \cref{lem:used-sigma-indep},
    we get $\used(c', f(\xi_n' \oplus \xi_n''), \xi_n')$ as desired.
  \item
    From $\getprsub(c'')\big(\getpvarsub(c')(\sigma_p, \xi_n'), \xi_n''\big) \neq 0$,
    there is $\sigma_r'' \in \State[\Var \setminus (\PVar \cup \dom(\xi_n''))]$
    such that \[\used(c'', \getpvarsub(c')(\sigma_p, \xi_n') \oplus \xi_n'' \oplus \sigma_r'', \xi_n'').\]
    Since $\used(c', f(\xi_n' \oplus \xi_n''), \xi_n')$, we have $\db{c'}(f(\xi_n' \oplus \xi_n'')) \in \State$ and
    $
      \getpvarsub(c')(\sigma_p, \xi_n')
      = \getpvarsub(c')\allowbreak(f(\xi_n' \oplus \xi_n'')|_\PVar, \xi_n') = \db{c'}(f(\xi_n' \oplus \xi_n''))|_\PVar;
      $
    also, by \cref{lem:cnt-increase-and-rv-nochange}, $\xi_n''= \db{c'}(f(\xi_n' \oplus \xi_n''))|_{\dom(\xi_n'')}$.
    Thus, \[(\getpvarsub(c')(\sigma_p, \xi_n') \oplus \xi_n'' \oplus \sigma_r'')|_{V''} = \db{c}(f(\xi_n' \oplus \xi_n''))|_{V''}\]
    for $V'' = \PVar \cup \dom(\xi_n'')$.
    From these and \cref{lem:used-sigma-indep},
    we get $\usedm(c'', \db{c'}(f(\xi_n' \oplus \xi_n'')), \xi_n'')$.
  \item
    From $\dom(\xi_n') \cap \dom(\xi_n'') = \emptyset$,
    $\usedm(c', f(\xi_n' \oplus \xi_n''), \xi_n')$, and $\usedm(c'', \db{c'}(f(\xi_n' \oplus \xi_n'')), \xi_n'')$,
    we can apply the third result proved above, and get $\usedm(c';c'', f(\xi_n' \oplus \xi_n''), \xi_n' \oplus \xi_n'')$.
    Since $f(\xi_n' \oplus \xi_n'')(\like) = 1$, we get $\used(c';c'', f(\xi_n' \oplus \xi_n''), \xi_n' \oplus \xi_n'')$ as desired.
  \end{itemize}
  This completes the proof.
\end{proof}

} 
\AtEndDocument{

\section{Deferred Results in \S\ref{sec:grad-estm-prog-trans}}

\subsection{Deferred Statements and Their Proofs}

\begin{lemma}
  \label{lem:suff-cond-grad-dens-zero}
  Let $c$ be a command and $\pi$ be a simple reparameterisation plan. 
  Suppose that for all $n \in \NameEx$, $d,d' \in \DistEx$, and $(\lambda y.e) \in \LamEx$ such that
  $\pi(n,d,\lambda y.e) = (d',\_)$, we have
  \begin{align}
    \label{eq:pi-use-const-dist-only}
    d' = \cnor(r_1', r_2') \quad\text{for some $r_1', r_2' \in \R$.}
  \end{align}
  Further, assume that for all $\sigma_n \in \State[\Name]$,
  the function
  \begin{align}
    \label{eq:p-repar-rv}
    \sigma_\theta \in \State[\theta] \longmapsto \pfun{\ctr{c}{\pi}, \sigma_\theta}{\repname(\pi)}(\sigma_n)
  \end{align}
  is continuous.
  Then, for all $\sigma_\theta \in \State[\theta]$ and $\sigma_n \in \State[\Name]$, 
  \begin{align}
    \label{eq:p-repar-rv-grad-zero}
    \nabla_\theta \pfun{\ctr{c}{\pi}, \sigma_\theta}{\repname(\pi)}(\sigma_n) = 0.
  \end{align}
\end{lemma}
\begin{proof}
  Consider $c$ and $\pi$ that satisfies the given conditions.
  Fix $\sigma_n \in \State[\Name]$.
  Let $f : \State[\theta] \to \R$ be the function in \cref{eq:p-repar-rv}.
  Suppose that \cref{eq:p-repar-rv-grad-zero} does not hold.
  Then, $f$ is not a constant function.
  
  On the one hand, since $f$ is continuous (by assumption) and not constant,
  the image of $f$ over its domain (i.e., $f(\State[\theta]) \subseteq \R$) is an uncountable set.
  This can be shown as follows:
  since the image of a connected set over a continuous function is connected,
  $f(\State[\theta])$ is a connected set in $\R$;
  since $f$ is not constant, $f(\State[\theta])$ contains at least two points;
  since $f(\State[\theta])$ is connected, it should contain a non-empty interval,
  so it should be an uncountable set.

  On the other hand, since $\pi$ is simple and satisfies \cref{eq:pi-use-const-dist-only},
  and since $c$ has only finitely many sample commands, $f(\State[\theta])$ is a finite set. 
  So this contradicts to that $f(\State[\theta])$ is an uncountable set.
  Hence, $f$ should satisfy \cref{eq:p-repar-rv-grad-zero}.
\end{proof}

\begin{theorem}
  \label{lem:suff-cond-int-diff}
  Let $f : \R \times \R^n \to \R$ be a measurable function that satisfies the next three conditions:
  \begin{itemize}
  \item For all $x \in \R^n$, $f(-, x) : \R \to \R$ is differentiable.
  \item For all $\theta \in \R$, $\smash{ \int_{\R^n} } f(\theta, x) \,dx$ is finite.
  \item For all $\theta \in \R$, there is an open $U \subseteq \R$ such that $\theta \in U$ and $\int_{\R^n} \lip\big(f(-,x)|_U\big) \,dx$ is finite.
  \end{itemize}
  Here $\lip(g)$ for a function $g : V \to \R$ with $V \subseteq \R$ denotes the smallest Lipschitz constant of $g$:
  \begin{align*}
    \lip(g) \defeq \sup_{r, r' \in V, \;r \neq r'} \frac{|g(r') - g(r)|}{|r'-r|} \in \R \cup \{\infty\}.
  \end{align*}
  Then, for all $\theta \in \R$, both sides of the following are well-defined and equal:
  \begin{align*}
    \nabla_\theta \int_{\R^n} f(\theta, x) \,dx
    &= \int_{\R^n} \nabla_\theta f(\theta, x) \,dx
  \end{align*}
  where $\nabla_\theta$ denotes the partial differentiation operator with respect to $\theta$.
\end{theorem}
\rev{\begin{proof}
    This theorem follows from \Cref{lem:suff-cond-int-diff-lip} due to the following:
    the first condition of this theorem implies the first and second conditions of \Cref{lem:suff-cond-int-diff-lip}
    (as differentiability implies continuity);
    the second and third conditions of this theorem
    are identical to the third and fourth conditions of \Cref{lem:suff-cond-int-diff-lip};
    and the conclusion of this theorem is identical to that of \Cref{lem:suff-cond-int-diff-lip}.
\end{proof}}

\subsection{Proof of \cref{thm:unbiased-grad}}
\label{sec:proof:lemmas:unbiased-grad}

The proof of \cref{thm:unbiased-grad} relies on the following two lemmas, which are proven in \cref{sec:proof:lemmas:unbiased-grad1}.
The first lemma states that if a command contains no observe commands,
then its (full) density function can be decomposed into its partial density functions over $S$ and $\Name \setminus S$ for any $S \subseteq \Name$.
The second lemma states that
if $\pi$ is simple and $c$ uses only $\lambda y.y$ as the third argument of its sample commands,
then the partial density function of $\ctr{c}{\pi}$ over non-transformed random variables (i.e., variables in $\Name \setminus \repname(\pi)$)
is connected to that of $c$ via the value function of~$\ctr{c}{\pi}$.
\begin{lemma}
  \label{lem:dens-decomp}
  Let $c$ be a command. 
  If $c$ does not contain observe commands, then, for all $S \subseteq \Name$, $\sigma_\theta \in \State[\theta]$, and $\sigma_n \in \State[\Name]$,
  \begin{align*}
    \pfun{c, \sigma_\theta}{}(\sigma_n)
    &= \pfun{c, \sigma_\theta}{S}(\sigma_n) \cdot \pfun{c, \sigma_\theta}{\Name \setminus S}(\sigma_n).
  \end{align*}
\end{lemma}

\begin{lemma}
  \label{lem:dens-misc}
  Let $c$ be a command 
  and $\pi$ be a reparameterisation plan.
  Suppose that every sample command in $c$ has $\lambda y.y$ as its third argument.
  Then, for all $\sigma_\theta \in \State[\theta]$ and $\sigma_n \in \State[\Name]$,
  if $\pfun{\ctr{c}{\pi}, \sigma_\theta}{}(\sigma_n) > 0$, then
  \begin{align*}
    \pfun{\ctr{c}{\pi}, \sigma_\theta}{\Name \setminus \repname(\pi)}(\sigma_n)
    &= \pfun{c, \sigma_\theta}{\Name \setminus \repname(\pi)}(\vfun{\ctr{c}{\pi}, \sigma_\theta}(\sigma_n)).
  \end{align*}
\end{lemma}

We now prove \cref{thm:unbiased-grad} using the two lemmas.
\begin{proof}[Proof of \cref{thm:unbiased-grad}]
  Let $S = \repname(\pi)$.
  Before starting the main derivation of the selective gradient estimator,
  we show the differentiability of several functions which are to be used in the derivation.
  From \cref{cond:dens-diff} and \cref{cond:val-diff},
  the next functions over $\State[\theta]$ are differentiable for all $\sigma_n$
  by the preservation of differentiability under function composition:
  \begin{align*}
    \sigma_\theta &\longmapsto \pfun{c_m, \sigma_\theta}{}(\vfun{\ctr{c_g}{\pi}, \sigma_\theta}(\sigma_n)),
    &
    \sigma_\theta &\longmapsto \pfun{c_g, \sigma_\theta}{S}(\vfun{\ctr{c_g}{\pi}, \sigma_\theta}(\sigma_n)),
    &
    \sigma_\theta &\longmapsto \pfun{c_g, \sigma_\theta}{\Name \setminus S}(\vfun{\ctr{c_g}{\pi}, \sigma_\theta}(\sigma_n)),
    \\
    &&
    \sigma_\theta & \longmapsto \pfun{\ctr{c_g}{\pi},\sigma_\theta}{S}(\sigma_n),
    &
    \sigma_\theta & \longmapsto \pfun{\ctr{c_g}{\pi},\sigma_\theta}{\Name \setminus S}(\sigma_n).
  \end{align*}
  From this, the next functions over $\State[\theta]$ are also differentiable for all $\sigma_n$
  by \cref{lem:dens-decomp} with $c_g$ and $\ctr{c_g}{\pi}$ 
  and by the fact that the multiplication of differentiable functions is differentiable:
  \begin{align*}
    \sigma_\theta &\longmapsto \pfun{c_g, \sigma_\theta}{}(\vfun{\ctr{c_g}{\pi}, \sigma_\theta}(\sigma_n)),
    &
    \sigma_\theta &\longmapsto \pfun{\ctr{c_g}{\pi},\sigma_\theta}{}(\sigma_n).
  \end{align*}
  These differentiability results are required in the below proof
  to apply several gradients rules (e.g., $\nabla_\theta (f(\theta) + g(\theta)) = \nabla_\theta f(\theta) + \nabla_\theta g(\theta)$)
  which may fail for non-differentiable functions.

  Fix $\sigma_\theta \in \State[\theta]$. Using the above differentiability results, we derive the selective gradient estimator as follows,
  where we write $\sigma'_n$ for $\vfun{\ctr{c_g}{\pi}, \sigma_\theta}(\sigma_n)$:
  \begin{align*}
    & \nabla_\theta \cL_\theta
    \\
    &= \nabla_\theta \int d \sigma_n \left(
    \pfun{c_g,\sigma_\theta}{}(\sigma_n)
    \cdot \log\frac{\pfun{c_m, \sigma_\theta}{}(\sigma_n)}{\pfun{c_g, \sigma_\theta}{}(\sigma_n)} \right)
    \\
    &= \nabla_\theta \int d \sigma_n \left(
    \pfun{\hln{\ctr{c_g}{\pi}},\sigma_\theta}{}(\sigma_n)
    \cdot \log\frac{\pfun{c_m, \sigma_\theta}{}(\hln{\sigma_n'})}{\pfun{c_g, \sigma_\theta}{}(\hln{\sigma_n'})} \right)
    \\
    &= \int d \sigma_n \, \hln{\nabla_\theta} \left(
    \pfun{\ctr{c_g}{\pi},\sigma_\theta}{}(\sigma_n)
    \cdot \log\frac{\pfun{c_m, \sigma_\theta}{}(\sigma_n')}{\pfun{c_g, \sigma_\theta}{}(\sigma_n')} \right)
    \\
    &= \int d \sigma_n \left(
    \hln{\nabla_\theta} \pfun{\ctr{c_g}{\pi},\sigma_\theta}{}(\sigma_n)
    \cdot \log\frac{\pfun{c_m, \sigma_\theta}{}(\sigma_n')}{\pfun{c_g, \sigma_\theta}{}(\sigma_n')}
     +
    \pfun{\ctr{c_g}{\pi},\sigma_\theta}{}(\sigma_n)
    \cdot \hln{\nabla_\theta} \log\frac{\pfun{c_m, \sigma_\theta}{}(\sigma_n')}{\pfun{c_g, \sigma_\theta}{}(\sigma_n')} \right)
    \\
    &= \int d \sigma_n \, \pfun{\ctr{c_g}{\pi},\sigma_\theta}{}(\sigma_n)
    \left(
    \hln{\nabla_\theta \log \pfun{\ctr{c_g}{\pi},\sigma_\theta}{}(\sigma_n)}
    \cdot \log\frac{\pfun{c_m, \sigma_\theta}{}(\sigma_n')}{\pfun{c_g, \sigma_\theta}{}(\sigma_n')}
    - \nabla_\theta \log \pfun{c_g, \sigma_\theta}{}(\sigma_n')
    + \nabla_\theta \log \pfun{c_m, \sigma_\theta}{}(\sigma_n') \right)
    \\
    &= \int d \sigma_n \, \pfun{\ctr{c_g}{\pi},\sigma_\theta}{}(\sigma_n)
    \bigg[ \left(\hln{\nabla_\theta \log \pfun{\ctr{c_g}{\pi},\sigma_\theta}{S}(\sigma_n)
      +\nabla_\theta \log \pfun{\ctr{c_g}{\pi},\sigma_\theta}{\Name \setminus S}(\sigma_n)}\right)
    \cdot \log\frac{\pfun{c_m, \sigma_\theta}{}(\sigma_n')}{\pfun{c_g, \sigma_\theta}{}(\sigma_n')}
    \\ & \qquad\qquad\qquad\qquad\quad
    - \left(\hln{\nabla_\theta \log \pfun{c_g, \sigma_\theta}{S}(\sigma_n')
      + \nabla_\theta \log \pfun{c_g, \sigma_\theta}{\Name \setminus S}(\sigma_n')}\right)
      + \nabla_\theta \log \pfun{c_m, \sigma_\theta}{}(\sigma_n') \bigg]
    \\
    &= \int d \sigma_n \, \pfun{\ctr{c_g}{\pi},\sigma_\theta}{}(\sigma_n)
    \bigg[ \left(\hln{0}+\nabla_\theta \log \pfun{\ctr{c_g}{\pi},\sigma_\theta}{\Name \setminus S}(\sigma_n)\right)
    \cdot \log\frac{\pfun{c_m, \sigma_\theta}{}(\sigma_n')}{\pfun{c_g, \sigma_\theta}{}(\sigma_n')}
    \\ &  \qquad\qquad\qquad\qquad\quad
    - \left(\nabla_\theta \log \pfun{c_g, \sigma_\theta}{S}(\sigma_n') + \hln{0}\right)
    + \nabla_\theta \log \pfun{c_m, \sigma_\theta}{}(\sigma_n') \bigg]
    \\
    &= \int d \sigma_n \, \pfun{\ctr{c_g}{\pi},\sigma_\theta}{}(\sigma_n)
    \bigg( \nabla_\theta \log \pfun{\hln{c_g},\sigma_\theta}{\Name \setminus S}(\hln{\sigma_n'})
    \cdot \log\frac{\pfun{c_m, \sigma_\theta}{}(\sigma_n')}{\pfun{c_g, \sigma_\theta}{}(\sigma_n')}
    \\ &  \qquad\qquad\qquad\qquad\quad
    - \nabla_\theta \log \pfun{c_g, \sigma_\theta}{S}(\sigma_n')
    + \nabla_\theta \log \pfun{c_m, \sigma_\theta}{}(\sigma_n') \bigg).
  \end{align*}
  We justify key steps of the above derivation below. 
  \begin{itemize}
  \item The second equality comes from \cref{thm:unbiased-val}
    and the fact that $\vfun{c_g,\sigma_\theta}$ is the identity function (since the third argument of every 
    sample command in $c_g$ is the identity function $\lambda y.y$).
  \item The third equality holds because differentiation there commutes with integration by \cref{cond:int-diff}.
  \item The fourth comes from the product rule for differentiation: $\nabla_\theta (f(\theta) \cdot g(\theta))
    = \nabla_\theta f(\theta) \cdot g(\theta) + f(\theta) \cdot \nabla_\theta g(\theta)$
    for all differentiable $f$ and $g$.
    Here $f$ and $g$ in the original equation are differentiable
    because differentiability is preserved under division and $\log$ for positive-valued functions.
  \item The fifth equality holds because $\nabla_\theta f(\theta) = f(\theta) \cdot  \nabla_\theta \log f(\theta)$
    for all differentiable and positive-valued $f$.
  \item The sixth equality follows from \cref{lem:dens-decomp} applied to $c_g$ and $\ctr{c_g}{\pi}$ (both of which do not contain observe commands),
    and from the linearity of differentiation: $\nabla_\theta (f(\theta) + g(\theta)) = \nabla_\theta f(\theta) + \nabla_\theta g(\theta)$
    for all differentiable $f$ and $g$.
    Here $f$ and $g$ in the original equation are differentiable
    because differentiability is preserved under $\log$ for positive-valued functions.
  \item The seventh equality follows from \cref{cond:grad-dens-zero} and
    \begin{align}
      \label{eq:zero-int:thm:unbiased-grad}
      \E_{\pfun{\ctr{c_g}{\pi}, \sigma_\theta}{}(\sigma_n)}
      \left[ \nabla_\theta \log \pfun{c_g, \sigma_\theta}{\Name \setminus S}(\sigma_n') \right]
      &= 0.
    \end{align}
    The proof of \cref{eq:zero-int:thm:unbiased-grad} will be given after we complete this justification of the derivation.
  \item The last equality comes from \cref{lem:dens-misc} applied to $c_g$.
  \end{itemize}
  
  The only remaining part is to prove \cref{eq:zero-int:thm:unbiased-grad}.
  We derive the equation as follows:
  \begin{align*}
    & \int d \sigma_n  \left( \pfun{\ctr{c_g}{\pi}, \sigma_\theta}{}(\sigma_n) \cdot
    \nabla_\theta \log \pfun{c_g, \sigma_\theta}{\Name \setminus S}(\sigma_n') \right)
    \\
    & {} = \int d \sigma_n \left( \pfun{\ctr{c_g}{\pi}, \sigma_\theta}{}(\sigma_n) \cdot
    \nabla_\theta \log \pfun{\hln{\ctr{c_g}{\pi}}, \sigma_\theta}{\Name \setminus S}(\hln{\sigma_n}) \right)
    \\
    & {} = \int d \sigma_n \left(\pfun{\ctr{c_g}{\pi}, \sigma_\theta}{}(\sigma_n)
    \cdot \left( \nabla_\theta \log \pfun{\ctr{c_g}{\pi}, \sigma_\theta}{\Name \setminus S}(\sigma_n)
    + \hln{\nabla_\theta \log \pfun{\ctr{c_g}{\pi}, \sigma_\theta}{S}(\sigma_n)} \right)\right)
    \\
    & {} = \int d \sigma_n \left( \pfun{\ctr{c_g}{\pi}, \sigma_\theta}{}(\sigma_n) \cdot
    \hln{\nabla_\theta \log \pfun{\ctr{c_g}{\pi}, \sigma_\theta}{}(\sigma_n)} \right)
    \\
    & {} = \int d\sigma_n \, \hln{\nabla_\theta \pfun{\ctr{c_g}{\pi}, \sigma_\theta}{}(\sigma_n)}
    \\
    & {} = \hln{\nabla_\theta} \int d\sigma_n \, \pfun{\ctr{c_g}{\pi}, \sigma_\theta}{}(\sigma_n)
    \\
    & {} = \nabla_\theta \hln{1}
    = 0.
  \end{align*}
  Here is the justification of the above derivation:
  \begin{itemize}
  \item The first equality comes from \cref{lem:dens-misc} applied to $c_g$.
  \item The second equality follows from \cref{cond:grad-dens-zero}.
  \item The third equality holds because of \cref{lem:dens-decomp} applied to $\ctr{c_g}{\pi}$ (which does not contain observe commands),
    and the linearity of differentiation: $\nabla_\theta (f(\theta) + g(\theta)) = \nabla_\theta f(\theta) + \nabla_\theta g(\theta)$
    for all differentiable $f$ and $g$.
    Here $f$ and $g$ in the original equation are differentiable
    because differentiability is preserved under $\log$ for positive-valued functions.
  \item The fourth equality holds because $\nabla_\theta f(\theta) = f(\theta) \cdot \nabla_\theta \log f(\theta)$
    for all differentiable and positive-valued $f$.
  \item The fifth equality uses  \cref{cond:int-diff}, which states the commutativity between differentiation and integration in the equality.
  \item The six equality comes from that $\pfun{\ctr{c_g}{\pi},{\sigma_\theta}}{}$ is a probability density by \cref{remark:p-probability}.
  \end{itemize}
  This completes the proof.
\end{proof}

\subsection{Proofs of \cref{lem:dens-decomp,lem:dens-misc}}
\label{sec:proof:lemmas:unbiased-grad1}

We define the partial density version of $\getprsubi{S}(c)$ for $S \subseteq \Name$:
\begin{align*}
  \getprsubi{S}(c) &: \State[\PVar] \times \Statesub[\Name] \to [0,\infty),
  \\
  \getprsubi{S}(c)(\sigma_p, \xi_n)
  &\defeq \begin{cases}
    \begin{array}{@{}l@{}}
      \prod_{\mu \in \dom(\xi_n) \cap S} \db{c}(\sigma_p \oplus \xi_n \oplus \sigma_r)(\pr_\mu)
    \end{array}
    & \text{if $\exists \sigma_r.\, \used(c, \sigma_p \oplus \xi_n \oplus \sigma_r, \xi_n)$}
    \\
    0 & \text{otherwise}.
  \end{cases}
\end{align*}
$\getprsubi{S}(c)$ enjoys many of the properties that $\getprsub(c)$ has.
For instance, $\getprsubi{S}(c)$ is a well-defined function (i.e., its value does not depend on the choice of $\sigma_r$), 
as $\getprsub(c)$ does.
Since the proof of those properties of $\getprsubi{S}(c)$ would be almost identical to that of $\getprsub(c)$,
we will use them in the following proofs without explicitly (re)proving them.

\begin{proof}[Proof of \cref{lem:dens-decomp}]
  Let $c$ be a command that has no observe commands.
  Let $S \subseteq \Name$, $\sigma_\theta \in \State[\theta]$, and $\sigma_n \in \State[\Name]$.
  We set $\sigma_0$ as in the definition of $\pfun{c,\sigma_\theta}{}$ in \cref{eq:density-of-c} (as a function of $\sigma_n$).
  Let $\sigma = \sigma_\theta \oplus \sigma_n \oplus \sigma_0$.
  If $\noerr(c, \sigma)$ does not hold, then the LHS and RHS of the desired equation become zero, so the equation holds.
  If $\noerr(c, \sigma)$ holds, we get the desired equation as follows:
  \begin{align*}
    \pfun{c, \sigma_\theta}{S}(\sigma_n) \cdot \pfun{c, \sigma_\theta}{\Name \setminus S}(\sigma_n)
    &= \textstyle
    \big( \db{c}\sigma(\like)
    \cdot \prod_{\mu \in S} \db{c}\sigma(\pr_\mu) \big)
    \cdot \big( \db{c}\sigma(\like)
    \cdot \prod_{\mu \in \Name \setminus S} \db{c}\sigma(\pr_\mu) \big)
    \\
    &= \textstyle
    \big( \db{c}\sigma(\like) \big)^2
    \cdot \prod_{\mu \in \Name} \db{c}\sigma(\pr_\mu)
    \\
    &= \textstyle
    \db{c}\sigma(\like)
    \cdot \pfun{c, \sigma_\theta}{}(\sigma_n)
    \\
    &= \sigma(\like) \cdot \pfun{c, \sigma_\theta}{}(\sigma_n)
    \\
    &= \pfun{c, \sigma_\theta}{}(\sigma_n).
  \end{align*}
  The second last equality uses \cref{lem:obs-free-like},
  and the last equality uses $\sigma(\like)=1$ (which holds by the definition of $\sigma_0$).
  This completes the proof.
\end{proof}

\begin{proof}[Proof of \cref{lem:dens-misc}]
  Consider a command $c$, a reparameterisation plan $\pi$, $\sigma_\theta \in \State[\theta]$, and $\sigma_n \in \State[\Name]$.
  Assume that all the sample commands of $c$ have $\lambda y.y$ as their third arguments,
  and $p_{\ctr{c}{\pi}, \sigma_\theta}(\sigma_n) > 0$.

  We first define several objects and make observations on them.
  Let $S \defeq \Name \setminus \repname(\pi)$.
  Define $f_* : \State[\Name] \to \State[\AVar]$ to be the function for constructing an initial state:
  \begin{align*}
    f_*(\sigma_n)(v)
    &\defeq
    \begin{cases}
      f_\pr(\sigma_n(\mu)) & \text{if $v \equiv \pr_\mu$ for some $\mu$}
      \\
      f_\val(\sigma_n(\mu)) & \text{if $v \equiv \val_\mu$ for some $\mu$}
      \\
      f_\sampled(\sigma_n(\mu)) & \text{if $v \equiv \sampled_\mu$ for some $\mu$}
      \\
      1 & \text{if $v \equiv \like$},
    \end{cases}
  \end{align*}
  where $f_\val(r) \defeq r$, $f_\pr(r) \defeq \cN(r; 0, 1)$, and $f_\sampled(r) \defeq 0$.
  Define initial states $\overline{\sigma}, \sigma \in \State$ for
  $\pfun{\ctr{c}{\pi}, \sigma_\theta}{}(\sigma_n)$ and
  $\pfun{{c}, \sigma_\theta}{}(\vfun{\ctr{c}{\pi}, \sigma_\theta}{}(\sigma_n))$, respectively, as
  \begin{align*}
    \overline{\sigma}
    &\defeq \sigma_p \oplus \sigma_n \oplus f_*(\sigma_n),
    &
    {\sigma}
    &\defeq \sigma_p \oplus \vfun{\ctr{c}{\pi}, \sigma_\theta}(\sigma_n) \oplus f_*(\vfun{\ctr{c}{\pi}, \sigma_\theta}(\sigma_n)),
  \end{align*}
  where $\sigma_p \defeq \sigma_\theta \oplus (\lambda v \in \PVar \setminus \theta.\, 0) \in \State[\PVar]$.
  Then, the assumption $\pfun{\ctr{c}{\pi}, \sigma_\theta}{}(\sigma_n) > 0$ implies $\noerr(\ctr{c}{\pi}, \overline{\sigma})$ by the definition of $p$.
  From this, $\overline{\sigma}(\like)=1$, and the definition of $\used$, there exists $\overline{\xi_n} \in \Statesub[\Name]$
  such that $\used(\ctr{c}{\pi}, \overline{\sigma}, \overline{\xi_n})$.
  From this, we have
  \begin{align*}
    \overline{\sigma} &= \sigma_p \oplus \overline{\xi_n} \oplus \overline{\sigma_r},
    &
    \used(\ctr{c}{\pi}, \sigma_p \oplus \overline{\xi_n} \oplus \overline{\sigma_r}, \overline{\xi_n}),
  \end{align*}
  for some $\overline{\sigma_r}$.
  Next, let \[\xi_n \defeq \getvalsub(\ctr{c}{\pi})(\sigma_p, \overline{\xi_n}).\]
  We can apply \cref{lem:id-lambda-props} to $\used(\ctr{c}{\pi}, \overline{\sigma}, \overline{\xi_n})$,
  since all the sample commands of $c$ have $\lambda y.y$ in their third arguments (by assumption).
  The application of the lemma gives:
  \begin{gather*}
    \forall \sigma_r' \in \State[\Var \setminus (\PVar \cup \dom(\xi_n))].\;\;
    \sigma_r'(\like)=1 \implies \used(c, \sigma_p \oplus \xi_n \oplus \sigma_r', \xi_n),
    \\
    \getprsubi{S}(\ctr{c}{\pi})(\sigma_p, \overline{\xi_n})
    = \getprsubi{S}(c)(\sigma_p, \xi_n),
  \end{gather*}
  where the for-all part comes from \cref{lem:used-sigma-indep}.
  
  We now show two claims.
  The first claims is: there exists $\sigma_r$ such that
  \begin{align*}
    \sigma &= \sigma_p \oplus \xi_n \oplus \sigma_r,
    &
    \used(c, \sigma_p \oplus \xi_n \oplus \sigma_r, \xi_n).    
  \end{align*}
  By the definition of $\sigma$, it suffices to show that $\xi_n = \big(\vfun{\ctr{c}{\pi}, \sigma_\theta}(\sigma_n)\big) |_{\dom(\xi_n)}$.
  This indeed holds as follows: for any $\mu \in \dom(\xi_n)$,
  $
    \xi_n(\mu)
    = \getvalsub(\ctr{c}{\pi})(\sigma_p, \overline{\xi_n})(\mu)
    = \db{\ctr{c}{\pi}} \overline{\sigma}(\val_\mu) = \vfun{\ctr{c}{\pi}, \sigma_\theta}(\sigma_n)(\mu),
  $
  where the second equality uses $\used(\ctr{c}{\pi}, \sigma_p \oplus \overline{\xi_n} \oplus \overline{\sigma_r}, \overline{\xi_n})$.
  The second claim is:
  for all $\mu \in S \setminus \dom({\xi_n})$,
  \[ \sigma(\pr_\mu) = \overline{\sigma}(\pr_\mu). \]
  Here is the proof of the claim:
  $
  \sigma(\pr_\mu)
  = f_\pr(\sigma(\mu))
  = f_\pr(\vfun{\ctr{c}{\pi}, \sigma_\theta}(\sigma_n)(\mu))
  = f_\pr(\db{\ctr{c}{\pi}} \overline{\sigma}(\val_\mu))
  = f_\pr(\overline{\sigma}(\val_\mu))
  = f_\pr(\overline{\sigma}(\mu));
  $
  and 
  $
  \overline{\sigma}(\pr_\mu)
  = f_\pr(\overline{\sigma}(\mu));
  $
  here the second last equality in the first equation uses
  \cref{lem:used-properties}-\cref{lem:used-properties-2}
  with $\mu \notin \dom(\overline{\xi_n})$
  and $\used(\ctr{c}{\pi}, \sigma_p \oplus \overline{\xi_n} \oplus \overline{\sigma_r}, \overline{\xi_n})$,
  and the last equality in the first equation uses $f_\val(r) = r$.

  Based on the observations made so far, we show the desired equation as follows:
  \begin{align*}
    \pfun{\ctr{c}{\pi}, \sigma_\theta}{S}(\sigma_n)
    &=
    \prod_{\mu \in S \cap \dom(\overline{\xi_n})} \db{\ctr{c}{\pi}} \overline{\sigma}(\pr_\mu)
    \cdot \prod_{\mu \in S \setminus \dom(\overline{\xi_n})} \db{\ctr{c}{\pi}} \overline{\sigma}(\pr_\mu)
    \\
    &=
    \getprsubi{S}(\ctr{c}{\pi})(\sigma_p, \overline{\xi_n})
    \cdot \prod_{\mu \in S \setminus \dom(\overline{\xi_n})} \overline{\sigma}(\pr_\mu)
    \\
    &=
    \getprsubi{S}(c)(\sigma_p, {\xi_n})
    \cdot \prod_{\mu \in S \setminus \dom({\xi_n})} {\sigma}(\pr_\mu)
    \\
    &=
    \prod_{\mu \in S \cap \dom({\xi_n})} \db{{c}} {\sigma}(\pr_\mu) 
    \cdot \prod_{\mu \in S \setminus \dom({\xi_n})} \db{{c}} {\sigma}(\pr_\mu) 
    \\
    &= \pfun{{c}, \sigma_\theta}{S}(\vfun{\ctr{c}{\pi}, \sigma_\theta}(\sigma_n)).
  \end{align*}
  The first and last equalities are by the definition of $p$.
  The second equality uses $\used(\ctr{c}{\pi}, \sigma_p \oplus \overline{\xi_n} \oplus \overline{\sigma_r}, \overline{\xi_n})$
  and \cref{lem:used-properties}-\cref{lem:used-properties-2} with $\mu \notin \dom(\overline{\xi_n})$.
  The third equality uses $\dom(\overline{\xi_n}) = \dom(\xi_n)$,
  the observation made in the first paragraph, and the second claim in the above.
  The fourth equality uses the first claim in the above,
  and \cref{lem:used-properties}-\cref{lem:used-properties-2} with $\mu \notin \dom({\xi_n})$.
\end{proof}

\begin{lemma}
  \label{lem:obs-free-like}
  Let $c$ be a command and $\sigma \in \State$.
  If $c$ has no observe commands and $\db{c}\sigma \in \State$, then
  \[ \db{c}\sigma(\like) = \sigma(\like). \]
\end{lemma}
\begin{proof}
  Let $c$ be a command that does not contain an observe command. We show the claim of the lemma by induction on the structure of $c$. 
  Pick $\sigma \in \State$ such that $\db{c}\sigma \in \State$. We will show that $\db{c}\sigma(\like) = \sigma(\like)$.

\paragraph{\bf Case $c \equiv \cskip$} In this case, $\db{c}\sigma(\like) = \sigma(\like)$ by the definition of the semantics.

\paragraph{\bf Case $c \equiv (x:=e)$} Again, $\db{c}\sigma(\like) = \sigma(\like)$ by the definition of the semantics. 

\paragraph{\bf Case $c \equiv (x:=\csample(n,d,\lambda y.e')$} Once more, $\db{c}\sigma(\like) = \sigma(\like)$ by the definition of the semantics. 

\paragraph{\bf Case $c \equiv (c';c'')$} We have $\db{c'}\sigma \in \State$ and $\db{c''}(\db{c'}\sigma) \in \State$. We apply induction hypothesis first to $(c',\sigma)$, and again to $(c',\db{c'}\sigma)$. The first application gives $\db{c'}\sigma(\like) = \sigma(\like)$, and the second $\db{c';c''}\sigma(\like) = \db{c'}\sigma(\like)$. The desired conclusion follows from these two equalities.

\paragraph{\bf Case $c \equiv (\cif\ b\ \{c'\}\ \celse\ \{c''\})$} We deal with the case that $\db{b}\sigma = \strue$. The other case of $\db{b}\sigma = \sfalse$ can be proved similarly. Since $\db{b}\sigma = \strue$, we have $\db{c'}\sigma = \db{c}\sigma \in \State$. Thus, we can apply induction hypothesis to $c'$.  If we do so, we get 
$\db{c'}\sigma(\like) = \sigma(\like)$. This gives the desired conclusion because $\db{c}\sigma = \db{c'}\sigma$.

\paragraph{\bf Case $c \equiv (\cwhile\ b\ \{c'\})$} Let $F$ be the operator on $[\State \to \State_\bot]$ such that $\db{c}$ is the least fixed point of $F$. Define a subset $\cT$ of $[\State \to \State_\bot]$ as follows: 
\[
f \in \cT \iff
\Big(
\forall \sigma' \in \State.\, f(\sigma') \in \State \implies f(\sigma')(\like) = \sigma'(\like).
\Big)
\]
The set $\cT$ contains the least function $\lambda \sigma.\bot$, and is closed under the least upper bound of any chain in $[\State \to \State_\bot]$. It is also closed under $F$. This $F$-closure follows essentially from our arguments for sequential composition, if command, and skip, and induction hypothesis on $c'$. What we have shown for $\cT$ implies that $\cT$ contains the least fixed point of $F$, which gives the desired property for $c$.
\end{proof}

\begin{lemma}
  \label{lem:id-lambda-props}
  Let $c$ be a command and $\pi$ be a reparameterisation plan.
  Suppose that every sample command in $c$ has $\lambda y.y$ as its third argument.
  Then, for all $\sigma_p \in \State[\PVar]$ and $\overline{\xi_n} \in \Statesub[\Name]$,
  if $\used(\ctr{c}{\pi}, \sigma_p \oplus \overline{\xi_n} \oplus \overline{\sigma_r}, \overline{\xi_n})$ for some $\overline{\sigma_r}$,
  then
  \begin{align*}
    \exists {\sigma_r}.\,\used(c, {} & {} \sigma_p \oplus \xi_n \oplus {\sigma_r}, \xi_n),
    \\
    \getpvarsub(\ctr{c}{\pi})(\sigma_p, \overline{\xi_n})
    &= \getpvarsub(c)(\sigma_p, \xi_n),
    \\
    \getprsubi{\Name \setminus \repname(\pi)}(\ctr{c}{\pi})(\sigma_p, \overline{\xi_n})
    &= \getprsubi{\Name \setminus \repname(\pi)}(c)(\sigma_p, \xi_n),
  \end{align*}
  where
  \begin{align*}
    \xi_n \defeq \getvalsub(\ctr{c}{\pi})(\sigma_p, \overline{\xi_n}).
  \end{align*}
\end{lemma}
\begin{proof}
  Fix a reparameterisation plan $\pi$.
  The proof proceeds by induction on the structure of $c$.
  Let $\sigma_p \in \State[\PVar]$, and $\overline{\xi_n} \in \Statesub[\Name]$.
  Assume that $c$ uses only $\lambda y.y$ in the third argument of its sample commands,
  and $\used(\ctr{c}{\pi}, \sigma_p \oplus \overline{\xi_n} \oplus \overline{\sigma_r}, \overline{\xi_n})$ for some $\overline{\sigma_r}$.
  Let $\overline{\sigma} \defeq \sigma_p \oplus \overline{\xi_n} \oplus \overline{\sigma_r}$
  and $S \defeq \Name \setminus \repname(\pi)$.
  Then, we simply have ${ \used(\ctr{c}{\pi}, \overline{\sigma}, \overline{\xi_n}) }$.
    
  \paragraph{\bf Cases $c \equiv \cskip$, $c \equiv (x:=e)$, or $c \equiv \cobs(d,r)$}
  In this case, $\db{c}\sigma(\sampled_\mu) = \sigma(\sampled_\mu)$ for all $\sigma \in \State$ and $\mu \in \Name$.
  So $\dom(\overline{\xi_n}) = \dom(\xi_n) = \emptyset$ and thus $\overline{\xi_n} = \xi_n$.
  We also know $\ctr{c}{\pi} \equiv c$.
  From these, all of the three conclusions follow immediately.

  \paragraph{\bf Case $c \equiv (x:=\csample(n,d,\lambda y.e))$}
  Since $\fv(n) \subseteq \PVar$, there exists $\mu \in \Name$ such that
  $\db{n}(\sigma_p \oplus \sigma_r) = \mu$ for all $\sigma_r \in \State[\Var \setminus \PVar]$.
  So, for all $\sigma_r \in \State[\Var \setminus \PVar]$ and $\mu' \in \Name \setminus \{\mu\}$,
  \begin{align}
    \label{eq:lem:id-lambda-props-sam-cnt}
    \db{c}(\sigma_p \oplus \sigma_r)(\sampled_{\mu'}) =
    \begin{cases}
      (\sigma_p \oplus \sigma_r)(\sampled_{\mu'}) + 1
      & \text{if $\mu' = \mu$}
      \\
      (\sigma_p \oplus \sigma_r)(\sampled_{\mu'})
      & \text{otherwise}.
    \end{cases}
  \end{align}
  From this, we get $\dom(\overline{\xi_n}) = \dom(\xi_n) = \{\mu\}$.
  Further, by assumption, we get $e \equiv y$.
  We now prove the three conclusions based on these observations and case analysis on $(n, d, \lambda y.e)$.

  First, assume $(n, d, \lambda y.e) \notin \dom(\pi)$.
  Then, $\ctr{c}{\pi} \equiv c$ and
  \begin{align*}
    \xi_n &= \getvalsub(\ctr{c}{\pi})(\sigma_p, \overline{\xi_n})
    = [\mu \mapsto \db{\ctr{c}{\pi}}\overline{\sigma}(\val_\mu)]
    = [\mu \mapsto \db{e[\overline{\sigma}(\mu)/y]}\overline{\sigma}]
    = [\mu \mapsto \overline{\xi_n}(\mu)] = \overline{\xi_n},
  \end{align*}
  where the second last equality uses $e \equiv y$.
  Hence, the three conclusions clearly hold.

  Next, assume $(n, d, \lambda y.e) \in \dom(\pi)$.
  Suppose that $\pi(n, d, \lambda y.e) = (\overline{d}, \lambda \overline{y}.\overline{e})$.
  Then, $\ctr{c}{\pi} \equiv (x:=\csample(n,\overline{d},\lambda \overline{y}.\overline{e}))$ and
  \begin{align*}
    \xi_n &= \getvalsub(\ctr{c}{\pi})(\sigma_p, \overline{\xi_n})
    = [\mu \mapsto \db{\ctr{c}{\pi}}\overline{\sigma}(\val_\mu)]
    = [\mu \mapsto \db{\overline{e}[\overline{\sigma}(\mu)/\overline{y}]}\overline{\sigma}]
    = [\mu \mapsto \db{\overline{e}[\overline{\xi_n}(\mu)/\overline{y}]}\overline{\sigma}].
  \end{align*}
  Since \cref{eq:lem:id-lambda-props-sam-cnt} holds also for $\db{\ctr{c}{\pi}}$,
  and since $\dom(\xi_n) = \{\mu\}$, we get the first conclusion: \[\used(c, \sigma_p \oplus \xi_n \oplus \overline{\sigma_r}, \xi_n).\]
  To prove the second conclusion, let $\sigma = \sigma_p \oplus \xi_n \oplus \overline{\sigma_r}$.
  Then, for all $v \in \PVar$,
  \begin{align*}
    \getpvarsub(\ctr{c}{\pi})(\sigma_p, \overline{\xi_n})(v)
    &= \db{\ctr{c}{\pi}}\overline{\sigma}(v)
    =
    \begin{cases}
      \db{\overline{e}[\overline{\sigma}(\mu)/\overline{y}]}\overline{\sigma}
      = \db{\overline{e}[\overline{\xi_n}(\mu)/\overline{y}]}\overline{\sigma}
      & \text{if $v \equiv x$}
      \\
      \overline{\sigma}(v)
      = \sigma_p(v)
      & \text{otherwise},
    \end{cases}
    \\
    \getpvarsub(c)(\sigma_p, \xi_n)(v)
    &= \db{c}{\sigma}(v)
    =
    \begin{cases}
      \db{{e}[{\sigma}(\mu)/{y}]}\sigma
      = \xi_n(\mu)
      = \db{\overline{e}[\overline{\xi_n}(\mu)/\overline{y}]}\overline{\sigma}
      & \text{if $v \equiv x$}
      \\
      {\sigma}(v)
      = \sigma_p(v),
      & \text{otherwise},
    \end{cases}
  \end{align*}
  where the second equation uses $e \equiv y$.
  Hence, the second conclusion holds.
  For the third conclusion, let $n = \cname(\alpha, \_)$.
  Then, $\mu = (\alpha, \_) \in \{(\alpha, i) \in \Name \mid i \in \N\} \subseteq \repname(\pi)$.
  Thus, $\dom(\overline{\xi_n}) \cap S = \dom(\xi_n) \cap S = \{\mu\} \cap S = \{\mu\} \cap (\Name \setminus \repname(\pi)) = \emptyset$.
  From this, we get the third conclusion:
  \begin{align*}
    \getprsubi{S}(\ctr{c}{\pi})(\sigma_p, \overline{\xi_n})
    = 1 = \getprsubi{S}(c)(\sigma_p, \xi_n).
  \end{align*}

  \paragraph{\bf Case $c \equiv (c';c'')$}
  First, we make several observations necessary to prove the conclusion.
  By $\used(\ctr{c}{\pi}, \overline{\sigma}, \overline{\xi_n})$,
  we have $\db{\ctr{c'}{\pi}}\overline{\sigma} \in \State$ and $\db{\ctr{c''}{\pi}}(\db{\ctr{c'}{\pi}}\overline{\sigma}) \in \State$.
  Let
  \begin{align*}
    \overline{\sigma'} &\defeq \db{\ctr{c'}{\pi}}\overline{\sigma},
    &
    \overline{\sigma''} &\defeq \db{\ctr{c''}{\pi}}\overline{\sigma'},
    &
    {\sigma_p'} &\defeq \overline{\sigma'}|_\PVar,
    &
    {\sigma_p''} &\defeq \overline{\sigma''}|_\PVar.
  \end{align*}
  Then, by $\used(\ctr{c}{\pi}, \overline{\sigma}, \overline{\xi_n})$
  and the claim in the proof of \cref{lem:used-properties} (for the sequential composition case),
  there exist $\overline{\xi_n'}$ and $\overline{\xi_n''}$ such that
  \begin{align*}
    &\overline{\xi_n} = \overline{\xi_n'} \oplus \overline{\xi_n''},
    \qquad
    \used(\ctr{c'}{\pi}, \overline{\sigma}, \overline{\xi_n'}),
    \qquad
    \usedm(\ctr{c''}{\pi}, \overline{\sigma'}, \overline{\xi_n''}).
  \end{align*}
  By the latter two, we can apply induction to $(c', \sigma_p, \overline{\xi_n'})$ and $(c'', \sigma_p', \overline{\xi_n''})$,
  and IH gives the following:
  \begin{align*}
    \getpvarsub(c')(\sigma_p, {\xi_n'})
    &=\getpvarsub(\ctr{c'}{\pi})(\sigma_p, \overline{\xi_n'})
    &\text{[By IH on $c'$]}
    \\
    &= \big( \db{\ctr{c'}{\pi}} (\overline{\sigma} [\like \mapsto 1]) \big)|_\PVar
    &\text{[By $\used(\ctr{c'}{\pi}, \overline{\sigma}, \overline{\xi_n'})$]}
    \\
    &= ( \db{\ctr{c'}{\pi}} \overline{\sigma} )|_\PVar
    = \overline{\sigma'}|_\PVar = \sigma_p',
    &\text{[By \cref{lem:used-properties}-\cref{lem:used-properties-3}]}
    \\
    \getpvarsub(c'')(\sigma_p', {\xi_n''})
    &= \getpvarsub(\ctr{c''}{\pi})(\sigma_p', \overline{\xi_n''})
    &\text{[By IH on $c''$]}
    \\
    &= \big( \db{\ctr{c''}{\pi}} (\overline{\sigma'} [\like \mapsto 1]) \big)|_\PVar
    &\text{[By $\usedm(\ctr{c''}{\pi}, \overline{\sigma'}, \overline{\xi_n''})$]}
    \\
    &= (\db{\ctr{c''}{\pi}} \overline{\sigma'} )|_\PVar
    = \overline{\sigma''}|_\PVar = \sigma_p'',
    &\text{[By \cref{lem:used-properties}-\cref{lem:used-properties-3}]}
  \end{align*}
  where
  \begin{align*}
    {\xi_n'} &\defeq \getvalsub(\ctr{c'}{\pi})(\sigma_p, \overline{\xi_n'}),
    &
    {\xi_n''} &\defeq \getvalsub(\ctr{c''}{\pi})(\sigma_p', \overline{\xi_n''}).
  \end{align*}
  By the former equation, we get
  \begin{align*}
    \xi_n
    &= \getvalsub(\ctr{c}{\pi})(\sigma_p, \overline{\xi_n})
    \\
    &= \getvalsub(\ctr{c'}{\pi}; \ctr{c''}{\pi})(\sigma_p, \overline{\xi_n'} \oplus \overline{\xi_n''})
    &\text{[By $\overline{\xi_n}  = \overline{\xi_n'} \oplus \overline{\xi_n''}$]}
    \\
    &=
    \getvalsub(\ctr{c'}{\pi})(\sigma_p, \overline{\xi_n'}) \oplus
    \getvalsub(\ctr{c''}{\pi})(\getpvarsub(\ctr{c'}{\pi})(\sigma_p, \overline{\xi_n'}), \overline{\xi_n''})
    \\
    &=
    \getvalsub(\ctr{c'}{\pi})(\sigma_p, \overline{\xi_n'}) \oplus
    \getvalsub(\ctr{c''}{\pi})(\sigma_p', \overline{\xi_n''})
    &\text{[By the former equation]}
    \\
    &=
    \xi_n' \oplus \xi_n''
  \end{align*}
  where the third equality uses
  $\used(\ctr{c'}{\pi}, \overline{\sigma}, \overline{\xi_n'})$,
  $\used(\ctr{c'}{\pi}; \ctr{c''}{\pi}, \overline{\sigma}, \overline{\xi_n'} \oplus \overline{\xi_n''})$,
  and the second claim in the proof of \cref{lem:integral-seq-decompose}.
  
  We now show the first conclusion.
  By IH on $(c', \sigma_p, \overline{\xi_n'})$ and $(c'', \sigma_p', \overline{\xi_n''})$, we get
  \begin{align*}
    &\exists {\sigma_r'}.\,\used(c', {}  {} \sigma_p \oplus {\xi_n'} \oplus {\sigma_r'}, {\xi_n'}),
    &
    &\exists {\sigma_r''}.\,\used(c'', {}  {} \sigma_p' \oplus {\xi_n''} \oplus {\sigma_r''}, {\xi_n''}).
  \end{align*}
  Let \[\sigma \defeq \sigma_p \oplus (\xi_n' \oplus \xi_n'') \oplus
  \sigma_r'|_{\dom(\sigma_r') \setminus \dom(\xi_n'')}.\]
  Then, $\db{c'}\sigma \in \State$ by $\used(c', {}  {} \sigma_p \oplus {\xi_n'} \oplus {\sigma_r'}, {\xi_n'})$
  and \cref{lem:used-properties}-\cref{lem:used-properties-1},
  and we have
  \begin{align*}
    \sigma|_\PVar &= \sigma_p,
    &
    (\db{c'}\sigma)|_\PVar
    &= (\db{c'}(\sigma_p \oplus {\xi_n'} \oplus {\sigma_r'}))|_\PVar
    &\text{[By \cref{lem:used-properties}-\cref{lem:used-properties-3}]}
    \\
    &&
    &= \getpvar(c')(\sigma_p, \xi_n') = \sigma_p',
    &\text{[By the above]}
    \\
    \sigma|_{\dom(\xi_n')} &= \xi_n',
    &
    (\db{c'}\sigma)|_{\dom(\xi_n'')}
    &= \sigma|_{\dom(\xi_n'')} = \xi_n''.
    &\text{[By \cref{lem:cnt-increase-and-rv-nochange}]}
  \end{align*}
  By these, $\used(c', {}  {} \sigma_p \oplus {\xi_n'} \oplus {\sigma_r'}, {\xi_n'})$,
  $\used(c'', {}  {} \sigma_p' \oplus {\xi_n''} \oplus {\sigma_r''}, {\xi_n''})$, and \cref{lem:used-sigma-indep}, we get
  \begin{align*}
    &\used(c', \sigma, \xi_n'),
    &
    &\usedm(c'', \db{c'}\sigma, \xi_n'').
  \end{align*}
  By these and the third claim in the proof of \cref{lem:integral-seq-decompose}, we get
  \begin{align*}
    \used(c';c'', \sigma, \xi_n' \oplus \xi_n'').
  \end{align*}
  By this and \cref{lem:used-sigma-indep}, we get the following as desired,
  since $\xi_n = \xi_n' \oplus \xi_n''$ (shown in the above)
  and $\sigma = \sigma_p \oplus (\xi_n' \oplus \xi_n'') \oplus \sigma_r$ for some $\sigma_r$:
  \begin{align*}
    \used(c, \sigma_p \oplus \xi_n \oplus \sigma_r, \xi_n).
  \end{align*}

  Next, we show the second conclusion as follows:
  \begin{align*}
    \getpvarsub(\ctr{c}{\pi})(\sigma_p, {\xi_n})
    &= ( \db{\ctr{c}{\pi}} \overline{\sigma} )|_\PVar
    &\text{[By $\used(\ctr{c}{\pi}, \overline{\sigma}, \overline{\xi_n})$]}
    \\
    &= \smash{ \big(\db{\ctr{c''}{\pi}} (\db{\ctr{c'}{\pi}} \overline{\sigma} )\big)|_\PVar }
    \\
    &= \smash{ \overline{\sigma''}|_\PVar = \sigma_p'', }
    \\
    \getpvarsub(c)(\sigma_p, \xi_n)
    &= \getpvarsub(c'; c'')(\sigma_p, \xi_n' \oplus \xi_n'')
    &\text{[By $\xi_n = \xi_n' \oplus \xi_n''$]}
    \\
    &= \getpvarsub(c'')(\getpvarsub(c')(\sigma_p, \xi_n'), \xi_n'')
    \\
    &= \getpvarsub(c'')(\sigma_p', \xi_n'')
    = \sigma_p'',
    &\text{[By the above]}
  \end{align*}
  where the second last equality uses
  $\used(c', \sigma, {\xi_n'})$,
  $\used(c';c'', \sigma, \xi_n' \oplus \xi_n'')$,
  and the second claim in the proof of \cref{lem:integral-seq-decompose}.

  Lastly, we show the third conclusion as follows:
  \begin{align*}
    \smash{ \getprsubi{S}(\ctr{c}{\pi})(\sigma_p, \overline{\xi_n}) }
    &= \smash{ \getprsubi{S}(\ctr{c';c''}{\pi})(\sigma_p, \overline{\xi_n'} \oplus \overline{\xi_n''}) }
    &\text{[By $\overline{\xi_n} = \overline{\xi_n'} \oplus \overline{\xi_n''}$]}
    \\
    &= \smash{ \getprsubi{S}(\ctr{c'}{\pi})(\sigma_p, \overline{\xi_n'}) 
    \cdot \getprsubi{S}(\ctr{c''}{\pi})(\getpvar(\ctr{c'}{\pi})(\sigma_p, \overline{\xi_n'}), \overline{\xi_n''}) }
    \\
    &= \smash{ \getprsubi{S}(\ctr{c'}{\pi})(\sigma_p, \overline{\xi_n'})
    \cdot \getprsubi{S}(\ctr{c''}{\pi})(\getpvar({c'}{})(\sigma_p, {\xi_n'}), \overline{\xi_n''}) }
    &\text{[By IH on $c'$]}
    \\
    &= \getprsubi{S}({c'}{})(\sigma_p, {\xi_n'})
    \cdot \getprsubi{S}({c''}{})(\getpvar({c'}{})(\sigma_p, {\xi_n'}), {\xi_n''})
    &\text{[By IH on $c'$ and $c''$]}
    \\
    &= \getprsubi{S}(c)(\sigma_p, \xi_n' \oplus \xi_n'')
    \\
    &= \getprsubi{S}(c)(\sigma_p, \xi_n).
    &\text{[By ${\xi_n} = {\xi_n'} \oplus {\xi_n''}$]}
  \end{align*}
  Here the second and fifth equalities use the second claim in the proof of \cref{lem:integral-seq-decompose} with the following:
  $\used(\ctr{c'}{\pi}, \overline{\sigma}, \overline{\xi_n'})$,
  $\used(\ctr{c'}{\pi}; \ctr{c''}{\pi}, \overline{\sigma}, \overline{\xi_n'} \oplus \overline{\xi_n''})$,
  $\used(c', \sigma, {\xi_n'})$, and
  $\used(c';c'', \sigma, \xi_n' \oplus \xi_n'')$.
  
  \paragraph{\bf Case $c \equiv (\cif\ b\ \{c'\}\ \celse\ \{c''\})$} 
  In this case, $\ctr{c}{\pi} \equiv (\cif\ b\ \{\ctr{c'}{\pi}\}\ \celse\ \{\ctr{c''}{\pi}\})$.
  Since $\fv(b) \subseteq \PVar$, $\db{b}(\sigma_p \oplus \sigma_r)$ is constant for all $\sigma_r \in \State[\Var \setminus \PVar]$.
  Without loss of generality, assume $\db{b}(\sigma_p \oplus \sigma_r) = \strue$.
  Then, $\db{c}(\sigma_p \oplus \sigma_r) = \db{c'}(\sigma_p \oplus \sigma_r)$
  and $\db{\ctr{c}{\pi}}(\sigma_p \oplus \sigma_r) = \db{\ctr{c'}{\pi}}(\sigma_p \oplus \sigma_r)$
  for all $\sigma_r \in \State[\Var \setminus \PVar]$.
  Hence, by IH on $(c', \sigma_p, \overline{\xi_n})$, we get the three conclusions directly.

  \paragraph{\bf Case $c \equiv (\cwhile\ b\ \{c'\})$}
  In this case, $\ctr{c}{\pi} \equiv (\cwhile\ b\ \{ \ctr{c'}{\pi} \})$.
  Consider the version of $\getprsub(-)$ where the parameter can be a state transformer $\overline{f} : \State \to \State_\bot$, instead of a command.
  Similarly, consider the version of the three conclusions where we use two state transformers $\overline{f}, f : \State \to \State_\bot$,
  again instead of a command.
  We denote the versions by $\getprsub(\overline{f})$ and $\varphi(\overline{f}, f, \sigma_p, \overline{\xi_n})$.
  We write $\overline{f} \sim f$
  if $\getprsub(\overline{f})(\sigma_p, \overline{\xi_n}) > 0$ implies $\varphi(\overline{f}, f, \sigma_p, \overline{\xi_n})$
  for all $\sigma_p \in \State[\PVar]$ and $\overline{\xi_n} \in \Statesub[\Name]$.
  Further, we define $F^{\pi'} : [\State \to \State_\bot] \to [\State \to \State_\bot]$ as
  \begin{align*}
    F^{\pi'}(f)(\sigma)
    &\defeq { \text{if}\ (\db{b}\sigma = \strue)\ \text{then}\ (f{}^\dagger \circ \db{\ctr{c'}{\pi'}})(\sigma)\ \text{else}\ \sigma. }
  \end{align*}
  Note that $F^\pi$ and $F^{\pi_0}$ are the operators used in the semantics of the loops $\smash{ \db{\ctr{c}{\pi}} }$ and $\db{c}$, respectively,
  where $\pi_0$ denotes the empty reparameterisation plan.
  We will show three claims: $\lambda \sigma. \bot \sim \lambda \sigma. \bot$;
  if $\overline{f} \sim f$, then $F^\pi(\overline{f}) \sim F^{\pi_0}(f)$;
  and if increasing sequences $\{\overline{f_k}\}_{k \in \N}$ and $\{f_k\}_{k \in \N}$ satisfy
  $\overline{f_k} \sim f_k$ for all $k \in \N$, then $\overline{f_\infty} \sim f_\infty$ holds
  for $\overline{f_\infty} = \bigsqcup_{k \in \N} \overline{f_k}$ and $f_\infty = \bigsqcup_{k \in \N} f_k$.
  These three claims imply $\db{\ctr{c}{\pi}} \sim \db{c}$, which in turn proves the desired three conclusions.
  
  The first claim holds simply because $\getprsub(\lambda \sigma.\bot)(-,-)$ is always 0.
  To show the second claim, consider $\overline{f}, f : \State \to \State_\bot$ such that $\overline{f} \sim f$.
  We first replay our proof for the sequential-composition case on $(\overline{f}, f)$
  after viewing $\overline{f}{}^\dagger \circ \db{\ctr{c'}{\pi}}$ and ${{f}{}^\dagger \circ \db{c'}}$
  as the sequential composition of $\ctr{c'}{\pi}$ and ${ \overline{f} }$, and of $c'$ and $f$, respectively.
  This replay, then, gives the relationship ${ \overline{f}{}^\dagger \circ \db{\ctr{c'}{\pi}} \sim f{}^\dagger \circ \db{c'} }$.
  Next, we replay our proof for the if case on $\smash{ (F^\pi(\overline{f}), F^{\pi_0}({f})) }$,
  after viewing ${ \overline{f}{}^\dagger \circ \db{\ctr{c'}{\pi}} }$ and $f{}^\dagger \circ \db{c'}$ as the true branches,
  and $\lambda \sigma.\,\sigma = \db{\cskip}$ as the false branch.
  This replay implies the required relationship $F^{\pi}(\overline{f}) \sim {F}^{\pi_0}({f})$.
  
  To show the third condition, consider increasing sequences $\{\overline{f_k}\}_{k \in \N}$ and $\{f_k\}_{k \in \N}$
  such that $\overline{f_k} \sim f_k$ for all $k \in \N$.
  Let $\overline{f_\infty} = \bigsqcup_{k \in \N} \overline{f_k}$ and $f_\infty = \bigsqcup_{k \in \N} f_k$.
  Consider any $\sigma_p$ and $\overline{\xi_n}$ such that $\getprsub(\overline{f_\infty})(\sigma_p, \overline{\xi_n}) > 0$.
  We should show $\varphi(\overline{f_\infty}, f_\infty, \sigma_p, \overline{\xi_n})$.
  Pick any $\sigma_r \in \State[\Var \setminus (\PVar \cup \dom(\overline{\xi_n}))]$ with $\sigma_r(\like)=1$.
  Let $\overline{\sigma} = \sigma_p \oplus \overline{\xi_n} \oplus \sigma_r$
  and $\sigma = \sigma_p \oplus \xi_n \oplus \sigma_r$.
  Note that the value of each term in $\varphi(\cdots)$ (i.e., $\used(\cdots)$, $\getpvarsub(\cdots)$, and $\getprsubi{S}(\cdots)$)
  is independent of the choice of $\sigma_r$ by \cref{lem:used-sigma-indep} and the well-definedness of $\getpvarsub$ and $\getprsubi{S}$.
  Since the two given sequences are increasing, there exists $K \in \N$ such that
  $\overline{f_\infty}(\overline{\sigma}) = \overline{f_K}(\overline{\sigma})$ and
  ${f_\infty}(\sigma) = {f_K}(\sigma)$.
  From this and $\getprsub(\overline{f_\infty})(\sigma_p, \overline{\xi_n}) > 0$,
  we have $\getprsub(\overline{f_K})(\sigma_p, \overline{\xi_n}) > 0$.
  This in turn gives  $\varphi(\overline{f_K}, f_K, \sigma_p, \overline{\xi_n})$ since $\overline{f_K} \sim f_K$.
  Lastly, again by  $\overline{f_\infty}(\overline{\sigma}) = \overline{f_K}(\overline{\sigma})$ and
  ${f_\infty}(\sigma) = {f_K}(\sigma)$, we obtain $\varphi(\overline{f_\infty}, f_\infty, \sigma_p, \overline{\xi_n})$ as desired.
\end{proof}
} 
\AtEndDocument{

\section{Deferred Results in \S\ref{sec:local-lips-req}}

\subsection{Deferred Statements and Their Proofs}

\begin{lemma}
  \label{lem:diff-rule-loclip}
  Let $f, g: \R^n \to \R$ be locally Lipschitz functions.
  Then, the following differentiation rules hold for almost every $x \in \R^n$:
  \begin{align*}
    \nabla (f+g)(x) &= \nabla f(x) + \nabla g(x),
    \\
    \nabla (f \cdot g)(x) &= \nabla f(x) \cdot g(x) + f(x) \cdot \nabla g(x),
    \\
    \nabla \log f(x) &= 1/f(x) \cdot \nabla f(x),
  \end{align*}
  where for the third rule we assume $f(y)>0$ for all $y \in \R^n$.
\end{lemma}
\begin{proof}
  Note that the three functions $+ : \R^2 \to \R$, $\cdot : \R^2 \to \R$, and $\log : \R_{>0} \to \R$ are all differentiable.
  Hence, applying \cref{lem:chain-rule-diff-loclip} produces the claim.
\end{proof}

\begin{theorem}
  \label{lem:suff-cond-int-diff-lip}
  Let $f : \R \times \R^n \to \R$ be a measurable function that satisfies the next four conditions:
  \begin{itemize}
  \item For all $x \in \R^n$, $f(-, x) : \R \to \R$ is continuous.
  \item For all $\theta \in \R$, $\nabla_\theta f(\theta, x)$ is well-defined for almost all $x \in \R^n$.
  \item For all $\theta \in \R$, $\smash{\int_{\R^n}} f(\theta, x) \,dx$ is finite.
  \item For all $\theta \in \R$, there is an open $U \subseteq \R$ such that $\theta \in U$ and $\int_{\R^n} \lip\big(f(-,x)|_U\big) \,dx$ is finite.
  \end{itemize}
  Here $\nabla_\theta(-)$ and $\lip(-)$ are defined as in \Cref{lem:suff-cond-int-diff},
  and ``almost all'' is with respect to the Lebesgue measure.
  Then, for all $\theta \in \R$, both sides of the following are well-defined and equal:
  \begin{align*}
    \nabla_\theta \int_{\R^n} f(\theta, x) \,dx
    &= \int_{\R^n} \nabla_\theta f(\theta, x) \,dx.
  \end{align*}
\end{theorem}
\rev{\begin{proof}
    Before starting the main proof, we show that for any open $U' \subseteq \R$,
    the following function $L : \R^n \to \R \cup \{\infty\}$ is measurable:
    \[ L(x) \defeq \lip\big( f(-,x)|_{U'} \big). \]
    Define $V \subseteq \R^2$ and $\ell : V \times \R^n \to \R$ as 
    \[ V \defeq \{v \in U' \times U' \mid v_1 \neq v_2\} \subseteq \R^2,
    \qquad \ell(v,x) \defeq \frac{|f(v_1,x) - f(v_2,x)|}{|v_1-v_2|}.\]
    Let $V'$ be a countable, dense subset of $V$.
    Then, for all $x \in \R^n$,
    \begin{align*}
      L(x) = \sup_{v \in V} \ell(v,x) = \sup_{v' \in V'} \ell(v',x),
    \end{align*}
    where the first equality is by the definition of $\lip(-)$,
    and the second equality holds since $\ell(-,x) : V \to \R$ is continuous for all $x \in \R^n$ by the first condition of this theorem.
    Since $L$ is the supremum of countably many measurable functions $\{ \ell(v',-) : \R^n \to \R \mid v' \in V'\}$,
    $L$ itself is a measurable function as desired.
    Note that the measurability of $L$ ensures that the integral $\int_{\R^n} L(x)\,dx$ in the fourth condition of this theorem
    is well-defined (as a value in $\R \cup \{\infty\}$).

    We now start the main proof.
    Pick any $\theta' \in \R$.
    Define $g : \R \setminus \{\theta'\} \to \R$ as
    \[ g(\theta) \defeq \int_{\R^n} \frac{f(\theta, x) - f(\theta', x)}{\theta-\theta'} \,dx, \]
    where $g(\theta)$ is finite by the third condition of this theorem.
    Then,
    \begin{align*}
      \Big( \nabla_\theta \int_{\R^n} f(\theta, x) \,dx \Big) \Big|_{\theta = \theta'}
      &=
      \lim_{\theta \to \theta'} \frac{1}{\theta-\theta'} \Big( \int_{\R^n} f(\theta, x) \,dx  - \int_{\R^n} f(\theta', x) \,dx \Big )
      =
      \lim_{\theta \to \theta'} g(\theta),
    \end{align*}
    where each equality denotes that LHS is well-defined if and only if RHS is well-defined, and if so, LHS and RHS are equal.
    Here the first equality is from the definition of $\nabla_\theta$,
    and the second equality from the third condition of this theorem.
    Note that the following are equivalent for any $r \in \R$:
    \begin{itemize}
    \item[(i)] $\lim_{\theta \to \theta'} g(\theta) = r$.
    \item[(ii)] For any $\{\theta_i\}_{i\in \N} \subseteq \R \setminus \{\theta'\}$,
      $\lim_{i \to \infty} \theta_i = \theta'$ implies $\lim_{i \to \infty} g(\theta_i) = r$.
    \end{itemize}
    So it suffices to show that (ii) holds with an appropriate choice of $r$.

    To do so, consider $\{\theta_i\}_{i\in \N} \subseteq \R \setminus \{\theta'\}$ such that $\lim_{i \to \infty} \theta_i = \theta'$.
    Define $h_i : \R^n \to \R$ as
    \[  h_i(x) \defeq \frac{f(\theta_i,x) - f(\theta',x)}{\theta_i - \theta'}. \]
    Then, $g(\theta_i) = \int_{\R^n} h_i(x) \,dx$ by the definition of $h_i$,
    and $\lim_{i \to \infty} h_i(x) = (\nabla_\theta f(\theta, x))|_{\theta=\theta'}$
    for almost all $x \in \R^n$ by the second condition of this theorem.
    So, if the dominated convergence theorem is applicable to $\lim_{i \to \infty} \int_{\R^n} h_i(x) \,dx$,
    we would have
    \begin{align*}
      \lim_{i \to \infty} g(\theta_i)
      = \lim_{i \to \infty}\int_{\R^n} h_i(x) \,dx
      &= \int_{\R^n} \Big(\nabla_\theta f(\theta, x)\Big)\Big|_{\theta=\theta'} \,dx,
    \end{align*}
    where each equality denotes that both LHS and RHS are well-defined and are equal.
    Therefore, it suffices to show that the preconditions of the dominated convergence theorem
    for $\lim_{i \to \infty} \int_{\R^n} h_i(x) \,dx$ are satisfied.
    
    We show that the preconditions indeed hold, which concludes the proof:
    \begin{itemize}
    \item ``$h_i$ is measurable for all $i \in \N$'': This holds by the measurability of $f$.
    \item ``$\lim_{i \to \infty} h_i(x) = (\nabla_\theta f(\theta, x))|_{\theta=\theta'}$ for almost all $x \in \R^n$'':
      This was already shown above.
    \item ``There exist $H : \R^n \to \R \cup \{\infty\}$ and $I \in \N$ such that
      (a) $\int_{\R^n} H(x) \,dx$ is finite, and (b) for all $i \geq I$, $|h_i(x)| \leq H(x)$ for almost all $x \in \R^n$'':
      By the fourth condition of this theorem, there is an open $U \subseteq \R$ such that $\theta' \in U$ and
      $\int_{\R^n} \lip( f(-,x)|_U ) \,dx$ is finite.
      Let $H : \R^n \to \R \cup \{\infty\}$ be \[ H(x) \defeq \lip\big( f(-,x)|_U \big). \]
      Then, (a) clearly holds since $\smash{\int_{\R^n}} \lip( f(-,x)|_U ) \,dx$ is finite.
      Further, (b) holds as follows:
      by $\lim_{i \to \infty} \theta_i = \theta'$ and $U$ being an open neighborhood of $\theta'$,
      there is $I \in \N$ such that $\theta_i \in U$ for all $i \geq I$;
      therefore, for all $i \geq I$,
      \begin{align*}
        |h_i(x)| &=
        \frac{|f(\theta_i,x) - f(\theta',x)|}{|\theta_i - \theta'|}
        \leq
        \lip\big( f(-,x)|_U \big) = H(x)
        \quad\text{for all $x \in \R^n$},
      \end{align*}
      where the inequality holds by the definition of $\lip(-)$ with $\theta_i, \theta' \in U$ and $\theta_i \neq \theta'$.
      \qedhere
    \end{itemize}
\end{proof}}

\rev{\begin{remark}
    The second condition $(*)$ of \Cref{lem:suff-cond-int-diff-lip} is weaker than
    the following, corresponding condition $(*')$ of a standard theorem
    for interchanging differentiation and integral (e.g., \cite[Corollary 2.8.7]{Bogachev07}):
    ``for almost all $x \in \R^n$, $\nabla_\theta f(\theta,x)$ is well-defined for all $\theta \in \R$.''
    The difference arises from whether a proof uses the mean value theorem or not:
    a proof of the standard theorem finds the function $H$ for the dominated convergence theorem,
    by applying the mean value theorem which requires the stronger condition $(*')$;
    the proof of \Cref{lem:suff-cond-int-diff-lip} finds $H$ not by applying the mean value theorem
    (but by using the fourth condition of the theorem),
    so the weaker condition $(*)$ is sufficient for the proof.
    In this sense, \Cref{lem:suff-cond-int-diff-lip} is close to \cite[Exercise 2.12.68]{Bogachev07}.
    \qed
\end{remark}}

\begin{lemma}
  \label{lem:chain-rule-diff-loclip}
  Let $f : X_1 \to X_2$ and $g : X_2 \to X_3$ for some open sets $X_i \subseteq \R^{n_i}$.
  Suppose that $f$ is locally Lipschitz and $g$ is differentiable.
  Then, $g \circ f : X_1 \to X_3$ is differentiable almost everywhere
  and the chain rule for $g \circ f$ holds almost everywhere, i.e.,
  \begin{align*}
    D(g \circ f)(x) = D(g)(f(x)) \cdot D(f)(x)
  \end{align*}
  for almost every $x \in X_1$.
  Here we use the Lebesgue measure as an underlying measure.
\end{lemma}
\begin{proof}
  Since local Lipschitzness is preserved under a function composition,
  $g \circ f$ is locally Lipschitz and thus differentiable almost everywhere.
  Since $f$ is also differentiable almost everywhere and $g$ is differentiable everywhere,
  the set
  \begin{align*}
    U = X_1 \setminus \{x \in X_1 \mid {}
    & (\text{$g \circ f$ is differentiable at $x$})
    \\
    & \land (\text{$g$ is differentiable at $f(x)$})
    \\
    & \land (\text{$f$ is differentiable at $x$}) \}
  \end{align*}
  has Lebesgue measure zero.
  Note that the differentiability of $g$ is importantly used here;
  if $g$ is non-differentiable even at a point, $U$ can have positive measure.
  The chain rule for $g \circ f$ holds for each $x \in U$ and this concludes the proof.
\end{proof}

\subsection{Proof of \cref{thm:unbiased-grad-lip}}

\begin{proof}[Proof of \cref{thm:unbiased-grad-lip}]
  The proof is essentially the same as the proof of \cref{thm:unbiased-grad},
  except that we invoke the following properties of local Lipschitzness (instead of differentiability):
  the composition of locally Lipschitz functions is again locally Lipschitz,
  and the differentiation rules for $+$, $\times$, and $log$ hold almost everywhere for locally Lipschitz functions (\cref{lem:diff-rule-loclip}).
\end{proof}
} 
\AtEndDocument{

\section{Deferred Results in \S\ref{:2:sa}}

\subsection{Proof of \cref{thm:abstract-semantics-well-formedness}}

\begin{proof}[Proof of \cref{thm:abstract-semantics-well-formedness}]
    We prove the theorem by induction on the structure of $c$.  Let $(p,d,V) \defeq \cdb{c}$, and pick $v \in \Var$. We have to show that $p(v) \supseteq d(v)^c$ and $d(v) \supseteq V$. We call these two requirements as conditions (i) and (ii).
    
    \paragraph{\bf Case $c \equiv \cskip$.} In this case, $p(v) = \Var$ and $V = \emptyset$, from which the conditions (i) and (ii) follow.
    
    \paragraph{\bf Case $c \equiv x:=e$.} In this case, $V = \emptyset$. So, the condition (ii) holds. For the proof of the condition (i), we do case analysis on $v$. If $v$ is the updated variable $x$, we have $p(v) = \edb{e}$ and $d(v) = \fv(e)$. Since $\edb{e} \supseteq \fv(e)^c$, the condition holds. If $v$ is different from $x$, then $p(v)$ is $\Var$, and so it includes $d(v)^c$.  
    
    \paragraph{\bf Case $c \equiv \cobs(\cnor(e_1,e_2),r)$.} The proof of this case is similar to the one for the assignments. Since $V = \emptyset$, the condition (ii) holds. If $v$ is the variable $\like$, then
    \begin{align*}
    p(v) = \edb{\like \times \cpdfnor(r;e_1,e_2)}
    \supseteq \fv(\like \times \cpdfnor(r;e_1,e_2))^c = \Big(\{\like\} \cup \fv(e_1) \cup \fv(e_2)\Big)^c = d(v)^c. 
    \end{align*}
    So, the condition (i) holds in this case. If $v$ is not the variable $\like$, then $p(v) = \Var$, from which the condition (i) follows.
      
    \paragraph{\bf Case $c \equiv x:=\csample(\cname(\alpha,e),\cnor(e_1,e_2),\lambda y.e')$.} In this case, $V = \emptyset$, from which the condition (ii) follows. We do case analysis on whether $e$ is a real constant $r$ or not. During the case analysis, we use the assumption that $\fv(e)^c \subseteq \edb{e}$ for all $e$, without mentioning it explicitly. 
    
    First, we deal with the case that $e \equiv r$. Let $\mu \defeq \mathit{create\_name}(\alpha,r)$. If $v$ is none of $x$, $\val_\mu$, $\pr_\mu$, and  $\sampled_\mu$, we have $p(v) = \Var$, which gives the condition (i). If $v \in \{x,\val_\mu\}$, we prove the condition (i) as follows:
    \begin{align*}
    d(v)^c = \fv(e'[\mu/y])^c \subseteq \edb{e'[\mu/y]} = p(v).
    \end{align*}
    If $v \equiv \pr_\mu$, we calculate the condition (i) as follows:
    \begin{align*}
    d(\pr_\mu)^c = (\{\mu\} \cup \fv(e_1) \cup \fv(e_2))^c = \fv(\cpdfnor(\mu;e_1,e_2))^c \subseteq \edb{\cpdfnor(\mu;e_1,e_2)} = p(\pr_\mu).
    \end{align*}
    If $v \equiv \sampled_\mu$, we derive the condition (i) as follows:
    \begin{align*}
    d(\sampled_\mu)^c = \{\sampled_\mu\}^c = \fv(\sampled_\mu + 1)^c \subseteq \edb{\sampled_\mu + 1} = p(\sampled_\mu).
    \end{align*}
    
    Next, we handle the case that $e$ is not a real constant. If $v$ is none of $x$, $\val_\mu$, $\pr_\mu$, and $\sampled_\mu$ for some $\mu = (\alpha,\_)$, we have $p(v) = \Var$, which implies the condition (i). If $v \equiv x$, we show the condition (i) as follows:
    \begin{align*}
    p(v) = \Big(\fv(e)^c \cap \bigcap_{\mu = (\alpha,\_) \in \Name} \edb{e'[\mu/y]}\Big) 
    \supseteq \Big(\fv(e) \cup \bigcup_{\mu = (\alpha,\_) \in \Name} \fv(e'[\mu/y])\Big)^c
    = d(v)^c.
    \end{align*}
    If $v \equiv \val_\mu$ for some $\mu = (\alpha,\_)$, we calculate the condition (i) as follows:
    \begin{align*}
    p(v) = \Big(\fv(e)^c \cap \edb{e'[\mu/y]}\Big) \supseteq \Big(\fv(e) \cup \{\val_\mu\} \cup \fv(e'[\mu/y])\Big)^c = d(v)^c.
    \end{align*}
    If $v \equiv \pr_\mu$ for some $\mu = (\alpha,\_)$, we prove the condition (i) as follows:
    \begin{align*}
    p(v) = \fv(e)^c \cap \edb{\cpdfnor(\mu;e_1,e_2)} 
    & {} \supseteq 
    \fv(e)^c \cap \fv(\cpdfnor(\mu;e_1,e_2))^c 
    \\
    & {}
    \supseteq
    \Big(\fv(e) \cup \{\mu,\pr_\mu\} \cup \fv(e_1) \cup \fv(e_2)\Big)^c = p(\pr_\mu).
    \end{align*}
    Finally, if $v \equiv \sampled_\mu$ for some $\mu = (\alpha,\_)$, we show the condition (i) as follows:
    \begin{align*}
    p(v) = \fv(e)^c \cap \edb{\sampled_\mu + 1} \supseteq \fv(e)^c \cap \fv(\sampled_\mu +1)^c = (\fv(e) \cup \{\sampled_\mu\})^c = d(v)^c.
    \end{align*}
        
    \paragraph{\bf Case $c \equiv c';c''$.} Let $(p',d',V') \defeq \cdb{c'}$ and $(p'',d'',V'') \defeq \cdb{c''}$. The condition (ii) holds since 
    \begin{align*}
    d(v) = V' \cup d'_\cup(d''(v)) \supseteq V' \cup d'_\cup(V'') = V.
    \end{align*} 
    For the condition (ii), we prove the required subset relationship as follows: 
    \begin{align*}
    d(v)^c 
    = (V' \cup d'_\cup(d''(v)))^c 
    & {} = (V')^c \cap \bigcap_{w \in d''(v)}d'(w)^c
    \\
    & {} \subseteq (V')^c \cap \bigcap_{w \in d''(v)}p'(w) \cap \bigcap_{w \in p''(v)^c} d'(w)^c 
    \\
    & {} = \Big(V' \cup p'_\cap(d''(v))^c \cup d'_\cup(p''(v)^c)\Big)^c = p(v). 
    \end{align*}
    The subset relationship in the above derivation holds because $d'(w)^c \subseteq p'(w)$ and $d''(v) \supseteq p''(v)^c$ by induction hypothesis.
    
    \paragraph{\bf Case $c \equiv \cif\ b\ \{c'\}\ \celse\ \{c''\}$.} Let $(p',d',V') \defeq \cdb{c'}$ and $(p'',d'',V'') \defeq \cdb{c''}$. Then, by 
    induction hypothesis,
    \begin{align*}
    V = \Big(\fv(b) \cup V' \cup V''\Big) \subseteq \Big(\fv(b) \cup d'(v) \cup d''(v)\Big) = d(v),
    \end{align*}
    which implies the condition (ii). Also, by induction hypothesis again,
    \begin{align*}
    d(v)^c = \Big(\fv(b)^c \cap d'(v)^c \cap d''(v)^c\Big) \subseteq \Big(\fv(b)^c \cap p'(v) \cap p''(v)\Big) = p(v),
    \end{align*}
    which shows the condition (ii).
    
    \paragraph{\bf Case $c \equiv \cwhile\ b\ \{c'\}$.} Let  $(p',d',V') \defeq \cdb{c'}$, and $F^\sharp$ be the operator in the abstract semantics of $c$. Note that the abstract domain $\aD$ contains $(p_\bot,d_\bot,V_\bot) = ((\lambda v.\Var),\,(\lambda v.\emptyset),\,\emptyset)$. Thus, it is sufficient to show that $F^\sharp$ is a well-defined monotone function on $\aD$, because then the least fixed point of $F^\sharp$ is also in $\aD$ and satisfies the conditions (i) and (ii). The monotonicity of $F^\sharp$ holds because when $(p_1,d_1,V_1) \defeq F^\sharp(p_0,d_0,V_0)$, the inputs $p_0$, $d_0$, and $V_0$ are used in the right polarity in the definitions of $p_1$, $d_1$, and $V_1$; for instance, $p_0$ is used only in the positive position (with respect to the subset order) when it is used to define $p_1$. To prove well-definedness of $F^\sharp$, assume that $p_0(v_0) \supseteq d_0(v_0)^c$ and $d_0(v_0) \supseteq V_0$ for all $v_0 \in \Var$, and pick a variable $v_1 \in \Var$. Then, since $V_0 \subseteq d_0(v_1)$, 
    \begin{align*}
    V_1 = \Big(\fv(b) \cup d'_\cup(V_0) \cup V'\Big)
    \subseteq \Big(\fv(b) \cup d'_\cup(d_0(v_1)) \cup V' \cup \{v_1\}\Big) = d_1(v_1).
    \end{align*}
    Also, by the induction hypothesis on the loop body $c'$ and the relationship $d_0(v_1) \supseteq p_0(v_1)^c$,
    \begin{align*}
    d_1(v_1)^c & = \fv(b)^c \cap (V')^c \cap \bigcap_{w \in d_0(v_1)} d'(w)^c  \cap\{v_1\}^c
    \\
    & {} \subseteq \fv(b)^c \cap (V')^c \cap \bigcap_{w \in d_0(v_1)} p'(w) \cap \bigcap_{w \in p_0(v_1)^c} d'(w)^c
    \\
    & {} =  \fv(b)^c \cap \Big(V' \cup p'_\cap(d_0(v_1))^c \cup d'_\cup(p_0(v_1)^c)\Big)^c = p_1(v_1).
    \end{align*}
    Thus, $(p_1,d_1,V_1)$ is also in $\aD$.
\end{proof}
}  
\AtEndDocument{


\subsection{Proof of \cref{thm:soundness-analysis}}

\label{:3:proof}

Our program analysis consists of two parts, one for tracking the dependency information and the other for tracking the smoothness information. The first part does not depend on the second, although it is used crucially by the second part. We exploit this one-way relationship between the two parts of our analysis, and prove the soundness of the dependency-tracking part first and then that of the other smoothness-tracking part. Consider a command $c$, and let $(p,d,V) \defeq \cdb{c}$. Then, we have:

\begin{theorem}
  \label{thm:soundness-delta}
  For all  $v \in \Var$, we have ${} \models \Delta(\db{c},d(v),\{v\})$. Also, ${} \models \Delta(\db{c},V,\emptyset)$.
\end{theorem}

\begin{theorem}
  \label{thm:soundness-phi}
  For all  $v \in \Var$, we have ${} \models \Phi(\db{c},p(v),\{v\})$.
\end{theorem}
\noindent
We prove the two soundness results in \cref{sec:pf:thm:soundness-delta} and \cref{sec:pf:thm:soundness-phi}.
From these, we immediately obtain the main soundness theorem:
\begin{proof}[Proof of \cref{thm:soundness-analysis}]
  Let $c$ be a command and $(p,d,V) = \cdb{c}$.
  Then, by \cref{thm:soundness-delta,thm:soundness-phi}, 
  we have ${} \models \Delta(\db{c},d(v),\{v\})$, ${} \models \Delta(\db{c},V,\emptyset)$, and ${} \models \Phi(\db{c},p(v),\{v\})$ for all $v \in \Var$.
  Hence, by the definition of $\gamma$ (i.e., \cref{eqn:definition-gamma}), we have $\db{c} \in \gamma(\cdb{c})$ as desired.
\end{proof}

\subsection{Proof of \cref{thm:soundness-delta}}
\label{sec:pf:thm:soundness-delta}

\begin{proof}[Proof of \cref{thm:soundness-delta}]
We prove the theorem by induction on the structure of $c$. Let $(p,d,V) \defeq \cdb{c}$. Pick a variable $v \in \Var$ and states $\sigma,\sigma',\sigma_0,\sigma_0' \in \State$ such that 
\[
\sigma \sim_{d(v)} \sigma'
\ \text{and}\ 
\sigma_0 \sim_V \sigma'_0.
\]
We will show that (i) if $\db{c}\sigma_0 \in \State$, then $\db{c}\sigma'_0 \in \State$, and (ii) if $\db{c}\sigma \in \State$ and $\db{c}\sigma' \in \State$, then $\db{c}\sigma \sim_{\{v\}} \db{c}\sigma'$, i.e., $\db{c}\sigma(v) = \db{c}\sigma'(v)$. Since $V \subseteq d(v)$, these two imply the claim of the theorem. We refer to these two properties as conditions (i) and (ii) in the rest of the proof.

\paragraph{\bf Case $c \equiv \cskip$.}
In this case, $d(v)=\{v\}$ and $V = \emptyset$. The condition (i) holds since $\cskip$ always terminates. The condition (ii) also holds because $\db{c}\sigma'' = \sigma''$ for all $\sigma''$, and the relation $\sim_{d(v)}$ coincides with $\sim_{\{v\}}$.
 
\paragraph{\bf Case $c \equiv (x:=e)$.} In this case,
$V = \emptyset$, and the condition (i) holds since the assignments always terminate. For the condition (ii), we do case analysis on the variable $v$.
  \begin{itemize}
  \item {Case $v \equiv x$.} In this case, $d(v)=\fv(e)$. This implies
    $\db{e}\sigma=\db{e}\sigma'$. Thus,
    $\db{c}\sigma(x) = \db{e}\sigma = \db{e}\sigma' = \db{c}\sigma'(x)$. This implies
    the desired $\db{c}\sigma \sim_{\{x\}} \db{c}\sigma'$.
  \item {Case $v \not\equiv x$.} In this case, $d(v)=\{v\}$, and so $\sigma(v)=\sigma'(v)$.
    This implies that $\db{c}\sigma(v) = \sigma(v) = \sigma'(v) = \db{c}\sigma'(v)$, 
    which gives the desired relationship.
  \end{itemize}

\paragraph{\bf Case $c \equiv \cobs(\cnor(e_1,e_2),r)$.}
The observe commands always terminate. Thus, the condition (i) holds. We prove the condition (ii) by  case  analysis on the variable $v$. If $v$ is not $\like$, then $d(v) = \{v\}$, $\db{c}\sigma(v) = \sigma(v)$, and $\db{c}\sigma'(v) = \sigma'(v)$. Thus, in this case, the assumption $\sigma \sim_{d(v)} \sigma'$ implies $\db{c}\sigma(v) = \db{c}\sigma'(v)$,  as desired. If $v$ is $\like$, then $d(v) = \fv(e_1) \cup \fv(e_2) \cup \{\like\}$, and for some function $g : \R^4 \to \R$,
\[
\db{c}\sigma(v) = g(\sigma(\like),r,\db{e_1}\sigma,\db{e_2}\sigma)
\ \text{and}\ 
\db{c}\sigma'(v) = g(\sigma'(\like),r,\db{e_1}\sigma',\db{e_2}\sigma').
\]
Therefore, from the assumption $\sigma \sim_{d(v)} \sigma'$, it follows that $\db{c}\sigma(v) = \db{c}\sigma'(v)$, as desired.

\paragraph{\bf Case $c \equiv (x := \csample(n,\cnor(e_1,e_2),\lambda y.e'))$.} The sample 
  commands always terminate. So, the condition (i)  holds. We prove the condition (ii)
  by case analysis on $n$. 
  
  The first case is that $n$ is a constant expression, i.e., it is an expression
  of the form $\cname(\alpha,r)$ for some $\alpha \in \Str$ and real number $r$.
  Let $\mu \defeq \mathit{create\_name}(\alpha,r)$. If $v$ is not one of $x$, $\val_\mu$, and $\pr_\mu$, then $d(v) = \{v\}$, and $\db{c}\sigma(v) = g(\sigma(v))$ and $\db{c}\sigma'(v) = g(\sigma'(v))$ for some function $g : \R \to \R$, so that the assumption $\sigma \sim_{d(v)} \sigma'$ implies that $\db{c}\sigma(v) = \db{c}\sigma'(v)$, as desired. If $v$ is $x$ or $\val_\mu$, then $d(v) = \fv(e'[\mu/y])$, $\db{c}\sigma(v) = \db{e'[\mu/y]}\sigma$, and $\db{c}\sigma'(v) = \db{e'[\mu/y]}\sigma'$, so that the required $\db{c}\sigma(v) = \db{c}\sigma'(v)$ holds. Finally,  if $v = \pr_\mu$, then $d(v) = \{\mu\} \cup \fv(e_1) \cup \fv(e_2)$, and so, the assumption $\sigma \sim_{d(v)} \sigma'$ implies that
  \[
  \db{c}\sigma(\pr_\mu) =
  \db{\cpdfnor(\mu; e_1, e_2)}\sigma = 
  \db{\cpdfnor(\mu; e_1,e_2)}\sigma' =
  \db{c}\sigma'(\pr_\mu),
  \]
  which is precisely the equality that we want.
  
  The next case is that $n$ is not a constant expression. Let $\cname(\alpha,e)$ be the form of $n$. If $v$ is not one of $x$, $\val_\mu$, $\pr_\mu$, and $\sampled_\mu$ for some $\mu$ of the form $(\alpha,\_)$, then $d(v) = \{v\}$, $\db{c}\sigma(v) = \sigma(v)$, and $\db{c}\sigma'(v) = \sigma'(v)$, so that the assumption $\sigma \sim_{d(v)} \sigma'$ implies the desired $\db{c}\sigma(v) = \db{c}\sigma'(v)$. Assume that $v$ is one of $x$, $\val_\mu$, $\pr_\mu$, and $\sampled_\mu$ for some $\mu$ with $\mu = (\alpha,\_)$. Let $\mu_0  \defeq \db{n}\sigma$ and $\mu'_0  \defeq \db{n}\sigma'$. Since $d(v) \supseteq \fv(n)$ in this case, the assumption $\sigma \sim_{d(v)} \sigma'$ ensures that $\mu_0 =  \mu'_0$. If $v$ is $x$, then $\fv(e'[\mu_0/y]) \subseteq d(v)$, so that the assumption $\sigma \sim_{d(v)} \sigma'$ gives  the desired
  \[
  \db{c}\sigma(v) = \db{e'[\mu_0/y]}\sigma = \db{e'[\mu'_0/y]}\sigma' = \db{c}\sigma'(v).
  \]
  If $v$ is $\sampled_\mu$ for some $\mu$ of the form $(\alpha,\_)$, then $\sampled_\mu \in  d(v)$, and $\db{c}\sigma(\sampled_\mu) = g(\sigma(\sampled_\mu))$ and $\db{c}\sigma'(\sampled_\mu) = g(\sigma'(\sampled_\mu))$ for some function $g : \R \to \R$, so that $\db{c}\sigma(v) = \db{c}\sigma'(v)$, as desired. If $v$ is $\val_\mu$ for $\mu = (\alpha,\_)$, then $d(v) \supseteq \{\val_\mu\} \cup \fv(e'[\mu/y])$, and $\db{c}\sigma(\val_\mu) = h(\db{e'[\mu/y]}\sigma,\sigma(\val_\mu))$ and $\db{c}\sigma'(\val_\mu) = h(\db{e'[\mu/y]}\sigma',\sigma'(\val_\mu))$ for some $h : \R \times \R \to \R$, so that $\db{c}\sigma(\val_\mu)  =  \db{c}\sigma'(\val_\mu)$ as desired. Finally,  if $v$ is $\pr_\mu$ for some  $\mu$ of the form $(\alpha,\_)$,  then $d(v) \supseteq \{\pr_\mu,\mu\} \cup \fv(e_1) \cup \fv(e_2)$, and for some $k : \R^4 \to \R$,
  \[
  \db{c}\sigma(v) = k(\sigma(\pr_\mu),\sigma(\mu),\db{e_1}\sigma,\db{e_2}\sigma)
  \ \text{and}\ 
  \db{c}\sigma'(v) = k(\sigma'(\pr_\mu),\sigma'(\mu),\db{e_1}\sigma',\db{e_2}\sigma'),
  \]
  so that the assumption $\sigma \sim_{d(v)} \sigma'$ guarantees that $\db{c}\sigma(v) = \db{c}\sigma'(v)$, as desired.

\paragraph{\bf Case $c \equiv (c';c'')$.}
  Let $(p',d',V') \defeq \cdb{c'}$ and $(p'',d'',V'') \defeq \cdb{c''}$. Recall that
  \[
  d(v)=V' \cup (d')_\cup(d''(v)) = V' \cup \bigcup \{ d'(v'') \mid v'' \in d''(v) \}
  \ \text{and}\ 
  V = V' \cup (d')_\cup(V'').
  \]

  Let us handle the condition (i) first. Since $\db{c';c''}\sigma_0 \in \State$, we have
  $\db{c'}\sigma_0 \in \State$. But $\sigma_0 \sim_{V'} \sigma'_0$, because $\sigma_0$ and $\sigma'_0$ are $\sim_V$-related and $V$ includes $V'$. Thus, $\db{c'}\sigma_0' \in \State$ as well by induction hypothesis, and it is sufficient to show $\db{c'}\sigma_0 \sim_{V''} \db{c'}\sigma_0'$. Note that for every $v'' \in V''$, by the definition of $V$, we have $V \supseteq d'(v'')$, and so $\sigma_0 \sim_{d'(v'')} \sigma'_0$, which implies, by induction hypothesis, that $\db{c'}\sigma_0 \sim_{\{v''\}} \db{c'}\sigma_0'$. As a result, we have the desired $\db{c'}\sigma_0 \sim_{V''} \db{c'}\sigma_0'$.
  
  Next, we deal with the condition (ii). Since $\db{c';c''}\sigma$ and $\db{c';c''}\sigma'$ are both in $\State$, there exist states $\sigma_1,\sigma_1'$ such that
  $\db{c'}\sigma = \sigma_1$ and $\db{c'}\sigma' = \sigma'_1$. We apply the induction hypothesis to $c'$ and get $\sigma_1 \sim_{d''(v)} \sigma_1'$. Since $\db{c''}\sigma_1$ and $\db{c''}\sigma'_1$ are in $\State$, we apply the induction hypothesis again but this time to $c''$, $\sigma_1$, and $\sigma_1'$, and obtain $\db{c''}\sigma_1 \sim_{\{v\}} \db{c''}\sigma'_1$, which implies the desired
  \[
  \db{c}\sigma(v) = \db{c''}\sigma_1(v) = \db{c''}\sigma'_1(v) = \db{c}\sigma'(v).
  \]
 
\paragraph{\bf Case $c \equiv (\cif\,b\,\{c'\}\,\celse\,\{c''\})$.}
Let $(p',d',V') \defeq \cdb{c'}$ and $(p'',d'',V'') \defeq \cdb{c''}$. Then,
$d(v) =  \fv(b) \cup d'(v) \cup d''(v)$ and
$V = \fv(b) \cup V' \cup V''$.

We prove the condition $(i)$ under the assumption that $\db{b}\sigma_0 = \strue$. Essentially the same proof applies to the other case that $\db{b}\sigma_0 = \sfalse$. Since $V$ includes $\fv(b)$, we also have $\db{b}\sigma'_0 = \strue$. Furthermore, since $V' \subseteq V$ and so $\sigma_0 \sim_{V'} \sigma'_0$ by the induction hypothesis, we get that $\db{c'}\sigma'_0 \in \State$.

Next we show the condition $(ii)$ under the assumption that $\db{b}\sigma = \strue$. As before, the proof of the other case $\db{b}\sigma = \sfalse$ is essentially the same. Since $d(v)$ includes $\fv(b)$ and $d'(v)$, we have $\db{b}\sigma' = \strue$ and $\sigma \sim_{d'(v)} \sigma'$. Also, because $\db{c}\sigma = \db{c'}\sigma$ and $\db{c}\sigma' = \db{c'}\sigma'$, both $\db{c'}\sigma$ and $\db{c'}\sigma'$ are in $\State$.
Thus, by induction hypothesis, $\db{c'}\sigma \sim_{\{v\}} \db{c'}\sigma'$, which implies that 
\[
\db{c}\sigma(v) = \db{c'}\sigma(v) = \db{c'}\sigma'(v) = \db{c}\sigma'(v),
\]
as desired.
 
\paragraph{\bf Case $c \equiv (\cwhile\,b\,\{c_0\})$.}
Let $(d_0,p_0,V_0) \defeq \cdb{c_0}$, and $F^\sharp$ be the operator in the abstract semantics of $c$ such that $(p,d,V)$ is the least fixed point of $F^\sharp$. Also, let $F$ be the operator in the concrete semantics of $c$ such that $\db{c}$ is the least fixed point of $F$. Now define
\[
T \defeq \{f \in [\State \to \State_\bot] \,\mid\,\text{for all $v \in \Var$}, {} \models \Delta(f,d(v),\{v\}) \ \text{and} {} \models \Delta(f,V,\emptyset)\}.
\]
We will show that (i) $T$ contains the empty function $\bot_{\State \to \State_\bot} \defeq \lambda \sigma.\,\text{undefined}$, (ii) it is closed under the least upper bounds of increasing chains, and (iii) the function $F$ maps functions in $T$ to some functions in the same set. These three imply that the least fixed point of $F$, namely, $\db{c}$, is in $T$, which gives the desired conclusion. 

The membership of $\bot_{\State \to \State_\bot}$ to $T$ is immediate, since we have
${} \models \Delta(\bot_{\State \to \State_\bot},U,U')$ for all $U,U' \subseteq \Var$. 

To prove the next requirement, namely, the closure under the least upper bounds of increasing chains, consider a chain $(f_n)_{n \in \N}$ in $T$, i.e., a sequence such that $f_n(\sigma) = f_{n+1}(\sigma)$ for all $n \in \N$ and $\sigma$ with $f_n(\sigma) \in \State$. Let $f_\infty$ be the least upper bound of the $f_n$'s (i.e., $f_\infty(\sigma) = f_n(\sigma)$ if $f_n(\sigma) \in \State$ and $f_\infty(\sigma) = \bot$ if $f_n(\sigma) = \bot$ for all $n \in \N$). As in all the other cases so far, we pick an arbitrary variable $v \in \Var$ and arbitrary states  $\sigma_0$, $\sigma_0'$, $\sigma$, and $\sigma'$ such that
\begin{align*}
\sigma_0 & \sim_V \sigma_0',
&
f_\infty(\sigma_0) & \in \State,
&
\sigma & \sim_{d(v)} \sigma',
&
f_\infty(\sigma),f_\infty(\sigma') & \in \State.
\end{align*}
We will show that $f_\infty(\sigma_0') \in \State$ and $f_\infty(\sigma) \sim_{\{v\}} f_\infty(\sigma')$, which correspond to what we have called conditions (i) and (ii) in the previous cases. Since $f_\infty(\sigma_0) \in \State$, there exists $n \in \N$ such that $f_n(\sigma_0) \in \State$. Because ${} \models \Delta(f_n,V,\emptyset)$ and $\sigma_0 \sim_V \sigma'_0$, we have $f_n(\sigma'_0) \in \State$, which implies that $f_\infty(\sigma'_0) = f_n(\sigma'_0) \in \State$, as desired. Our proof of the condition (ii) has a similar form. Since both $f_\infty(\sigma)$ and $f_\infty(\sigma')$ are in $\State$, there exists $n \in \N$ such that $f_\infty(\sigma) = f_n(\sigma)$ and $f_\infty(\sigma') = f_n(\sigma')$. By assumption, $\sigma \sim_{d(v)} \sigma'$, and $f_n \in T$. Thus, 
$f_n(\sigma) \sim_{\{v\}} f_n(\sigma')$, which gives the desired $f_\infty(\sigma) \sim_{\{v\}} f_\infty(\sigma')$.

It remains to show the last requirement, i.e., the closure under $F$. Pick an arbitrary $f \in T$. Consider a variable $v \in \Var$ and states  $\sigma_0$, $\sigma_0'$, $\sigma$, and $\sigma'$ such that
\begin{align*}
\sigma_0 & \sim_V \sigma_0',
&
F(f)(\sigma_0) & \in \State,
&
\sigma & \sim_{d(v)} \sigma',
&
F(f)(\sigma),F(f)(\sigma') & \in \State.
\end{align*}
We will show that $F(f)(\sigma_0') \in \State$ and $F(f)(\sigma) \sim_{\{v\}} F(f)(\sigma')$, while referring to these two desired properties as conditions (i) and (ii), as we have done before. 

Let us handle the condition (i) first.  If $\db{b}\sigma_0 = \sfalse$, we have $\db{b}\sigma_0' = \sfalse$, because $\sigma_0 \sim_V \sigma_0'$ and $\fv(b) \subseteq V$. Thus, in this case, $F(f)(\sigma_0') = \sigma'_0 \in \State$. If $\db{b}\sigma_0 = \strue$, then $\db{b}\sigma_0'$ is also $\strue$. Furthermore, in this case, by induction hypothesis, $\db{c_0}\sigma_0' \in \State$ since $V \supseteq V_0$, $\sigma_0 \sim_V \sigma_0'$,
and $\db{c_0}\sigma_0 \in \State$. Also, by induction hypothesis again, 
$\db{c_0}\sigma_0 \sim_V \db{c_0}\sigma'_0$, since 
 $V \supseteq  (d_0)_\cup(V)$ and $\sigma_0 \sim_V \sigma_0'$. Since $f \in T$ and $f(\db{c_0}\sigma_0) \in  \State$, we have $f(\db{c_0}\sigma'_0) \in \State$, which implies that $F(f)(\sigma'_0) \in \State$, as desired.

Next, we prove the condition (ii). If $\db{b}\sigma = \sfalse$, we have $\db{b}\sigma' = \sfalse$ since $\fv(b) \subseteq d(v)$ and $\sigma \sim_{d(v)} \sigma'$. Thus, in this case,  $F(f)(\sigma) = \sigma$ and $F(f)(\sigma') = \sigma'$. Also, $\{v\} \subseteq d(v)$, and so, $\sigma \sim_{d(v)} \sigma'$ implies that $F(f)(\sigma) = \sigma \sim_{\{v\}} \sigma' = F(f)(\sigma')$, as desired. Now assume that $\db{b}\sigma = \strue$. Then, $\db{b}\sigma' = \strue$ by the reason that $\fv(b) \subseteq d(v)$ and $\sigma \sim_{d(v)} \sigma'$. Also, $\db{c_0}\sigma$ and $\db{c_0}\sigma'$ are in $\State$, so that $F(f)(\sigma) = f(\db{c_0}\sigma)$ and $F(f)(\sigma') = f(\db{c_0}\sigma')$. Furthermore, since $d(v) \supseteq (d_0)_\cup(d(v))$ and $\sigma \sim_{d(v)} \sigma'$, we have
$\db{c_0}\sigma \sim_{d(v)} \db{c_0}\sigma'$. We then use the fact that $f \in T$ and $f(\db{c_0}\sigma), f(\db{c_0}\sigma') \in \State$, and conclude that $f(\db{c_0}\sigma) \sim_{\{v\}} f(\db{c_0}\sigma')$, which gives the desired $F(f)(\sigma) \sim_{\{v\}} F(f)(\sigma')$.
\end{proof}


\subsection{Proof of \cref{thm:soundness-phi}}
\label{sec:pf:thm:soundness-phi}

Let $\sseq$ be the following operator, which models sequential composition:
\begin{align*}
\sseq & : [\State \to \State_\bot] \times [\State \to \State_\bot] \to [\State \to \State_\bot]
\\
\sseq(f,g) & \defeq g^\dagger \circ f.
\end{align*}
Also, define an operator $\scond$ for modelling if commands as follows:
\begin{align*}
\scond & : [\State \to \B] \times [\State \to \State_\bot] \times [\State \to \State_\bot] \to [\State \to \State_\bot]
\\
\scond(h,f,g)(\sigma) & \defeq 
\begin{cases}
f(\sigma) & \text{if } h(\sigma)=\strue,
\\
g(\sigma) & \text{if } h(\sigma)=\sfalse.
\end{cases}
\end{align*}

\begin{proof}[Proof of \cref{thm:soundness-phi}]
  We prove the theorem by induction on the structure of $c$.
  Let ${ (p,d,V) \defeq \cdb{c} }$.   

  \paragraph{\bf Case $c \equiv \cskip$.}
  In this case, $\db{c}(\sigma) = \sigma$ for all $\sigma \in \State$, and $p(v) = \Var$ for all $v \in \Var$.
  To prove the conclusion, consider $v \in \Var$ and $\tau \in \State[p(v)^c] = \State[\emptyset]$.
  We should show $g \in \phi_{p(v), \{v\}}$, where
  \begin{align*}
    g(\sigma) &= 
    \begin{cases}
      (\pi_{\Var, \{v\}} \circ \db{c})(\sigma \oplus \tau)
      & \text{if $\db{c}(\sigma \oplus \tau) \in \State$}
      \\
      \text{undefined} & \text{otherwise}.
    \end{cases}
  \end{align*}
  Since $\db{c}(\sigma \oplus \tau) = \db{c}(\sigma) = \sigma \in \State$ for all $\sigma \in \State$,
  we have $g = \pi_{\Var,\{v\}}$.
  Thus, \cref{assm:projection} implies \[g = \pi_{\Var,\{v\}} \in \phi_{\Var, \{v\}} = \phi_{p(v), \{v\}}.\]

  \paragraph{\bf Case $c \equiv (x := e)$.}
  In this case, $\db{c}(\sigma) = \sigma[x \mapsto \db{e}\sigma]$ for all $\sigma \in \State$. Also,
  $p(v) = \Var$ if $v \not\equiv x$, and ${ \edb{e} }$ if $v \equiv x$.
  Consider $v \in \Var$ and $\tau \in \State[p(v)^c]$.
  We should show $g \in \phi_{p(v), \{v\}}$, where
  \begin{align*}
    g(\sigma) &= 
    \begin{cases}
      (\pi_{\Var, \{v\}} \circ \db{c})(\sigma \oplus \tau)
      & \text{if $\db{c}(\sigma \oplus \tau) \in \State$}
      \\
      \text{undefined} & \text{otherwise}.
    \end{cases}
  \end{align*}
  If $v \not\equiv x$, then
  $g(\sigma) = \pi_{\Var,\{v\}}( \db{c}(\sigma) ) = \pi_{\Var,\{v\}}( \sigma[x \mapsto \db{e}\sigma] ) = \pi_{\Var,\{v\}}( \sigma )$
  for all $\sigma \in \State$,
  where the first equality uses $p(v) = \Var$, and the last uses $v \not\equiv x$.
  Hence, \cref{assm:projection} implies \[g = \pi_{\Var,\{v\}} \in \phi_{\Var,\{v\}} = \phi_{p(v),\{v\}}.\]
  If $v \equiv x$, then
  $g(\sigma)
  = (\pi_{\Var, \{x\}} \circ \db{c})(\sigma \oplus \tau)
  = \pi_{\Var, \{x\}} ( (\sigma \oplus \tau)[x \mapsto \db{e}(\sigma \oplus \tau)] )
  = [x \mapsto \db{e}(\sigma \oplus \tau)]$
  for all $\sigma \in \State$.
  Since ${ \tau \in \State[(\edb{e})^c] }$ and  ${ p(v) = \edb{e} }$,
  \cref{assm:expression-analysis-soundness} implies
  \[ g = \lambda \sigma.\, [x \mapsto \db{e}(\sigma \oplus \tau)] \in \phi_{\edb{e}, \{x\}} = \phi_{p(v), \{v\}}. \]

  \paragraph{\bf Case $c \equiv (c';c'')$.}
  Let $(p',d',V') \defeq \cdb{c'}$ and $(p'',d'',V'') \defeq \cdb{c''}$. Then,
  \begin{align*}
    p(v) &= \Big( V' \cup (p')_\cap(d''(v))^c \cup (d')_\cup(p''(v)^c) \Big)^c
    \\
    &= \Big((p')_\cap(d''(v))\Big) \setminus \Big( V' \cup (d')_\cup(p''(v)^c) \Big)
     \; \text{ for all $v \in \Var$}.
  \end{align*}
  Also, we have $\db{c} = \sseq(\db{c'},\db{c''})$.
  To prove the conclusion, let $v \in \Var$.
  It suffices to apply \cref{lem:sound-seq-phi} to
  $f=\db{c'}$, $g=\db{c''}$, $K=p(v)$, $L=d''(v) \cap p''(v)$, $L'=d''(v)$, and $M=\{v\}$.
  What remains is to show the preconditions of the lemma:
  \begin{itemize}
  \item[(a)] ${}\models \Phi(\db{c'}, p(v), d''(v) \cap p''(v))$.
  \item[(b)] ${}\models \Phi(\db{c''}, d''(v) \cap p''(v), \{v\})$.
  \item[(c)] ${}\models \Delta(\db{c'}, p(v)^c, d''(v) \setminus (d''(v) \cap p''(v)))$.
  \item[(d)] ${}\models \Delta(\db{c''}, d''(v), \{v\})$.
  \end{itemize}
  We obtain (b) as follows: by induction hypothesis on $c''$, we have ${}\models \Phi(\db{c''}, p''(v), \{v\})$,
  and by the weakening lemma for $\Phi$ (\cref{lem:weakening-phi}), we have (b).
  We obtain (d) directly by \cref{thm:soundness-delta} on $c''$.
  For (a), consider induction hypothesis on $c'$, which says that ${}\models \Phi(\db{c'}, p'(w), \{w\})$ for all $w \in \Var$.
  By the merging lemma for $\Phi$ (\cref{lem:merging-phi}), we have ${}\models \Phi(\db{c'}, (p')_\cap(d''(v) \cap p''(v)), d''(v) \cap p''(v))$.
  Since \[p(v) \subseteq { (p')_\cap(d''(v)) } \subseteq { (p')_\cap(d''(v) \cap p''(v)) },\]
  we obtain (a) by the weakening lemma for $\Phi$ (\cref{lem:weakening-phi}).
  For (c), observe that
  \begin{align}
    \label{eq:seq-delta-sound}
    p(v)^c \supseteq V' \cup (d')_\cup(p''(v)^c) \; \text{ and } \;
    d''(v) \setminus (d''(v) \cap p''(v)) = p''(v)^c
  \end{align}
  where the second equality follows from $p''(v) \supseteq d''(v)^c$.
  By \cref{thm:soundness-delta} on $c'$, we have
  ${}\models \Delta(\db{c'}, V', \emptyset)$ and ${}\models \Delta(\db{c'}, d'(w), \{w\})$ for all $w \in \Var$.
  If $p''(v)^c = \emptyset$, then ${}\models \Delta(\db{c'}, V', p''(v)^c)$ holds,
  and if $p''(v)^c \neq \emptyset$, then ${}\models \Delta(\db{c'}, (d')_\cup(p''(v)^c), p''(v)^c)$ holds by the merging lemma for $\Delta$
  (\cref{lem:merging-delta}).
  By \cref{eq:seq-delta-sound} and the weakening lemma for $\Delta$ (\cref{lem:weakening-delta}), we obtain (c) for both cases.
  Note that we crucially used $p(v)^c \supseteq V'$ to handle the case $p''(v)^c = \emptyset$.

  \paragraph{\bf Case $c \equiv (\cif\,b\,\{c'\},\celse\,\{c''\})$.} 
  Let $(p',d',V') \defeq \cdb{c'}$ and $(p'',d'',V'') \defeq \cdb{c''}$. Then,
  \begin{align}
    \label{eq:cond-delta-sound}
    p(v) = \fv(b)^c \cap p'(v) \cap p''(v) \; \text{ for all $v \in \Var$}.
  \end{align}
  Also, $\db{c} = \scond(\db{b}, \db{c'}, \db{c''})$.
  To prove the conclusion, let $v \in \Var$.
  It suffices to apply \cref{lem:sound-cond-phi} to
  $f=\db{c'}$, $g=\db{c''}$, $K=p(v)$, $L=\{v\}$, and $b$.
  What remains is to show the preconditions of the lemma:
  \begin{itemize}
  \item[(a)] ${}\models \Phi(\db{c'}, p(v), \{v\})$.
  \item[(b)] ${}\models \Phi(\db{c''}, p(v), \{v\})$.
  \item[(c)] $p(v)^c \supseteq \fv(b)$.
  \end{itemize}
  We obtain (a) and (b) as follows:
  by induction hypothesis on $c'$ and $c''$, we have ${}\models \Phi(\db{c'}, p'(v), \{v\})$ and ${}\models \Phi(\db{c''}, p''(v), \{v\})$,
  and by \cref{eq:cond-delta-sound} and the weakening lemma for $\Phi$ (\cref{lem:weakening-phi}), we have (a) and (b).
  We obtain (c) directly by \cref{eq:cond-delta-sound}.
  
  \paragraph{\bf Case $c \equiv (\cwhile\,b\,\{c_0\})$.}
  The proof starts by decomposing $\db{c}$ and $\cdb{c}$ into smaller pieces.
  Let \[ (p_0, d_0, V_0) \defeq \cdb{c_0}. \]
  Define $F : [\State \to \State_\bot] \to [\State \to \State_\bot]$ and $F^\sharp : \aD \to \aD$ as in \cref{sec:setup} and \cref{f:asem}:
  \begin{align*}
    F(t)(\sigma) &\defeq
    \begin{cases}
      \sigma & \text{if $\db{b}\sigma = \sfalse$}
      \\
      t^\dagger(\db{c_0}\sigma) & \text{if $\db{b}\sigma = \strue$},
    \end{cases}
    \\
    F^\sharp(p,d,V) &\defeq \left(\hspace{-0.3em}
    \begin{array}{l}
      \lambda v.\, \fv(b)^c \cap (V_0 \cup (p_0)_\cap(d(v))^c \cup (d_0)_\cup(p(v)^c))^c,
      \\
      \lambda v.\, \fv(b) \cup V_0 \cup (d_0)_\cup(d(v)) \cup \{v\},
      \\
      \fv(b) \cup (d_0)_\cup(V) \cup V_0 
    \end{array}
    \hspace{-0.5em}\right).
  \end{align*}
  Define $t'_n \in [\State \to \State_\bot]$ and $(p'_n, d'_n, V'_n) \in \aD$ for $n \in \N \cup \{\infty\}$ as
  \begin{align*}
    t'_n &\defeq
    \begin{cases}
      \lambda \sigma.\,F^n(t_\bot)(\sigma)
      & \text{if $n \in \N$}
      \\
      \bigsqcup_{i \in \N} t'_i
      & \text{if $n = \infty$},
    \end{cases}
    &
    (p'_n, d'_n, V'_n) &\defeq
    \begin{cases}
      (F^\sharp)^n(p_\bot, d_\bot, V_\bot),
      & \text{if $n \in \N$}
      \\
      \bigsqcup_{i \in \N} (p'_i, d'_i, V'_i)
      & \text{if $n = \infty$},
    \end{cases}
  \end{align*}
  where $t_\bot = \lambda \sigma.\, \bot$ and $(p_\bot, d_\bot, V_\bot) = (\lambda v.\,\Var,\, \lambda v.\,\emptyset,\, \emptyset)$.
  Then, we have
  \begin{align*}
    \db{c} = t'_\infty \quad\text{and}\quad \cdb{c} = (p'_\infty, d'_\infty, V'_\infty).
  \end{align*}
  
  The proof is organized as follows.
  Define $T, T' \subseteq [\State \to \State_\bot]$ as
  \begin{align*}
    T  &\defeq \{ f \in [\State \to \State_\bot] \,\mid\, \forall v \in \Var.\, {}\models\Delta(f,d'_\infty(v),\{v\}) \land {}\models\Delta(f,V'_\infty,\emptyset)\},
    \\
    T' &\defeq \{ f \in [\State \to \State_\bot] \,\mid\, \forall v \in \Var.\, {}\models\Phi  (f,p'_\infty(v),\{v\}) \}.
  \end{align*}
  In \cref{thm:soundness-delta}, we proved
  \begin{align}
    \label{eq:loop-delta-sound}
    t'_n \in T \text{ for all $n \in \N \cup \{\infty\}$}.
  \end{align}
  In this theorem, our goal is to show $t'_\infty \in T'$.
  To do so, we prove the next three statements:
  \begin{itemize}
  \item[(a)] $t'_0 \in T'$.
  \item[(b)] If $t' \in T' \cap T$, then $F(t') \in T'$.
  \item[(c)] If $t'_n \in T'$ for all $n \in \N$, then $t'_\infty \in T'$.
  \end{itemize}
  It suffices to prove the three because (a), (b), and \cref{eq:loop-delta-sound} imply $t'_n \in T'$ for all $n \in \N$,
  and this and (c) imply $t'_\infty \in T'$.
  We now prove (a), (b), and (c) as follows.
  
  First, (a) follows directly from \cref{lem:sound-loop-base-phi}.
  
  Next, we prove (b). Consider $t' \in T' \cap T$.
  Our goal is to show ${}\models \Phi(F(t'), p'_\infty(v), \{v\})$ for all $v \in \Var$.
  Observe that \[ F(t') = \scond(\db{b}, \sseq(\db{c_0}, t'), \db{\cskip}). \]
  By \cref{thm:soundness-delta} and induction hypothesis on $c_0$, we have \[\db{c_0} \in \gamma(p_0, d_0, V_0),\]
  and by assumption, we have \[t' \in \gamma(p'_\infty, d'_\infty, V'_\infty) = T' \cap T.\]
  By applying to these the proofs of skip, sequential composition, and conditional cases, we have
  \[ {}\models \Phi(F(t'), p''(v), \{v\}) \; \text{ for all $v \in \Var$} \]
  where
  \[
  p''(v) = \fv(b)^c \cap \Big(V_0 \cup (p_0)_\cap(d'_\infty(v))^c \cup (d_0)_\cup(p'_\infty(v)^c)\Big)^c \cap \Var.
  \]
  Since $p''$ is the $p$ part of $F^\sharp(p'_\infty, d'_\infty, V'_\infty)$ and $(p'_\infty, d'_\infty, V'_\infty)$ is a fixed point of $F^\sharp$,
  we have $p''= p'_\infty$.
  Hence, we obtain
  ${}\models \Phi(F(t'), p'_\infty(v), \{v\})$ for all $v \in \Var$.
  This completes the proof of (b).
  
  Finally, we prove (c). Suppose that $t'_n \in T'$ for all $n \in \N$, and let $v \in \Var$.
  Our goal is to show 
  \[
  {}\models \Phi(t'_\infty, p'_\infty(v), \{v\}).
  \]
  Observe that \cref{lem:sound-loop-limit-phi} implies the goal
  when applied to $K = p'_\infty(v)$, $L=\{v\}$, and $\{f_n\}_{n \in \N} = \{t'_n\}_{n \in \N}$.
  Hence, it suffices to show the three preconditions of the lemma:
  \begin{itemize}
  \item $\{t'_n\}_{n \in \N}$ is an $\omega$-chain.
  \item For all $\tau \in \State[p'_\infty(v)^c]$ and $n \in \N$,
    the set $\{\sigma' \in \State[p'_\infty(v)] \mid t'_n(\sigma' \oplus \tau) \in \State\}$ is $\emptyset$ or $\State[p'_\infty(v)]$.
  \item For all $n \in \N$, we have ${} \models \Phi(t'_n, p'_\infty(v), \{v\})$.
  \end{itemize}
  The first precondition was already observed in \cref{sec:setup}.
  The third one holds by the assumption that $t'_n \in T'$ for all $n \in \N$.
  For the second one, it is enough to show the next two statements:
  \begin{itemize}
  \item[(i)] For all $U \subseteq \Var$ with $U \supseteq V'_\infty$, and for all $\tau \in \State[U]$ and $n \in \N$,
    the next set is $\emptyset$ or $\State[U^c]$:
    \[
    \{\sigma' \in \State[U^c] \mid t'_n(\sigma' \oplus \tau) \in \State\}.
    \]
  \item[(ii)] For all $v \in \Var$, \[p'_\infty(v)^c \supseteq V'_\infty.\]
  \end{itemize}
  We give the proof of the two statements below. This completes the proof of the while-loop case.
  
  \paragraph{Proof of (ii).}
  We prove a stronger statement: for all $n \in \N$ and $v \in \Var$, $p'_n(v)^c \supseteq V'_n.$
  This statement implies (ii) because
  $p'_\infty(v)^c = (\bigcap_{n \in \N} p'_n(v))^c = \bigcup_{n \in \N} p'_n(v)^c
  \supseteq \bigcup_{n \in \N} V'_n = V'_\infty.$
  We prove the statement by induction on $n$.
  For $n=0$, we have \[p'_n(v)^c = \Var^c = \emptyset \supseteq \emptyset = V'_n \text{ \; for all $v \in \Var$}.\]
  For $n>0$, let $v \in \Var$.
  By induction hypothesis, $p'_{n-1}(v)^c \supseteq V'_{n-1}$ holds.
  Using this, we have
  \begin{align*}
    p'_n(v)^c &= \fv(b) \cup V_0 \cup (p_0)_\cap(d'_{n-1}(v))^c \cup (d_0)_\cup(p'_{n-1}(v)^c)
    \\
    &\supseteq \fv(b) \cup V_0 \cup (d_0)_\cup(V'_{n-1}) = V'_n.
  \end{align*}
  This completes the proof of (ii).

  \paragraph{Proof of (i).}
  Consider $U \subseteq \Var$, $\tau \in \State[U]$, and $n \in \N$ such that $U \supseteq V'_\infty$.
  Let $\Sigma \defeq \{\sigma' \in \State[U^c] \mid t'_n(\sigma' \oplus \tau) \in \State\}$.
  If $\Sigma = \emptyset$, there is nothing left to prove.
  So assume $\Sigma \neq \emptyset$.
  To prove $\Sigma = \State[U^c]$, we need to show that $t'_n(\sigma' \oplus \tau) \in \State$ for any $\sigma' \in \State[U^c]$.
  Choose $\sigma' \in \State[U^c]$.
  We show $t'_n(\sigma' \oplus \tau) \in \State$ using the next two claims:
  \begin{itemize}
  \item[(iii)] For all $n \in \N$ and $\sigma \in \State$, 
    \begin{align}
      \label{eq:loop-unroll-sem}
      t'_n(\sigma) &=
      \begin{cases}
        \db{c_0^{(i)}}\sigma
        &\text{if } i \in I_n(\sigma)
        \\
        \bot &\text{otherwise},
      \end{cases}
      \\ \nonumber
      \text{where } c_0^{(i)} & \defeq (\cskip; c_0; \cdots; c_0) \text{ that has $i$ copies of $c_0$, and}
      \\ \nonumber
      I_n(\sigma)
      & {} \defeq \{i \in [0, n-1] \mid \db{c_0^{(i)}}\sigma \in \State
      \land \db{b}(\db{c_0^{(i)}}\sigma)= \sfalse
      \\ \nonumber
      & \phantom{{} \defeq \{i \in [0, n-1] \mid \db{c_0^{(i)}}\sigma \in \State}
         {} \land \db{b}(\db{c_0^{(i-1)}}\sigma) = \cdots = \db{b}(\db{c_0^{(0)}}\sigma) = \strue\}.
    \end{align}
    Note that \cref{eq:loop-unroll-sem} is well-defined since $I_n(\sigma)$ has at most one element.
        
  \item[(iv)] For all $n \in \N$,
    \[
      {}\models \Delta(\db{c_0^{(n)}}, V'_\infty, \fv(b)).
      \]
  \end{itemize}
  We give the proof of the two claims below, and for now we just assume them.
  
  Since $\Sigma \neq \emptyset$, there is some $\sigma'' \in \State[U^c]$ such that $t'_n(\sigma'' \oplus \tau) \in \State$.
  Since $t'_n(\sigma'' \oplus \tau) \in \State$, (iii) implies that
  \[{t'_n(\sigma'' \oplus \tau) = \db{c_0^{(m)}}(\sigma'' \oplus \tau) \in \State} \]
  for some $m \in I_n(\sigma'' \oplus \tau)$.
  Since $\sigma' \oplus \tau \sim_{V'_\infty} \sigma'' \oplus \tau$ (by $U \supseteq V'_\infty$)
  and $\smash{ \db{c_0^{(i)}}(\sigma'' \oplus \tau) \in \State }$ for all $i \in [0, m]$ (by $\smash{ \db{c_0^{(m)}} (\sigma'' \oplus \tau) \in \State }$),
  (iv) implies that
  \begin{align*}
    {\db{c_0^{(i)}}(\sigma' \oplus \tau) \in \State}
    \text{ \; and \; }
    {\db{c_0^{(i)}}(\sigma' \oplus \tau) \sim_{\fv(b)} \db{c_0^{(i)}}(\sigma'' \oplus \tau)}
    \text{ \; for all $i \in [0, m]$.}
  \end{align*}
  By combining these with $m \in I_n(\sigma'' \oplus \tau)$, we get $m \in I_n(\sigma' \oplus \tau)$.
  Hence, by (iii), we have \[{t'_n(\sigma' \oplus \tau) = \db{c_0^{(m)}}(\sigma' \oplus \tau) \in \State}.\]
  This completes the proof of (i).

  \paragraph{Proof of (iii).}
  We prove this by induction on $n$.
  For $n=0$, $t'_n(\sigma) = \bot$ and $I_n(\sigma) = \emptyset$ for all $\sigma \in \State$.
  Hence, \cref{eq:loop-unroll-sem} holds.
  For $n>0$, we have
  \begin{align*}
    t'_n(\sigma) &= F(t'_{n-1})(\sigma)
    \\ &=
    \begin{cases}
      \sigma & \text{if $\db{b}\sigma = \sfalse$}
      \\
      (t'_{n-1})^\dagger (\db{c_0}\sigma) & \text{if $\db{b}\sigma = \strue$}
    \end{cases}
    \\ &=
    \begin{cases}
      \sigma & \text{if $\db{b}\sigma = \sfalse$ \;$\cdots (*_1)$}
      \\
      \db{c_0^{(i)}}(\db{c_0}\sigma) &
      \text{if $\db{b}\sigma = \strue$, $\db{c_0}\sigma \in \State$ and $i \in I_{n-1}(\db{c_0}\sigma)$ \;$\cdots (*_2)$}
      \\
      \bot & \text{otherwise}
    \end{cases}
    \\ &=
    \begin{cases}
      \db{c_0^{(0)}}\sigma & \text{if $0 \in I_n(\sigma)$ \;$\cdots (*_1')$}
      \\
      \db{c_0^{(i+1)}}\sigma & \text{if $i+1 \in I_n(\sigma)$ and $i+1 \geq 1$ \;$\cdots (*_2')$}
      \\
      \bot & \text{otherwise}.
    \end{cases}
  \end{align*}
  The second equality is by the definition of $F$,
  the third by induction hypothesis,
  and the last by the following:
  for any $\sigma \in \State$ and $j\in\{1,2\}$, $(*_j)$ holds iff $(*_j')$ holds;
  and for any $i \in \N$, $\db{c_0}\sigma \in \State$ implies $\smash{ \db{c_0^{(i)}}(\db{c_0}\sigma) = \db{c_0^{(i+1)}}\sigma }$,
  and $\smash{ \db{c_0^{(i+1)}}\sigma \in \State }$ implies  $\db{c_0}\sigma \in \State$.
  Hence, \cref{eq:loop-unroll-sem} holds.
  This completes the proof of (iii).

  \paragraph{Proof of (iv).}
  To prove this, we prove a stronger statement: for all $n \in \N$, $\smash{ {}\models \Delta(\db{ c_0^{(n)}}, V'_{n+1}, \fv(b)) }.$
  This statement implies (iv) since $V'_{n} \subseteq V'_\infty$ for all $n \in \N$ and we have the weakening lemma for $\Delta$ (\cref{lem:weakening-delta}).
  We prove the statement by induction on $n$.
  For $n=0$, $\smash{ \db{c_0^{(n)}} = \db{\cskip} }$ and $V'_{n+1} = \fv(b) \cup V_0$.
  By \cref{thm:soundness-delta} on $\cskip$, we have ${}\models \Delta(\db{\cskip}, \{v\}, \{v\})$ for all $v \in \Var$,
  and then by the merging lemma for $\Delta$ (\cref{lem:merging-delta}), we have \[{}\models \Delta(\db{\cskip}, \fv(b), \fv(b)).\]
  Since $\smash{ V'_{n+1} \supseteq \fv(b) }$,
  $\smash{ {}\models \Delta(\db{c_0^{(n)}}, V'_{n+1}, \fv(b)) }$ holds by the weakening lemma for $\Delta$ (\cref{lem:weakening-delta}).
  Next, for $n>0$, $\smash{ \db{c_0^{(n)}} = \db{c_0; c_0^{(n-1)}} }$ and $\smash{ V'_{n+1} = \fv(b) \cup V_0 \cup (d_0)_\cup(V'_n) }$.
  By ${ \cdb{c_0}=(p_0, d_0, V_0) }$, \cref{thm:soundness-delta} on $c_0$, and induction hypothesis of the theorem (not that of the claim (iv)), we have
  \[ \db{c_0} \in \gamma(p_0, d_0, V_0).\]
  Also, by induction hypothesis of our strengthening of the claim (iv) and the weakening lemma for $\Delta$ (\cref{lem:weakening-delta}), 
  \[{}\models \Delta({ \db{c_0^{(n-1)}} }, V'_n, \{v\}) \text{ for all $v \in \fv(b)$}.\]
  By applying to these the proof of \cref{thm:soundness-delta} (on the sequential composition case), we have
  \[{}\models \Delta({ \db{c_0; c_0^{(n-1)}} }, V_0 \cup (d_0)_\cup(V'_n), \{v\}) \; \text{ for all $v \in \fv(b)$}.\]
  By the merging lemma for $\Delta$ (\cref{lem:merging-delta}), 
  $\smash{ {}\models \Delta({ \db{c_0^{(n)}} }, V_0 \cup (d_0)_\cup(V'_n), \fv(b)) }$ holds.
  Since $V'_{n+1}$ includes $V_0 \cup \smash{ (d_0)_\cup(V'_n) }$,
  we get $\smash{ {}\models \Delta(\db{c_0^{(n)}}, V'_{n+1}, \fv(b)) }$ by the weakening lemma for $\Delta$ (\cref{lem:weakening-delta}).
  This completes the proof of (iv).
  
  \paragraph{\bf Case $c \equiv (x:=\csample(\cname(\alpha, e),\cnor(e_1,e_2),\lambda y.e'))$.} 
  To prove the conclusion, consider $v \in \Var$ and $\tau \in \State[p(v)^c]$.
  We should show $g \in \phi_{p(v),\{v\}}$, where
  \[ g(\sigma) = \pi_{\Var, \{v\}} (\db{c}(\sigma \oplus \tau)) = [v \mapsto \db{c}(\sigma \oplus \tau)(v)]. \]
  We prove this by case analysis on $v$.

  First, suppose $v \not\in \{x\} \cup \{\val_\mu, \pr_\mu, \sampled_\mu \mid \mu \in \Name, \mu = (\alpha, \_) \}$.
  Then, $p(v) = \Var$ and
  \[ g(\sigma) = [v \mapsto \db{c}\sigma(v)] = [v \mapsto \sigma(v)] = \pi_{\Var, \{v\}}(\sigma) \]
  for all $\sigma \in \State[p(v)]$.
  Here the first equality follows from $\tau \in \State[\emptyset]$, and the second equality holds
  since $\db{c}$ does not change the value of $v$.
  Hence, by \cref{assm:projection}, $g = \pi_{\Var, \{v\}} \in \phi_{\Var, \{v\}} = \phi_{p(v), \{v\}}$.

  Next, suppose $v \in \{x\} \cup \{\val_\mu, \pr_\mu, \sampled_\mu \mid \mu \in \Name, \mu = (\alpha, \_)\}$.
  Then, we have $p(v)^c \supseteq \fv(e)$:
  if $e$ is a constant, $\fv(e) = \emptyset$ holds,
  and if $e$ is not a constant, the definition of $p(v)$ ensures this.
  Thus, there exists $\mu_0 \in \Name$ such that $\mathit{create\_name}(\alpha, \db{e}(\sigma \oplus \tau)) = \mu_0$ for all $\sigma \in \State[p(v)]$.
  We now do refined case analysis on $v$ using $\mu_0$.
  \begin{itemize}
  \item Case $v \in \{\val_\mu, \pr_\mu, \sampled_\mu \mid \mu \in \Name, \mu = (\alpha, \_), \mu \neq \mu_0\}$.
    In this case, 
    \[ g(\sigma) = [v \mapsto (\sigma \oplus \tau)(v)] = \pi_{\Var, \{v\}}(\sigma \oplus \tau) \]
    for all $\sigma \in \State[p(v)]$.
    Here the first equality holds since $\db{c}$ does not change the value of $v$.
    By \cref{assm:projection}, we have $\pi_{\Var, \{v\}} \in \phi_{\Var, \{v\}}$.
    Then, by \cref{assm:restriction}, we obtain $g \in \phi_{p(v), \{v\}}$.
    Note that this argument does not depend on the value of $p(v)$
    (which can be $\Var$, $\fv(e)^c \cap \edb{v+1}$, etc., depending on $e$ and $v$).
    
  \item Case $v \in \{x, \val_{\mu_0}\}$.
    Define $K_1 \defeq \edb{e'[\mu_0/y]}$. 
    Then, $p(v)^c \supseteq (\edb{e'[\mu_0/y]})^c = K_1^c$. So, there exist $\tau_1 \in \State[K_1^c]$ and $\tau_2 \in \State[p(v)^c \setminus K_1^c]$
    such that $\tau = \tau_1 \oplus \tau_2$.
    Let $h : \State[K_1] \to \State[\{v\}]$ be a function defined by
    \[
       h(\sigma') \defeq [v \mapsto \db{e'[\mu_0/y]}(\sigma' \oplus \tau_1)].
    \]
    Then,
    \begin{align*}
      g(\sigma)
      = \big[v \mapsto \db{e'[\mu_0/y]}(\sigma \oplus \tau)\big]
      = h(\sigma \oplus \tau_2)
    \end{align*}
    for all $\sigma \in \State[p(v)]$.
    By \cref{assm:expression-analysis-soundness}, we have $h \in \phi_{K_1, \{v\}}$,
    Then, by \cref{assm:restriction}, we obtain $g \in \phi_{p(v), \{v\}}$.

  \item Case $v \equiv \pr_{\mu_0}$. 
    In this case, $p(v) = \fv(e)^c \cap K$ for $K = \edb{\cpdfnor({\mu_0};e_1,e_2)}$.
    Since $\tau \in \State[p(v)^c]$ and $p(v)^c = (K \setminus \fv(e)^c) \uplus K^c$,
    there exist $\tau_1 \in \State[K \setminus \fv(e)^c]$ and $\tau_2 \in \State[K^c]$ such that $\tau = \tau_1 \oplus \tau_2$.
    Using $\tau_1$ and $\tau_2$, we have
    \begin{align*}
      g(\sigma)
      &= \big[v \mapsto \db{\cnor(e_1,e_2)}(\sigma \oplus \tau)\big((\sigma \oplus \tau)({\mu_0})\big)\big]
      \\&
      = \big[v \mapsto \db{\cpdfnor({\mu_0};e_1,e_2)}(\sigma \oplus \tau)\big]
      = h(\sigma \oplus \tau_1)
    \end{align*}
    for all $\sigma \in \State[p(v)]$,
    where $h : \State[K] \to \State[\{v\}]$ is defined by
    \[h(\sigma') = [v \mapsto \db{\cpdfnor({\mu_0};e_1,e_2)}(\sigma' \oplus \tau_2)].\]
    Here the first equality follows from the definition of $\db{c}$,
    the second equality holds because $\cpdfnor$ is the density function of a normal distribution,
    and the third equality comes from $\tau = \tau_1 \oplus \tau_2$.
    By \cref{assm:expression-analysis-soundness}, we have $h \in \phi_{K, \{v\}}$.
    Then, by \cref{assm:restriction}, we obtain $g \in \phi_{p(v), \{v\}}$.

  \item Case $v \equiv \sampled_{\mu_0}$.
    The proof is similar to the above case $v \equiv \pr_{\mu_0}$.
    In this case, $p(v) = \fv(e)^c \cap K$ for $K = \edb{\sampled_{\mu_0} + 1}$.
    As in the above case,
    there exist $\tau_1 \in \State[K \setminus \fv(e)^c]$ and $\tau_2 \in \State[K^c]$ such that $\tau = \tau_1 \oplus \tau_2$,
    and we have
    \begin{align*}
      g(\sigma)
      &= [v \mapsto (\sigma \oplus \tau)(\sampled_{\mu_0})+1]
      \\&
      = [v \mapsto \db{\sampled_{\mu_0}+1}(\sigma \oplus \tau)]
      = h(\sigma \oplus \tau_1)
    \end{align*}
    for all $\sigma \in \State[p(v)]$,
    where $h : \State[K] \to \State[\{v\}]$ is defined by
    \[h(\sigma') = [v \mapsto \db{\sampled_{\mu_0}+1}(\sigma' \oplus \tau_2)].\]
    By \cref{assm:expression-analysis-soundness}, we have $h \in \phi_{K, \{v\}}$.
    Then, by \cref{assm:restriction}, we obtain $g \in \phi_{p(v), \{v\}}$.
  \end{itemize}

  \paragraph{\bf Case $c \equiv \cobs(\cnor(e_1,e_2),r)$.}
  To prove the conclusion, consider $v \in \Var$ and $\tau \in \State[p(v)^c]$.
  We should show $g \in \phi_{p(v),\{v\}}$, where
  \[ g(\sigma) = \pi_{\Var, \{v\}} (\db{c}(\sigma \oplus \tau)) = [v \mapsto \db{c}(\sigma \oplus \tau)(v)]. \]
  We prove this by case analysis on $v$.

  First, suppose $v \not\equiv \like$.
  Then, $p(v) = \Var$ and
  \[ g(\sigma) = [v \mapsto \db{c}\sigma(v)] = [v \mapsto \sigma(v)] = \pi_{\Var, \{v\}}(\sigma) \]
  for all $\sigma \in \State[p(v)]$.
  Here the first equality is by $\tau \in \State[\emptyset]$, and the second equality holds
  since $\db{c}$ does not change the value of $v$.
  Hence, by \cref{assm:projection}, $g = \pi_{\Var, \{v\}} \in \phi_{\Var, \{v\}} = \phi_{p(v), \{v\}}$.

  Next, suppose $v \equiv \like$.
  Then, $p(v) = \edb{\like \times \cpdfnor(r;e_1,e_2)}$ and
  \begin{align*}
    g(\sigma)
    &= [v \mapsto (\sigma \oplus \tau)(\like) \cdot \db{\cnor(e_1,e_2)}(\sigma \oplus \tau)(r)]
    \\&
    = [v \mapsto \db{\like \times \cpdfnor(r;e_1;e_2)}(\sigma \oplus \tau)]
  \end{align*}
  for all $\sigma \in \State[p(v)]$.
  Here the first equality is by the definition of $\db{c}$,
  and the second equality holds because $\cpdfnor$ is the density function of a normal distribution.
  Hence, by \cref{assm:expression-analysis-soundness},
  we have $g \in \phi_{p(v),\{v\}}$.
\end{proof}


\subsection{Proofs of Lemmas for \cref{thm:soundness-phi}}

Here are the lemmas used to prove \cref{thm:soundness-phi}:

\begin{lemma}[Weakening; $\Delta$]
  \label{lem:weakening-delta}
  Let $f \in [\State \to \State_\bot]$ and $K,K',L,L' \subseteq \Var$. Then,
  \begin{align*}
    {}\models \Delta(f, K, L) \land (K \subseteq K') \land (L \supseteq L')
    &\implies {}\models \Delta(f, K', L').
  \end{align*}
\end{lemma}
\begin{proof}
Consider $\sigma,\sigma' \in \State$ with $\sigma \sim_{K'} \sigma'$. Then, $\sigma \sim_{K} \sigma'$ because $K \subseteq K'$. Since ${} \models \Delta(f,K,L)$, 
\[
(f(\sigma) \in \State \iff f(\sigma') \in \State)\ \text{and}\
(f(\sigma) \in \State \implies f(\sigma) \sim_L f(\sigma')).
\]
Note that the conclusion of the second conjunct implies $f(\sigma) \sim_{L'} f(\sigma')$ since $L' \subseteq L$. From what we have just shown, the desired conclusion ${} \models \Delta(f,K',L')$ follows.
\end{proof}

\begin{lemma}[Weakening; $\Phi$]
  \label{lem:weakening-phi}
  Let $f \in [\State \to \State_\bot]$ and $K,K',L,L' \subseteq \Var$. Then,
  \begin{align*}
    {}\models \Phi(f, K, L) \land (K \supseteq K') \land (L \supseteq L')
    &\implies {}\models \Phi(f, K', L').
  \end{align*}
\end{lemma}
\begin{proof}
  We prove the lemma using \cref{assm:restriction,assm:projection,assm:composition}.
  Consider $\tau \in \State[(K')^c]$, and let $g$ be the following partial function:
  \begin{align*}
  g & : \State[K'] \rightharpoonup \State[L'],
  &
  g(\sigma') \defeq
  \begin{cases}
  (\pi_{\Var,L'} \circ f)(\sigma' \oplus \tau) & \text{if}\ f(\sigma' \oplus \tau) \in \State
  \\
  \text{undefined} & \text{otherwise}.
  \end{cases}
  \end{align*}
  We should show $g \in \phi_{K',L'}$. Note that $(K')^c \supseteq K^c$. Thus, there exist $\tau_1 \in \State[(K')^c \setminus K^c]$
  and $\tau_2 \in \State[K^c]$ such that $\tau = \tau_1 \oplus \tau_2$. Define a partial function $h : \State[K] \to \State[L]$ by
  \[
  h(\sigma'') \defeq 
  \begin{cases}
  (\pi_{\Var,L} \circ f)(\sigma'' \oplus \tau_2) & \text{if}\ f(\sigma'' \oplus \tau_2) \in \State
  \\
  \text{undefined} & \text{otherwise}.
  \end{cases}
  \] 
  Then, since ${} \models \Phi(f,K,L)$, we have $h \in \phi_{K,L}$. Note that for all $\sigma' \in \State[K']$,
  \[
  g(\sigma') = (\pi_{L,L'} \circ h)(\sigma' \oplus \tau_1).
  \]
  By \cref{assm:restriction,assm:projection,assm:composition}, the above equation implies $g \in \phi_{K',L'}$, as desired.
\end{proof}

\begin{lemma}[Merging; $\Delta$]
  \label{lem:merging-delta}
  Let $f \in [\State \to \State_\bot]$ and $K,K',L,L' \subseteq \Var$. Then,
  \begin{align*}
    {}\models \Delta(f, K, L) \land {}\models \Delta(f, K', L') &\implies {}\models \Delta(f, K \cup K', L \cup L').
  \end{align*}
\end{lemma}
\begin{proof}
  Consider $\sigma,\sigma' \in \State$ with $\sigma \sim_{K \cup K'} \sigma'$. Then, $\sigma \sim_K \sigma'$, and by the assumption that ${} \models \Delta(f,K,L)$, we have
  \[
  f(\sigma) \in \State \iff f(\sigma') \in \State.
  \]
  It remains to show that if $f(\sigma),f(\sigma') \in \State$, then $f(\sigma) \sim_{L \cup L'} f(\sigma')$. Assume $f(\sigma),f(\sigma') \in \State$. 
  Since $\sigma \sim_{K \cup K'} \sigma'$ and we have ${} \models \Delta(f,K,L)$ and ${} \models \Delta(f,K',L')$ by assumption, 
  \[
  f(\sigma) \sim_L f(\sigma')\ \text{and}\ f(\sigma) \sim_{L'} f(\sigma').
  \]
  This implies that $f(\sigma) \sim_{L \cup L'} f(\sigma')$, as desired.
\end{proof}

\begin{lemma}[Merging; $\Phi$]
  \label{lem:merging-phi}
  Let $f \in [\State \to \State_\bot]$ and $K,K',L,L' \subseteq \Var$. Then,
  \begin{align*}
    {}\models \Phi(f, K, L) \land {}\models \Phi(f, K', L') &\implies {}\models \Phi(f, K \cap K', L \cup L').
  \end{align*}
\end{lemma}
\begin{proof}
  Uses the weakening lemma for $\Phi$ (\cref{lem:weakening-phi}), we have
  \[
  {}\models \Phi(f,K \cap K', L)\ \text{and}\ {}\models \Phi(f,K \cap K',L').
  \]
  This and \cref{assm:pairing} then imply the desired conclusion. Concretely, for all $\tau \in \State[(K \cap K')^c]$, if $g$, $g_1$, and $g_2$
  are the following partial functions
  \begin{align*}
  g & : \State[K \cap K'] \rightharpoonup \State[L \cup L'],
  &
  g(\sigma') & \defeq 
  \begin{cases}
  (\pi_{\Var, L \cup L'} \circ f)(\sigma' \oplus \tau) & \text{if}\ f(\sigma'\oplus \tau) \in \State
  \\
  \text{undefined} & \text{otherwise},
  \end{cases}
  \\
  g_1 & : \State[K \cap K'] \rightharpoonup \State[L],
  &
  g_1(\sigma') & \defeq 
  \begin{cases}
  (\pi_{\Var, L} \circ f)(\sigma' \oplus \tau) & \text{if}\ f(\sigma'\oplus \tau) \in \State
  \\
  \text{undefined} & \text{otherwise},
  \end{cases}
  \\
  g_2 & : \State[K \cap K'] \rightharpoonup \State[L'],
  &
  g_2(\sigma') & \defeq 
  \begin{cases}
  (\pi_{\Var, L'} \circ f)(\sigma' \oplus \tau) & \text{if}\ f(\sigma'\oplus \tau) \in \State
  \\
  \text{undefined} & \text{otherwise},
  \end{cases}
  \end{align*}
  then $g_1 \in \phi_{K\cap K',L}$, $g_2 \in \phi_{K \cap K',L'}$, and $g = \langle g_1, g_2 \rangle$, so that 
  by \cref{assm:pairing}, we have $g \in \phi_{K \cap K', L \cup L'}$ as desired.
\end{proof}

\begin{lemma}[Sequence]
  \label{lem:sound-seq-phi}
  Let $f,g \in [\State \to \State_\bot]$ and $K,L,L',M \subseteq \Var$. Then,
  \begin{gather*}
    {}\models \Phi(f, K, L)
    \land {}\models \Phi(g, L, M)
    \land {}\models \Delta(f, K^c, L' \setminus L)
    \land {}\models \Delta(g, L', M)
    \implies {}\models \Phi(\sseq(f,g), K, M).
  \end{gather*}
\end{lemma}
\begin{proof}
  Consider $f,g \in [\State \to \State_\bot]$ and $K,L,L',M \subseteq \Var$ that satisfy the given conditions:
  \[
  {}\models \Phi(f, K, L), 
  \quad
  {}\models \Phi(g, L, M),
  \quad
  {}\models \Delta(f, K^c, L' \setminus L),
  \quad 
  \text{and}
  \quad 
  {}\models \Delta(g, L', M). 
  \]
  To prove the conclusion, pick an arbitrary $\tau \in \State[K^c]$.
  We have to show $h \in \phi_{K,M}$, where
  \begin{align*}
    h(\sigma) &=
    \begin{cases}
      \pi_{\Var, M}((g^\dagger \circ f)(\sigma \oplus \tau)) & \text{if $(g^\dagger \circ f)(\sigma \oplus \tau) \in \State$}
      \\
      \text{undefined} & \text{otherwise}.
    \end{cases}
  \end{align*}
  Observe that since ${}\models \Phi(f, K, L)$ and ${}\models \Phi(g, L, M)$, we have $h_1 \in \phi_{K,L}$ and $h_2 \in \phi_{L,M}$
  for any $\tau_1 \in \State[K^c]$ and $\tau_2 \in \State[L^c]$,
  where $h_1$ and $h_2$ are parameterised by $\tau_1$ and $\tau_2$, and defined by
  \begin{align*}
    h_1 & : \State[K] \rightharpoonup \State[L],
    &
    h_1(\sigma) & \defeq
    \begin{cases}
      \pi_{\Var, L}(f(\sigma \oplus \tau_1)) & \text{if $f(\sigma \oplus \tau_1) \in \State$}
      \\
      \text{undefined} & \text{otherwise},
    \end{cases}
    \\
    h_2 & : \State[L] \rightharpoonup \State[M],
    &
    h_2(\sigma) & \defeq
    \begin{cases}
      \pi_{\Var, M}(g(\sigma \oplus \tau_2)) & \text{if $g(\sigma \oplus \tau_2) \in \State$}
      \\
      \text{undefined} & \text{otherwise}.
    \end{cases}
  \end{align*}
  Given these, it suffices to show the claim that $h = h_2 \circ h_1$ for some $\tau_1$ and $\tau_2$:
  if the claim holds, then we have $h = h_2 \circ h_1 \in \phi_{K,M}$ by \cref{assm:composition}, since $h_1 \in \phi_{K,L}$ and $h_2 \in \phi_{L,M}$.
  We prove the claim by case analysis on $f(- \oplus \tau)$.

  \paragraph{\bf Case $f(\sigma \oplus \tau) \not\in \State$ for all $\sigma \in \State[K]$.}
  In this case, we set $\tau_1 \defeq \tau$ and pick any $\tau_2 \in \State[L^c]$.
  Then, for all $\sigma \in \State[K]$, $h(\sigma)$ and $(h_2 \circ h_1)(\sigma)$ are both undefined, as desired.
  Note that the latter term is undefined since $h_1(\sigma)$ is undefined.

  \paragraph{\bf Case $f(\sigma' \oplus \tau) \in \State$ for some $\sigma' \in \State[K]$.}
  In this case, we set $\tau_1 \defeq \tau$ and $\tau_2 \defeq \pi_{\Var, L^c}(f(\sigma' \oplus \tau))$.
  To show $h = h_2 \circ h_1$, consider any $\sigma \in \State[K]$.
  If $f(\sigma \oplus \tau) \not\in \State$, then by the same argument for the above case,
  $h(\sigma)$ and $(h_2 \circ h_1)(\sigma)$ are both undefined.
  So, assume that $f(\sigma \oplus \tau) \in \State$.
  Then, 
  \begin{align}
    \label{eq:sound-seq-phi-1}
    h(\sigma) &=
    \begin{cases}
      \pi_{\Var, M}(g(\sigma_1)) & \text{if $g(\sigma_1) \in \State$}
      \\
      \text{undefined} & \text{otherwise},
    \end{cases}
    &
    (h_2 \circ h_1)(\sigma) &=
    \begin{cases}
      \pi_{\Var, M}(g(\sigma_2)) & \text{if $g(\sigma_2) \in \State$}
      \\
      \text{undefined} & \text{otherwise},
    \end{cases}
  \end{align}
  where
  \begin{align*}
    \sigma_1 &= f(\sigma \oplus \tau) \in \State,
    &
    \sigma_2 &= \pi_{\Var,L}(f(\sigma \oplus \tau)) \oplus \pi_{\Var,L^c}(f(\sigma' \oplus \tau)) \in \State.
  \end{align*}
  Our goal is to show that $h(\sigma)$ and $(h_2 \circ h_1)(\sigma)$ are both undefined, or they are both defined and are the same.
  By ${}\models \Delta(g, L', M)$ and \cref{eq:sound-seq-phi-1}, it suffices to show $\sigma_1 \sim_{L'} \sigma_2$.
  To prove this, we show a stronger statement: $\sigma_1 \sim_L \sigma_2$ and $\sigma_1 \sim_{L' \setminus L} \sigma_2$.
  The former relation holds since $\pi_{\Var, L}(\sigma_1) = \pi_{\Var, L}(f(\sigma \oplus \tau)) = \pi_{\Var,L}(\sigma_2)$.
  The latter relation is equivalent to $f(\sigma \oplus \tau) \sim_{L' \setminus L} f(\sigma' \oplus \tau)$,
  and this holds by ${}\models \Delta(f, K^c, L' \setminus L)$ and $\sigma \oplus \tau \sim_{K^c} \sigma' \oplus \tau$.
  Hence, $h=h_2 \circ h_1$ as desired.
\end{proof}

\begin{lemma}[Conditional]
  \label{lem:sound-cond-phi}
  Let $f,f' \in [\State \to \State_\bot]$ and $K,L \subseteq \Var$. Then, for any boolean expression~$b$,
  \begin{align*}
    {}\models \Phi(f, K, L) \land {}\models \Phi(f', K, L) \land (K^c \supseteq \fv(b))
    &\implies {}\models \Phi(\scond(\db{b}, f, f'), K, L).
  \end{align*}
\end{lemma}
\begin{proof}
  Let $f$, $f'$, $K$, $L$, and $b$ be the functions, sets and a boolean expression such that
  \[
      {}\models \Phi(f, K, L),
      \quad
      {}\models \Phi(f', K, L),
      \quad
      \text{and}
      \quad
      K^c \supseteq \fv(b).
  \]
  Consider $\tau \in \State[K^c]$. Define $f'' \defeq \scond(\db{b},f,f')$, and also partial functions $g$, $g'$, and $g''$ as follows:
  \begin{align*}
  g & : \State[K] \rightharpoonup \State[L],
  &
  g(\sigma) & \defeq 
  \begin{cases}
  (\pi_{\Var,L} \circ f)(\sigma \oplus \tau) & \text{if}\ f(\sigma \oplus \tau) \in \State
  \\
  \text{undefined} & \text{otherwise},
  \end{cases}
  \\
  g' & : \State[K] \rightharpoonup \State[L],
  &
  g'(\sigma) & \defeq 
  \begin{cases}
  (\pi_{\Var,L} \circ f')(\sigma \oplus \tau) & \text{if}\ f'(\sigma \oplus \tau) \in \State
  \\
  \text{undefined} & \text{otherwise},
  \end{cases}
  \\
  g'' & : \State[K] \rightharpoonup \State[L],
  &
  g''(\sigma) & \defeq 
  \begin{cases}
  (\pi_{\Var,L} \circ f)(\sigma \oplus \tau) & \text{if}\ f(\sigma \oplus \tau) \in \State
  \\
  \text{undefined} & \text{otherwise},
  \end{cases}  
  \end{align*}
  We should show $g'' \in \phi_{K,L}$. Since $K^c \supseteq \fv(b)$, either 
  $\db{b}(\sigma \oplus \tau) = \strue$ for all $\sigma \in \State[K]$ or 
  $\db{b}(\sigma \oplus \tau) = \sfalse$ for all $\sigma \in \State[K]$. In the former case,
  $g'' = g$, and in the latter case, $g'' = g'$. Since both $g$ and $g'$ are in $\phi_{K,L}$,
  we have the desired $g'' \in \phi_{K,L}$ in both cases.
\end{proof}

\begin{lemma}[Loop; base]
  \label{lem:sound-loop-base-phi}
  Let $K,L \subseteq \Var$. Then,
  \begin{align*}
    {}\models \Phi((\lambda \sigma \in \State.\,\bot), K, L).
  \end{align*}
\end{lemma}
\begin{proof}
  Consider $\tau \in \State[K^c]$. Define a partial function $g : \State[K] \rightharpoonup \State[L]$ by
  \begin{align*}
  g(\sigma) & \defeq 
  \begin{cases}
  (\pi_{\Var,L} \circ (\lambda \sigma \in \State.\bot))(\sigma) & \text{if}\  (\lambda \sigma \in \State.\bot)(\sigma) \in \State
  \\
  \text{undefined} & \text{otherwise}
  \end{cases}
  \\
  & =
  \text{undefined}.
  \end{align*}
  Then, $g \in \phi_{K,L}$ by \cref{assm:strictness}.
\end{proof}

\begin{lemma}[Loop; limit]
  \label{lem:sound-loop-limit-phi}
  Let $K,L \subseteq \Var$ and $\{f_n \in [\State \to \State_\bot]\}_{n \in \N}$ be an $\omega$-chain
  (i.e., $f_n \sqsubseteq f_{n+1}$ for all $n \in \N$).
  Here we write $f \sqsubseteq g$ if $f(\sigma) \sqsubseteq g(\sigma)$ for all $\sigma \in \State$.
  Suppose that for any $\tau \in \State[K^c]$ and $n \in \N$, the set
  \(
  \{ \sigma \in \State[K] \mid f_n(\sigma \oplus \tau) \in \State \}
  \)
  is either $\emptyset$ or $\State[K]$.
  Then,
  \begin{align*}
    \bigwedge_{n \in \N} {}\models \Phi(f_n, K, L)
    &\implies {}\models \Phi(\bigsqcup_{n \in \N} f_n, K, L).
  \end{align*}
\end{lemma}
\begin{proof}
  Consider an $\omega$-chain $\{f_n \in [\State \to \State_\bot]\}_{n \in \N}$ such that ${} \models \Phi(f_n,K,L)$ for all $n$. 
  Pick an arbitrary $\tau \in \State[K^c]$. Let $f_\infty \defeq \bigsqcup_{n \in \N} f_n$, and define partial functions $g_\infty$ and $g_n$ for all $n$ as follows:
  \begin{align*}
  g_\infty & : \State[K] \rightharpoonup \State[L],
  &
  g_\infty(\sigma) & \defeq 
  \begin{cases}
  (\pi_{\Var,L} \circ f_\infty)(\sigma \oplus \tau) & \text{if}\ f_\infty(\sigma \oplus \tau) \in \State
  \\
  \text{undefined} & \text{otherwise},
  \end{cases}
  \\
  g_n & : \State[K] \rightharpoonup \State[L],
  &
  g_n(\sigma) & \defeq 
  \begin{cases}
  (\pi_{\Var,L} \circ f_n)(\sigma \oplus \tau) & \text{if}\ f_n(\sigma \oplus \tau) \in \State
  \\
  \text{undefined} & \text{otherwise}.
  \end{cases}
  \end{align*}
  Then, $\{g_n\}_{n \in \N}$ is an $\omega$-chain when we order $g_n$'s by graph inclusion, and $g_\infty$ is the least upper bound of this
  chain for the same order. We will show that $g_\infty = g_n$ for some $n \in \N$ by case analysis on 
  $\Sigma_n \defeq \{ \sigma \in \State[K] \mid f_n(\sigma \oplus \tau) \in \State \}$. Note that
  this implies the desired $g_\infty \in \phi_{K,L}$ because $g_n \in \phi_{K,L}$ for all $n \in \N$.
  If $\Sigma_n = \emptyset$ for all $n \in \N$, then $g = g_n$ for any $n$, since both $g$ and $g_n$ are the same empty partial function.
  Otherwise, by the assumption of the lemma, $\Sigma_n = \State[K]$ for some $n$. This means that $g_n$ is the total function, and so $g_{n} = g_{m}$ for all $m \geq n$, which implies that $g = g_n$, as desired.
\end{proof}
}  
\AtEndDocument{

\section{Deferred Results in \S\ref{:4:inst}}

\subsection{Proof of \cref{thm:families-satisfy-assumptions}}

\begin{proof}[Proof of \cref{thm:families-satisfy-assumptions}]
We go through the assumptions, and show that they are satisfied by $\phi^{(d)}$ and $\phi^{(l)}$.

\paragraph{\bf Case of \cref{assm:projection}}        
        Let $K,L \subseteq \Var$ such that $L \subseteq K$. The projection $\pi_{K,L}$ is total and has an open set as
        its domain. Furthermore, the projection $\pi_{K,L}$ is differentiable and $1$-Lipschitz continuous. Since Lipschitz
        continuity implies local Lipschitz continuity, we have both $\pi_{K,L} \in \phi_{K,L}^{(d)}$
        and $\pi_{K,L} \in \phi_{K,L}^{(l)}$, as desired.

\paragraph{\bf Case of \cref{assm:pairing}}
        Let $K,L_0,L_1 \subseteq \Var$ such that $L_0 \cap L_1 = \emptyset$. Consider
        $f_0,g_0 \in [\State[K] \rightharpoonup \State[L_0]]$ and
        $f_1,g_1 \in [\State[K] \rightharpoonup \State[L_1]]$ such that all of the following hold:
        \[
        f_0 \in \phi^{(d)}_{K,L_0},
        \qquad
        f_1 \in \phi^{(d)}_{K,L_1},
        \qquad
        g_0 \in \phi^{(l)}_{K,L_0},
        \qquad\text{and}\qquad
        g_1 \in \phi^{(l)}_{K,L_1}.
        \]
        Let 
        \begin{align*}
        L & \defeq L_0 \cup L_1;
        \\
        f & : \State[K] \rightharpoonup \State[L],
        &
        f(\sigma) & \defeq 
        \begin{cases}
        f_0(\sigma) \oplus f_1(\sigma) & \text{if } \sigma \in \dom(f_0) \cap \dom(f_1),
        \\
        \text{undefined} & \text{otherwise};
        \end{cases}
        \\
        g & : \State[K] \rightharpoonup \State[L],
        &
        g(\sigma) & \defeq 
        \begin{cases}
        g_0(\sigma) \oplus g_1(\sigma) & \text{if } \sigma \in \dom(g_0) \cap \dom(g_1),
        \\
        \text{undefined} & \text{otherwise}.
        \end{cases}  
        \end{align*}
        We should show that $f$ and $g$ satisfy $\phi^{(d)}_{K,L}$ and $\phi^{(l)}_{K,L}$, respectively. In both cases, $\dom(f)$ and $\dom(g)$ are the intersections of two open sets, so that they are open as required. 
        
        To prove the differentiability of $f$, consider $\sigma \in \dom(f)$. Let $h_0$ and $h_1$ be the linear functions in $[\State[K] \rightharpoonup \State[L_0]]$ and $[\State[K] \rightharpoonup \State[L_1]]$, respectively, such that
        their domains are open and contain $0$, and for all $i \in \{0,1\}$,
        \[
                \lim_{\sigma' \to 0} \frac{\norm{f_i(\sigma + \sigma') - f_i(\sigma) - h_i(\sigma')}}{\norm{\sigma'}} = 0.
        \]
        Let $h$ be the linear function in $[\State[K] \rightharpoonup \State[L]]$ defined by
        \[
        		h(\sigma) \defeq 
		\begin{cases}
		 h_0(\sigma) \oplus h_1(\sigma) & \text{if } \sigma \in \dom(h_0) \cap \dom(h_1);
		 \\
		 \text{undefined} & \text{otherwise}.
		 \end{cases}
        \]
        Then, $\dom(h)$ is open and contains $\sigma$. Furthermore, 
	\begin{align*}
	                & 
	                \lim_{\sigma' \to 0} \frac{\norm{f(\sigma + \sigma') - f(\sigma) - h(\sigma')}}{\norm{\sigma'}}
	                \\ 
	                &
	                \qquad {} =
	                \lim_{\sigma' \to 0} \frac{
	                \sqrt{\norm{f_0(\sigma + \sigma') - f_0(\sigma) - h_0(\sigma')}^2 + 
	                \norm{f_1(\sigma + \sigma') - f_1(\sigma) - h_1(\sigma')}^2}
	                }{\norm{\sigma'}} 
	                	\\ 
	                &
	                \qquad {} =
	                \lim_{\sigma' \to 0} 
	                \sqrt{
	                \left(\frac{\norm{f_0(\sigma + \sigma') - f_0(\sigma) - h_0(\sigma')}}{\norm{\sigma'}}\right)^2 +
	                \left(\frac{\norm{f_1(\sigma + \sigma') - f_1(\sigma) - h_1(\sigma')}}{\norm{\sigma'}}\right)^2}
	                \\
	                &
	                \qquad {} =
	                \sqrt{
	                \left(\lim_{\sigma' \to 0} \frac{\norm{f_0(\sigma + \sigma') - f_0(\sigma) - h_0(\sigma')}}{\norm{\sigma'}}\right)^2 +
	                \left(\lim_{\sigma' \to 0} \frac{\norm{f_1(\sigma + \sigma') - f_1(\sigma) - h_1(\sigma')}}{\norm{\sigma'}}\right)^2} 
	                \\
	                &
	                \qquad {} = 0.
	\end{align*}
	Thus, $f$ is differentiable at $\sigma$, as desired.
	
	It remains to prove the local Lipschitzness of $g$. Pick $\sigma \in \dom(g)$. Then, $g_0$ and $g_1$ are defined at $\sigma$
	and they are locally Lipschitz. Thus, there exist open sets $O_0 \subseteq \dom(g_0)$ and $O_1 \subseteq \dom(g_1)$
	and constants $B_0, B_1 > 0$ such that $\sigma$ belongs to both $O_0$ and $O_1$, and for all $\sigma_0,\sigma_0' \in O_0$ and $\sigma_1,\sigma_1' \in O_1$,
	\[
	\norm{g_0(\sigma_0) - g_0(\sigma_0')} \leq B_0 \norm{\sigma_0 - \sigma_0'}
	\quad\text{and}\quad
	\norm{g_1(\sigma_1) - g_1(\sigma_1')} \leq B_1 \norm{\sigma_1 - \sigma_1'}.
	\] 
	Let $O \defeq O_0 \cap O_1$. The set $O$ is open, and contains $\sigma$. Furthermore, for all $\sigma',\sigma'' \in O$, 
	\begin{align*}
	\norm{g(\sigma') - g(\sigma'')}
	& {} = \sqrt{\norm{g_0(\sigma') - g_0(\sigma'')}^2 + \norm{g_1(\sigma') - g_1(\sigma'')}^2}
	\\
	& {} \leq  \sqrt{B_0^2\norm{\sigma' - \sigma''}^2 + B_1^2 \norm{\sigma' -\sigma''}^2}
	\\
	& {} = \sqrt{B_0^2 + B_1^2} \cdot \norm{\sigma' - \sigma''}.
	\end{align*}
	Thus, $g$ is Lipschitz in $O$, as desired.

\paragraph{\bf Case of \cref{assm:restriction}}
        Consider $K,K',L \subseteq \Var$ with $K \subseteq K'$, 
        and $\tau \in \State[K'\setminus K]$. Let $f$ and $g$ be partial functions in $[\State[K'] \rightharpoonup \State[L]]$ such
        that $f \in \phi^{(d)}_{K',L}$ and $g \in \phi^{(l)}_{K',L}$. Let
        \begin{align*}
                f_1 & : \State[K] \rightharpoonup \State[L],
                &
                f_1(\sigma) & \defeq 
                \begin{cases}
                f(\sigma \oplus \tau) & \text{if } \sigma \oplus \tau \in \dom(f)
                \\
                \text{undefined} & \text{otherwise},
                \end{cases}
                \\
                g_1 & : \State[K] \rightharpoonup \State[L],
                &
                g_1(\sigma) & \defeq 
                \begin{cases} 
                g(\sigma \oplus \tau) & \text{if } \sigma \oplus \tau \in \dom(g)
                \\
                \text{undefined} & \text{otherwise}.
                \end{cases}
        \end{align*}
        We should show that $f_1$ and $g_1$ satisfy $\phi^{(d)}_{K,L}$ and $\phi^{(l)}_{K,L}$, respectively. Note that
        \[
                \dom(f_1) = \{\sigma \,\mid\,\sigma \oplus \tau \in \dom(f)\}
                \quad\text{and}\quad
                \dom(g_1) = \{\sigma \,\mid\,\sigma \oplus \tau \in \dom(g)\}.
        \]
        These two sets are open since $\dom(f)$ and $\dom(g)$ are open and for any open $O$, 
        the slice $\{\sigma \in \State[K] \mid \sigma \oplus \tau \in O\}$ is open. 
        Let $\sigma_0 \in \dom(f_1)$ and $\sigma_1 \in \dom(g_1)$. We will show
        that $f_1$ is differentiable at $\sigma_0$, and $g_1$ is Lipschitz
        in an open neighbourhood of $\sigma_1$.

        Since $f$ is differentiable and $\sigma_0 \oplus \tau \in \dom(f)$, 
        there exists a linear map $h : \State[K'] \rightharpoonup \State[L]$ such that $\dom(h)$ is open and contains $0$, and
        \[
                \lim_{\sigma' \to 0} \frac{\norm{f(\sigma_0 \oplus \tau + \sigma') - f(\sigma_0 \oplus \tau) - h(\sigma')}}{\norm{\sigma'}} = 0.
        \]
        Let $\tau_0 \defeq \lambda v \in K' \setminus K.\,0$, and $h_1$ be the partial function from $\State[K]$ to $\State[L]$ defined by
        \begin{align*}
                h_1(\sigma) \defeq 
                \begin{cases}
                h(\sigma \oplus \tau_0) & \text{if } \sigma \oplus \tau_0 \in \dom(h),
                \\
                \text{undefined} & \text{otherwise}.
                \end{cases}
        \end{align*}
        Then, $h_1$ is linear, its domain is open (since taking a slice of an open set in $\State[K] \cong \R^{|K'|}$ by fixing some coordinate variables gives an open set), and 
        \begin{align*}
                \lim_{\sigma'' \to 0} \frac{\norm{f_1(\sigma_0 + \sigma'') - f_1(\sigma_0) - h_1(\sigma'')}}{\norm{\sigma''}} 
                & {} =
                \lim_{\sigma'' \to 0} 
                \frac{\norm{f(\sigma_0 \oplus \tau + \sigma'' \oplus \tau_0) - 
                            f(\sigma_0 \oplus \tau) - h(\sigma'' \oplus \tau_0)}}{\norm{\sigma'' \oplus \tau_0}}
                \\
                & {} = 0.
        \end{align*}
        This means that $f_1$ is differentiable at $\sigma_0$.

        Since $g$ is locally Lipschitz and $\sigma_1 \oplus \tau \in \dom(g)$, 
        there exists an open subset $O$ of $\dom(g)$ such that $O$ contains $\sigma_1 \oplus \tau$
        and $g$ is Lipschitz in $O$, that is, there exists a real number $B > 0$ such that
        \[
                \norm{g(\sigma) - g(\sigma')} \leq B \cdot \norm{\sigma - \sigma'}
        \]
        for all $\sigma,\sigma' \in O$. Let
        \[
                O' \defeq \{\sigma \in \State[K] \,\mid\, \sigma \oplus \tau \in O\}.
        \]
        Then, $O'$ is open, and it contains $\sigma_1$. Furthermore, for all $\sigma,\sigma' \in O'$,
        \begin{align*}
                \norm{g_1(\sigma) - g_1(\sigma')} 
                = \norm{g(\sigma \oplus \tau) - g(\sigma' \oplus \tau)} 
                \leq B \cdot \norm{\sigma \oplus \tau - \sigma' \oplus \tau}
                = B \cdot \norm{\sigma - \sigma'}.
        \end{align*}
        Thus, $g_1$ is Lipschitz in $O'$, as desired.

\paragraph{\bf Case of \cref{assm:composition}}
        For the composition condition, we handle the differentiability case only. The other case can be proved similarly.
        Consider 
        \[
                K,L,M \subseteq \Var,\quad
                f \in [\State[K] \rightharpoonup \State[L]],\quad\text{and}\quad
                g \in [\State[L] \rightharpoonup \State[M]].
        \]
        Assume that $f \in \phi^{(d)}_{K,L}$ and $g \in \phi^{(d)}_{L,M}$. 
        Let $h$ be the standard composition of partial functions $g$ and $f$.
        We should show that $h \in \phi^{(d)}_{K,M}$ as well, that is, $\dom(h)$ is open and $h$
        is differentiable on its domain. Note that
        \[
                \dom(h) = \dom(f) \cap f^{-1}(\dom(g)).
        \]
        Since $g \in \phi^{(d)}_{L,M}$, the set $\dom(g)$ is open. Because  
        $f \in \phi^{(d)}_{K,L}$, $\dom(f)$ is open and $f$ is continuous on its domain. The latter implies
        that $f^{-1}(\dom(g))$ is open as well.  Thus, the intersection of $\dom(f)$ and 
        $f^{-1}(\dom(g))$ is open as desired. The differentiability of $h$ on its domain
        holds since the restriction of $f$ to $\dom(h)$ gives a differentiable total function from $\dom(h)$ to
        $\dom(g)$, that of $g$ to $\dom(g)$ is also a differentiable total function, and
        the composition of two differentiable functions is differentiable.
        
\paragraph{\bf Case of \cref{assm:strictness}}
        The empty set is open, and the empty function is jointly differentiable and locally Lipschitz continuous. 
        Thus, the strictness assumption holds for both predicate families.
\end{proof}
}  
\AtEndDocument{

\section{Deferred Results in \S\ref{sec:var-sel-alg}}

\subsection{Proof of \cref{thm:soundness-reparam}}

\begin{proof}[Proof of \cref{thm:soundness-reparam}]
  Suppose that the algorithm returns $\pi$ (without an error message).
  We should show the conclusion that $\pi$ is simple and satisfies \cref{cond:dens-diff,cond:val-diff}.

  We make several observations before proving the conclusion.
  Let 
  \begin{align*}
  ({\mathbb{p}_m}, {\mathbb{d}_m}, {\mathbb{V}_m}) & \defeq \cdb{c_m},
  &
  ({\mathbb{p}_g}, {\mathbb{d}_g}, {\mathbb{V}_g}) & \defeq \cdb{c_g},
  &
  (\overline{\mathbb{p}_g}, \overline{\mathbb{d}_g}, \overline{\mathbb{V}_g}) & \defeq \cdb{\ctr{c_g}{\pi}},
  \end{align*}
  and $K \subseteq \Var$ be the set defined in \cref{eq:var-sel-alg-check1}. Also, let $S_r$ be the set of names that the algorithm uses to construct the returned reparameterisation plan $\pi$.
  Then, by the algorithm, we have the two inclusions in \cref{eq:var-sel-alg-check1} and \cref{eq:var-sel-alg-check2},
  and also $\pi = \pi_0[S_r]$. In addition, $S_r \subseteq K$ since
  \[
  S_r \subseteq \{(\alpha, i) \in \Name \mid \text{for all $i' \in \N$, $(\alpha,i') \in \Name \implies (\alpha,i') \in K$}\} \subseteq K.
  \]

  We now prove the conclusion in three parts.
  
  \paragraph{\bf First part:} We show that $\pi$ is simple. To show this, consider $(n,d,l), (n',d',l') \in \NameEx \times \DistEx \times \LamEx$
  such that $n = \cname(\alpha,e)$ and $n'=\cname(\alpha,e')$ for some $\alpha \in \Str$, $e$, and $e'$.
  Suppose that $(n,d,l) \in \dom(\pi)$.
  We should show $(n',d',l') \in \dom(\pi)$.
  Since $(n,d,l) \in \dom(\pi) = \dom(\pi_0[S_r])$, we have $(n,d,l) \in \dom(\pi_0)$ and $(\alpha, \_) \in S_r$.
  This implies that $(n',d',l') \in \dom(\pi_0)$ because $\pi_0$ is simple.
  Since $n'=\cname(\alpha,\_)$ and $(\alpha, \_) \in S_r$, we have $(n',d',l') \in \dom(\pi)$ as desired.

  \paragraph{\bf Second part:} We show that $\pi$ satisfies \cref{cond:dens-diff} in three steps.

  {\vspace{1mm} \it First step:} We prove $\theta \cup \repname(\pi) \subseteq K$.
  To do so, observe that the $S$ in the algorithm always satisfies the following property:
  \begin{align}
    \label{eq:var-sel-alg-I-invariant}
    (\alpha,i) \in S \implies (\alpha,i') \in S \quad\text{for any $(\alpha,i), (\alpha,i') \in \Name$}.
  \end{align}
  We can prove this by induction:
  the initial $S$
  (i.e., $S = \{(\alpha, i) \in \Name \mid {}$for all $i' \in \N$, $(\alpha,i') \in \Name \implies$ $(\alpha,i') \in K \}$) satisfies the property,
  and each update of $S$ (i.e., $S \leftarrow S \setminus \{(\alpha,i) \in \Name\}$ for some $(\alpha,\_) \in S$) preserves the property.
  From this, we obtain $\repname(\pi) = \repname(\pi_0) \cap S$:
  \begin{align*}
    \repname(\pi)
    &= \repname(\pi_0[S])
    \\
    &= \{(\alpha,i) \in \Name \mid \exists e,d,l.\, (\cname(\alpha,e),d,l) \in \dom(\pi_0[S])\}
    \\
    &= \big\{(\alpha,i) \in \Name \mid \exists e,d,l.\, \big((\cname(\alpha,e),d,l) \in \dom(\pi_0) \land \exists i'.\, (\alpha, i') \in S\big) \big\}
    \\
    &= \big\{(\alpha,i) \in \Name \mid \big(\exists e,d,l.\, \big((\cname(\alpha,e),d,l) \in \dom(\pi_0)\big) \land \big(\exists i'.\, (\alpha, i') \in S\big) \big\}
    \\
    &= \{(\alpha,i) \in \Name \mid (\cname(\alpha,\_),\_,\_) \in \dom(\pi_0)\} \cap \{(\alpha,i) \in \Name \mid (\alpha, \_) \in S \}
    \\
    &= \repname(\pi_0) \cap S,
  \end{align*}
  where the second and third equalities use the definitions of $\repname(-)$ and $\pi_0[S]$, respectively,
  and the last equality uses \cref{eq:var-sel-alg-I-invariant}.
  Since $\repname(\pi) \subseteq S \subseteq K$ and $\theta \subseteq K$ by \cref{eq:var-sel-alg-check1} (the inclusion of $\theta$),
  we get $\theta \cup \repname(\pi) \subseteq K$ as desired.

  {\vspace{1mm} \it Second step:} We prove that for all $u \in \{\like\} \cup \{\pr_\mu \mid \mu \in \Name\}$ and $v \in \{\pr_\mu \mid \mu \in \Name\}$,
  the following functions (which are total since $c_m$ and $c_g$ always terminate) are differentiable with respect to the variables in $\theta \cup \repname(\pi)$ jointly:
  \begin{align}
    \label{eq:var-sel-alg-dens}
    \begin{array}{r@{}l}
      (\sigma_\theta, \sigma_n) &{} \in \State[\theta] \times \State[\Name] \longmapsto \db{c_m}(\sigma_{p\setminus \theta} \oplus \sigma_\theta \oplus \sigma_n \oplus g(\sigma_n))(u),
      \\[0.3em]
      (\sigma_\theta, \sigma_n) &{} \in \State[\theta] \times \State[\Name] \longmapsto \db{c_g}(\sigma_{p\setminus \theta} \oplus \sigma_\theta \oplus \sigma_n \oplus g(\sigma_n))(v),
    \end{array}
  \end{align}
  where $\sigma_{p\setminus \theta} \defeq (\lambda v \in \PVar \setminus \theta.\, 0)$ and the function $g : \State[\Name] \to \State[\AVar]$ takes $\sigma_n$ and returns $\sigma_a$ such that $\sigma_a$ maps $\like$ to $1$, $\pr_\mu$ to $\cN(\sigma_n(\mu); 0,1)$, $\val_\mu$ to $\sigma_n(\mu)$, and all the other variables to $0$. The state $\sigma_{p\setminus \theta} \oplus g(\sigma_n)$ is the very initialisation used in \cref{eq:density-of-c}.
  This step consists of two substeps.

  {First substep:}
  We first show that
  for all $u \in \{\like\} \cup \{\pr_\mu \mid \mu \in \Name\}$ and $v \in \{\pr_\mu \mid \mu \in \Name\}$,
  the functions $f_m, f_g : \State[\theta] \times \State[\Name] \times \State[\AVar] \to \R$
  are differentiable with respect to the variables in $\theta \cup \repname(\pi)$ jointly:
  \begin{align}
    \label{eq:var-sel-alg-dens-fixed}
    \begin{array}{r@{}l}
      f_m(\sigma_\theta, \sigma_n, \sigma_a) \defeq \db{c_m}(\sigma_{p\setminus \theta} \oplus \sigma_\theta \oplus \sigma_n \oplus \sigma_{a})(u),
      \\[0.3em]
      f_g(\sigma_\theta, \sigma_n, \sigma_a) \defeq \db{c_g}(\sigma_{p\setminus \theta} \oplus \sigma_\theta \oplus \sigma_n \oplus \sigma_{a})(v).
    \end{array}
  \end{align}
  Note that in $f_m$ and $f_g$, the $\AVar$ part does not depend on the $\Name$ part (unlike in \cref{eq:var-sel-alg-dens}).
  For the proof, pick arbitrary $u \in \{\like\} \cup \{\pr_\mu \mid \mu \in \Name\}$ and $v \in \{\pr_\mu \mid \mu \in \Name\}$.
  Then, $\theta \cup \repname(\pi) \subseteq K \subseteq \mathbb{p}_m(u) \cap \mathbb{p}_g(v)$,
  where the first inclusion is from the above result and the second from \cref{eq:var-sel-alg-check1} (the definition of $K$).
  By the soundness of differentiability analysis (\cref{thm:soundness-analysis}), 
  ${}\models \Phi(\db{c_m}, \mathbb{p}_m(u), \{u\})$ and ${}\models \Phi(\db{c_g}, \mathbb{p}_g(v), \{v\})$.
  From this, and by the weakening lemma of $\Phi$ with $\theta \cup \repname(\pi) \subseteq \mathbb{p}_m(u) \cap \mathbb{p}_g(v)$
  (\cref{lem:weakening-phi}), we have
  ${}\models \Phi(\db{c_m}, \theta \cup \repname(\pi), \{u\})$ and
  ${}\models \Phi(\db{c_g}, \theta \cup \repname(\pi), \{v\})$.
  Hence, the functions in \cref{eq:var-sel-alg-dens-fixed} are differentiable with respect to $\theta \cup \repname(\pi)$ jointly as desired,
  by the definition of $\Phi$ (\cref{:1:abs}) and
  the definition of ``$f : \State[L] \to \R$ for $L \subseteq \Var$ is differentiable with respect to $L'\subseteq L$ jointly''
  (\cref{sec:grad-estm-prog-trans}).

  {Second substep:}
  We now prove that the claim of the second step follows from the first substep just proved.
  Pick any $u \in \{\like\} \cup \{\pr_\mu \mid \mu \in \Name\}$.
  We should show that the first function in \cref{eq:var-sel-alg-dens} is differentiable with respect to $\theta \cup \repname(\pi)$.
  Note that we should also show the same for the second function (for any $v$), but the proof is similar to the first function so we omit this case.
  To prove the claim for the first function,
  pick any $\xi_{n,0}' \in \State[\Name \setminus \repname(\pi)]$ and $\sigma_{a,0} \in \State[\AVar]$.
  Define $f' : \State[\theta] \times \State[\repname(\pi)] \times \State[\AVar]$ as
  \begin{align*}
    f'(\sigma_\theta, \xi_n, \sigma_a) \defeq f(\sigma_\theta, \xi_n \oplus \xi_{n,0}', \sigma_a).
  \end{align*}
  Then, by \cref{lem:used-properties}-\cref{lem:used-properties-2} and \cref{lem:used-properties}-\cref{lem:used-properties-3},
  \begin{align*}
    f'(\sigma_\theta, \xi_n, \sigma_a) &\defeq
    \begin{cases}
      f(\sigma_\theta, \xi_n \oplus \xi_{n,0}', \sigma_{a,0})
      & \text{if $(\sigma_\theta, \xi_n) \in U$}
      \\
      \smath{proj}(\sigma_a)
      & \text{if $(\sigma_\theta, \xi_n) \notin U$}
    \end{cases}
  \end{align*}
  for some $U \subseteq \State[\theta] \times \State[\repname(\pi)]$ and some projection map $\smath{proj} : \State[\AVar] \to \R$.
  Also, since $f$ is differentiable with respect to $\theta \cup \repname(\pi)$,
  $f'(-,-,\sigma_a) : \State[\theta] \times \State[\repname(\pi)]$ is differentiable and thus continuous for all $\sigma_a \in \State[\AVar]$.
  From these, \cref{lem:split-func-cont} is applicable to $f'$,
  implying that $U$ should be either $\emptyset$ or $\State[\theta] \times \State[\repname(\pi)]$.
  We now consider $f'' : \State[\theta] \times \State[\repname(\pi)] \to \R$ defined by
  \begin{align*}
    f''(\sigma_\theta, \xi_n) \defeq f'(\sigma_\theta, \xi_n, g(\xi_n) \oplus g(\xi_{n,0}')),
  \end{align*}
  where $g$ is extended to accept a substate in $\Statesub[\Name]$ and return a substate for the corresponding auxiliary part.
  Then, to prove the claim, it suffices to show that $f''$ is differentiable (since $\xi_{n,0}'$ was chosen arbitrarily).
  We do case analysis on $U$.
  If $U = \State[\theta] \times \State[\repname(\pi)]$, then
  \[f''(\sigma_\theta,  \xi_n) = f(\sigma_\theta, \xi_n \oplus \xi_{n,0}', \sigma_{a,0}) \quad\text{for all $\sigma_\theta$ and $\xi_n$};\]
  since $f$ is differentiable with respect to $\theta \cup \repname(\pi)$, $f''$ is differentiable.
  If $U = \emptyset$, then
  \[f''(\sigma_\theta,  \xi_n) = \smath{proj}(g(\xi_n) \oplus g(\xi_{n,0}')) \quad\text{for all $\sigma_\theta$ and $\xi_n$};\]
  since $g$, $\oplus$, and $\smath{proj}$ are all differentiable (because $g$ only uses projection and the density of the standard normal distribution),
  $f''$ is differentiable.
  Hence, $f''$ is differentiable in both cases, and this shows the claim of the second step.
  
  {\vspace{1mm} \it Third step:} We prove that $\pi$ satisfies \cref{cond:dens-diff},
  i.e., the functions in \cref{cond:dens-diff} are differentiable with respect to $\theta \cup \repname(\pi)$ jointly.
  This holds because:
  $c_m$ and $c_g$ do not have a double-sampling error, so
  each function in \cref{cond:dens-diff} is a multiplication of some of the functions in \cref{eq:var-sel-alg-dens} (for different $u$ and $v$);
  the functions in \cref{eq:var-sel-alg-dens} are differentiable with respect to  $\theta \cup \repname(\pi)$ jointly (by the above result);
  and multiplication preserves differentiability.
  
  \paragraph{\bf Third part:} We show that $\pi$ satisfies \cref{cond:val-diff},
  i.e., the functions in \cref{cond:val-diff} are differentiable with respect to $\theta$ jointly.
  For this, it suffices to show the claim that for all
  $v \in { \{\pr_\mu, \val_\mu \mid \mu \in \Name\} }$
  and $\sigma_n \in \State[\Name]$,
  the following function is differentiable with respect to $\theta$ jointly:
  \begin{align}
    \label{eq:var-sel-alg-val}
    \sigma_\theta &{} \in \State[\theta] \longmapsto
    \db{\ctr{c_g}{\pi}}(\sigma_{p \setminus \theta} \oplus \sigma_\theta \oplus \sigma_n \oplus g(\sigma_n))(v).
  \end{align}
  This implies \cref{cond:val-diff} because:
  $\ctr{c_g}{\pi}$ does not have a double-sampling error, so
  each function in \cref{cond:dens-diff} is either a multiplication or a pairing of the function in \cref{eq:var-sel-alg-val} (for different $v$);
  and multiplication and pairing preserve differentiability.
  To show the claim, consider any
  $v \in { \{\pr_\mu, \val_\mu \mid \mu \in \Name\} }$.
  Then, $\theta \subseteq \overline{\mathbb{p}_g}(v)$ by \cref{eq:var-sel-alg-check2}.
  By the soundness of differentiability analysis (\cref{thm:soundness-analysis}), 
  ${}\models \Phi(\db{\ctr{c_g}{\pi}}, \overline{\mathbb{p}_g}(v), \{v\})$.
  From this, and by the weakening lemma of $\Phi$ with  $\theta \subseteq \overline{\mathbb{p}_g}(v)$ (\cref{lem:weakening-phi}), we have
  ${}\models \Phi(\db{\ctr{c_g}{\pi}}, \theta, \{v\})$.
  Hence, the function in \cref{eq:var-sel-alg-val} is differentiable with respect to $\theta$ jointly.
  This completes the overall proof.
\end{proof}

\begin{lemma}
  \label{lem:split-func-cont}
  Let $f : \R^n \times \R^m \to \R$ be a function such that
  \begin{align*}
    f(x,y) &=
    \begin{cases}
      f_1(x) & \text{if $x \in U$}
      \\
      f_2(y) & \text{if $x \notin U$}
    \end{cases}
  \end{align*}
  for some $f_1 : \R^n \to \R$, $f_2 : \R^m \to \R$, and $U \subseteq R^n$.
  Suppose that $f_2(\R^m) = \R$ and $f(-,y) : \R^n \to \R$ is continuous for all $y \in \R^m$.
  Then, $U$ is either $\emptyset$ or $\R^n$.
\end{lemma}
\begin{proof}
  Here is a sketch of the proof. We prove the lemma by contradiction.
  Suppose that $U$ is neither $\emptyset$ nor $\R^n$.
  Then, the boundary of $U$ (i.e., $\mathrm{bd}(U) \subseteq \R^n$) is nonempty,
  since the boundary of a set is empty if and only if the set is both open and closed,
  and since $\emptyset$ and $\R^n$ are the only subsets of $\R^n$ that are both open and closed.
  Let $x \in \mathrm{bd}(U)$ and consider two cases: $x \in U$ or $x \notin U$.
  In each of the two cases, we can show that there exists $y \in \R^m$ such that $f(-,y)$ is not continuous at $x$.
  When showing the discontinuity, we use the following:
  $x \in \mathrm{bd}(U)$; the specific way that $f$ is defined (in terms of $U$, $f_1$, and $f_2$); and
  the assumption that $f_2(\R^m) = \R$.
  By assumption, $f(-,y)$ should be continuous over the entire $\R^n$, so we get contradiction.
\end{proof}
}  
\AtEndDocument{

\section{Deferred Results in \S\ref{sec:impl-eval}}
\label{sec:impl-eval-more}

\subsection{Deferred Experiment Details and Results}

\begin{table*}[!ht]
  \centering
  \renewcommand{\arraystretch}{.90}
  \aboverulesep=0.2ex
  \belowrulesep=0.2ex
  \small
  \caption{%
    Pyro examples used in experiments and their key features (continued from \cref{t:examples}).
    The last five columns show the total number of code lines (excluding comments),
    loops, sample commands, observe commands, and learnable parameters
    (declared explicitly by \mytt{pyro.param} or implicitly by a neural network module).
    Each number is the sum of the counts in the model and guide.
  }
  \label{t:examples-more}
  \vspace{-0.5em}
  \begin{tabular}{llrrrrr}
    \toprule
    Name & Probabilistic model & LoC
    & $\cwhile$ & $\csample$ & $\cobs$ & \mytt{param}
    \\ \midrule
    \mytt{dpmm} & Dirichlet process mixture models
    & 27 
    & 0 & 6 & 1 & 4 
    \\
    \mytt{vae} & Variational autoencoder (VAE)
    & 35 
    & 0 
    &   2 & 1 & 5 
    \\
    \mytt{csis} & Inference compilation
    & 38 
    & 0 
    &   2 & 2 & 5 
    \\
    \mytt{br} & Bayesian regression
    & 42 
    & 0 
    &  10 & 1 & 5 
    \\
    \mytt{lda} & Amortised latent Dirichlet allocation
    & 57 
    & 0 
    &   8 & 1 & 5 
    \\
    \mytt{prodlda} & Probabilistic topic modelling
    & 58 
    & 0 & 2 & 1 & 5
    \\
    \mytt{ssvae} & Semi-supervised VAE
    & 60 
    & 0 
    & 4 & 1 & 7 
    \\ \bottomrule
  \end{tabular}
  \vspace{-0.5em}
\end{table*}

\begin{table*}[!ht]
  \centering
  \renewcommand{\arraystretch}{.90}
  \aboverulesep=0.2ex
  \belowrulesep=0.2ex
  \small
  \caption{%
    Results of smoothness analyses (continued from \cref{tab:analysis-result}).
    ``Manual'' and ``Ours'' denote the number of continuous random variables and learnable parameters
    in which the density of the program is smooth, computed by hand and by our analyser.
    ``Time'' denotes the runtime of our analyser in seconds.
    ``\#CRP'' denotes the total number of continuous random variables and learnable parameters in the program.
    \mytt{-m} and \mytt{-g} denote model and guide.
    We consider $\{(\alpha,i) \in \Name\}$ as one random variable for each $\alpha \in \Name$.
  }
  \label{tab:analysis-result-more}
  \vspace{-0.5em}
  \begin{tabular}{lrrrrrrr}
    \toprule
    & \multicolumn{3}{c}{Differentiable}
    & \multicolumn{3}{c}{Locally Lipschitz}
    \\ \cmidrule(lr){2-4} \cmidrule(lr){5-7}
    \multicolumn{1}{l}{Name}
    & \multicolumn{1}{r}{Manual}
    & \multicolumn{1}{r}{Ours}
    & \multicolumn{1}{l}{Time}
    & \multicolumn{1}{r}{Manual}
    & \multicolumn{1}{r}{Ours}
    & \multicolumn{1}{l}{Time}
    & \multicolumn{1}{r}{\#CRP}
    \\ \midrule
    \mytt{dpmm-m}    &  2 &  2 & 0.002 &  2 &  2 & 0.002 &  2 
    \\
    \mytt{dpmm-g}    &  6 &  6 & 0.003 &  6 &  6 & 0.003 &  6 
    \\
    \mytt{vae-m}     &  3 &  3 & 0.002 &  3 &  3 & 0.003 &  3 
    \\
    \mytt{vae-g}     &  4 &  4 & 0.002 &  4 &  4 & 0.002 &  4 
    \\
    \mytt{csis-m}    &  1 &  1 & 0.001 &  1 &  1 & 0.001 &  1 
    \\
    \mytt{csis-g}    &  \hltorg{2} &  \hltorg{2} & 0.004 &  6 &  6 & 0.004 &  6 
    \\
    \mytt{br-m}      &  5 &  5 & 0.002 &  5 &  5 & 0.002 &  5 
    \\
    \mytt{br-g}      & 10 & 10 & 0.004 & 10 & 10 & 0.004 & 10 
    \\
    \mytt{lda-m}     &  3 &  3 & 0.002 &  3 &  3 & 0.002 &  3 
    \\
    \mytt{lda-g}     &  7 &  7 & 0.007 &  7 &  7 & 0.007 &  7 
    \\
    \mytt{prodlda-m} &  2 &  2 & 0.008 &  2 &  2 & 0.007 &  2 
    \\
    \mytt{prodlda-g} &  5 &  5 & 0.007 &  5 &  5 & 0.006 &  5 
    \\
    \mytt{ssvae-m}   &  3 &  3 & 0.004 &  3 &  3 & 0.003 &  3 
    \\
    \mytt{ssvae-g}   &  6 &  6 & 0.007 &  6 &  6 & 0.009 &  6 
    \\ \bottomrule
  \end{tabular}
  \vspace{1.5em}
\end{table*}

\begin{table*}[!ht]
  \centering
  \renewcommand{\arraystretch}{.90}
  \aboverulesep=0.2ex
  \belowrulesep=0.2ex
  \small
  \caption{%
    Results of variable selections (continued from \cref{tab:reparam-result}).
    ``Ours-Time'' denote the runtime of our variable selector in seconds.
    ``Ours-Sound'' and ``Pyro $\setminus$ Ours'' denote the number of random variables in the example
    that are in $\pi_\ours$, and that are in $\pi_0$ but not in $\pi_\ours$, respectively,
    where $\pi_\ours$ and $\pi_0$ denote the reparameterisation plans given by our variable selector and by Pyro.
    ``Pyro $\setminus$ Ours'' is partitioned into ``Sound'' and ``Unsound'':
    the latter denotes the number of random variables that make \cref{cond:dens-lips} or \cref{cond:val-lips} violated when added to $\pi_\ours$,
    and the former denotes the number of the rest.
    ``\#CR'' and ``\#DR'' denote the total number of continuous and discrete random variables in the example.
    We consider $\{(\alpha,i) \in \Name\}$ as one random variable for each $\alpha \in \Name$.
  }
  \label{tab:reparam-result-more}
  \vspace{-0.5em}
  \begin{tabular}{lrrrrr|r}
    \toprule
    & \multicolumn{2}{c}{Ours} & \multicolumn{2}{c}{Pyro $\setminus$ Ours}
    & \multicolumn{1}{c|}{}
    \\ \cmidrule(lr){2-3} \cmidrule(lr){4-5}
    Name
    & Time
    & Sound
    & Sound
    & Unsound
    & \#CR
    & \#DR
    \\ \midrule
    \mytt{dpmm}    & 0.007 & 2 & 0 & 0 & 2 & 1 
    \\
    \mytt{vae}     & 0.004 & 1 & 0 & 0 & 1 & 0 
    \\
    \mytt{csis}    & 0.014 & 1 & 0 & 0 & 1 & 0 
    \\
    \mytt{br}      & 0.009 & 5 & 0 & 0 & 5 & 0 
    \\
    \mytt{lda}     & 0.011 & 3 & 0 & 0 & 3 & 1 
    \\
    \mytt{prodlda} & 0.018 & 1 & 0 & 0 & 1 & 0 
    \\
    \mytt{ssvae}   & 0.013 & 1 & 0 & 0 & 1 & 1 
    \\ \bottomrule
  \end{tabular}
  \vspace{-0.5em}
\end{table*}
}   


\end{document}